\documentclass[10pt,a4paper]{article}

\usepackage{ifthen}

\newcommand{\typeof}{1} %

\newcommand{\longv}[1]{\ifthenelse{\equal{\typeof}{0}}{}{#1}}
\newcommand{\shortv}[1]{\ifthenelse{\equal{\typeof}{0}}{#1}{}}
\newcommand{\longshortv}[2]{\ifthenelse{\equal{\typeof}{0}}{#2}{#1}}

\longv{\usepackage{a4wide}}

\usepackage{amsfonts}
\usepackage{amsmath}
\usepackage{amssymb} 
\usepackage{amsthm}
\usepackage{mathrsfs}
\usepackage{stmaryrd}
\usepackage{cmll} 
\usepackage{proof} 
\usepackage{array}
\usepackage[all]{xy}
\SelectTips{eu}{11}
\usepackage{tikz}
\usepackage{url}
\usepackage{verbatim} 
\newif\ifshowtikz
\ifshowtikz
\newenvironment{tk}[1]
{\begin{tikzpicture}#1}{\end{tikzpicture}}
\else

\fi

\tikzset{
  obj/.style={scale=0.66}
}

\newtheorem{definition}{Definition}[section]
\newtheorem{lemma}{Lemma}[section]
\newtheorem{proposition}{Proposition}[section]
\newtheorem{theorem}{Theorem}[section]
\newtheorem{corollary}{Corollary}[section]
\newtheorem{remark}{Remark}[section]

\newtheorem{convention}{Convention}[section]
\newcommand{\PCFSS}{\mathbf{PCFSS}}
\newcommand{\ANGLICAN}{\ensuremath{\textbf{ANGLICAN}}}
\newcommand{\CHURCH}{\ensuremath{\textbf{CHURCH}}}
\newcommand{\ttreal}{\mathtt{Real}}
\newcommand{\ttunit}{\mathtt{Unit}}
\newcommand{\sem}[1]{\llbracket #1 \rrbracket}
\newcommand{\psem}[1]{\llparenthesis #1 \rrparenthesis}
\newcommand{\bsem}[1]{\llbracket #1 \rrbracket_{\mathrm{d}}}
\newcommand{\Bsem}[1]{\llparenthesis #1 \rrparenthesis_{\mathrm{d}}}
\newcommand{\letin}[3]{\mathtt{let}\;#1\;\mathtt{be}\;#2\;\mathtt{in}\;#3}
\newcommand{\id}{\mathrm{id}}
\newcommand{\ifterm}[3]{\mathtt{ifz}(#1, #2, #3)} 

\newcommand{\arity}[1]{\vert #1\vert} 
\newcommand{\detred}{\stackrel{\mathrm{red}}{\longrightarrow}}
\newcommand{\nulldot}{\makebox[0pt]{\quad .}} 
\newcommand{\nullcomma}{\makebox[0pt]{\quad ,}} 
\newcommand{\state}[1]{S_{#1}} 
\newcommand{\tran}[1]{\tau_{#1}} 
\newcommand{\init}[1]{s_{#1}} 
\newcommand{\hh}{\bullet} 
\newcommand{\mm}{\circ} 
\newcommand{\midd}{\; \; \mbox{\Large{$\mid$}}\;\;} 
\newcommand{\cons}{\mathbin{::}}
\newenvironment{varitemize}
{
\begin{list}{\labelitemi}
{\setlength{\itemsep}{0pt}
 \setlength{\topsep}{0pt}
 \setlength{\parsep}{0pt}
 \setlength{\partopsep}{0pt}
 \setlength{\leftmargin}{15pt}
 \setlength{\rightmargin}{0pt}
 \setlength{\itemindent}{0pt}
 \setlength{\labelsep}{5pt}
 \setlength{\labelwidth}{10pt}
}}
{
 \end{list} 
}

\newcounter{numberone}

\begin{document}

\title{The Geometry of Bayesian Programming}

\shortv{ \author{\IEEEauthorblockN{Ugo Dal Lago}
    \IEEEauthorblockA{University of Bologna, Italy \& INRIA Sophia Antipolis, France \\
      Email: ugo.dallago@unibo.it} \and \IEEEauthorblockN{Naohiko
      Hoshino}
    \IEEEauthorblockA{Kyoto University, Japan \\
      Email: naophiko@kurims.kyoto-u.ac.jp} } } \longv{ \author{Ugo
    Dal Lago \and Naohiko Hoshino} }

\shortv{
\IEEEoverridecommandlockouts
\IEEEpubid{\makebox[\columnwidth]{978-1-7281-3608-0/19/\$31.00~
    \copyright2019 IEEE \hfill} \hspace{\columnsep}\makebox[\columnwidth]{ }}
}

\maketitle

\begin{abstract}
  We give a geometry of interaction model for a typed
  $\lambda$-calculus endowed with operators for sampling from a
  continuous uniform distribution and soft conditioning, namely a
  paradigmatic calculus for higher-order Bayesian programming. The
  model is based on the category of measurable spaces and partial
  measurable functions, and is proved adequate with respect to both a
  distribution-based and a sampling based operational semantics.
\end{abstract}

\shortv{\IEEEpeerreviewmaketitle}

\section{Introduction}
Randomisation provides the most efficient algorithmic solutions, at
least concretely, in many different contexts. A typical example is the
one of primality testing, where the Miller-Rabin
test~\cite{Miller1976,Rabin1980} remains the preferred choice despite
polynomial time deterministic algorithms are available from many years
now~\cite{AKS2002}. Probability theory can be exploited even more
fundamentally in programming, by way of so-called probabilistic (or,
more specifically, Bayesian) programming, as popularized by languages
like, among others, \ANGLICAN~\cite{WMM2014} or
\CHURCH~\cite{GMRBT08}. This has stimulated research about
probabilistic programming languages and their
semantics~\cite{jones1990,DH2002,EPT2018}, together with type
systems~\cite{DLG2017,BDL2018}, equivalence
methodologies~\cite{DLSA2014,CDL2014}, and verification
techniques~\cite{SABGGH2019}.

Giving a satisfactory denotational semantics to higher-order
functional languages is already problematic in presence of
probabilistic choice~\cite{jones1990,JT1998}, and becomes even more
challenging when continuous distributions and scoring are present.
Recently, quasi-Borel spaces~\cite{hksy2017} have been proposed as a
way to give semantics to calculi with all these features, and only
very recently~\cite{VKS2019} this framework has been shown to be
adaptable to a fully-fledged calculus for probabilistic programming,
in which continuous distributions \emph{and} soft-conditioning are
present. Probabilistic coherent spaces~\cite{DE2011} are fully
abstract~\cite{EPT2018} for $\lambda$-calculi with discrete
probabilistic choice, and can, with some effort, be adapted to calculi
with sampling from continuous distributions~\cite{EPT2018POPL},
although without scoring.

A research path which has been studied only marginally, so far,
consists in giving semantics to Bayesian higher-order programming
languages through \emph{interactive forms} of semantics, e.g. game
semantics~\cite{HO2000,AJM2000} or the geometry of
interaction~\cite{girard1989}. One of the very first models for
higher-order calculi with discrete probabilistic choice was in fact a
game model, proved fully abstract for a probabilistic calculus with
global ground references~\cite{DH2002}. After more than ten years, a
parallel form of Geometry of Interaction (GoI) and some game models
have been introduced for $\lambda$-calculi with probabilistic
choice~\cite{DLFVY2017,CCPW2018,CP2018}, but in all these cases only
discrete probabilistic choice can be handled, with the exception of a
recent work on concurrent games and continuous
distributions~\cite{PW2018}.

In this paper, we will report on some results about GoI models of
higher-order Bayesian languages. The distinguishing features of the
introduced GoI model can be summarised as follows:
\begin{varitemize}
\item \textbf{Simplicity}. The category on which the model is defined
  is the one of measurable spaces and \emph{partial} measurable
  functions, so it is completely standard from a measure-theoretic
  perspective.
\item \textbf{Expressivity}. As is well-known, the GoI
  construction~\cite{jsv,ahs2002} allows to give semantics to calculi
  featuring higher-order functions and recursion. Indeed, our GoI
  model can be proved adequate for $\PCFSS$, a fully-fledged calculus
  for probabilistic programming.
\item \textbf{Flexibility}. The model we present is quite flexible, in
  the sense of being able to reflect the operational behaviour of
  programs as captured by \emph{both} the distribution-based and the
  sampling-based semantics.
\item \textbf{Intuitiveness}. GoI visualises the structure of programs
  in terms of graphs, from which dependencies between subprograms can
  be analyzed. Adequacy of our model provides diagrammatic reasoning
  principle about observational equivalence of $\PCFSS$.
\end{varitemize}
This paper's contributions, beside the model's definition, are two
adequacy results which precisely relate our GoI model to the
operational semantics, as expressed (following \cite{bdlgs2016}), in
both the distribution and sampling styles. As a corollary of our
adequacy results, we show that the distribution induced by
sampling-based operational semantics coincides with distribution-based
operational semantics.

\subsection{Turning Measurable Spaces into a GoI
  Model}\label{sec:turning}
Before entering into the details of our model, it is worthwhile to
give some hints about how the proposed model is obtained, and why it
differs from similar GoI models from the literature.

The thread of work the proposed model stems from is the one of
so-called memoryful geometry of interaction~\cite{hmh2014,mhh2016}.
The underlying idea of this paper is precisely the same: program
execution is modelled as an interaction between the program and its
environment, and memoisation takes place inside the program as a
result of the interaction.

In the previous work on memoryful GoI by the second author with Hasuo
and Muroya, the goal consisted in modelling a $\lambda$-calculus with
\emph{algebraic} effects. Starting from a monad together with some
algebraic effects, they gave an adequate GoI model for such a
calculus, which is applicable to wide range of algebraic effects.  In
principle, then, their recipe could be applicable to $\PCFSS$,
sinc sampling-based operational semantics enables us to see scoring and
sampling as algebraic effects acting on global states. However, the
that would not work for $\PCFSS$, since the category
$\mathbf{Meas}$ of measurable spaces\footnote{We need to work on
  $\mathbf{Meas}$ because we want to give adequacy for
  distribution-based semantics.} is not cartesian closed, and we thus
\emph{cannot} define a state monad by way of the exponential $S
\Rightarrow S \times (-)$.

In this paper, we side step this issue by a series of translations, to
be described in Section~\ref{sec:towards} below. Instead of looking
for a state monad on $\mathbf{Meas}$, we embed $\mathbf{Meas}$ into
the category $\mathbf{Mealy}$ of $\mathbf{Int}$-objects and Mealy
machines (Section~\ref{sec:mealy}) and use a state monad \emph{on
  this} category. This is doable because $\mathbf{Mealy}$ is a
\emph{compact closed category} given by the
\emph{$\mathbf{Int}$-construction}~\cite{ahs2002}. The use of such
compact closed categories (or, more generally, of traced monoidal
categories) is the way GoI models capture higher-order functions.

\subsection{Outline}

The rest of the paper is organised as follows. After giving some
necessary measure-theoretic preliminaries in Section~\ref{sec:measure}
below, we introduce in Section~\ref{sec:syntax} the language $\PCFSS$,
together with the two kinds of operational semantics we were referring
to above. In Section~\ref{sec:towards}, we introduce our GoI model
informally, while in Section~\ref{sec:mealy} a more rigorous treatment
of the involved concepts is given, together with the adequacy results.
We discuss in Section~\ref{sec:sfinite} an alternative way of giving a
GoI semantics to $\PCFSS$ based on s-finite kernels, and we conclude
in Section~\ref{sec:conclusion}. \shortv{More details on our GoI model
  and details proofs of most results in this paper are available in
  this paper's extended version~\cite{dlh2019}.}

\section{Measure-Theoretic Preliminaries}
\label{sec:measure}
We recall some basic notions in measure theory that will be needed in
the following. We also fix some useful notations. For more about
measure theory, see standard text books such as
\cite{billingsley1986}.

A \emph{$\sigma$-algebra} on a set $X$ is a family $\Sigma$ consisting
of subsets of $X$ such that $\emptyset \in \Sigma$; and if
$A \in \Sigma$, then the complement $X \setminus A$ is in $\Sigma$;
and for any family $\{A_{n}\in \Sigma\}_{n \in \mathbb{N}}$, the
intersection $\bigcap_{n \in \mathbb{N}} A_{n}$ is in $\Sigma$.
A \emph{measurable space} $X$ is a set $|X|$ equipped with a
$\sigma$-algebra $\Sigma_{X}$ on $|X|$. We often confuse a measurable
space $X$ with its underlying set $|X|$. For example, we simply write
$x \in X$ instead of $x \in |X|$. For measurable spaces $X$ and $Y$,
we say that a partial function $f \colon X \to Y$ (in this paper, we
use $\to$ for both partial functions and total functions) is
\emph{measurable} when for all $A \in \Sigma_{Y}$, the inverse image
\begin{equation*}
  \{x \in X : f(x) \textnormal{ is defined and is equal to
    an element of $A$}\}
\end{equation*}
is in $\Sigma_{X}$. A \emph{measurable function} from $X$ to $Y$ is a
totally defined partial measurable function. A (partial) measurable
function $f \colon X \to Y$ is \emph{invertible} when there is a
measurable function $g \colon Y \to X$ such that $g \circ f$ and
$f \circ g$ are identities. In this case, we say that $f$ is an
isomorphism from $X$ to $Y$ and say that $X$ is isomorphic to $Y$.

We denote a singleton set $\{\ast\}$ by $1$, and we regard the latter as a
measurable space by endowing it with the trivial $\sigma$-algebra. We
also regard the empty set $\emptyset$ as a measurable space in the
obvious way. In this paper, $\mathbb{N}$ denotes the measurable set of
all non-negative integers equipped with the $\sigma$-algebra
consisting of all subsets of $\mathbb{N}$, and $\mathbb{R}$ denotes
the measurable set of all real numbers equipped with the
$\sigma$-algebra consisting of \emph{Borel sets}, that is, the least
$\sigma$-algebra that contains all open subsets of $\mathbb{R}$. By
the definition of $\Sigma_{\mathbb{R}}$, a function
$f \colon \mathbb{R} \to \mathbb{R}$ is measurable whenever
$f^{-1}(U) \in \Sigma_{\mathbb{R}}$ for all open subsets
$U \subseteq \mathbb{R}$. Therefore, all continuous functions on
$\mathbb{R}$ are measurable.

When $Y$ is a subset of the underlying set of a measurable space $X$,
we can equip $Y$ with a $\sigma$-algebra
$\Sigma_{Y} = \{A \cap Y : A \in \Sigma_{X}\}$. This way, we regard
the unit interval and the set of all non-negative real numbers as measurable
spaces, and indicate them as follows:
\begin{equation*}
  \mathbb{R}_{[0,1]} = \{a \in \mathbb{R} : 0 \leq a \leq 1\},
  \quad
  \mathbb{R}_{\geq 0} = \{a \in \mathbb{R} : a \geq 0\}
\end{equation*}
For measurable spaces $X$ and $Y$, we define the product measurable
space $X \times Y$ and the coproduct measurable space $X + Y$ by
\begin{align*}
  |X \times Y|
  &= |X| \times |Y|, \\
  |X + Y|
  &= \{(\hh,x) : x \in X\} \cup \{(\mm,y) : y \in Y\}
\end{align*}
where the underlying $\sigma$-algebras are:
\begin{align*}
  \Sigma_{X \times Y}
  &= \textnormal{the least $\sigma$-algebra
    such that } A \times B \in \Sigma_{X \times Y} \\
  &\quad
    \textnormal{ for all } A \in \Sigma_{X} \textnormal{ and }
    B \in \Sigma_{Y}, \\
  \Sigma_{X + Y}
  &= \{\{\hh\} \times A \cup \{\mm\} \times B : A \in \Sigma_{X}
    \textnormal{ and }
    B \in \Sigma_{Y}\}.
\end{align*}
We assume that $\times$ has higher precedence than $+$, i.e., we write
$X + Y \times Z$ for $X + (Y \times Z)$. In this paper, we always
regard finite products $\mathbb{R}^{n}$ as the product measurable
space on $\mathbb{R}$. It is well-known that the $\sigma$-algebra
$\Sigma_{\mathbb{R}^{n}}$ is the set of all Borel sets, i.e.,
$\Sigma_{\mathbb{R}^{n}}$ is the least one that contains all open
subsets of $\mathbb{R}^{n}$. Partial measurable functions are closed
under compositions, products and coproducts.

Let $X$ be a measurable space. A \emph{measure} $\mu$ on $X$ is a
function from $\Sigma_{X}$ to $[0,\infty]$ that is the set of all
non-negative real numbers extended with $\infty$, such that
\begin{varitemize}
\item $\mu(\emptyset) = 0$; and
\item for any mutually disjoint family
  $\{A_{n} \in \Sigma_{X}\}_{n \in \mathbb{N}}$, we have
  $\sum_{n \in \mathbb{N}} \mu(A_{n}) = \mu\left(\bigcup_{n \in
      \mathbb{N}} A_{n}\right)$.
\end{varitemize}
We say that a measure $\mu$ on $X$ is \emph{finite} when $\mu(X)<\infty$ and
that it is \emph{$\sigma$-finite} if
$X = \bigcup_{n \in \mathbb{N}} X_{n}$ for some family
$\{X_{n} \in \Sigma_{X}\}_{n \in \mathbb{N}}$ satisfying
$\mu(X_{n}) < \infty$.


For a measurable space $X$, we write $\varnothing_{X}$ for a measure
on $X$ given by $\varnothing_{X}(A) = 0$ for all $A \in \Sigma_{X}$.
If $\mu$ is a measure on a measurable space $X$, then for any
non-negative real number $a$, the function
\begin{math}
 (a\,\mu)(A) = a (\mu(A))
\end{math}
is also a measure on $X$. The \emph{Borel measure}
$\mu_{\mathrm{Borel}}$ on $\mathbb{R}^{n}$ is the unique measure that
satisfies
\begin{equation*}
  \mu_{\mathrm{Borel}}([a_{1},b_{1}] \times \cdots \times [a_{n},b_{n}])
  = \prod_{1 \leq i \leq n} |a_{i} - b_{i}|.
\end{equation*}
We define the Borel measure $\mu_{\mathrm{Borel}}$ on $1$ by
$\mu_{\mathrm{Borel}}(1) = 1$.
For a measurable
function $f \colon \mathbb{R}^{n} \to \mathbb{R}$ and
a measurable subset $X \subseteq \mathbb{R}^{n}$, we
denote the integral of $f$ with respect to the Borel measure
restricted to $X$ by
\begin{equation*}
  \int_{X} f(u)\; \mathrm{d}u.
\end{equation*}
For a measurable space $X$ and for an element $x \in X$, a \emph{Dirac
  measure} $\delta_{x}$ on $X$ is given by
\begin{equation*}
  \delta_{x}(A) = [x \in A] =
  \begin{cases}
    1, & \textnormal{if } x \in A; \\
    0, & \textnormal{if } x \notin A.
  \end{cases}
\end{equation*}
The square bracket notation in the right hand side is called Iverson's
bracket. In general, for a proposition $P$, we have $[P] = 1$ when $P$
is true and $[P] = 0$ when $P$ is false.

\begin{proposition}
  For every $\sigma$-finite measures $\mu$ on a measurable space $X$
  and $\nu$ on a measurable space $Y$, there is a unique measure $\mu
  \times \nu$ on $X \times Y$ such that $(\mu \times \nu)(A \times B)
  = \mu(A) \nu(B)$ for all $A \in \Sigma_{X}$ and $B \in \Sigma_{Y}$.
\end{proposition}
The measure $\mu \times \nu$ is called the \emph{product measure} of
$\mu$ and $\nu$. For example, the Borel measure on $\mathbb{R}^{2}$ is
the product measure of the Borel measure on $\mathbb{R}$. 

Finally, let us recall the notion of a kernel, which is a well-known
concept in the theory of stochastic processes. For measurable spaces
$X$ and $Y$, a \emph{kernel} from $X$ to $Y$ is a function $k \colon X
\times \Sigma_{Y} \to [0,\infty]$ such that for any $x \in X$, the
function $k(x,-)$ is a measure on $Y$, and for any $A \in \Sigma_{Y}$,
the function $k(-,A)$ is measurable. Notions of \emph{finite} and \emph{$\sigma$-finite}
kernels can be naturally given, following the emponymous constraint on
measures. Those kernels which can be expressed as the sum of countably
many finite kernels are said to be
\emph{s-finite}~\cite{staton2017}.
We use kernels to give semantics
for our probabilistic programming language, to be defined in the next
section.

\section{Syntax and Operational Semantics}\label{sec:syntax}
\subsection{Syntax and Type System}

Our language $\PCFSS$ for higher order Bayesian programming can be seen
as Plotkin's $\mathbf{PCF}$ endowed with real numbers, measurable
functions, sampling from the uniform distribution on
$\mathbb{R}_{[0,1]}$ and soft-conditioning. We first define types
$\mathtt{A},\mathtt{B},\ldots$, values $\mathtt{V},\mathtt{W},\ldots$
and terms $\mathtt{M},\mathtt{N},\ldots$ as follows:
\begin{align*}
  \mathtt{A},\mathtt{B}
  &::= \ttunit \midd \ttreal \midd \mathtt{A} \to \mathtt{B}, \\
  \mathtt{V},\mathtt{W}
  &::= \mathtt{skip} \midd \mathtt{x} \midd 
    \lambda \mathtt{x}^{\mathtt{A}}.\,\mathtt{M} \midd 
    \mathtt{r}_{a} \midd
    \mathtt{fix}_{\mathtt{A},\mathtt{B}}(\mathtt{f},\mathtt{x},\mathtt{M}), \\
  \mathtt{M}, \mathtt{N}
  &::= \mathtt{\mathtt{V}} \midd \mathtt{V}\,\mathtt{W} \midd
    \letin{\mathtt{x}}{\mathtt{M}}{\mathtt{N}} \midd
    \ifterm{\mathtt{V}}{\mathtt{M}}{\mathtt{N}} \\
  &\hspace{15pt}
    \midd \mathtt{F} (\mathtt{V}_{1},\dots, \mathtt{V}_{\arity{\mathtt{F}}})
    \midd \mathtt{sample}
    \midd \mathtt{score}(\mathtt{V}).
\end{align*}
Here, $\mathtt{x}$ varies over a countably infinite set of variable
symbols, and $a$ varies over the set $\mathbb{R}$ of all real numbers.
Each \emph{function identifier} $\mathtt{F}$ is associated with a
measurable function $\mathrm{fun}_{\mathtt{F}}$ from
$\mathbb{R}^{\vert \mathtt{F} \vert}$ to $\mathbb{R}$. For terms
$\mathtt{M}$ and $\mathtt{N}$, we write
$\mathtt{M}\{\mathtt{N}/\mathtt{x}\}$ for the capture-avoiding
substitution of $\mathtt{x}$ in $\mathtt{M}$ by $\mathtt{N}$.

Terms in $\PCFSS$ are restricted to be \emph{A-normal forms}, in order
to make some of the arguments on our semantics simpler. This
restriction is harmless for the language's expressive power, thanks to
the presence of $\mathtt{let}$-bindings.  For example, term
application $\mathtt{M}\,\mathtt{N}$ can be defined to be
$\letin{\mathtt{x}}{\mathtt{M}}{
  \letin{\mathtt{y}}{\mathtt{N}}{\mathtt{x}\,\mathtt{y}}}$.

The term constructor $\mathtt{score}$ and the constant
$\mathtt{sample}$ enable probabilistic programming in $\PCFSS$.
Evaluation of $\mathtt{score}(\mathtt{r}_{a})$ has the effect of
multiplying the weight of the current probabilistic branch by $|a|$,
this way enabling a form of \emph{soft-conditioning}. The constant $\mathtt{sample}$
generates a real number randomly drawn from the uniform distribution
on $\mathbb{R}_{[0,1]}$. Only one sampling mechanism is sufficient
because we can model sampling from other standard distributions by
composing $\mathtt{sample}$ with measurable functions \cite{wcgc2018}.

Terms can be typed in a natural way. A \emph{context}
$\mathtt{\Delta}$ is a finite sequence consisting of pairs of a
variable and a type such that every variable appears in
$\mathtt{\Delta}$ at most once. A \emph{type judgement} is a triple
$\mathtt{\Delta} \vdash \mathtt{M} : \mathtt{A}$ consisting of a
context $\mathtt{\Delta}$, a term $\mathtt{M}$ and a type
$\mathtt{A}$. We say that a type judgement
$\mathtt{\Delta} \vdash \mathtt{M} : \mathtt{A}$ is \emph{derivable}
when we can derive $\mathtt{\Delta} \vdash \mathtt{M} : \mathtt{A}$
from the typing rules in Figure~\ref{fig:types}. Here, the type of
$\mathtt{sample}$ is $\ttreal$, and the type of
$\mathtt{score}(\mathtt{V})$ is $\ttunit$ because $\mathtt{sample}$
returns a real number, and the purpose of scoring is its side effect.

In the sequel, we only consider derivable type judgements and \emph{typable}
closed terms, that is, closed terms $\mathtt{M}$ such that
$\vdash \mathtt{M} : \mathtt{A}$ is derivable for some type
$\mathtt{A}$.
\begin{figure}[t]
  \begin{center}
    \fbox{
      \begin{minipage}{.959\columnwidth}
        \centering
        $\begin{array}{c}
           \infer{\mathtt{\Delta} \vdash \mathtt{x} : \mathtt{A}}{
           \mathtt{x} : \mathtt{A} \in \mathtt{\Delta}}
           \hspace{8pt}
           \infer{\mathtt{\Delta} \vdash \mathtt{r}_{a} : \ttreal}{a \in \mathbb{R}}
           \hspace{8pt}
           \infer{\mathtt{\Delta}\vdash \mathtt{F}
           (\mathtt{V}_{1},\dots, \mathtt{V}_{\arity{\mathtt{F}}}) :\ttreal}
           {\mathtt{\Delta}\vdash \mathtt{V}_{i}:\ttreal
           \text{ for all $i\leq \arity{\mathtt{F}}$}
           }
           \\[6pt]
           \infer{\mathtt{\Delta} \vdash \mathtt{V}\,\mathtt{W} : \mathtt{B}}{
           \mathtt{\Delta} \vdash \mathtt{V} : \mathtt{A} \to \mathtt{B}
           &
             \mathtt{\Delta} \vdash \mathtt{W} :\mathtt{A}
             }
             \hspace{8pt}
             \infer{\mathtt{\Delta} \vdash \letin{\mathtt{x}}{\mathtt{M}}{\mathtt{N}} : \mathtt{A}}
             {
             \mathtt{\Delta} \vdash \mathtt{M} : \mathtt{B}
           &
             \mathtt{\Delta}, \mathtt{x} : \mathtt{B} \vdash \mathtt{N} : \mathtt{A}	
             }
           \\[6pt]
           \infer{\mathtt{\Delta} \vdash \lambda \mathtt{x}^{\mathtt{A}}.\,\mathtt{M}
           : \mathtt{A} \to \mathtt{B}}{
           \mathtt{\Delta},\, \mathtt{x} : \mathtt{A} \vdash \mathtt{M} : \mathtt{B}}
           \hspace{10pt}
           \infer{\mathtt{\Delta} \vdash
           \mathtt{fix}_{\mathtt{A},\mathtt{B}}(\mathtt{f},\mathtt{x},\mathtt{M})
           : \mathtt{A} \to \mathtt{B}}{
           \mathtt{\Delta},\, \mathtt{f} : \mathtt{A} \to \mathtt{B},\, \mathtt{x} : \mathtt{A}
           \vdash \mathtt{M} : \mathtt{B}}
           \\[6pt]
           \infer{\mathtt{\Delta} \vdash \mathtt{skip} : \ttunit}{}
           \hspace{10pt}
           \infer{\mathtt{\Delta}\vdash \ifterm{\mathtt{V}}{\mathtt{M}}{\mathtt{N}}: \mathtt{A}}{
           \mathtt{\Delta}\vdash \mathtt{V} :\ttreal
           &
             \mathtt{\Delta}\vdash \mathtt{M} :\mathtt{A}
           &
             \mathtt{\Delta}\vdash \mathtt{N}:\mathtt{A}
             }
           \\[6pt]
           \infer{\mathtt{\Delta} \vdash \mathtt{sample} : \ttreal}{}
           \quad
           \infer{\mathtt{\Delta}\vdash \mathtt{score}(\mathtt{V}) : \ttunit}{
           \mathtt{\Delta}\vdash \mathtt{V}:\ttreal}
         \end{array}$
      \end{minipage}}
  \end{center}
  \caption{Typing Rules}
  \label{fig:types}
\end{figure}

\subsection{Distribution-Based Operational Semantics}

We define distribution-based operational semantics following
\cite{bdlgs2016} where, however, a $\sigma$-algebra on the set of
terms is necessary so as to define evaluation results of terms to
be distributions (i.e. measures) over values. In this paper, we only consider
evaluation of terms of type $\ttreal$ and avoid introducing
$\sigma$-algebras on sets of closed terms, thus greatly simplifying
the overall development.

Distribution-based operational semantics is a function that sends a
closed term $\mathtt{M}:\ttreal$ to a measure $\mu$ on $\mathbb{R}$.
Because of the presence of $\mathtt{score}$, the measure may not be a
probabilistic measure, i.e., $\mu(\mathbb{R})$ may be larger than $1$,
but the idea of distribution-based operational semantics is precisely
that of associating each closed term of type $\ttreal$ with a measure
over $\mathbb{R}$.

As common in call-by-value programming languages, evaluation is
defined by way of \emph{evaluation contexts}:
\begin{equation*}
  \mathtt{E}[-] ::= [-] \midd \letin{\mathtt{x}}{\mathtt{E}[-]}{\mathtt{M}}.
\end{equation*}
The distribution-based operational semantics of $\PCFSS$ is a family
of binary relations $\{\Rightarrow_{n}\}_{n \in \mathbb{N}}$ between
closed terms of type $\ttreal$ and measures on $\mathbb{R}$
inductively defined by the evaluation rules in
Figure~\ref{fig:distribution-based} where the evaluation rule for
$\mathtt{score}$ is inspired from the one in \cite{staton2017}. The binary
relation $\detred$ in the precondition of the third rule in
Figure~\ref{fig:distribution-based} is called \emph{deterministic
  reduction} and is defined as follows as a relation on
closed terms:
\begin{align*}
  (\lambda \mathtt{x}^{\mathtt{A}}.\,\mathtt{M})\,\mathtt{V}
  &\detred \mathtt{M}\{\mathtt{V}/\mathtt{x}\},\\
  \letin{\mathtt{x}}{\mathtt{V}}{\mathtt{M}}
  &\detred \mathtt{M}\{\mathtt{V}/\mathtt{x}\},\\
  \mathtt{fix}_{\mathtt{A},\mathtt{B}}(\mathtt{f},\mathtt{x},\mathtt{M})\,\mathtt{V}
  &\detred \mathtt{M}\{\mathtt{fix}_{\mathtt{A},\mathtt{B}}
    (\mathtt{f},\mathtt{x},\mathtt{M})/\mathtt{f},\mathtt{V}/\mathtt{x}\}, \\
  \ifterm{\mathtt{r}_{a}}{\mathtt{M}}{\mathtt{N}}
  &\detred
    \begin{cases}
      \mathtt{M}, & \textnormal{if } a = 0, \\
      \mathtt{N}, & \textnormal{if } a \neq 0,
    \end{cases} \\
  \mathtt{F}(\mathtt{r}_{a}, \ldots \mathtt{r}_{b})
  &\detred
    \mathtt{r}_{\mathrm{fun}_{\mathtt{F}}(a,\ldots,b)}.
\end{align*}
The last evaluation rule in Figure~\ref{fig:distribution-based} makes
sense because $k$ in the precondition is a kernel from
$\mathbb{R}_{[0,1]}$ to $\mathbb{R}$:
\begin{lemma}\label{lem:distribution-based}
  For any $n \in \mathbb{N}$ and for any term
  \begin{equation*}
    \mathtt{x}_{1}:\ttreal, \ldots,\mathtt{x}_{m}:\ttreal \vdash
    \mathtt{M}:\ttreal, 
  \end{equation*}
  there is a finite kernel $k$ from $\mathbb{R}^{m}$ to $\mathbb{R}$ such
  that for any $u \in \mathbb{R}^{m}$ and for any measure $\mu$
  on $\mathbb{R}$,
  \begin{equation*}
    \mathtt{M}\{\mathtt{r}_{a_{1}}/\mathtt{x}_{1},
    \ldots,\mathtt{r}_{a_{m}}/\mathtt{x}_{m}\}
    \Rightarrow_{n} \mu
    \iff \mu = k(u,-)
  \end{equation*}
  where $u=(a_{1},\ldots,a_{m})$.
\end{lemma}
\longv{
  \begin{proof}
    Let $\mathtt{\Delta}$ be a context of the form
    $\mathtt{x}_{1}:\ttreal,\ldots,\mathtt{x}_{m}:\ttreal$. In this proof,
    for a finite sequence
    $u = (a_{1},\ldots,a_{n}) \in \mathbb{R}^{m}$, and for a term
    $\mathtt{\Delta} \vdash \mathtt{M} : \mathtt{A}$, we denote
    \begin{equation*}
      \mathtt{M}\{\mathtt{r}_{a_{1}}/\mathtt{x}_{1},\ldots,\mathtt{r}_{a_{m}}/\mathtt{x}_{m}\}
    \end{equation*}
    by $\mathtt{M}\{\mathtt{r}_{u}/\mathtt{\Delta}\}$.
    We prove the statement by induction on $n \in \mathbb{N}$. (Base
    case) Let $k$ be a kernel from $\mathbb{R}^{m}$ to $\mathbb{R}$
    given by
    \begin{equation*}
      k(u,A) = 0.
    \end{equation*}
    Then for any $u=(a_{1},\ldots,a_{m}) \in \mathbb{R}^{m}$,
    \begin{equation*}
      \mathtt{M}\{\mathtt{r}_{a_{1}}/\mathtt{x}_{1},\ldots,\mathtt{r}_{a_{m}}/\mathtt{x}_{m}\}
      \Rightarrow_{0} \mu
      \iff \mu = \varnothing_{\mathbb{R}}
      \iff \mu = k(u,-).
    \end{equation*}
    (Induction step) We define a \emph{redex} $\mathtt{R}$ by
    \begin{align*}
      \mathtt{R} ::=
      &
        \mathtt{score}(\mathtt{V})
        \midd
        \mathtt{sample}
        \midd
        (\lambda \mathtt{x}^{\mathtt{A}}.\,\mathtt{M})\,\mathtt{V}
        \midd
        \mathtt{fix}_{\mathtt{A},\mathtt{B}}(\mathtt{f},\mathtt{x},\mathtt{M})\,\mathtt{V} \\
      &
        \midd
        \mathtt{F}(\mathtt{V}, \ldots \mathtt{W})
        \midd
        \letin{\mathtt{x}}{\mathtt{V}}{\mathtt{M}}
        \midd
        \ifterm{\mathtt{r}_{a}}{\mathtt{M}}{\mathtt{N}}.
    \end{align*}
    We note that $\mathtt{V},\mathtt{W}$ in the above BNF can be
    variables. By induction on the size of type derivation, we can
    show that every term
    $\mathtt{\Delta} \vdash \mathtt{M}:\mathtt{A}$ is either a value
    or of the form $\mathtt{E}[\mathtt{R}]$ for some evaluation
    context $\mathtt{E}[-]$ and some redex $\mathtt{R}$. Given a term
    $\mathtt{\Delta} \vdash \mathtt{M} : \mathtt{A}$ where
    $\mathtt{\Delta} =
    \mathtt{x}_{1}:\ttreal,\ldots,\mathtt{x}_{m}:\ttreal$, we prove
    the induction step by case analysis.
    \begin{varitemize}
    \item If $\mathtt{\Delta} \vdash \mathtt{M} : \ttreal$ is a value,
      then $\mathtt{M}$ is either a variable $\mathtt{x}_{i}$ or a
      constant $\mathtt{r}_{a}$. When $\mathtt{M}$ is a variable
      $\mathtt{x}_{i}$, we have
      \begin{equation*}
        \mathtt{x}_{i}\{\mathtt{r}_{a_{1}}/\mathtt{x}_{1},\ldots,\mathtt{r}_{a_{m}}/\mathtt{x}_{m}\}
        \equiv \mathtt{r}_{a_{i}}
        \Rightarrow_{n+1} \mu
        \iff \mu = \delta_{a_{i}}.
      \end{equation*}
      When $\mathtt{M}$ is a constant $\mathtt{r}_{a}$, we have
      \begin{equation*}
        \mathtt{r}_{a}\{\mathtt{r}_{a_{1}}/\mathtt{x}_{1},\ldots,\mathtt{r}_{a_{m}}/\mathtt{x}_{m}\}
        \equiv \mathtt{r}_{a}
        \Rightarrow_{n+1} \mu
        \iff \mu = \delta_{a}.
      \end{equation*}
      Both
      $k,h \colon \mathbb{R}^{m} \times \Sigma_{\mathbb{R}} \to
      [0,\infty]$ given by
      \begin{equation*}
        k((a_{1},\ldots,a_{m}),A) = \delta_{a_{i}}(A),
        \qquad
        h((a_{1},\ldots,a_{m}),A) = \delta_{a}(A)
      \end{equation*}
      are kernels from $\mathbb{R}^{m}$ to $\mathbb{R}$.
    \item If $\mathtt{\Delta} \vdash \mathtt{M} : \ttreal$ is of the
      form $\mathtt{E}[\mathtt{sample}]$, then by induction
      hypothesis, there is a kernel from $\mathbb{R}^{m+1}$ to
      $\mathbb{R}$ such that for any $u \in \mathbb{R}^{m+1}$,
      \begin{equation*}
        \mathtt{E}[\mathtt{y}]
        \{\mathtt{r}_{u}/(\mathtt{\Delta},\mathtt{y}:\ttreal)\}
        \Rightarrow_{n} \mu
        \iff \mu = k(u,-).
      \end{equation*}
      We define a kernel $h$ from $\mathbb{R}^{m}$ to $\mathbb{R}$ by
      \begin{equation*}
        h((a_{1},\ldots,a_{m}),A) = \int_{\mathbb{R}_{[0,1]}}
        k((a_{1},\ldots,a_{m},a),A)\; \mathrm{d}a.
      \end{equation*}
      \longv{This is a kernel because if
        $f \colon \mathbb{R} \times \cdots \times \mathbb{R} \to
        \mathbb{R}$ is a non-negative measurable function, then
        \begin{equation*}
          (b,\ldots,c) \mapsto \int_{\mathbb{R}} f(a,b,\ldots,c)
          \; \mathrm{d}a
        \end{equation*}
        is measurable. See \cite[Theorem~18.3]{billingsley1986}.}
      Then, for any $u=(a_{1},\ldots,a_{m}) \in \mathbb{R}^{m}$,
      \begin{align*}
        \mathtt{E}[\mathtt{sample}]
        \{\mathtt{r}_{u}/\mathtt{\Delta}\}
        \Rightarrow_{n+1} \mu
        &\iff \mu = \int_{\mathbb{R}_{[0,1]}} k((a_{1},\ldots,a_{m},a),-) \; \mathrm{d}a \\
        &\iff \mu = h(u,-).
      \end{align*}
    \item If $\mathtt{\Delta} \vdash \mathtt{M} : \mathtt{B}$ is of
      the form $\mathtt{E}[\mathtt{score}(\mathtt{x}_{i})]$ for some
      $i \in \{1,2,\ldots,m\}$, then by induction hypothesis, there is
      a kernel $k$ from $\mathbb{R}^{m}$ to $\mathbb{R}$ such that for
      any $u \in \mathbb{R}^{m}$,
      \begin{equation*}
        \mathtt{E}[\mathtt{skip}]
        \{\mathtt{r}_{u}/\mathtt{\Delta}\}
        \Rightarrow_{n} \mu
        \iff \mu = k(u,-).
      \end{equation*}
      We define a kernel $h \colon \mathbb{R}^{m}$ to $\mathbb{R}$ by
      \begin{equation*}
        h((a_{1},\ldots,a_{m}),A) =
        |a_{i}|\,k((a_{1},\ldots,a_{m}),A).
      \end{equation*}
      Then, for any $u = (a_{1},\ldots,a_{m}) \in \mathbb{R}^{m}$,
      \begin{align*}
        \mathtt{E}[\mathtt{score}(\mathtt{x}_{i})]
        \{\mathtt{r}_{u}/\mathtt{\Delta}\}
        \Rightarrow_{n+1}
        \mu
        &\iff \mathtt{E}[\mathtt{skip}]
          \{\mathtt{r}_{u}/\mathtt{\Delta}\}
          \Rightarrow_{n}
          \nu
          \textnormal{ and }
          \mu = |a_{i}|\,\nu \\
        &\iff \mu = h(u,-).
      \end{align*}
    \item If $\mathtt{\Delta} \vdash \mathtt{M} : \mathtt{B}$ is of
      the form $\mathtt{E}[\mathtt{score}(\mathtt{r}_{a})]$ for some
      $a \in \mathbb{R}$, then by induction hypothesis, there is a
      kernel $k$ from $\mathbb{R}^{m}$ to $\mathbb{R}$ such that for
      any $u \in \mathbb{R}^{m}$,
      \begin{equation*}
        \mathtt{E}[\mathtt{skip}]
        \{\mathtt{r}_{u}/\mathtt{\Delta}\}
        \Rightarrow_{n} \mu
        \iff \mu = k(u,-).
      \end{equation*}
      We define a kernel $h \colon \mathbb{R}^{m}$ to $\mathbb{R}$ by
      \begin{equation*}
        h((a_{1},\ldots,a_{m}),A) =
        |a|\,k((a_{1},\ldots,a_{m}),A).
      \end{equation*}
      Then, for any $u = (a_{1},\ldots,a_{m}) \in \mathbb{R}^{m}$,
      \begin{align*}
        \mathtt{E}[\mathtt{score}(\mathtt{x}_{i})]
        \{\mathtt{r}_{u}/\mathtt{\Delta}\}
        \Rightarrow_{n+1}
        \mu
        &\iff \mathtt{E}[\mathtt{skip}]
          \{\mathtt{r}_{u}/\mathtt{\Delta}\}
          \Rightarrow_{n}
          \nu
          \textnormal{ and }
          \mu = |a|\,\nu \\
        &\iff \mu = h(u,-).
      \end{align*}
    \item If $\mathtt{\Delta} \vdash \mathtt{M} : \mathtt{B}$ is of
      the form
      $\mathtt{E}[(\lambda
      \mathtt{x}^{\mathtt{A}}.\,\mathtt{N})\,\mathtt{V}]$, then by
      induction hypothesis, there is a kernel $k$ from
      $\mathbb{R}^{m}$ to $\mathbb{R}$ such that for all
      $u \in \mathbb{R}^{m}$,
      \begin{equation*}
        \mathtt{E}[\mathtt{N}\,\{\mathtt{V}/\mathtt{x}\}]\{\mathtt{r}_{u}/\mathtt{\Delta}\}
        \Rightarrow_{n} \mu
        \iff \mu = k(u,-).
      \end{equation*}
      Hence,
      \begin{align*}
        \mathtt{E}[(\lambda \mathtt{x}^{\mathtt{A}}.\,\mathtt{N})\,\mathtt{V}]
        \{\mathtt{r}_{u}/\mathtt{\Delta}\}
        \Rightarrow_{n+1} \mu
        &\iff \mathtt{E}[\mathtt{N}\{\mathtt{V}/\mathtt{x}\}]
          \{\mathtt{r}_{u}/\mathtt{\Delta}\}
          \Rightarrow_{n} \mu \\
        &\iff \mu = k(u,-).
      \end{align*}
    \item If $\mathtt{\Delta} \vdash \mathtt{M} : \mathtt{B}$ is of
      the form
      $\mathtt{E}[\mathtt{fix}_{\mathtt{A},\mathtt{B}}(\mathtt{f},\mathtt{x},\mathtt{N})\,\mathtt{V}]$,
      then by induction hypothesis, there is a kernel $k$ from
      $\mathbb{R}^{m}$ to $\mathbb{R}$ such that for all
      $u \in \mathbb{R}^{m}$,
      \begin{equation*}
        \mathtt{E}[\mathtt{N}
        \{\mathtt{fix}_{\mathtt{A},\mathtt{B}}(\mathtt{f},\mathtt{x},\mathtt{N})/\mathtt{f},
        \mathtt{V}/\mathtt{x}\}]\{\mathtt{r}_{u}/\mathtt{\Delta}\}
        \Rightarrow_{n} \mu
        \iff \mu = k(u,-).
      \end{equation*}
      Hence,
      \begin{align*}
        \mathtt{E}[\mathtt{fix}_{\mathtt{A},\mathtt{B}}(\mathtt{f},\mathtt{x},\mathtt{N})\,\mathtt{V}]
        \Rightarrow_{n+1} \mu
        &\iff \mathtt{E}[\mathtt{N}
          \{\mathtt{fix}_{\mathtt{A},\mathtt{B}}(\mathtt{f},\mathtt{x},\mathtt{N})/\mathtt{f},
          \mathtt{V}/\mathtt{x}\}]\{\mathtt{r}_{u}/\mathtt{\Delta}\}
          \Rightarrow_{n} \mu \\
        &\iff \mu = k(u,-).
      \end{align*}
    \item If $\mathtt{\Delta} \vdash \mathtt{M} : \ttreal$ is of the
      form
      $\mathtt{E}[\mathtt{F}(\mathtt{V}_{1},\ldots,\mathtt{V}_{|\mathtt{F}|})]$,
      then $\mathtt{V}_{i}$ is equal to either a variable or a
      constant $\mathtt{r}_{a}$. For simplicity, we suppose that
      $|\mathtt{F}| = 2$ and $\mathtt{V}_{1} = \mathtt{x}_{i}$ and
      $\mathtt{V}_{2} = \mathtt{r}_{a}$. By induction hypothesis,
      there is a kernel from $\mathbb{R}^{m+1}$ to $\mathbb{R}$ such
      that for all $u \in \mathbb{R}^{m+1}$,
      \begin{equation*}
        \mathtt{E}[\mathtt{y}]
        \{\mathtt{r}_{u}/(\mathtt{\Delta},\mathtt{y}:\ttreal)\}
        \Rightarrow_{n} \mu
        \iff \mu = k(u,-).
      \end{equation*}
      We define a kernel $h$ from $\mathbb{R}^{m}$ to $\mathbb{R}$ by
      \begin{equation*}
        h((a_{1},\ldots,a_{m}),A)
        = k((a_{1},\ldots,a_{m},\mathrm{fun}_{\mathtt{F}}(a_{i},a)),A).
      \end{equation*}
      Then, for any $u=(a_{1},\ldots,a_{m}) \in \mathbb{R}^{m}$,
      \begin{align*}
        \mathtt{E}[\mathtt{F}[\mathtt{x}_{i},\mathtt{r}_{a}]]
        \{\mathtt{r}_{u}/\mathtt{\Delta}\}
        \Rightarrow_{n+1} \mu
        &\iff \mathtt{E}[\mathtt{y}]
          \{\mathtt{r}_{u}/\mathtt{\Delta},\mathtt{r}_{\mathrm{fun}_{\mathtt{F}}(a_{i},a)}/\mathtt{y}\}
          \Rightarrow_{n} \mu \\
        &\iff \mu = k((u,\mathrm{fun}_{\mathtt{F}}(a_{i},a)),-) = h(u,-).
      \end{align*}
    \item If $\mathtt{\Delta} \vdash \mathtt{M} : \ttreal$ is of the
      form $\letin{\mathtt{x}}{\mathtt{V}}{\mathtt{N}}$, then by
      induction hypothesis, there is a kernel $k$ from
      $\mathbb{R}^{m}$ to $\mathbb{R}$ such that for all
      $u \in \mathbb{R}^{m}$,
      \begin{equation*}
        \mathtt{E}[\mathtt{N}\,\{\mathtt{V}/\mathtt{x}\}]\{\mathtt{r}_{u}/\mathtt{\Delta}\}
        \Rightarrow_{n} \mu
        \iff \mu = k(u,-).
      \end{equation*}
      Hence,
      \begin{align*}
        \mathtt{E}[\letin{\mathtt{x}}{\mathtt{V}}{\mathtt{N}}]
        \{\mathtt{r}_{u}/\mathtt{\Delta}\}
        \Rightarrow_{n+1} \mu
        &\iff \mathtt{E}[\mathtt{N}\{\mathtt{V}/\mathtt{x}\}]
          \{\mathtt{r}_{u}/\mathtt{\Delta}\}
          \Rightarrow_{n} \mu \\
        &\iff \mu = k(u,-).
      \end{align*}
    \item If $\mathtt{\Delta} \vdash \mathtt{M} : \ttreal$ is of the
      form
      $\mathtt{E}[\ifterm{\mathtt{x}_{i}}{\mathtt{N}}{\mathtt{L}}]$
      for some $i \in \{1,2,\ldots,m\}$, then by induction hypothesis,
      there are kernels $k$ and $k'$ from $\mathbb{R}^{m}$ to
      $\mathbb{R}$ such that for any $u \in \mathbb{R}^{m}$,
      \begin{align*}
        \mathtt{E}[\mathtt{N}]\{\mathtt{r}_{u}/\mathtt{\Delta}\}
        \Rightarrow_{n} \mu
        &\iff \mu = k(u,-), \\
        \mathtt{E}[\mathtt{L}]\{\mathtt{r}_{u}/\mathtt{\Delta}\}
        \Rightarrow_{n} \mu
        &\iff \mu = k'(u,-).
      \end{align*}
      We define a kernel $h$ from $\mathbb{R}^{m}$ to $\mathbb{R}$ by
      \begin{equation*}
        h(u,A) =
        \begin{cases}
          k(u,A), & \textnormal{if } a_{i} = 0, \\
          k'(u,A), & \textnormal{if } a_{i} \neq 0
        \end{cases}
        \qquad
        \textnormal{ where }
        u = (a_{1},\ldots,a_{n}).
      \end{equation*}
      Then, for any $u \in \mathbb{R}^{m}$,
      \begin{align*}
        \mathtt{E}[\ifterm{\mathtt{x}_{i}}{\mathtt{N}}{\mathtt{L}}]
        \{\mathtt{r}_{u}/\mathtt{\Delta}\}
        \Rightarrow_{n+1} \mu
        &\iff \bigl(
          \mathtt{E}[\mathtt{N}]
          \{\mathtt{r}_{u}/\mathtt{\Delta}\}
          \Rightarrow_{n} \mu
          \textnormal{ and } a_{i} = 0
          \bigr) \\
        &\hspace{28pt} \textnormal{ or } \bigl(
          \mathtt{E}[\mathtt{L}]
          \{\mathtt{r}_{u}/\mathtt{\Delta}\}
          \Rightarrow_{n} \mu
          \textnormal{ and } a_{i} \neq 0
          \bigr) \\
        &\iff \mu = h(u,-).
      \end{align*}
    \item If $\mathtt{\Delta} \vdash \mathtt{M} : \ttreal$ is of the
      form
      $\mathtt{E}[\ifterm{\mathtt{r}_{0}}{\mathtt{N}}{\mathtt{L}}]$,
      then by induction hypothesis, there is a kernel $k$ from
      $\mathbb{R}^{m}$ to $\mathbb{R}$ such that for any
      $u \in \mathbb{R}^{m}$,
      \begin{equation*}
        \mathtt{E}[\mathtt{N}]\{\mathtt{r}_{u}/\mathtt{\Delta}\}
        \Rightarrow_{n} \mu
        \iff \mu = k(u,-).
      \end{equation*}
      Hence,
      \begin{align*}
        \mathtt{E}[\ifterm{\mathtt{r}_{0}}{\mathtt{N}}{\mathtt{L}}]
        \{\mathtt{r}_{u}/\mathtt{\Delta}\}
        \Rightarrow_{n+1} \mu
        &\iff
          \mathtt{E}[\mathtt{N}]
          \{\mathtt{r}_{u}/\mathtt{\Delta}\}
          \Rightarrow_{n} \mu \\
        &\iff \mu = k(u,-).
      \end{align*}
    \item If $\mathtt{\Delta} \vdash \mathtt{M} : \ttreal$ is of the
      form
      $\mathtt{E}[\ifterm{\mathtt{r}_{a}}{\mathtt{N}}{\mathtt{L}}]$
      for some real number $a \neq 0$, then by induction hypothesis,
      there is a kernel $k$ from $\mathbb{R}^{m}$ to $\mathbb{R}$ such
      that
      \begin{equation*}
        \mathtt{E}[\mathtt{L}]\{\mathtt{r}_{u}/\mathtt{\Delta}\}
        \Rightarrow_{n} \mu
        \iff \mu = k(u,-).
      \end{equation*}
      Hence,
      \begin{align*}
        \mathtt{E}[\ifterm{\mathtt{r}_{a}}{\mathtt{N}}{\mathtt{L}}]
        \{\mathtt{r}_{u}/\mathtt{\Delta}\}
        \Rightarrow_{n+1} \mu
        &\iff
          \mathtt{E}[\mathtt{L}]
          \{\mathtt{r}_{u}/\mathtt{\Delta}\}
          \Rightarrow_{n} \mu \\
        &\iff \mu = k(u,-).
      \end{align*}
    \end{varitemize}
  \end{proof}
} Lemma~\ref{lem:distribution-based} implies that the relations
$\Rightarrow_{n}$ can be seen as functions from the set of closed
terms of type $\ttreal$ to the set of measures on
$\mathbb{R}$. \shortv{We can prove Lemma~\ref{lem:distribution-based}
  by induction on $n \in \mathbb{N}$. In the proof, the following fact
  is crucial: if $f \colon \mathbb{R} \times \cdots \times \mathbb{R}
  \to \mathbb{R}$ is a non-negative measurable function, then
\begin{equation*}
  (b,\ldots,c) \mapsto \int_{\mathbb{R}} f(a,b,\ldots,c)
  \; \mathrm{d}a
\end{equation*}
is measurable (see \cite[Theorem~18.3]{billingsley1986}).
In the proof of this fact, $\sigma$-finiteness of the Borel measure is exploited
in an essential way.}

\begin{figure}[t]
  \begin{center}
    \fbox{
      \begin{minipage}{.959\columnwidth}
        \centering
        $\begin{array}{c}
           \infer{\mathtt{M} \Rightarrow_{0} \varnothing_{\mathbb{R}}}{}
           \quad
           \infer{\mathtt{r}_{a} \Rightarrow_{n} \delta_{a}}{n > 0}
           \quad
           \infer{\mathtt{E}[\mathtt{M}] \Rightarrow_{n+1} \mu}{
           \mathtt{M} \detred \mathtt{N}
           &
             \mathtt{E}[\mathtt{N}] \Rightarrow_{n} \mu}
           \\[4pt]
           \infer{\mathtt{E}[\mathtt{score}(\mathtt{r}_{a})]
           \Rightarrow_{n+1}
           |a|\;\mu}{
           \mathtt{E}[\mathtt{skip}] \Rightarrow_{n} \mu
           }
           \\[4pt]
           \infer{\mathtt{E}[\mathtt{sample}]
           \Rightarrow_{n+1}
           \int_{\mathbb{R}_{[0,1]}} k(a,-)\;
           \mathrm{d}a
           }{\mathtt{E}[\mathtt{r}_{a}]
           \Rightarrow_{n} k(a,-)
           &
             \textnormal{for all }a \in \mathbb{R}_{[0,1]}
             }
         \end{array}$
        \end{minipage}}
  \end{center}
  \caption{Evaluation Rules of Distribution-Based Operational Semantics}
  \label{fig:distribution-based}
\end{figure}

The step-indexed distribution-based operational semantics approximates
the evaluation of closed terms by restricting the number of reduction
steps. Thus, the limit of the step-indexed distribution-based
operational semantics represents the ``true'' result of evaluating the
underlying term.

\begin{definition}
  For a closed term $\mathtt{M}:\ttreal$
  and a measure $\mu$ on $\mathbb{R}$,
  we write $\mathtt{M} \Rightarrow_{\infty} \mu$
  when there is a family of measures
  $\{\mu_{n}\}_{n \in \mathbb{N}}$ on $\mathbb{R}$
  such that $\mathtt{M} \Rightarrow_{n} \mu_{n}$
  and for all $A \in \Sigma_{\mathbb{R}}$,
  \begin{equation*}
    \mu(A) = \sup_{n \in \mathbb{N}} \mu_{n}(A).
  \end{equation*}
\end{definition}

The binary relation $\Rightarrow_{\infty}$ is a function from the set
of closed terms of type $\ttreal$ to the set of measures on
$\mathbb{R}$. This follows from Lemma~\ref{lem:distribution-based} and
that the family of measures $\{\mu_{n}\}_{n \in \mathbb{N}}$ on
$\mathbb{R}$ such that $\mathtt{M} \Rightarrow_{n} \mu_{n}$ forms an
ascending chain $\mu_{0} \leq \mu_{1} \leq \cdots$ with respect to the
pointwise order. Moreover, it can be proved that
for any $\mathtt{x_1}:\ttreal,\ldots,\mathtt{x}_m:\ttreal \vdash \mathtt{M}:
\ttreal$, $k$ given by $\mathtt{M}\{\mathtt{r}_{a_1}/\mathtt{x_1},\ldots,\mathtt{r}_{a_m}/\mathtt{x}_m\}
\Rightarrow_{\infty} k((a_1,\ldots,a_m),-)$ is an s-finite kernel.

\subsection{Sampling-Based Operational Semantics}
\label{sec:sample}

$\PCFSS$ can be endowed with another form of operational semantics,
closer in spirit to inference algorithms, called the
\emph{sampling-based} operational semantics. The way we formulate it
is deeply inspired from the one in \cite{bdlgs2016}.

The idea behind sampling-based operational semantics is to give the
evaluation result of each probabilistic branch somehow independently.
We specify each probabilistic branch by two parameters: one is a
sequence of random draws, which will be consumed by $\mathtt{sample}$;
the other is a likelihood measure called weight, which will be
modified by $\mathtt{score}$.

\begin{definition}\label{def:configuration}
  A \emph{configuration} is a triple $(\mathtt{M},a,u)$ consisting of
  a closed term $\mathtt{M}:\ttreal$, a real number $a \geq 0$ called
  the configuration's \emph{weight}, and a finite sequence $u$ of real
  numbers in $\mathbb{R}_{[0,1]}$, called its \emph{trace}.
\end{definition}

Below, we write $\varepsilon$ for the empty sequence. For a real
number $a$ and a finite sequence $u$ consisting of real numbers, we
write $a \mathbin{::} u$ for the finite sequence obtained by putting
$a$ on the head of $u$. In Figure~\ref{fig:sampling-based}, we give
the evaluation rules of sampling-based operational semantics where
$\detred$ is the deterministic reduction relation introduced in the
previous section. We denote the reflective transitive closure of $\to$
by $\to^{\ast}$. Intuitively, $(\mathtt{M},1,u) \to^{\ast}
(\mathtt{r}_{a},b,\varepsilon)$ means that by evaluating $\mathtt{M}$,
we get the real number $a$ with weight $b$ consuming \emph{all} the
random draws in $u$.

\begin{figure}[t]
  \begin{center}
    \fbox{\begin{minipage}{.959\columnwidth}
        \centering
        $\begin{array}{c}
            \infer{(\mathtt{M},b,u) \to (\mathtt{N},b,u)}{\mathtt{M} \detred \mathtt{N}}
            \\[6pt]
            (\mathtt{E}[\mathtt{score}(\mathtt{r}_{a})],b,u) \to
            (\mathtt{E}[\mathtt{skip}],|a|\,b,u)
            \\[6pt]
            (\mathtt{E}[\mathtt{sample}], a,b \mathbin{::}u)
            \to (\mathtt{E}[\mathtt{r}_{b}],a,u)
          \end{array}$
      \end{minipage}}
  \end{center}
  \caption{Evaluation Rules of Sampling-Based Operational Semantics}
  \label{fig:sampling-based}
\end{figure}

\section{Towards Mealy Machine Semantics}
\label{sec:towards}

In this section, we give some intuitions about our GoI model, which we
also call \emph{Mealy machine semantics}. Giving Mealy machine
semantics for $\PCFSS$ requires translating $\PCFSS$ into the linear
$\lambda$-calculus. This is because GoI is a semantics for linear
logic, and is thus tailored for calculi in which terms are treated as
resources. Schematically, Mealy machine semantics for $\PCFSS$
translates terms in $\PCFSS$ into Mealy machines in the following way.
\begin{equation*}
  \xymatrix@R=7pt{
    \textnormal{$\PCFSS$}
    \ar[d]^-{(1) \; \textnormal{Moggi's translation}} \\
    \textnormal{Moggi's meta-language\,$+\mathtt{sample}+\mathtt{score}$
    }
    \ar[d]^-{(2) \; \textnormal{Girard translation}} \\
    \textnormal{the linear $\lambda$-calculus\,$+\mathtt{sample}+\mathtt{score}$}
    \ar[d]^-{(3)} \\
    \textnormal{proof structures$+\mathtt{sample}+\mathtt{score}$}
    \ar[d]^-{(4)} \\
    \textnormal{Mealy machines}
    \nulldot
  }
\end{equation*}
In Section~\ref{sec:linear}, we explain the first three steps. The
last step deserves to be explained in more detail, which we do in
Section~\ref{sec:pn}. For the sake of simplicity, we ignore the
translation of conditional branching and the fixed point operator.

\subsection{From $\PCFSS$ to Proof Structures}
\label{sec:linear}

\subsubsection{Moggi's Translation}

In the first step, we translate $\PCFSS$ into an extension of the
Moggi's meta-language by Moggi's translation \cite{moggi1991}. Here,
in order to translate scoring and sampling in $\PCFSS$,
we equip Moggi's meta-language with base types $\ttunit$ and
$\ttreal$ and the following terms:
\begin{equation*}
  \infer{\mathtt{\Delta} \vdash \mathtt{r}_{a} : \ttreal}{a \in \mathbb{R}},
  \hspace{5pt}
  \infer{\mathtt{score}(\mathtt{M}) : \mathtt{T} \, \ttunit}{
    \mathtt{\Delta} \vdash \mathtt{M} : \ttreal},
  \hspace{5pt}
  \infer{\mathtt{\Delta} \vdash \mathtt{sample} : \mathtt{T} \, \ttreal}{}
\end{equation*}
where $\mathtt{T}$ is the monad of Moggi's meta-language.
Any type 
$\mathtt{A}$ of $\PCFSS$ is translated into the type
$\mathtt{A}^{\sharp}$ defined as follows:
\begin{equation*}
  \ttunit^{\sharp}
  = \ttunit,
  \quad
  \ttreal^{\sharp}
  = \ttreal,
  \quad
  (\mathtt{A} \to \mathtt{B})^{\sharp}
  = \mathtt{A}^{\sharp} \to
  \mathtt{T}\,\mathtt{B}^{\sharp}.
\end{equation*}
Terms $\mathtt{sample}$ and $\mathtt{score}(-)$ in $\PCFSS$ are
translated into $\mathtt{sample}$ and $\mathtt{score}(-)$ in
Moggi's meta-language respectively. See \cite{moggi1991} for more
detail about Moggi's translation.

\subsubsection{Girard Translation}

We next translate the extended Moggi's meta-language into an extension
of the linear $\lambda$-calculus, by way of the so-called Girard
translation~\cite{girard1987}. Types are given by
\begin{equation*}
  \mathtt{A},\mathtt{B}::= \ttunit \midd \ttreal \midd \mathtt{State}
  \midd \mathtt{A}^{\bot}
  \midd \mathtt{A} \otimes \mathtt{B} \midd \mathtt{A} \mathbin{\wp} \mathtt{B}
  \midd \oc \mathtt{A}
\end{equation*}
where $\ttunit$, $\ttreal$ and $\mathtt{State}$ are base types, and
terms are generated by the standard term constructors of the linear
$\lambda$-calculus, plus the following rules:
\begin{equation*}
  \begin{array}{c}
    \infer{\mathtt{\Delta} \vdash \mathtt{r}_{a}:\ttreal}{a \in \mathbb{R}} \\[6pt]
    \infer{\mathtt{\Delta} \vdash
    \mathtt{score}(\mathtt{M}): \mathtt{State} \multimap
    \mathtt{State} \otimes \oc \ttunit}{\mathtt{\Delta} \vdash \mathtt{M}:\oc\ttreal}\\[6pt]
    \infer{\mathtt{\Delta} \vdash \mathtt{sample}:\mathtt{State} \multimap \mathtt{State}
    \otimes \oc \ttreal}{} 
  \end{array}
\end{equation*}
(as customary in linear logic, $\mathtt{A} \multimap \mathtt{B}$ is an abbreviation of
$\mathtt{A}^{\bot} \mathbin{\wp} \mathtt{B}$). These typing rules are
derived from the following translation $(-)^{\flat}$ of types of the
extended Moggi's meta-language into types of the extended linear
$\lambda$-calculus:
\begin{equation*}
  \begin{array}{c}
    \ttunit^{\flat}
    = \ttunit,
    \hspace{8pt}
    \ttreal^{\flat}
    = \ttreal,
    \hspace{8pt}
    (\mathtt{A} \to \mathtt{B})^{\flat}
    =
    \oc \mathtt{A}^{\flat} \multimap
    \mathtt{B}^{\flat},
    \\[3pt]
    (\mathtt{T}\,\mathtt{A})^{\flat}
    =
    \mathtt{State} \multimap
    \mathtt{State} \otimes \oc \mathtt{A}^{\flat}
  \end{array}
\end{equation*}
The definition of $(\mathtt{T}\,\mathtt{A})^{\flat}$ is motivated by
the following categorical observation: let $\mathcal{L}$ be the
syntactic category of the extended linear $\lambda$-calculus,
which is a symmetric monoidal closed category endowed with a comonad
$\oc \colon \mathcal{L} \to \mathcal{L}$ with certain coherence
conditions (see e.g. \cite{hs2003}), and let $\mathcal{L}_{\oc}$ be
the coKleisli category $\mathcal{L}_{\oc}$ of the comonad $\oc$. Then,
by composing the adjunction between $\mathcal{L}$ and
$\mathcal{L}_{\oc}$ with a state monad
$\mathsf{State} \multimap \mathsf{State} \otimes (-)$ on
$\mathcal{L}$, we obtain a monad on $\mathcal{L}_{\oc}$:
\begin{equation*}
  \xymatrix{
    \mathcal{L} \ar@/^{3mm}/[r] \ar@{}[r]|{\top}
    \ar@(dl,ul)[]^{\mathtt{State} \multimap \mathtt{State} \otimes (-)}&
    \mathcal{L}_{\oc} \ar@/^{3mm}/[l]
  },
\end{equation*}
which sends an object $\mathtt{A} \in \mathcal{L}_{\oc}$ to
$\mathtt{State} \multimap \mathtt{State} \otimes \oc \mathtt{A}$. This
use of the state monad is motivated by sampling-based operational
semantics: we can regard $\PCFSS$ as a call-by-value
$\lambda$-calculus with global states consisting of pairs of a
non-negative real number and a finite sequence of real numbers, and we
can regard $\mathtt{score}$ and $\mathtt{sample}$ as effectful
operations interacting with those states.

\subsubsection{The Third Step}

We translate terms in the extended linear $\lambda$-calculus into (an
extension of proof structures) \cite{lafont1995}, which are graphical
presentations of type derivation trees of linear $\lambda$-terms. We
can also understand proof structures as string diagrams for compact
closed categories \cite{selinger2011}. Operators of the pure, linear,
$\lambda$-calculus, can be translated as usual \cite{lafont1995}. For
example, type derivation trees
\begin{equation*}
  \infer{\mathtt{x}:\mathtt{A} \vdash \mathtt{x}:\mathtt{A}}{},
  \qquad
  \infer{\vdash \lambda \mathtt{x}^{\mathtt{A}}.\,
    \mathtt{x}:\mathtt{A} \multimap \mathtt{A}}{
    \mathtt{x}:\mathtt{A} \vdash \mathtt{x}:\mathtt{A}}
  \qquad
  \infer{\vdash \mathtt{M} \otimes \mathtt{N}:\mathtt{A} \otimes \mathtt{B}}{
    \vdash \mathtt{M}:\mathtt{A}&\vdash \mathtt{N}:\mathtt{B}}
\end{equation*}
are translated into proof structures
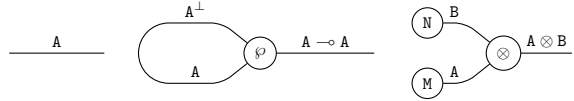
\begin{center}
  \begin{tikzpicture}[rounded corners]
    \begin{scope}[xshift=5.5cm]
      \node (M) [draw,circle,obj] at (0,0) {$\mathtt{M}$};
      \node (N) [draw,circle,obj] at (0,0.8) {$\mathtt{N}$};
      \node (ox) [draw,circle,obj] at (1,0.4) {$\otimes$};
      \draw (M) -- node[above,obj]{$\mathtt{A}$} ++ (0.5,0) -- (ox);
      \draw (N) -- node[above,obj]{$\mathtt{B}$} ++ (0.5,0) -- (ox);
      \draw (ox) -- node[above,obj]{$\mathtt{A} \otimes \mathtt{B}$} ++ (0.9,0);
    \end{scope}
    \begin{scope}[xshift=2.1cm]
      \node (wp) [draw,circle,obj] at (1.2,0.4) {$\wp$};
      \draw (0,0) -- node[above,obj] {$\mathtt{A}$}++ (0.7,0) -- (wp);
      \draw (0,0.8) -- node[above,obj] {$\mathtt{A}^{\bot}$}++ (0.7,0) -- (wp);
      \draw (0,0.8) arc(90:270:0.4);
      \draw (wp) -- node[above,obj] {$\mathtt{A} \multimap \mathtt{A}$} ++ (1.5,0);
    \end{scope}
    \draw (0,0.4) -- node[above,obj] {$\mathtt{A}$}++ (1.25,0);
  \end{tikzpicture}
\end{center}
respectively where nodes labelled with $\mathtt{M}$ and $\mathtt{N}$
are proof structures associated to type derivations of $\mathtt{M}$
and $\mathtt{N}$. Terms of the form  $\mathtt{r}_{a}$,
$\mathtt{sample}(\mathtt{M})$ and $\mathtt{score}$, require
new kinds of nodes:
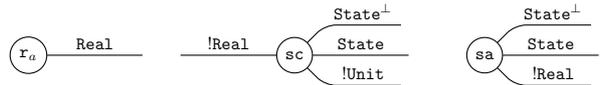
\begin{center}
  \begin{tikzpicture}[rounded corners]
    \node (a) [draw,circle,obj] at (-3,0) {$\mathtt{r}_{a}$};
    \draw (a) -- node[above,obj] {$\ttreal$} ++(1.5,0);
    \node (sample) [draw,circle,obj] at (3,0) {$\mathtt{sa}$};
    \draw (sample) -- ++(0.4,-0.4) -- node[above,obj] {$\oc \ttreal$} ++(1,0);
    \draw (sample) -- ++(0.4,0.4) -- node[above,obj] {$\mathtt{State}^{\bot}$} ++(1,0);
    \draw (sample) -- node[above,obj] {$\mathtt{State}$} ++(1.5,0);
    \node (score) [draw,circle,obj] at (0.5,0) {$\mathtt{sc}$};
    \draw (score) -- node[above,obj] {$\mathtt{State}$} ++ (1.5,0);
    \draw (score) -- ++(0.4,0.4) -- node[above,obj] {$\mathtt{State}^{\bot}$} ++ (1,0);
    \draw (score) -- ++(0.4,-0.4) -- node[above,obj] {$\oc \ttunit$} ++ (1,0);
    \draw (score) -- node[above,obj] {$\oc\ttreal$} ++ (-1.5,0);
  \end{tikzpicture}
  .
\end{center}
This is not a direct adaptation of typing rules for $\mathtt{score}$
and $\mathtt{sample}$ in the linear $\lambda$-calculus, but the
correspondence can be recovered by way of multiplicatives:
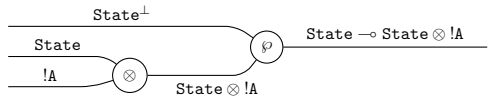
\begin{center}
  \begin{tikzpicture}[rounded corners]
    \node (oxscore) [draw,circle,obj] at (2,-0.25) {$\otimes$};
    \node (wpscore) [draw,circle,obj] at (3.8,0.125) {$\wp$};
    \draw (0.4,0) -- node[above,obj] {$\mathtt{State}$} ++ (1.3,0) -- (oxscore);
    \draw (0.4,0.4) -- node[above,obj] {$\mathtt{State}^{\bot}$} ++ (3,0) -- (wpscore);
    \draw (0.4,-0.4) -- node[above,obj] {$\oc \mathtt{A}$} ++ (1.1,0) -- (oxscore);
    \draw (oxscore) -- ++(0.8,0) -- node[below,obj] {$\mathtt{State} \otimes \oc \mathtt{A}$}
    ++(0.7,0) -- (wpscore)
    -- node[above,obj] {$\mathtt{State} \multimap
      \mathtt{State} \otimes \oc \mathtt{A}$} ++ (2.9,0);
  \end{tikzpicture}
  .
\end{center}

\subsection{From Proof Structures to Mealy Machines}
\label{sec:pn}

The series of translations from $\PCFSS$ to proof structures is
agnostic as for the computational meaning of $\mathtt{score}$ and
$\mathtt{sample}$ in $\PCFSS$ because $\mathtt{score}$ and
$\mathtt{sample}$ introduced in these translations are just constant
symbols. In other words, the translation from $\PCFSS$ to the extended
proof structures is not sound with respect to either form of
operational semantics for $\PCFSS$. In the last translation step, we
assign proof structures a computational meaning, respecting the
operational semantics of the underlying $\PCFSS$ term.

We do this by associating proof structures with Mealy machines. A
\emph{Mealy machine} is an input/output-machine whose evolution may
depend on its current state. In this paper, for the sake of supporting
intuition and of enabling graphical reasoning, we depict a Mealy
machine $\mathsf{M}$ as a node with some input/output-ports:
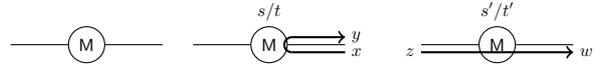
\begin{center}
  \begin{tikzpicture}
    \node (M) [draw,circle,obj] at (0,0) {$\mathsf{M}$};
    \draw (M) -- ++ (1,0);
    \draw (M) -- ++ (-1,0);
    \begin{scope}[xshift=2.4cm]
      \node (M) [draw,circle,obj] at (0,0) {$\mathsf{M}$};
      \draw (M) -- ++ (1,0);
      \draw (M) -- ++ (-1,0);
      \draw[thick] (1,-0.1) node[right,obj] {$x$} -- ++(-0.7,0);
      \draw[<-,thick] (1,0.1) node[right,obj] {$y$} -- ++(-0.7,0)
      arc (90:270:0.1);
      \node [above,obj] at (M.north) {$s/t$};
    \end{scope}
    \begin{scope}[xshift=5.4cm]
      \node (M) [draw,circle,obj] at (0,0) {$\mathsf{M}$};
      \draw (M) -- ++ (1,0);
      \draw (M) -- ++ (-1,0);
      \draw[thick,->] (-1,-0.1) node[left,obj] {$z$} -- ++(2,0) node[right,obj]{$w$};
      \node [above,obj] at (M.north) {$s'/t'$};
    \end{scope}
  \end{tikzpicture}
  .
\end{center}
For example, the thick arrow in the middle diagram indicates that if
the current state is $s$ and the given input is $x$, then the Mealy
machine outputs $y$ and changes its state to $t$. In the GoI jargon,
data traveling along edges of proof structures are often called
\emph{tokens}.

For the standard proof structures, we can follow \cite{laurent2001}
where Mealy machines associated with proof structures are built up
from Mealy machines associated to each nodes. For example, the
following nodes
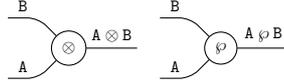
\begin{center}
  \begin{tikzpicture}[rounded corners]
    \node (x) [draw,circle,obj] at (0.8,0.4) {$\otimes$};
    \draw (0,0) -- node[above,obj] {$\mathtt{A}$} ++(0.4,0) -- (x);
    \draw (0,0.8) -- node[above,obj] {$\mathtt{B}$} ++(0.4,0) -- (x);
    \draw (x) -- node[above,obj] {$\mathtt{A} \otimes \mathtt{B}$} ++ (0.9,0);
    \begin{scope}[xshift=2cm]
      \node (v) [draw,circle,obj] at (0.8,0.4) {$\wp$};
      \draw (0,0) -- node[above,obj] {$\mathtt{A}$} ++(0.4,0) -- (v);
      \draw (0,0.8) -- node[above,obj] {$\mathtt{B}$} ++(0.4,0) -- (v);
      \draw (v) -- node[above,obj] {$\mathtt{A} \mathbin{\wp} \mathtt{B}$} ++ (0.9,0);
    \end{scope}
  \end{tikzpicture}
\end{center}
are both associated with a one-state Mealy machine that behaves in
the following manner:
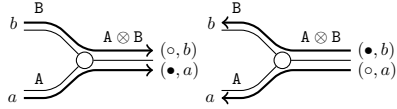
\begin{center}
  \begin{tikzpicture}[rounded corners]
    \begin{scope}[xshift=3.2cm]
      \node (x) [draw,circle,obj] at (0.8,0.4) {};
      \draw (0,0) -- node[above,obj] {$\mathtt{A}$} ++(0.4,0) -- (x);
      \draw (0,0.8) -- node[above=4,obj] {$\mathtt{B}$} ++(0.4,0) -- (x);
      \draw (x) -- node[above=4,obj] {$\mathtt{A} \mathbin{\otimes} \mathtt{B}$} ++ (0.9,0);
      \draw[thick,->] (0,0.9) node[left,obj] {$b$}
      -- ++(0.4,0) -- (0.8,0.53) -- ++(0.9,0) node[right,obj] {$(\mm,b)$};
      \draw[thick,->] (0,-0.1) node[left,obj] {$a$}
      -- ++(0.4,0) -- (0.8,0.27) -- ++(0.9,0) node[right,obj] {$(\hh,a)$};
    \end{scope}
    \begin{scope}[xshift=5.8cm]
      \node (x) [draw,circle,obj] at (0.8,0.4) {};
      \draw (0,0) -- node[above,obj] {$\mathtt{A}$} ++(0.4,0) -- (x);
      \draw (0,0.8) -- node[above=4,obj] {$\mathtt{B}$} ++(0.4,0) -- (x);
      \draw (x) -- node[above=4,obj] {$\mathtt{A} \mathbin{\otimes} \mathtt{B}$} ++ (0.9,0);
      \draw[thick,<-] (0,0.9) node[left,obj] {$b$}
      -- ++(0.4,0) -- (0.8,0.53) -- ++(0.9,0) node[right,obj] {$(\hh,b)$};
      \draw[thick,<-] (0,-0.1) node[left,obj] {$a$}
      -- ++(0.4,0) -- (0.8,0.27) -- ++(0.9,0) node[right,obj] {$(\mm,a)$};
    \end{scope}
  \end{tikzpicture}
  .
\end{center}
Namely, the Mealy machine forwards each input from the left hand side
to the right hand side endowing it with a tag that tells where the
token came from. The Mealy machine handles inputs from the right
hand side in the reverse way.

Soundness of Mealy machine semantics states that if two (pure) linear
$\lambda$-terms are $\beta$-equivalent, then the behaviours of the Mealy
machines associated to these terms are the same. As an example, let us
consider a $\beta$-reduction step
\begin{math}
  (\lambda \mathtt{x}^{\mathtt{A}}.\,\mathtt{x})\, \mathtt{y} \to
  \mathtt{y}.
\end{math}
The proof structure associated to
$(\lambda \mathtt{x}^{\mathtt{A}}.\,\mathtt{x})\, \mathtt{y}$ is the
graph in the left hand side, and the arrow in the right hand side
illustrates a trace of a run of this Mealy machine for an input $a$ from
the right edge:
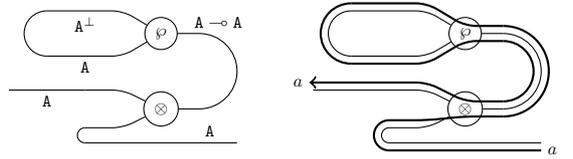
\begin{center}
  \begin{tikzpicture}[rounded corners]
    \node (wp)[draw,circle,obj] at (1,1) {$\wp$};
    \node (ox)[draw,circle,obj] at (1,0) {$\otimes$};
    \draw (-0.5,0.7) -- node[below,obj] {$\mathtt{A}$} ++(1,0) -- (wp);
    \draw (-0.5,1.3) -- node[below,obj] {$\mathtt{A}^{\bot}$} ++ (1,0) -- (wp);
    \draw (-0.5,1.3) arc(90:270:0.3);
    \draw (0,0.25) -- ++(0.5,0) -- (ox);
    \draw (ox) -- ++(0.5,0);
    \draw (wp) -- node[above right,obj] {$\mathtt{A}\multimap\mathtt{A}$} ++(0.5,0);
    \draw (1.5,1) arc(90:-90:0.5);
    \draw (0,-0.25) -- ++(0.5,0) -- (ox);
    \draw (0,-0.25) arc(90:270:0.1);
    \draw (0,-0.45) -- ++ (1,0) -- node[above right,obj] {$\mathtt{A}$} ++(1,0);
    \draw (0,0.25) -- node[below,obj] {$\mathtt{A}$} ++(-1,0);
    \begin{scope}[xshift=4cm]
      \node (wp)[draw,circle,obj] at (1,1) {$\wp$};
      \node (ox)[draw,circle,obj] at (1,0) {$\otimes$};
      \draw (-0.5,0.7) -- ++(1,0) -- (wp);
      \draw (-0.5,1.3) -- ++ (1,0) -- (wp);
      \draw (-0.5,1.3) arc(90:270:0.3);

      \draw (0,0.25) -- ++(0.5,0) -- (ox);
      \draw (ox) -- ++(0.5,0);
      \draw (wp) -- ++(0.5,0);
      \draw (1.5,1) arc(90:-90:0.5);
      \draw (0,-0.25) -- ++(0.5,0) -- (ox);
      \draw (0,-0.25) arc(90:270:0.1);
      \draw (0,-0.45) -- ++ (1,0) -- ++(1,0);
      \draw (0,0.25) -- ++(-1,0);

      \draw[thick] (0,-0.55) -- ++ (2,0) node[right,obj] {$a$};
      \draw[thick] (0,-0.55) arc(270:90:0.2);
      \draw[thick] (0,-0.15) -- ++(0.5,0) -- (1,-0.1) -- ++ (0.5,0);
      \draw[thick] (1.5,-0.1) arc(-90:90:0.6);
      \draw[thick] (1.5,1.1) -- (1,1.1) -- ++(-0.5,0.3) -- ++(-1,0);
      \draw[thick] (-0.5,1.4) arc(90:270:0.4);
      \draw[thick] (-0.5,0.6) -- ++(1,0) -- (1,0.9) -- ++(0.5,0);
      \draw[thick] (1.5,0.9) arc(90:-90:0.4);
      \draw[thick,->] (1.5,0.1) -- ++(-0.5,0) -- (0.5,0.35) -- ++(-0.3,0)
      -- ++(-1.25,0) node[left,obj] {$a$};
    \end{scope}
  \end{tikzpicture}
  .
\end{center}
This Mealy machine forwards any input from the right hand side to
the left hand side as indicated by the thick arrow, and it also
forwards any input from the left hand side to the right hand side.
Hence, the behaviour of this Mealy machine is equivalent to the
behaviour of the following trivial Mealy machine:
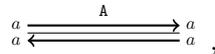
\begin{center}
  \begin{tikzpicture}
    \draw (0,0) -- node[above=4,obj] {$\mathtt{A}$} ++(2,0);
    \draw[thick,->] (0,0.1) node[left,obj] {$a$} -- ++(2,0) node [right,obj] {$a$};
    \draw[thick,->] (2,-0.1) node[right,obj] {$a$} -- ++(-2,0) node [left,obj] {$a$};
  \end{tikzpicture}
  ,
\end{center}
which is the interpretation of $\mathtt{y}:\mathtt{A} \vdash
\mathtt{y}:\mathtt{A}$. This is in fact a symptom of a general
phenomenon: Mealy machine semantics for the linear $\lambda$-calculus
captures $\beta$-reduction
$(\lambda\mathtt{x}^{\mathtt{A}}.\,\mathtt{x})\,\mathtt{y} \to
\mathtt{y}$.

But how can we extend this Mealy machine semantics to
$\mathtt{score}$ and $\mathtt{sample}$? Here, we borrow the idea from
Game semantics \cite{am1996} that models computation in terms of
interaction between programs and environments.
For scoring and sampling, we can infer
how they interact with the environment from sampling-based operational
semantics. For scoring, we associate $\mathtt{score}$
with a one-state Mealy machine that has the following
transitions:
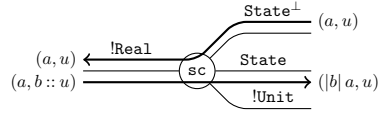
\begin{center}
  \begin{tikzpicture}[rounded corners]
    \node (score) [draw,circle,obj] at (0.5,0) {$\mathtt{sc}$};
    \draw (score) -- node[above,obj] {$\mathtt{State}$} ++ (1.5,0);
    \draw (score) -- ++(0.5,0.5) -- node[above=4,obj] {$\mathtt{State}^{\bot}$} ++ (1,0);
    \draw (score) -- ++(0.5,-0.5) -- node[above,obj] {$\oc \ttunit$} ++ (1,0);
    \draw (score) -- node[above=4,obj] {$\oc\ttreal$} ++ (-1.5,0);

    \draw[thick,->] (2,0.65) node[right,obj] {$(a,u)$}
    -- ++(-1,0) -- ++(-0.5,-0.5) -- ++ (-1.5,0) node[left,obj] {$(a,u)$};
    \draw[thick,->] (-1,-0.15) node[left,obj] {$(a,b\mathbin{::}u)$}
    -- (2,-0.15) node[right,obj] {$(|b|\,a,u)$};
  \end{tikzpicture}
\end{center}
where $u$ is a finite sequence of real numbers and $a,b$ are real
numbers such that $a \geq 0$. We can read these transitions as
follows: for each ``configuration'' $(-,a,u)$, the Mealy machine sends
a query $(a,u)$ to environment in order to know the value of its
argument, and if environment answers that the value is $b$, i.e., if
the Mealy machine receives $(a,b\mathbin{::}u)$, then it outputs
$(|b|\,a,u)$, which is the evaluation result of
$(\mathtt{score}(\mathtt{r}_{b}),a,u)$.

For sampling, we associate $\mathtt{sample}$ with a Mealy machine that
has the following transitions:
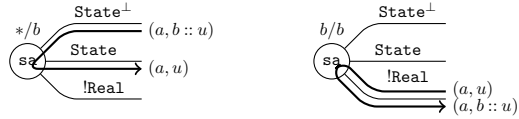
\begin{center}
  \begin{tikzpicture}[rounded corners]
    \node (sample) [draw,circle,obj] at (3,0) {$\mathtt{sa}$};
    \node [obj] at (3,0.4) {$\ast/b$};

    \draw (sample) -- ++(0.5,-0.5) -- node[above,obj] {$\oc \ttreal$} ++(1,0);
    \draw (sample) -- node[above,obj] {$\mathtt{State}$} ++(1.5,0);
    \draw (sample) -- ++(0.5,0.5) -- node[above,obj] {$\mathtt{State}^{\bot}$} ++(1,0);

    \draw[thick,->] (4.5,0.4) node[right,obj] {$(a,b\mathbin{::}u)$}
    -- ++(-1,0) -- ++ (-0.5,-0.5) -- ++ (1.5,0)
    node[right,obj] {$(a,u)$};

    \begin{scope}[xshift=4cm]
      \node (sample) [draw,circle,obj] at (3,0) {$\mathtt{sa}$};
      \node [obj] at (3,0.4) {$b/b$};

      \draw (sample) -- ++(0.5,-0.5) -- node[above=4,obj] {$\oc \ttreal$} ++(1,0);
      \draw (sample) -- node[above,obj] {$\mathtt{State}$} ++(1.5,0);
      \draw (sample) -- ++(0.5,0.5) -- node[above,obj] {$\mathtt{State}^{\bot}$} ++(1,0);

      \draw[thick,->] (4.5,-0.4) node[right,obj] {\raisebox{5pt}{$(a,u)$}}
      -- ++(-0.95,0) -- ++(-0.4,0.4) -- ++ (-0.15,-0.15) -- ++ (0.45,-0.45) -- ++ (1.05,0)
      node[right,obj] {\raisebox{-5pt}{$(a,b\mathbin{::}u)$}};
    \end{scope}
  \end{tikzpicture}
\end{center}
where $u$ is a finite sequence of real numbers and $a,b$ are real
numbers such that $a \geq 0$. The first transition means that in the
initial state $\ast$, given a ``configuration''
$(-,a,b \mathbin{::} u)$, the Mealy machine pops the first element of
$b \mathbin{::} u$ and memorises the value $b$ by changing its state
from $\ast$ to $b$. After this transition, for any query $(a,u)$
asking the result of sampling, it answers the value memorised in the
first transition.

For example, a Mealy machine
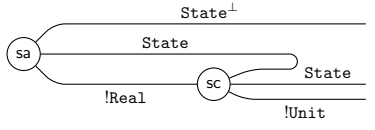
\begin{center}
  \begin{tikzpicture}[rounded corners]
    \node (sa) [draw,circle,obj] at (0,0) {$\mathsf{sa}$};
    \node (sc) [draw,circle,obj] at (2.5,-0.4) {$\mathsf{sc}$};
    \draw (sa) -- ++(0.4,-0.4) -- node[below,obj]{$\oc\ttreal$} (sc);
    \draw (sa) -- ++(0.4,0.4) -- node[above,obj]{$\mathtt{State}^{\bot}$} ++(4.1,0);
    \draw (sa) -- node[above,obj]{$\mathtt{State}$}++ (3.5,0);
    \draw (3.5,0) arc(90:-90:0.1);
    \draw (3.5,-0.2) -- ++(-0.5,0) -- (sc);
    \draw (sc) -- ++(0.4,-0.2) -- node[below,obj]{$\oc\ttunit$}++(1.6,0);
    \draw (sc) -- ++(1,0) -- node[above,obj]{$\mathtt{State}$} ++(1,0);
  \end{tikzpicture}
  ,
\end{center}
which is a denotation of the term
\begin{equation*}
  \mathtt{M}=\letin{\mathtt{x}}{\mathtt{sample}}{\mathtt{score}(\mathtt{x})}, 
\end{equation*}
and behaves as follows:
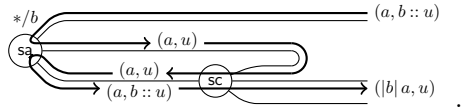
\begin{center}
  \begin{tikzpicture}[rounded corners]
    \node (sa) [draw,circle,obj] at (0,0) {$\mathsf{sa}$};
    \node (sc) [draw,circle,obj] at (2.5,-0.4) {$\mathsf{sc}$};
    \draw (sa) -- ++(0.4,-0.4) -- (sc);
    \draw (sa) -- ++(0.4,0.4) -- ++(4.1,0);
    \draw (sa) -- ++ (3.5,0);
    \draw (3.5,0) arc(90:-90:0.1);
    \draw (3.5,-0.2) -- ++(-0.5,0) -- (sc);
    \draw (sc) -- ++(0.4,-0.3) -- ++(1.6,0);
    \draw (sc) -- ++(1,0) -- ++(1,0);

    \node (1) [obj] at (2,0.1) {\raisebox{5pt}{$(a,u)$}};
    \node (2) [obj] at (1.5,-0.3) {\raisebox{5pt}{$(a,u)$}};
    \node (3) [obj] at (1.5,-0.5) {\raisebox{-12pt}{$(a,b\mathbin{::}u)$}};
    \draw[thick,->] (4.5,0.5) node[right,obj] {$(a,b\mathbin{::}u)$}
    -- ++(-4.1,0) -- ++ (-0.4,-0.4) -- (1);
    \draw[thick] (1) -- ++ (1.5,0);
    \draw[thick] (3.5,0.1) arc(90:-90:0.2);
    \draw[thick,->] (3.5,-0.3)--++(-1,0)--++(-0.3,0) -- (2);
    \draw[thick,->] (2) -- ++ (-1.05,0) -- ++(-0.3,0.3)
    -- ++ (-0.15,-0.15) -- ++ (0.35,-0.35) -- (3);
    \draw[thick,->] (3) -- ++ (1,0) -- ++ (0.4,0) -- ++ (1.6,0) node[right,obj]{$(|b|\,a,u)$};

    \node [above,obj] at (sa.north) {$\ast/b$};
  \end{tikzpicture}
  .
\end{center}
Our adequacy theorem says that the evaluation result of a term coincides
with the execution result of the associated Mealy machine. In fact,
for this case, the outcome $(|b|\,a,u)$ of the above Mealy machine is
equal to the evaluation result of $(\mathtt{M},a,b\mathbin{::}u)$,
that is,
$(\mathtt{M},a,b\mathbin{::}u) \to^{\ast} (\mathtt{skip},|b|\,a,u)$.
In this interaction process, the memoisation mechanism of the
$\mathsf{sa}$-node is necessary, otherwise the $\mathsf{sa}$-node
can not tell the $\mathsf{sc}$-node that the result of sampling is
$b$.

\begin{remark}
  Two notions of \emph{state} (the one coming from
  the state monad and the one of the of the Mealy
  machine itself) are used for different purpose here: the first notion is needed
  to model the call-by-value evaluation strategy where we need to
  store intermediate effects that are invoked during the
  evaluation. The second notion of state is needed to model sampling.
  More concretely, each Mealy machine for sampling need to remember
  the already sampled values in the current probabilistic
  branch.
\end{remark}

\section{Mealy Machines and their Compositions}
\label{sec:mealy}

After having described Mealy machine semantics briefly and
informally, it is now time to get more formal. In this section, we introduce
the notion of a Mealy machine and some constructions on
Mealy machines. We also introduce a way of diagramatically presenting Mealy
machines which is behaviourally sound.

\subsection{Mealy Machines, Formally}
\label{sec:diagram}

In this paper, we call a pair of measurable spaces an
\emph{$\mathbf{Int}$-object}. We use sans-serif capital letters
$\mathsf{X}, \mathsf{Y},\mathsf{Z},\ldots$ to denote
$\mathbf{Int}$-objects, and we denote the positive/negative part of an
$\mathbf{Int}$-object by the same italic letter superscripted by
$+/-$. For example, $\mathsf{X}$ denotes an $\mathbf{Int}$-object
$(X^{+},X^{-})$ consisting of two measurable spaces $X^{+}$ and $X^{-}$.
The name ``$\mathbf{Int}$-object'' comes from the so-called
\emph{$\mathbf{Int}$-construction} \cite{jsv}.
Definition~\ref{def:pmm} and the definition of monoidal products in
Section~\ref{sec:const} are also motivated by
$\mathbf{Int}$-construction.

\begin{definition}\label{def:pmm}
  For $\mathbf{Int}$-objects $\mathsf{X}$ and $\mathsf{Y}$, a
  \emph{Mealy machine} $\mathsf{M}$ from $\mathsf{X}$ to
  $\mathsf{Y}$ consists of
  \begin{varitemize}
  \item a measurable space $\state{\mathsf{M}}$ called the \emph{state
      space} of $\mathsf{M}$;
  \item an element $\init{\mathsf{M}} \in \state{\mathsf{M}}$ called
    the \emph{initial state} of $\mathsf{M}$;
  \item a partial measurable function
    \begin{equation*}
      \tran{\mathsf{M}} \colon (X^{+} + Y^{-}) \times
      \state{\mathsf{M}} \to (Y^{+} + X^{-}) \times
      \state{\mathsf{M}} 
    \end{equation*}
    called the \emph{transition function}.
  \end{varitemize}
  If $\mathsf{M}$ is a Mealy machine from $\mathsf{X}$ to
  $\mathsf{Y}$, we write $\mathsf{M} \colon \mathsf{X} \multimap \mathsf{Y}$.
\end{definition}

The transition function $\tran{\mathsf{M}}$ of a
Mealy machine $\mathsf{M}$ describes a mapping between inputs and
outputs which can also alter the underlying state. For
$x \in X^{+} + Y^{-}$ and $s \in \state{\mathsf{M}}$,
$\tran{\mathsf{M}}(x,s) =(y,t)$ means that when the current state of
$\mathsf{M}$ is $s$, given an input $x$, there is an output $y$ and
the next state is $t$.

Readers may wonder why $X^{-}$ appears in the target and $Y^{-}$
appears in the source of the transition function of a Mealy machine
from $\mathsf{X}$ to $\mathsf{Y}$. In short, this is because we are
interested in Mealy machines that handle bidirectional computation.
The diagrammatic presentation of Mealy machines clarifies the meaning of
``bidirectional.'' Let
$\mathsf{M} \colon \mathsf{X} \multimap \mathsf{Y}$ be a Mealy
machine. In this paper, we depict $\mathsf{M}$ as follows:
\begin{center}
  \begin{tikzpicture}
    \node (M) [draw,circle,obj] at (0,0) {$\mathsf{M}$};
    \draw (M) -- node[above,obj] {$\mathsf{Y}$} ++ (2,0);
    \draw (M) -- node[above,obj] {$\mathsf{X}$} ++ (-2,0);
  \end{tikzpicture}
  .
\end{center}
Intuitively, each label on an edge indicates the type of data
traveling along the edge. Namely, on the $\mathsf{X}$-edge (on the
$\mathsf{Y}$-edge), elements in $X^{+}$ (in $Y^{+}$) go from left to
right, and elements in $X^{-}$ (in $Y^{-}$) go from right to left. For
example, we depict the following transitions
\begin{align*}
  \tran{\mathsf{M}}((\mm,y),s_{0}) = ((\mm,x),s_{1}),
                                     \hspace{5pt}
  \tran{\mathsf{M}}((\mm,y'),s_{0}) = ((\hh,y''),s_{2})
\end{align*}
for some $y,y' \in Y^{-}$, $x \in X^{-}$,
$y'' \in Y^{+}$ and $s_{0},s_{1},s_{2} \in \state{\mathsf{M}}$ as the
following thick arrows
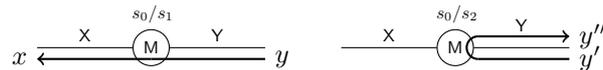
\begin{center}
  \begin{tikzpicture}
    \node (M) [draw,circle,obj] at (0,0) {$\mathsf{M}$};
    \node [obj] at (0.02,0.45) {$s_{0}/s_{1}$};
    \draw (M) -- node[above,obj] {$\mathsf{Y}$} ++ (1.5,0);
    \draw (M) -- node[above,obj] {$\mathsf{X}$} ++ (-1.5,0);
    \draw[thick] (1.5,-0.15) node[right,label] {$y$} -- ++ (-1.5,0);
    \draw[thick,->] (0,-0.15) -- ++ (-1.5,0) node[left,label] {$x$};

    \begin{scope}[xshift=4cm]
      \node (M) [draw,circle,obj] at (0,0) {$\mathsf{M}$};
      \node [obj] at (0.02,0.45) {$s_{0}/s_{2}$};
      \draw (M) -- node[above=4,obj] {$\mathsf{Y}$} ++ (1.5,0);
      \draw (M) -- node[above,obj] {$\mathsf{X}$} ++ (-1.5,0);
      \draw[thick] (1.5,-0.15) node[right,label] {$y'$} -- ++ (-1.2,0);
      \draw[thick,<-] (1.5,0.15) node[right,label] {$y''$} -- ++ (-1.2,0);
      \draw[thick] (0.3,0.15) arc (90:270:0.15);
    \end{scope}
  \end{tikzpicture}
  .
\end{center}
(Recall that the white/black bullet indicates the left/right part of
the disjoint sum.) 
The expressions $s_{0}/s_{1}$ and $s_{0}/s_{2}$ on the Mealy machine
$\mathsf{M}$ stands for transitions of states. We omit states
transitions when we can infer them.

We will give some Mealy machines whose state spaces are trivial,
namely $1$. We call such a Mealy machine \emph{token machine}. Our
usage of the term token machine is along the lines of that in other
papers on GoI such as \cite{mackie1995,laurent2001}. Since we can
identify the transition function of a token machine $\mathsf{M} \colon
\mathsf{X} \multimap \mathsf{Y}$ with the following partial measurable
function
\begin{equation*}
  X^{+} + Y^{-} \cong (X^{+} + Y^{-}) \times 1
  \xrightarrow{\tran{\mathsf{M}}} (Y^{+} + X^{-}) \times 1
  \cong Y^{+} + X^{-},
\end{equation*}
giving partial measurable function of this type
is enough to specify a token machine.
\begin{convention}
  We define a token machine
  $\mathsf{M} \colon \mathsf{X} \multimap \mathsf{Y}$ by giving a
  partial measurable function from $X^{+} + Y^{-}$ to $Y^{+} + X^{-}$,
  and we also call this partial measurable function \emph{transition
    function} of $\mathsf{M}$. Abusing notation, we write
  $\tran{\mathsf{M}}$ for this transition function.
\end{convention}

\subsection{Behavioural Equivalence}
\label{sec:bhe}

We are now ready to give an equivalence relation between Mealy machines
which identifies machines which \emph{behave the same way}. Identifying
Mealy machines in terms of their behaviour is important to reason about
compositions of Mealy machines in the following part of this paper.
Here, we are inspired by behavioural equivalence from coalgebraic theory of
modelling transition systems \cite{jacobs2016}.

Let $\mathsf{M}$ and $\mathsf{N}$ be Mealy machines from $\mathsf{X}$
to $\mathsf{Y}$. We write
$\mathsf{M} \preceq_{\mathsf{X},\mathsf{Y}} \mathsf{N}$ when there is
a measurable function
$f \colon \state{\mathsf{M}} \to \state{\mathsf{N}}$ satisfying
$f(\init{\mathsf{M}}) = \init{\mathsf{N}}$ and
\begin{equation*}
  \xymatrix@C=10mm@R=5mm{
    (X^{+} + Y^{-}) \times \state{\mathsf{M}}
    \ar[r]^-{\id \times f}
    \ar[d]_{\tran{\mathsf{M}}}
    & (X^{+} + Y^{-}) \times \state{\mathsf{N}}
    \ar[d]^{\tran{\mathsf{N}}}
    \\
    (Y^{+} + X^{-}) \times \state{\mathsf{M}}
    \ar[r]^-{\id \times f}
    &
    (Y^{+} + X^{-}) \times \state{\mathsf{N}}
    \nulldot
  }
\end{equation*}
The definition means that if we have $\mathsf{M}
\preceq_{\mathsf{X},\mathsf{Y}} \mathsf{N}$, then no observer can
distinguish between $\mathsf{M}$ and $\mathsf{N}$ from
their input/output behaviour, although their internal structure can be
quite different. We define an equivalence relation
$\simeq_{\mathsf{X},\mathsf{Y}}$ to be the reflective symmetric
transitive closure of $\preceq_{\mathsf{X},\mathsf{Y}}$. Below, if we
can infer the subscript $\mathsf{X},\mathsf{Y}$ from
the context, we write $\simeq$ instead of $\simeq_{\mathsf{X},\mathsf{Y}}$.
\begin{definition}
  For Mealy machines
  $\mathsf{M},\mathsf{N} \colon \mathsf{X} \multimap \mathsf{Y}$, we
  say that $\mathsf{M}$ is \emph{behaviourally equivalent} to
  $\mathsf{N}$ when $\mathsf{M} \simeq \mathsf{N}$.
\end{definition}

For a Mealy machine $\mathsf{M} \colon \mathsf{X} \multimap \mathsf{Y}$,
we write $[\mathsf{M}]$ for its equivalence class with respect to
behavioural equivalence. We define a binary relation $\leq$ between
equivalence classes of Mealy machines from $\mathsf{X}$ to $\mathsf{Y}$
by $[\mathsf{M}] \leq [\mathsf{N}]$ if and only if
there are $\mathsf{M}' \simeq \mathsf{M}$ and
$\mathsf{N}' \simeq \mathsf{N}$ such that
$\state{\mathsf{M}'} = \state{\mathsf{N}'}$ and
$\init{\mathsf{M}'} = \init{\mathsf{N}'}$, and the graph
relation of $\tran{\mathsf{M}'}$ is a subset of the graph
relation of $\tran{\mathsf{N}'}$.
\begin{proposition}\label{prop:wcpo}
  The set of equivalence classes for $\simeq_{\mathsf{X},\mathsf{Y}}$
  with $\leq$ is a pointed $\omega$cpo.
\end{proposition}
We can characterize interpretation of the fixed point operator in
$\PCFSS$ in terms of least fixed points, see \cite{dlh2019}. \longv{We
  give a proof of Proposition~\ref{prop:wcpo} in
  Section~\ref{sec:wcpo}.}

\longv{
  \subsection{Proof of Proposition~\ref{prop:wcpo}}
  \label{sec:wcpo}

  For a partially defined expressions $E$ and $E'$, we write
  $E \approx E'$ when $E$ is defined if and only if $E'$ is defined,
  and if both expressions are defined, then they are the same. For
  example, we have $(1 - x)^{-1} \approx \sum_{n = 0}^{\infty} x^{n}$
  for all $x \in \mathbb{R}_{[0,1]}$. For a measurable space $X$, we
  write $LX$ for the measurable space of all finite sequences over $X$
  endowed with the following $\sigma$-algebra:
  \begin{equation*}
    A \in \Sigma_{LX} \iff
    \textnormal{ for all } n \in \mathbb{N},\,
    A \cap X^{n} \in \Sigma_{X^{n}}.
  \end{equation*}
  We write $\varepsilon$ for the empty sequence. For $a \in X$ and
  $u \in LX$, we denote the list obtained by appending $a$ to $u$ by
  $a \cons u$. 

  Let $\mathsf{M} \colon \mathsf{X} \multimap \mathsf{Y}$ be a Mealy
  machine. We write $Z$ for $X^{+} + Y^{-}$ and $W$ for
  $Y^{+} + X^{-}$. Then, the transition function of $\mathsf{M}$ is of
  the form
  \begin{equation*}
    \tran{\mathsf{M}} \colon Z \times \state{\mathsf{M}}
    \to W \times \state{\mathsf{M}}.
  \end{equation*}
  We define partial measurable functions
  $\alpha_{\mathsf{M}} \colon LZ \to \state{\mathsf{M}}$ and
  $\beta_{\mathsf{M}} \colon Z \times LZ \to W$ by
  \begin{align*}
    \alpha_{\mathsf{M}}(\varepsilon)
    &= \init{\mathsf{M}},\\
    \alpha_{\mathsf{M}}(z \cons u)
    &= 
      \begin{cases}
        s,
        & \textnormal{if } \alpha_{\mathsf{M}}(u) \textnormal{ and }
        \tran{\mathsf{M}}(z,\alpha_{\mathsf{M}}(u)) \textnormal{ are defined and } \\
        & \tran{\mathsf{M}}(z,\alpha_{\mathsf{M}}(u)) = (w,s)
        \textnormal{ for some } w \in W, \\
        \textnormal{undefined},
        & \textnormal{otherwise},
      \end{cases} \\
    \beta_{\mathsf{M}}(z,u)
    &= 
      \begin{cases}
        w,
        & \textnormal{if } \alpha_{\mathsf{M}}(u) \textnormal{ and }
        \tran{\mathsf{M}}(z,\alpha_{\mathsf{M}}(u)) \textnormal{ are defined and } \\
        & \tran{\mathsf{M}}(z,\alpha_{\mathsf{M}}(u)) = (w,s)
        \textnormal{ for some } s \in \state{\mathsf{M}}, \\
        \textnormal{undefined},
        & \textnormal{otherwise}.
      \end{cases}
  \end{align*}
  Below, for $x \in W \times \state{\mathsf{M}}$, we write
  $\mathrm{fst}(x)$ for the first entry of $x$, and we write
  $\mathrm{snd}(x)$ for the second entry of $x$. By the definition of
  $\alpha_{\mathsf{M}}$ and $\beta_{\mathsf{M}}$, we have
  \begin{equation*}
    \alpha_{\mathsf{M}}(z \cons u)
    \approx
    \mathrm{snd}(\tran{\mathsf{M}}(z,\alpha_{\mathsf{M}}(u))),
    \qquad
    \beta_{\mathsf{M}}(z, u)
    \approx
    \mathrm{fst}(\tran{\mathsf{M}}(z,\alpha_{\mathsf{M}}(u))).
  \end{equation*}

  \begin{lemma}\label{lem:preceq->beta}
    If $\mathsf{M} \preceq \mathsf{N}$, then
    $\beta_{\mathsf{M}} = \beta_{\mathsf{N}}$.
  \end{lemma}
  \begin{proof}
    Let $h \colon \state{\mathsf{M}} \to \state{\mathsf{N}}$ be a
    measurable function that realizes $\mathsf{M} \preceq \mathsf{N}$.
    We show $h(\alpha_{\mathsf{M}}(u)) \approx \alpha_{\mathsf{N}}(u)$
    and $\beta_{\mathsf{M}}(z,u) \approx \beta_{\mathsf{N}}(z,u)$ by
    induction on the size of $u$. (Base case)
    \begin{align*}
      h(\alpha_{\mathsf{M}}(\varepsilon))
      &= h(\init{\mathsf{M}}) \\
      &= \init{\mathsf{N}} \\
      &= \alpha_{\mathsf{N}}(\varepsilon)
    \end{align*}
    \begin{align*}
      \beta_{\mathsf{M}}(z,\varepsilon)
      &\approx \mathrm{fst}(\tran{\mathsf{M}}(z,\alpha_{\mathsf{M}}(\varepsilon))) \\
      &\approx \mathrm{fst}((W \times h)(\tran{\mathsf{M}}(z,\alpha_{\mathsf{M}}(\varepsilon)))) \\
      &\approx \mathrm{fst}(\tran{\mathsf{N}}(z,h(\alpha_{\mathsf{M}}(\varepsilon)))) \\
      &\approx \mathrm{fst}(\tran{\mathsf{N}}(z,\alpha_{\mathsf{N}}(\varepsilon))) \\
      &\approx \beta_{\mathsf{N}}(z,\varepsilon).
    \end{align*}
    (Induction step)
    \begin{align*}
      h(\alpha_{\mathsf{M}}(z \cons u))
      &\approx
        h(\mathrm{snd}(\tran{\mathsf{M}}(z,\alpha_{\mathsf{M}}(u)))) \\
      &\approx
        \mathrm{snd}((W \times h)(\tran{\mathsf{M}}(z,\alpha_{\mathsf{M}}(u)))) \\
      &\approx
        \mathrm{snd}(\tran{\mathsf{N}}(z,h(\alpha_{\mathsf{M}}(u)))) \\
      &\approx
        \mathrm{snd}(\tran{\mathsf{N}}(z,\alpha_{\mathsf{N}}(u))) \\
      &\approx
        \alpha_{\mathsf{N}}(z \cons u).
    \end{align*}
    \begin{align*}
      \beta_{\mathsf{M}}(z',z \cons u)
      &\approx \mathrm{fst}(\tran{\mathsf{M}}(z',\alpha_{\mathsf{M}}(z \cons u))) \\
      &\approx \mathrm{fst}((W \times h)(\tran{\mathsf{M}}(z',\alpha_{\mathsf{M}}(z \cons u)))) \\
      &\approx \mathrm{fst}(\tran{\mathsf{N}}(z',h(\alpha_{\mathsf{M}}(z \cons u)))) \\
      &\approx \mathrm{fst}(\tran{\mathsf{N}}(z',\alpha_{\mathsf{N}}(z \cons u))) \\
      &\approx \beta_{\mathsf{N}}(z',z \cons u).
    \end{align*}
  \end{proof}

  For a Mealy machine
  $\mathsf{M} \colon \mathsf{X} \multimap \mathsf{Y}$, we define Mealy
  machines $\mathsf{M}^{\#},\mathsf{M}^{@} \colon \mathsf{X} \multimap \mathsf{Y}$ by
  \begin{varitemize}
  \item $\state{\mathsf{M}^{\#}} = LZ$,
  \item $\init{\mathsf{M}^{\#}} = \varepsilon$,
  \item
    \begin{math}
      \tran{\mathsf{M}^{\#}}(z,u)
      =
      \begin{cases}
        (\beta_{\mathsf{M}}(z,u),z \cons u)
        , & \textnormal{if } \beta_{\mathsf{M}}(z,u)
        \textnormal{ is defined,} \\
        \textnormal{undefined}, & \textnormal{otherwise}, \\
      \end{cases}
    \end{math}
  \end{varitemize}
  and
  \begin{varitemize}
  \item
    $\state{\mathsf{M}^{@}} = \{\diamond\} \cup \state{\mathsf{M}}$,
  \item $\init{\mathsf{M}^{@}} = \init{\mathsf{M}}$,
  \item
    \begin{math}
      \tran{\mathsf{M}^{@}}(z,s) =
      \begin{cases}
        \tran{\mathsf{M}}(z,s) , & \textnormal{if } s \in
        \state{\mathsf{M}} \textnormal{ and } \tran{\mathsf{M}}(z,s)
        \textnormal{ is defined,} \\
        \textnormal{undefined}, & \textnormal{if } u = \diamond, \\
        \textnormal{undefined}, & \textnormal{otherwise}.
      \end{cases}
    \end{math}
  \end{varitemize}
  Here, the $\sigma$-algebra of $\state{\mathsf{M}^{@}}$ is the one
  induced by $\Sigma_{1 + \state{\mathsf{M}}}$ via the obvious
  bijection between $1 + \state{\mathsf{M}}$ and
  $\{\diamond\} \cup \state{\mathsf{M}}$.
  \begin{lemma}\label{lem:sharp}
    $\mathsf{M} \preceq \mathsf{M}^{@} \succeq \mathsf{M}^{\#}$.
  \end{lemma}
  \begin{proof}
    It is straightforward to check that the embedding
    $e \colon \state{\mathsf{M}} \to \{\diamond\}\cup
    \state{\mathsf{M}}$ is a measurable function that realizes
    $\mathsf{M} \preceq \mathsf{M}^{@}$. It remains to show
    $\mathsf{M}^{\#} \preceq \mathsf{M}^{@}$. We define a measurable
    function $h \colon LZ \to \{\diamond\} \cup \state{\mathsf{M}}$ by
    \begin{equation*}
      h(u)
      =
      \begin{cases}
        \alpha_{\mathsf{M}}(u), & \textnormal{if }
        \alpha_{\mathsf{M}}(u)
        \textnormal{ is defined}, \\
        \diamond, & \textnormal{otherwise}.
      \end{cases}
    \end{equation*}
    We show that for any $(z,u) \in Z \times LZ$,
    \begin{equation*}
      \tran{\mathsf{M}^{@}} ((Z \times h)(z,u))
      \approx
      (W \times h)(\tran{\mathsf{M}^{\#}}(z,u))
    \end{equation*}
    by induction on $u \in LZ$. (Base case)
    \begin{align*}
      \tran{\mathsf{M}^{@}}(z,h(\varepsilon))
      = \tran{\mathsf{M}^{@}}(z,\init{\mathsf{M}})
      = \tran{\mathsf{M}}(z,\init{\mathsf{M}})
      &\approx (\mathrm{fst}(\tran{\mathsf{M}}(z,\alpha_{\mathsf{M}}(\varepsilon))),
        \mathrm{snd}(\tran{\mathsf{M}}(z,\alpha_{\mathsf{M}}(\varepsilon)))) \\
      &\approx (\mathrm{fst}(\tran{\mathsf{M}}(z,\alpha_{\mathsf{M}}(\varepsilon))),
        \alpha_{\mathsf{M}}(z \cons \varepsilon)) \\
      &\approx (\beta_{\mathsf{M}}(z,\varepsilon),h(z \cons \varepsilon)) \\
      &\approx (W \times h)(\tran{\mathsf{M}^{\#}}(z,\varepsilon)).
    \end{align*}
    (Induction step)
    \begin{align*}
      \tran{\mathsf{M}^{@}}((Z \times h)(z,z' \cons u))
      &\approx
        \tran{\mathsf{M}}(z,\alpha_{\mathsf{M}}(z' \cons u)) \\
      &\approx
        \tran{\mathsf{M}}(z,\mathrm{snd}(\tran{\mathsf{M}}(z',\alpha_{\mathsf{M}}(u)))) \\
      &\approx
        \tran{\mathsf{M}^{@}}(z,\mathrm{snd}(\tran{\mathsf{M}}(z',\alpha_{\mathsf{M}}(u)))) \\
      &\approx
        \tran{\mathsf{M}^{@}}(z,\mathrm{snd}(\tran{\mathsf{M}^{@}}(z',h(u)))) \\
      &\approx
        \tran{\mathsf{M}^{@}}(z,h(\mathrm{snd}(\tran{\mathsf{M}^{\#}}(z',u)))) \\
      &\approx
        \tran{\mathsf{M}}(z,\alpha_{\mathsf{M}}(\mathrm{snd}(\tran{\mathsf{M}^{\#}}(z',u)))) \\
      &\approx
        \tran{\mathsf{M}}(z,\alpha_{\mathsf{M}}(z' \cons u)) \\
      &\approx
        (\beta_{\mathsf{M}}(z,z' \cons u),\alpha_{\mathsf{M}}(z \cons z' \cons u)) \\
      &\approx
        (\beta_{\mathsf{M}}(z,z' \cons u),h(z \cons z' \cons u)) \\
      &\approx
        (W \times h)(\tran{\mathsf{M}^{\#}}(z,z' \cons u)).
    \end{align*}
  \end{proof}
  \begin{proposition}
    For all Mealy machines $\mathsf{M},\mathsf{N} \colon \mathsf{X}
    \multimap \mathsf{Y}$, we have
    $\mathsf{M} \simeq \mathsf{N}$
    if and only if $\beta_{\mathsf{M}} = \beta_{\mathsf{N}}$.
  \end{proposition}
  \begin{proof}
    If $\mathsf{M} \simeq \mathsf{N}$, then we can show that
    $\beta_{\mathsf{M}} = \beta_{\mathsf{N}}$ by using
    Lemma~\ref{lem:preceq->beta}. If
    $\beta_{\mathsf{M}} = \beta_{\mathsf{N}}$, then we have
    $\mathsf{M}^{\#} = \mathsf{N}^{\#}$ by the definition of $(-)^{\#}$.
    Because we have $\mathsf{M} \simeq \mathsf{M}^{\#}$ and
    $\mathsf{N} \simeq \mathsf{N}^{\#}$ (Lemma~\ref{lem:sharp}), we see
    that $\mathsf{M}$ is behaviourally equivalent to $\mathsf{N}$.
  \end{proof}

  Hence, each equivalence class $[\mathsf{M}]$ of behavioural equivalence
  is represented by $\mathsf{M}^{\#}$, and $\mathsf{M}^{\#}$
  is independent of choice of $\mathsf{M}$. We extend this
  correspondence to order theoretic structure of Mealy machines.
  \begin{lemma}\label{lem:wcpo}
    Let $\mathsf{M},\mathsf{M}$ be Mealy machines from $\mathsf{X}$ to
    $\mathsf{Y}$ such that $\state{\mathsf{M}} = \state{\mathsf{N}}$.
    If $\tran{\mathsf{M}} \leq \tran{\mathsf{N}}$ and
    $\init{\mathsf{M}} = \init{\mathsf{N}}$, then
    $\tran{\mathsf{M}^{\#}} \leq \tran{\mathsf{N}^{\#}}$.
  \end{lemma}
  \begin{proof}
    By induction on the size of $u \in LZ$, we can show that if
    $\alpha_{\mathsf{M}}(u)$ is defined, then $\alpha_{\mathsf{N}}(u)$
    is defined and they are the same. Then
    $\tran{\mathsf{M}^{\#}} \leq \tran{\mathsf{N}^{\#}}$ follows from the
    definition of $(-)^{\#}$.
  \end{proof}
  \begin{theorem}\label{thm:leq}
    For Mealy machines $\mathsf{M},\mathsf{N}
    \colon \mathsf{X} \multimap \mathsf{Y}$,
    \begin{equation*}
      [\mathsf{M}] \leq [\mathsf{N}]
      \iff
      \tran{\mathsf{M}^{\#}} \leq \tran{\mathsf{N}^{\#}}.
    \end{equation*}
  \end{theorem}
  \begin{proof}
    If $\tran{\mathsf{M}^{\#}} \leq \tran{\mathsf{N}^{\#}}$, then
    because $\mathsf{M}^{\#}$ and $\mathsf{N}^{\#}$ are representatives
    of $[\mathsf{M}]$ and $[\mathsf{N}]$ respectively, we have
    $[\mathsf{M}] \leq [\mathsf{N}]$. If
    $[\mathsf{M}] \leq [\mathsf{N}]$, then there are
    $\mathsf{M}' \simeq \mathsf{M}$ and $\mathsf{N}' \simeq \mathsf{N}$
    such that
    \begin{varitemize}
    \item $\state{\mathsf{M}'} = \state{\mathsf{N}'}$ and $\init{\mathsf{M}'} = \init{\mathsf{N}'}$,
    \item the graph relation of $\tran{\mathsf{M}'}$ is a subset of the
      graph relation of $\tran{\mathsf{N}'}$.
    \end{varitemize}
    By Lemma~\ref{lem:wcpo}, we see that
    $\tran{{\mathsf{M}'}^{\#}} \leq \tran{{\mathsf{N}'}^{\#}}$.
  \end{proof}

  \begin{theorem}
    The set of equivalence classes
    of Mealy machines from $\mathsf{X}$ to $\mathsf{Y}$
    with the partial order $\leq$ is an $\omega$-cpo.
  \end{theorem}
  \begin{proof}
    Let $[\mathsf{N}]$ be an upper bound of an $\omega$-chain
    \begin{equation*}
      [\mathsf{M}_{1}] \leq [\mathsf{M}_{2}] \leq \cdots.
    \end{equation*}
    By Theorem~\ref{thm:leq}, we have
    \begin{equation*}
      \tran{\mathsf{M}_{1}^{\#}} \leq
      \tran{\mathsf{M}_{2}^{\#}} \leq \cdots \leq
      \tran{\mathsf{N}^{\#}}.
    \end{equation*}
    We define a Mealy machine
    $\mathsf{L} \colon \mathsf{X} \multimap \mathsf{Y}$ by
    \begin{varitemize}
    \item $\state{\mathsf{L}} = \state{\mathsf{M}_{1}^{\#}}$,
    \item $\init{\mathsf{L}} = \init{\mathsf{M}_{1}^{\#}}$,
    \item
      $\tran{\mathsf{L}} = \bigvee_{n \in
        \mathbb{N}}\tran{\mathsf{M}_{n}^{\#}}$.
    \end{varitemize}
    Because $\mathsf{M}_{n} \simeq \mathsf{M}_{n}^{\#}$, the
    equivalence class $[\mathsf{L}]$ is an upper bound of the
    $\omega$-chain
    $[\mathsf{M}_{1}] \leq [\mathsf{M}_{2}] \leq \cdots$. We also have
    $[\mathsf{L}] \leq [\mathsf{N}]$ because
    $\tran{\mathsf{L}} \leq \tran{\mathsf{M}_{n}^{\#}}$.
  \end{proof}
}

\subsection{Constructions on Mealy Machines}
\label{sec:const}
It is now time to give some constructions which are the basic building
blocks of our Mealy machine semantics. This section consists of three
parts. The first part (from Section~\ref{sec:mon} to
Section~\ref{sec:oc}) is related to the linear $\lambda$-calculus and
is serves to model the purely functional features of $\PCFSS$, such as
$\lambda$-abstraction and function application. In the second part
(Section~\ref{sec:num} and Section~\ref{sec:mf}), we give Mealy
machines modelling real numbers and measurable functions. In the
last part (from Section~\ref{sec:state} to
Section~\ref{sec:sampling}), we introduce a state monad and associate
the monad with Mealy machines modelling $\mathtt{score}$ and
$\mathtt{sample}$.

\subsubsection{Composition}

Let $\mathsf{X}$, $\mathsf{Y}$ and $\mathsf{Z}$ be
$\mathbf{Int}$-objects, and let
$\mathsf{M} \colon \mathsf{X} \multimap \mathsf{Y}$,
$\mathsf{N} \colon \mathsf{Y} \multimap \mathsf{Z}$ be Mealy machines.
We can now define their composition
$\mathsf{N} \circ \mathsf{M} \colon \mathsf{X} \multimap \mathsf{Z}$.
Before giving a precise definition, some intuitive explanation
about $\mathsf{N} \circ \mathsf{M}$ is in order. The main idea is to define
$\mathsf{N} \circ \mathsf{M}$ as a Mealy machine obtained by
connecting $\mathsf{N}$ and $\mathsf{M}$ in the following manner:
\begin{center}
  \begin{tikzpicture}[rounded corners]
    \node (M) [draw,circle,obj] at (0,0) {$\mathsf{M}$};
    \node (N) [draw,circle,obj] at (2,0) {$\mathsf{N}$};
    \draw (M) -- node[above,obj] {$\mathsf{X}$} ++(-1.5,0);
    \draw (M) -- node[above,obj] {$\mathsf{Y}$} (N);
    \draw (N) -- node[above,obj] {$\mathsf{Z}$} ++(1.5,0);
  \end{tikzpicture}
  .
\end{center}
The following series of thick arrows
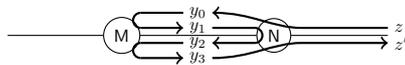
\begin{center}
  \begin{tikzpicture}[rounded corners]
    \node (M) [draw,circle,obj] at (0,0) {$\mathsf{M}$};
    \node (N) [draw,circle,obj] at (2,0) {$\mathsf{N}$};
    
    \node (0) [obj] at (1,0.3) {$y_{0}$};
    \node (1) [obj] at (1,0.1) {$y_{1}$};
    \node (2) [obj] at (1,-0.1) {$y_{2}$};
    \node (3) [obj] at (1,-0.3) {$y_{3}$};
    \draw (M) -- ++(-1.5,0);
    \draw (M) -- (N);
    \draw (N) -- ++(1.5,0);
    
    \draw[thick,->] (3.5,0.1) node[right,obj] {$z$} -- ++(-1,0) to[out=180,in=0] (0);
    \draw[thick,<-] (3.5,-0.1) node[right,obj] {$z'$} -- ++(-1,0) to[out=180,in=0] (3);
    \draw[thick] (0) -- ++(-0.75,0);
    \draw[thick,<-] (1) -- ++(-0.75,0);
    \draw[thick] (1) -- ++(0.75,0);
    \draw[thick] (2) -- ++(-0.75,0);
    \draw[thick,<-] (2) -- ++(0.75,0);
    \draw[thick,<-] (3) -- ++(-0.75,0);
    \draw[thick] (0.25,0.3) arc(90:270:0.1);
    \draw[thick] (1.75,0.1) arc(90:-90:0.1);
    \draw[thick] (0.25,-0.1) arc(90:270:0.1);
  \end{tikzpicture}
\end{center}
illustrates an execution of the obtained Mealy machine. Given an input from an
edge, $\mathsf{M}$ and $\mathsf{N}$ engage in some interactive communication,
and at some point, some output is  produced. Because
$\mathsf{N} \circ \mathsf{M}$ performs  ``parallel composition plus
connecting,'' the state space of $\mathsf{N} \circ \mathsf{M}$ should
be $\state{\mathsf{M}} \times \state{\mathsf{N}}$, and the initial
state should be $(\init{\mathsf{M}},\init{\mathsf{N}})$. The
transition function of $\mathsf{N} \circ \mathsf{M}$ should be given
by the collection of all possible interaction paths between $\mathsf{M}$
and $\mathsf{N}$.

Let us give a precise definition. For Mealy machines
$\mathsf{M} \colon \mathsf{X} \multimap \mathsf{Y}$ and
$\mathsf{N} \colon \mathsf{Y} \multimap \mathsf{Z}$, we define the
state space and the initial states of $\mathsf{N} \circ \mathsf{M}$ by
$\state{\mathsf{N} \circ \mathsf{M}} = \state{\mathsf{M}} \times
\state{\mathsf{N}}$,
$\init{\mathsf{N} \circ \mathsf{M}} =
(\init{\mathsf{M}},\init{\mathsf{N}})$ and we define
the transition function
$\tran{\mathsf{N} \circ \mathsf{M}}$ by
\longshortv{
  \begin{equation*}
    \tran{\mathsf{N} \circ \mathsf{M}} = f_{X^{+} , Z^{-},Z^{+} , X^{-}}
    \vee {} \\
    \bigvee_{n \in \mathbb{N}} f_{Y^{+} , Y^{-},Z^{+} , X^{-}} \circ
    f_{Y^{+} , Y^{-},Y^{+} , Y^{-}}^{n} \circ f_{X^{+} , Z^{-},Y^{+} ,
      Y^{-}}
  \end{equation*}}{
  \begin{multline*}
    \tran{\mathsf{N} \circ \mathsf{M}} = f_{X^{+} , Z^{-},Z^{+} , X^{-}}
    \vee {} \\
    \bigvee_{n \in \mathbb{N}} f_{Y^{+} , Y^{-},Z^{+} , X^{-}} \circ
    f_{Y^{+} , Y^{-},Y^{+} , Y^{-}}^{n} \circ f_{X^{+} , Z^{-},Y^{+} ,
      Y^{-}}    
  \end{multline*}
}
where the $f_{A,B,C,D} \colon (A + B) \times \state{\mathsf{N} \circ \mathsf{M}}
\to (C + D) \times \state{\mathsf{N} \circ \mathsf{M}}$ are restrictions of
the following partial measurable function
\begin{equation*}
  \xymatrix@R=3mm{
    (X^{+} + Z^{-} + Y^{+} + Y^{-}) \times
    \state{\mathsf{N} \circ \mathsf{M}}
    \ar[d]^-{\cong} \\
    (X^{+} + Y^{-}) \times \state{\mathsf{M}} \times \state{\mathsf{N}}
    +
    (Y^{+} + Z^{-}) \times \state{\mathsf{N}} \times \state{\mathsf{M}}
    \ar[d]^-{\tran{\mathsf{M}} \times \state{\mathsf{N}}
      + \tran{\mathsf{N}} \times \state{\mathsf{M}}} \\
    (Y^{+} + X^{-}) \times \state{\mathsf{M}} \times \state{\mathsf{N}}
    +
    (Z^{+} + Y^{-}) \times \state{\mathsf{N}} \times \state{\mathsf{M}}
    \ar[d]^-{\cong} \\
    (Z^{+} + X^{-} + Y^{+} + Y^{-}) \times
    \state{\mathsf{N} \circ \mathsf{M}}
    \nullcomma
  }
\end{equation*}
and the above join is with respect to the inclusion order between
graph relations. The above join is measurable because measurable sets
are closed under countable joins. It is tedious but doable to check
that the above join always exists and that the composition is
compatible with behavioural equivalence and satisfies associativity
modulo behavioural equivalence. We define a Mealy machine
$\mathsf{id}_{\mathsf{X}} \colon \mathsf{X} \multimap \mathsf{X}$ by
$\tran{\mathsf{id}_{\mathsf{X}}} = \id_{X^{+} + X^{-}}$. This is the
unit of the composition modulo behavioural equivalence. \shortv{This
  can also be checked by the underlying category theory \cite{hmh2014}.}

\subsubsection{Monoidal Products}
\label{sec:mon}

\paragraph{Monoidal Products of Int-objects}

We introduce monoidal products of $\mathbf{Int}$-objects and their
diagrammatic presentation. For $\mathbf{Int}$-objects $\mathsf{X}$ and
$\mathsf{Y}$, we define a $\mathbf{Int}$-object
$\mathsf{X} \otimes \mathsf{Y}$ by
\begin{equation*}
  \mathsf{X} \otimes \mathsf{Y} = (X^{+} + Y^{+},Y^{-} + X^{-}).
\end{equation*}
We define an $\mathbf{Int}$-object $\mathsf{I}$ to be
$(\emptyset,\emptyset)$. We write
$\mathsf{X} \otimes \mathsf{Y} \otimes \cdots$
for $\mathsf{X} \otimes (\mathsf{Y} \otimes \cdots )$.

Let
$\mathsf{X}_{1},\ldots,\mathsf{X}_{n},
\mathsf{Y}_{1},\ldots,\mathsf{Y}_{m}$ be $\mathbf{Int}$-object. We
depict a Mealy machine $\mathsf{M}$ from
$\mathsf{X}_{1} \otimes \cdots \otimes \mathsf{X}_{n}$ to
$\mathsf{Y}_{1} \otimes \cdots \otimes \mathsf{Y}_{m}$ as a node with
edges labeled by $\mathsf{X}_{1},\ldots,\mathsf{X}_{n}$ on the left
hand side and edges labeled by $\mathsf{Y}_{1},\ldots,\mathsf{Y}_{m}$
on the right hand side:
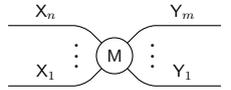
\begin{center}
  \begin{tikzpicture}[rounded corners]
    \node (M) [draw,circle,obj] at (0,0) {$\mathsf{M}$};
    \draw (M) -- ++ (0.4,0.4) -- node[above,obj] {$\mathsf{Y}_{m}$} ++ (1,0);
    \node at (0.5,0.1) {$\vdots$};
    \draw (M) -- ++ (0.4,-0.4) -- node[above,obj] {$\mathsf{Y}_{1}$} ++ (1,0);
    \draw (M) -- ++ (-0.4,0.4) -- node[above,obj] {$\mathsf{X}_{n}$} ++ (-1,0);
    \node at (-0.5,0.1) {$\vdots$};
    \draw (M) -- ++ (-0.4,-0.4) -- node[above,obj] {$\mathsf{X}_{1}$} ++ (-1,0);
  \end{tikzpicture}
  .
\end{center}
We do not draw any edges on the left/right hand side when
the domain/codomain of $\mathsf{M}$ is $\mathsf{I}$:
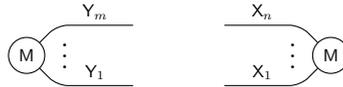
\begin{center}
  \begin{tikzpicture}[rounded corners]
    \node (M) [draw,circle,obj] at (0,0) {$\mathsf{M}$};
    \draw (M) -- ++ (0.4,0.4) -- node[above,obj] {$\mathsf{Y}_{m}$} ++ (1,0);
    \node at (0.5,0.1) {$\vdots$};
    \draw (M) -- ++ (0.4,-0.4) -- node[above,obj] {$\mathsf{Y}_{1}$} ++ (1,0);
    \begin{scope}[xshift=4cm]
      \node (M) [draw,circle,obj] at (0,0) {$\mathsf{M}$};
      \draw (M) -- ++ (-0.4,0.4) -- node[above,obj] {$\mathsf{X}_{n}$} ++ (-1,0);
      \node at (-0.5,0.1) {$\vdots$};
      \draw (M) -- ++ (-0.4,-0.4) -- node[above,obj] {$\mathsf{X}_{1}$} ++ (-1,0);
    \end{scope}
  \end{tikzpicture}
\end{center}
The diagrammatic presentation of monoidal products allows for an intuitive
description of transition functions. For example, we can depict
transitions
\begin{align*}
  \tran{\mathsf{M}}((\mm,(\mm,y)),s)
  &= ((\mm,(\mm,x')),t), \\
  \tran{\mathsf{M}}((\hh,(\hh,x)),s)
  &= ((\hh,(\mm,\cdots(\mm,y'))),t'),
\end{align*}
for some $y \in Y_{1}^{-}$, $x \in X_{1}^{+}$,
$x' \in X_{1}^{-}$, $y' \in Y_{m}^{+}$ and
$s,t,t' \in \state{\mathsf{M}}$
as follows:
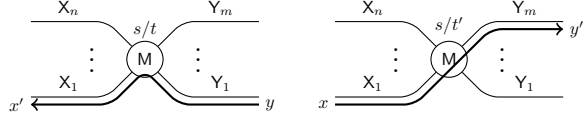
\begin{center}
  \begin{tikzpicture}[rounded corners]
    \node (M) [draw,circle,obj] at (0,0) {$\mathsf{M}$};
    \draw (M) -- ++ (0.5,0.5) -- node[above,obj] {$\mathsf{Y}_{m}$} ++ (1,0);
    \node at (0.7,0.1) {$\vdots$};
    \draw (M) -- ++ (0.5,-0.5) -- node[above,obj] {$\mathsf{Y}_{1}$} ++ (1,0);
    \draw (M) -- ++ (-0.5,0.5) -- node[above,obj] {$\mathsf{X}_{n}$} ++ (-1,0);
    \node at (-0.7,0.1) {$\vdots$};
    \draw (M) -- ++ (-0.5,-0.5) -- node[above,obj] {$\mathsf{X}_{1}$} ++ (-1,0);
    \draw[thick,->] (1.5,-0.6) node[right,obj] {$y$}
    -- ++(-1.05,0) -- ++(-0.45,0.45) -- ++ (-0.45,-0.45) -- ++(-1.05,0) node[left,obj] {$x'$};
    \node [obj] at (0,0.4) {$s/t$};

    \begin{scope}[xshift=4cm]
      \node (M) [draw,circle,obj] at (0,0) {$\mathsf{M}$};
      \draw (M) -- ++ (0.5,0.5) -- node[above,obj] {$\mathsf{Y}_{m}$} ++ (1,0);
      \node at (0.7,0.1) {$\vdots$};
      \draw (M) -- ++ (0.5,-0.5) -- node[above,obj] {$\mathsf{Y}_{1}$} ++ (1,0);
      \draw (M) -- ++ (-0.5,0.5) -- node[above,obj] {$\mathsf{X}_{n}$} ++ (-1,0);
      \node at (-0.7,0.1) {$\vdots$};
      \draw (M) -- ++ (-0.5,-0.5) -- node[above,obj] {$\mathsf{X}_{1}$} ++ (-1,0);
      \draw[thick,->] (-1.5,-0.6) node[left,obj] {$x$}
      -- ++ (1.05,0) -- ++ (1,1) -- ++ (0.95,0) node[right,obj] {$y'$};
      \node [above,obj] at (M.north) {$s/t'$};
  \end{scope}
  \end{tikzpicture}
\end{center}
\longv{
  We note that there are several ways to present a Mealy machine
  $\mathsf{M} \colon \mathsf{X}_{1} \otimes \cdots \otimes
  \mathsf{X}_{n} \multimap \mathsf{Y}_{1} \otimes \cdots \otimes \mathsf{Y}_{m}$
  such as
  
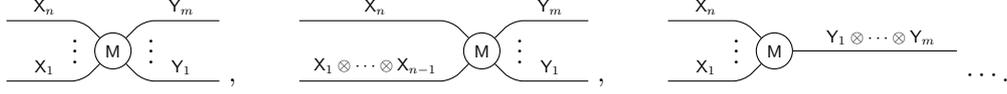
\begin{center}
    \begin{tikzpicture}[rounded corners]
      \node (M) [draw,circle,obj] at (0,0) {$\mathsf{M}$};
      \draw (M) -- ++ (0.4,0.4) -- node[above,obj] {$\mathsf{Y}_{m}$} ++ (1,0);
      \node at (0.5,0.1) {$\vdots$};
      \draw (M) -- ++ (0.4,-0.4) -- node[above,obj] {$\mathsf{Y}_{1}$} ++ (1,0);
      \draw (M) -- ++ (-0.4,0.4) -- node[above,obj] {$\mathsf{X}_{n}$} ++ (-1,0);
      \node at (-0.5,0.1) {$\vdots$};
      \draw (M) -- ++ (-0.4,-0.4) -- node[above,obj] {$\mathsf{X}_{1}$} ++ (-1,0);
    \end{tikzpicture}
    ,
    \qquad
    \begin{tikzpicture}[rounded corners]
      \node (M) [draw,circle,obj] at (0,0) {$\mathsf{M}$};
      \draw (M) -- ++ (0.4,0.4) -- node[above,obj] {$\mathsf{Y}_{m}$} ++ (1,0);
      \node at (0.5,0.1) {$\vdots$};
      \draw (M) -- ++ (0.4,-0.4) -- node[above,obj] {$\mathsf{Y}_{1}$} ++ (1,0);
      \draw (M) -- ++ (-0.4,0.4) -- node[above,obj] {$\mathsf{X}_{n}$} ++ (-2,0);
      \draw (M) -- ++ (-0.4,-0.4) --
      node[above,obj] {$\mathsf{X}_{1} \otimes \cdots \otimes \mathsf{X}_{n-1}$} ++ (-2,0);
    \end{tikzpicture}
    ,
    \qquad
    \begin{tikzpicture}[rounded corners]
      \node (M) [draw,circle,obj] at (0,0) {$\mathsf{M}$};
      \draw (M) -- ++ (0.4,0) -- node[above,obj] {$\mathsf{Y}_{1} \otimes \cdots \otimes
        \mathsf{Y}_{m}$} ++ (2,0);
      \draw (M) -- ++ (-0.4,0.4) -- node[above,obj] {$\mathsf{X}_{n}$} ++ (-1,0);
      \node at (-0.5,0.1) {$\vdots$};
      \draw (M) -- ++ (-0.4,-0.4) -- node[above,obj] {$\mathsf{X}_{1}$} ++ (-1,0);
    \end{tikzpicture}
    $\cdots$.
  \end{center}
}
\paragraph{Monoidal Product of Mealy Machines}

We give monoidal products of Mealy machines.
For Mealy machines $\mathsf{M} \colon \mathsf{X} \multimap \mathsf{Z}$
and $\mathsf{N} \colon \mathsf{Y} \multimap \mathsf{W}$, we define a
Mealy machine $\mathsf{M} \otimes \mathsf{N}\colon
\mathsf{X} \otimes \mathsf{Y} \multimap \mathsf{Z} \otimes \mathsf{W}$ by:
$\state{\mathsf{M} \otimes \mathsf{N}} = \state{\mathsf{M}} \times
\state{\mathsf{N}}$,
$\init{\mathsf{M} \otimes \mathsf{N}} =
(\init{\mathsf{M}},\init{\mathsf{N}})$ and
$\tran{\mathsf{M} \otimes \mathsf{N}}$ is given by
\begin{equation*}
  \xymatrix@R=3mm{
    ((X^{+} + Y^{+}) + (W^{-} + Z^{-})) \times
    \state{\mathsf{M} \otimes \mathsf{N}}
    \ar[d]^-{\cong} \\
    (X^{+} + Z^{-}) \times \state{\mathsf{M}} \times \state{\mathsf{N}}
    +
    (Y^{+} + W^{-}) \times \state{\mathsf{N}} \times \state{\mathsf{M}}
    \ar[d]^-{\tran{\mathsf{M}} \times \state{\mathsf{N}}
      + \tran{\mathsf{N}} \times \state{\mathsf{M}}} \\
    (Z^{+} + X^{-}) \times \state{\mathsf{M}} \times \state{\mathsf{N}}
    +
    (W^{+} + Y^{-}) \times \state{\mathsf{N}} \times \state{\mathsf{M}}
    \ar[d]^-{\cong} \\
    ((Z^{+} + W^{+}) + (Y^{-} + X^{-})) \times
    \state{\mathsf{M} \otimes \mathsf{N}}
    \nulldot
  }
\end{equation*}
It is not difficult to check that the monoidal product is
compatible with behavioural equivalence.

We depict
$\mathsf{M} \otimes \mathsf{N} \colon (\mathsf{X} \otimes \mathsf{Y})
\multimap (\mathsf{Z} \otimes \mathsf{W})$
as follows:
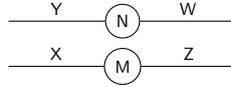
\begin{center}
  \begin{tikzpicture}
    \node (M) [draw,circle,obj] at (0,0) {$\mathsf{M}$};
    \node (N) [draw,circle,obj] at (0,0.6) {$\mathsf{N}$};
    \draw (M) -- node[above,obj]{$\mathsf{Z}$} ++(1.5,0);
    \draw (M) -- node[above,obj]{$\mathsf{X}$} ++(-1.5,0);
    \draw (N) -- node[above,obj]{$\mathsf{W}$} ++(1.5,0);
    \draw (N) -- node[above,obj]{$\mathsf{Y}$} ++(-1.5,0);
  \end{tikzpicture}
\end{center}
As indicated by the above diagram, $\mathsf{M} \otimes \mathsf{N}$
consists of two sub-machines $\mathsf{M}$ and $\mathsf{N}$ working
independently. For example, if we have
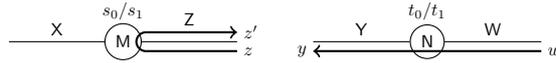
\begin{center}
  \begin{tikzpicture}
    \node (M) [draw,circle,obj] at (0,0) {$\mathsf{M}$};
    \node (N) [draw,circle,obj] at (4,0) {$\mathsf{N}$};
    \draw (M) -- node[above=4,obj]{$\mathsf{Z}$} ++(1.5,0);
    \draw (M) -- node[above,obj]{$\mathsf{X}$} ++(-1.5,0);
    \draw (N) -- node[above,obj]{$\mathsf{W}$} ++(1.5,0);
    \draw (N) -- node[above,obj]{$\mathsf{Y}$} ++(-1.5,0);
    \node [obj] at (0,0.4) {$s_{0}/s_{1}$};
    \node [obj] at (4,0.4) {$t_{0}/t_{1}$};

    \draw[thick] (1.5,-0.125) node[right,obj] {$z$} -- ++(-1.2,0);
    \draw[thick,<-] (1.5,0.125) node[right,obj] {$z'$} -- ++(-1.2,0);
    \draw[thick] (0.3,0.125) arc(90:270:0.125);
    \draw[thick,->] (5.5,-0.125) node[right,obj] {$w$} -- ++(-3,0) node[left,obj] {$y$};
  \end{tikzpicture}
\end{center}
then $\mathsf{M} \otimes \mathsf{N}$ has the following transitions:
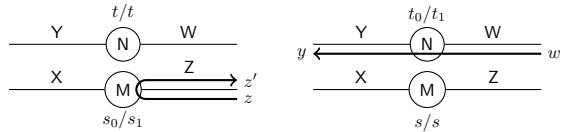
\begin{center}
  \begin{tikzpicture}
    \node (M) [draw,circle,obj] at (0,0) {$\mathsf{M}$};
    \node (N) [draw,circle,obj] at (0,0.6) {$\mathsf{N}$};
    \draw (M) -- node[above=4,obj]{$\mathsf{Z}$} ++(1.5,0);
    \draw (M) -- node[above,obj]{$\mathsf{X}$} ++(-1.5,0);
    \draw (N) -- node[above,obj]{$\mathsf{W}$} ++(1.5,0);
    \draw (N) -- node[above,obj]{$\mathsf{Y}$} ++(-1.5,0);
    \node [obj] at (0,-0.4) {$s_{0}/s_{1}$};
    \node [obj,above] at (N.north) {$t/t$};
    \draw[thick] (1.5,-0.125) node[right,obj] {$z$} -- ++(-1.2,0);
    \draw[thick,<-] (1.5,0.125) node[right,obj] {$z'$} -- ++(-1.2,0);
    \draw[thick] (0.3,0.125) arc(90:270:0.125);
    \begin{scope}[xshift=4cm]
      \node (M) [draw,circle,obj] at (0,0) {$\mathsf{M}$};
      \node (N) [draw,circle,obj] at (0,0.6) {$\mathsf{N}$};
      \draw (M) -- node[above,obj]{$\mathsf{Z}$} ++(1.5,0);
      \draw (M) -- node[above,obj]{$\mathsf{X}$} ++(-1.5,0);
      \draw (N) -- node[above,obj]{$\mathsf{W}$} ++(1.5,0);
      \draw (N) -- node[above,obj]{$\mathsf{Y}$} ++(-1.5,0);
      \node [obj] at (0,1) {$t_{0}/t_{1}$};
      \node [obj,below] at (M.south) {$s/s$};
      \draw[thick,->] (1.5,0.475) node[right,obj] {$w$} -- ++(-3,0) node[left,obj] {$y$};
    \end{scope}
  \end{tikzpicture}
\end{center}
for all $t \in \state{\mathsf{N}}$ and for all $s \in \state{\mathsf{M}}$.

\begin{convention}\label{conv:ox}
  We do the following identification:
  \begin{varitemize}
  \item We identity
    $\mathsf{X} \otimes (\mathsf{Y} \otimes \mathsf{Z})$ with
    $(\mathsf{X} \otimes \mathsf{Y}) \otimes \mathsf{Z}$ by
    the canonical isomorphism
    \begin{math}
      X + (Y + Z) \cong (X + Y) + Z.
    \end{math}
  \item We identify $\mathsf{I} \otimes \mathsf{X}$ and
    $\mathsf{X} \otimes \mathsf{I}$ with $\mathsf{X}$ by the unit laws
    $X^{+} + \emptyset \cong X^{+}$ and
    $\emptyset + X^{-} \cong X^{-}$.
  \end{varitemize}
\end{convention}

\subsubsection{Axiom Link and Cut Link}
\label{sec:id}
For an $\mathbf{Int}$-object $\mathsf{X}$, we define
$\mathsf{X}^{\bot}$ to be $(X^{-},X^{+})$, and we define token
machines
\begin{equation*}
  \mathsf{unit}_{\mathsf{X}} \colon \mathsf{I} \multimap \mathsf{X}
  \otimes \mathsf{X}^{\bot},
  \quad
  \mathsf{counit}_{\mathsf{X}} \colon \mathsf{X}^{\bot} \otimes
  \mathsf{X} \multimap \mathsf{I}
\end{equation*}
by
$\tran{\mathsf{unit}_{\mathsf{X}}} = \id_{X^{+} + X^{-}}$ and
$\tran{\mathsf{counit}_{\mathsf{X}}} = \mathrm{id}_{{X^{-} + X^{+}}}$.
We depict them by single edges
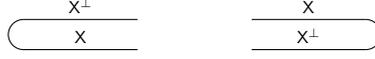
\begin{center}
  \begin{tikzpicture}
    \begin{scope}[xshift=3cm]
      \draw (0,0) -- node[above,obj] {$\mathsf{X}$} ++(1.5,0);
      \draw (0,0.4) -- node[above,obj] {$\mathsf{X}^{\bot}$} ++(1.5,0);
      \draw (0,0.4) arc(90:270:0.2);
    \end{scope}
    \begin{scope}[xshift=6cm]
      \draw (0,0) -- node[above,obj] {$\mathsf{X}^{\bot}$} ++(1.5,0);
      \draw (0,0.4) -- node[above,obj] {$\mathsf{X}$} ++(1.5,0);
      \draw (1.5,0) arc(-90:90:0.2);
    \end{scope}
  \end{tikzpicture}
\end{center}
respectively. This is compatible with behaviour of these Mealy machines:
if we give an input to an edge, then we will get the same value from
the other end of the edge. For example, for any $x \in X^{+}$, we have
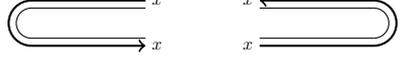
\begin{center}
  \begin{tikzpicture}
    \begin{scope}[xshift=3cm]
      \draw (0,0) -- ++(1.5,0);
      \draw (0,0.4) --  ++(1.5,0);
      \draw (0,0.4) arc(90:270:0.2);

      \draw[->,thick] (0,-0.1) -- ++(1.5,0) node[right,obj] {$x$};
      \draw[thick] (0,0.5) -- ++(1.5,0) node[right,obj] {$x$};
      \draw[thick] (0,0.5) arc(90:270:0.3);
    \end{scope}
    \begin{scope}[xshift=6cm]
      \draw (0,0) --  ++(1.5,0);
      \draw (0,0.4) -- ++(1.5,0);
      \draw (1.5,0) arc(-90:90:0.2);

      \draw[thick] (0,-0.1) node[left,obj] {$x$} -- ++(1.5,0);
      \draw[thick,<-] (0,0.5) node[left,obj] {$x$} -- ++(1.5,0);
      \draw[thick] (1.5,-0.1) arc(-90:90:0.3);
    \end{scope}
  \end{tikzpicture}
  .
\end{center}

\subsubsection{Symmetry}

Let $\mathsf{X}$ and $\mathsf{Y}$ be $\mathbf{Int}$-objects. We define
a token machine
$\mathsf{sym}_{\mathsf{X},\mathsf{Y}} \colon \mathsf{X} \otimes
\mathsf{Y} \multimap \mathsf{Y} \otimes \mathsf{X}$ by
letting its transition function be the canonical isomorphism
\begin{equation*}
  (X^{+} + Y^{+}) + (X^{-} + Y^{-})
  \xrightarrow{\cong}
  (Y^{+} + X^{+}) + (Y^{-} + X^{-}).
\end{equation*}
We depict $\mathsf{sym}_{\mathsf{X},\mathsf{Y}}$ by a crossing:
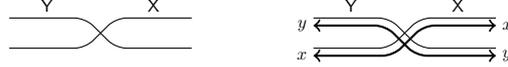
\begin{center}
  \begin{tikzpicture}[rounded corners]
    \draw (0,0)-- ++(1,0)-- ++(0.4,0.4) -- node[above,obj]{$\mathsf{X}$} ++(1,0);
    \draw (0,0.4) -- node[above,obj]{$\mathsf{Y}$} ++(1,0)-- ++(0.4,-0.4) -- ++(1,0);
    \begin{scope}[xshift=4cm]
      \draw (0,0) -- ++(1,0)-- ++(0.4,0.4) -- node[above,obj]{$\mathsf{X}$} ++(1,0);
      \draw (0,0.4) -- node[above,obj]{$\mathsf{Y}$} ++(1,0)-- ++(0.4,-0.4) -- ++(1,0);
      \draw[thick,<->] (0,-0.1) node[left,obj] {$x$}
      -- ++(1.05,0)-- ++(0.4,0.4) -- ++(0.95,0) node[right,obj] {$x$};
      \draw[thick,<->] (0,0.3) node[left,obj] {$y$}
      -- ++(0.95,0)-- ++(0.4,-0.4) -- ++(1.05,0) node[right,obj] {$y$};
    \end{scope}
  \end{tikzpicture}
\end{center}
As arrows in the right hand side indicate, given an input from an edge
in one side, $\mathsf{sym}_{\mathsf{X},\mathsf{Y}}$ outputs the same
value to the corresponding edge on other side.

\subsubsection{A Modal Operator}
\label{sec:oc}
We give a constructor on Mealy machines that corresponds to
the resource modality in linear logic.
For an $\mathbf{Int}$-object $\mathsf{X}$, we define
an $\mathbf{Int}$-object $\oc \mathsf{X}$ by
\begin{equation*}
  \oc \mathsf{X} = (\mathbb{N} \times X^{+}, \mathbb{N} \times X^{-}).
\end{equation*}
We can informally regard $\oc \mathsf{X}$ as a countable monoidal power
$\bigotimes_{n \in \mathbb{N}} \mathsf{X} \approx \mathsf{X} \otimes
\mathsf{X} \otimes \cdots$. Following this intuition, we extend the
action of $\oc(-)$ to Mealy machines. Let
$\mathsf{M} \colon \mathsf{X} \multimap \mathsf{Y}$ be a Mealy
machine. We define a Mealy machine
$\oc \mathsf{M} \colon \oc \mathsf{X} \multimap \oc \mathsf{Y}$ by:
the state space of $\oc \mathsf{M}$ is defined to be
$|\mathsf{M}|^{\mathbb{N}}$ associated with the least $\sigma$-algebra
such that for all $A_{1},A_{2},\ldots \in \Sigma_{\mathsf{M}}$,
\begin{equation*}
  A_{1} \times A_{2} \times \cdots \in \Sigma_{\oc\mathsf{M}};
\end{equation*}
the initial state $\init{\oc\mathsf{M}}$ is
$(\init{\mathsf{M}},\init{\mathsf{M}},\ldots)$; the transition
function $\tran{\oc \mathsf{M}}$ is the unique
partial measurable function satisfying
\longshortv{
  \begin{equation*}
    \xymatrix@C=70pt{
      (X^{+} + Y^{-})
      \times \state{\mathsf{M}} \times \state{\mathsf{M}}^{\mathbb{N}}
      \ar[r]^-{(\mathrm{inj}_{n}+\mathrm{inj}_{n}) \times \mathrm{ins}_{n}}
      \ar[d]_-{\tran{\mathsf{M}} \times \state{\mathsf{M}}^{\mathbb{N}}}
      &
      (\mathbb{N} \times X^{+} + \mathbb{N} \times Y^{-})
      \times \state{\mathsf{M}}^{\mathbb{N}}
      \ar[d]^-{\tran{\oc\mathsf{M}}} \\
      (Y^{+} + X^{-})
      \times \state{\mathsf{M}} \times \state{\mathsf{M}}^{\mathbb{N}}
      \ar[r]^-{(\mathrm{inj}_{n}+\mathrm{inj}_{n}) \times \mathrm{ins}_{n}}
      &
      (\mathbb{N} \times Y^{+} + \mathbb{N} \times X^{-})
      \times \state{\mathsf{M}}^{\mathbb{N}}
    }
  \end{equation*}}{
  \begin{equation*}
    \xymatrix@R=10pt@C=16pt{
      (X^{+} + Y^{-})
      \times \state{\mathsf{M}} \times \state{\mathsf{M}}^{\mathbb{N}}
      \ar[r]^-{\raisebox{6pt}{$\scriptstyle(\mathrm{inj}_{n}+\mathrm{inj}_{n}) \times \mathrm{ins}_{n}$}}
      \ar[d]_-{\tran{\mathsf{M}} \times \state{\mathsf{M}}^{\mathbb{N}}}
      &
      (\mathbb{N} \times X^{+} + \mathbb{N} \times Y^{-})
      \times \state{\mathsf{M}}^{\mathbb{N}}
      \ar[d]^-{\tran{\oc\mathsf{M}}} \\
      (Y^{+} + X^{-})
      \times \state{\mathsf{M}} \times \state{\mathsf{M}}^{\mathbb{N}}
      \ar[r]^-{\raisebox{6pt}{$\scriptstyle(\mathrm{inj}_{n}+\mathrm{inj}_{n}) \times \mathrm{ins}_{n}$}}
      &
      (\mathbb{N} \times Y^{+} + \mathbb{N} \times X^{-})
      \times \state{\mathsf{M}}^{\mathbb{N}}
    }
  \end{equation*}
} for all $n \in \mathbb{N}$. Here,
$\mathrm{inj}_{n} \colon (-) \to \mathbb{N} \times (-)$ are the $n$th
injections, and
$\mathrm{ins}_{n} \colon \state{\mathsf{M}} \times
\state{\mathsf{M}}^{\mathbb{N}} \to \state{\mathsf{M}}^{\mathbb{N}}$
sends $(s,\{s_{n}\}_{n \in \mathbb{N}})$ to
$(s_{0},\ldots,s_{n-1},s,s_{n},s_{n+1},\ldots)$.


      
As $\oc(-)$ is defined to be a countable monoidal power,
$\oc\mathsf{M}$ behaves as a parallel composition of countably
infinite copies of $\mathsf{M}$. For example, if we have
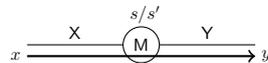
\begin{center}
  \begin{tikzpicture}
    \node (M) [draw,circle,obj] at (0,0) {$\mathsf{M}$};
    \draw (M) -- node[above,obj] {$\mathsf{Y}$} ++(1.5,0);
    \draw (M) -- node[above,obj] {$\mathsf{X}$} ++(-1.5,0);
    \node [obj] at (0.05,0.4) {$s/s'$};
    \draw [thick,<-] (1.5,-0.15) node[right,obj] {$y$} -- ++ (-3,0)
    node [left,obj] {$x$};
  \end{tikzpicture}
\end{center}
then for all $n \in \mathbb{N}$
and $t_{1},t_{2},\ldots \in \state{\mathsf{M}}$, we have
\begin{center}
  \begin{tikzpicture}[rounded corners]
      \node (M) [draw,circle,obj] at (0,0) {$\oc\mathsf{M}$};
    \draw (M) -- node[above,obj] {$\oc\mathsf{Y}$} ++(1.5,0);
    \draw (M) -- node[above,obj] {$\oc\mathsf{X}$} ++(-1.5,0);
    \node [obj] at (0.05,0.6) {$(t_{1},\ldots,t_{n-1},s,t_{n},t_{n+1},\ldots)/
      (t_{1},\ldots,t_{n-1},s',t_{n},t_{n+1},\ldots)$};
    \draw [thick,<-] (1.5,-0.15) node[right,obj] {$(n,y)$} -- ++ (-3,0)
    node [left,obj] {$(n,x)$};
  \end{tikzpicture}
  .
\end{center}
In other words, given an input whose first entry is $n$, then the
$n$th copy of $\mathsf{M}$ handles the input, and there is no side
effect to the other copies of $\mathsf{M}$.

\longv{
  \begin{proposition}\label{prop:oc_f}
    The operator $\oc(-)$ is compatible with
    the behavioral equivalence and is functorial. Namely,
    \begin{itemize}
    \item for all Mealy machines
      $\mathsf{M},\mathsf{N} \colon \mathsf{X} \multimap \mathsf{Y}$, if
      $\mathsf{M} \simeq \mathsf{M}'$, then
      $\oc \mathsf{M} \simeq \oc
      \mathsf{N}$; and
    \item for all Mealy machines
      $\mathsf{M} \colon \mathsf{X} \multimap \mathsf{Y}$
      and $\mathsf{N} \colon \mathsf{Y} \multimap \mathsf{Z}$,
      \begin{equation*}
        \oc (\mathsf{N} \circ \mathsf{M})
        \simeq
        \oc \mathsf{N} \circ \oc \mathsf{M};
      \end{equation*}
    \item $\oc \mathsf{id}_{\mathsf{X}}
      \simeq \mathsf{id}_{\oc \mathsf{X}}$.
    \end{itemize}
  \end{proposition}
}

\begin{convention}\label{conv:oc}
  For the sake of legibility and due to lack of space, we sometimes
  implicitly identify $\oc(\mathsf{X} \otimes \mathsf{Y})$ with
  $\oc \mathsf{X} \otimes \oc \mathsf{Y}$ by the canonical isomorphism
  \begin{math}
    \mathbb{N} \times (X + Y) \cong \mathbb{N} \times X + \mathbb{N}
    \times Y.
  \end{math}
\end{convention}

\longv{
  Under the above convention, for Mealy macines
  $\mathsf{M} \colon \oc (\mathsf{X} \otimes \mathsf{Y}) \multimap
  \mathsf{Z}$ and $\mathsf{N} \colon \mathsf{W} \multimap \mathsf{X}$,
  we can simply write
  $\mathsf{M} \circ (\oc \mathsf{N} \otimes \mathsf{Y}_{\oc
    \mathsf{Y}}) \colon \oc \mathsf{W} \otimes \oc \mathsf{Y}
  \multimap \mathsf{Z}$. It is not difficult to see that
  when $\mathsf{Z} = \oc \mathsf{Z}'$ and
  $\mathsf{M} = \oc \mathsf{M}'$ for some $\mathsf{M}' \colon \mathsf{X} \otimes \mathsf{Y}
  \multimap \mathsf{Z}'$, we have
  $\mathsf{M} \circ (\oc \mathsf{N} \otimes \mathsf{id}_{\oc \mathsf{Y}})
  \simeq \oc (\mathsf{M} \circ (\mathsf{N} \otimes \mathsf{id}_{\mathsf{Y}}))$.
}

\shortv{ We can construct token machines related to inference rules
  for resource modality in linear logic. Their definitions are
  essentially the same with the one given in \cite{laurent2001},
  and they work as we expect modulo behavioural equivalence.
  \begin{proposition}\label{prop:oc}
    For every $\mathbf{Int}$-object $\mathsf{X}$,
    there are token machines
    \begin{equation*}
      \mathsf{d} \colon \oc \mathsf{X} \multimap \mathsf{X},
      \hspace{5pt}
      \mathsf{dg} \colon \oc \mathsf{X} \multimap \oc\oc\mathsf{X},
      \hspace{5pt}
      \mathsf{c} \colon \oc \mathsf{X} \multimap \oc \mathsf{X}
      \otimes \oc \mathsf{X},
      \hspace{5pt}
      \mathsf{w} \colon \oc \mathsf{X} \multimap \mathsf{I}
    \end{equation*}
    such that for any Mealy machine $\mathsf{M} \colon \mathsf{I}
    \multimap \mathsf{X}$,
    \begin{equation*}
      \mathsf{d} \circ \oc \mathsf{M} \simeq \mathsf{M},
      \hspace{5.5pt}
      \mathsf{dg} \circ \oc \mathsf{M} \simeq \oc\oc\mathsf{M},
      \hspace{5.5pt}
      \mathsf{c} \circ \oc \mathsf{M} \simeq \oc \mathsf{M} \otimes \oc \mathsf{M},
      \hspace{5.5pt}
      \mathsf{w} \circ \oc \mathsf{M} \simeq \mathsf{id}_{\mathsf{I}}.
    \end{equation*}
  \end{proposition}
}
\longv{
  \paragraph{Dereliction}
  \label{sec:der}

  For an $\mathbf{Int}$-object $\mathsf{X}$, we define a token machine
  $\mathsf{d}_{\mathsf{X}} \colon \oc\mathsf{X} \multimap \mathsf{X}$ by
  defining
  $\tran{\mathsf{d}_{\mathsf{X}}} \colon (\mathbb{N} \times X^{+}) +
  X^{-} \to X^{+} + (\mathsf{N} \times X^{-})$ by
  \begin{equation*}
    \tran{\mathsf{d}_{\mathsf{X}}}
    (\hh,(n,x)) = (\hh,x),
    \qquad
    \tran{\mathsf{d}_{\mathsf{X}}}
    (\mm,x) = (\mm,(0,x)).
  \end{equation*}
  The Mealy machine $\mathsf{d}_{\mathsf{X}}$ pops/pushes indices with
  probability $1$. Namely, we have
  \begin{center}
    \begin{tikzpicture}
      \node (d) [draw,circle,obj] at (0,0) {$\mathsf{d}_{\mathsf{X}}$};
      \draw (d) -- node[above,obj] {$\mathsf{X}$} ++ (1.5,0);
      \draw (d) -- node[above,obj] {$\oc\mathsf{X}$} ++ (-1.5,0);
      \draw[thick,<-] (1.5,-0.15) node[right,obj]{$x$}
      -- ++ (-3,0) node[left,obj]{$(n,x)$};
    \end{tikzpicture}
    \qquad
    \begin{tikzpicture}
      \node (d) [draw,circle,obj] at (0,0) {$\mathsf{d}_{\mathsf{X}}$};
      \draw (d) -- node[above,obj] {$\mathsf{X}$} ++ (1.5,0);
      \draw (d) -- node[above,obj] {$\oc\mathsf{X}$} ++ (-1.5,0);
      \draw[thick,->] (1.5,-0.15) node[right,obj]{$x$}
      -- ++ (-3,0) node[left,obj]{$(0,x)$};
    \end{tikzpicture}
  \end{center}
  for all $n \in \mathbb{N}$, $x \in X^{+}$ and $x' \in X^{-}$. Hence,
  for any Mealy machine
  $\mathsf{M} \colon \mathsf{I} \multimap \mathsf{X}$, if we have
  \begin{center}
    \begin{tikzpicture}[rounded corners]
      \node (M) [draw,circle,obj] at (0,0) {$\mathsf{M}$};
      \draw (M) -- node[above=4pt,obj]{$\mathsf{X}$} ++(1.5,0);
      \draw[thick,->] (1.5,-0.15) node[right,obj]{$x$}
      -- ++ (-1.3,0) -- ++ (0,0.3) -- ++ (1.3,0) node[right,obj]{$x'$};
      \node [above,obj] at (M.north) {$s/s'$};
    \end{tikzpicture}
  \end{center}
  for some $x \in X^{-}$, $x' \in X^{+}$ and
  $s,s' \in \state{\mathsf{M}}$, then
  $\mathsf{d}_{\mathsf{X}} \circ \oc \mathsf{M}$ has the following
  transition:
  \begin{center}
    \begin{tikzpicture}[rounded corners]
      \node (M) [draw,circle,obj] at (0,0) {$\oc \mathsf{M}$};
      \draw (M) -- node[above=4pt,obj]{$\oc\mathsf{X}$} ++(1.5,0);
      \node (x) [obj] at (1.5,-0.15) {$(0,x)$};
      \node (x') [obj] at (1.5,0.15) {$(0,x')$};
      \draw[thick,->] (x) -- ++ (-1.3,0) -- ++ (0,0.3) -- (x');
      \node [left,obj] at (M.west)
      {$(s,s_{1},s_{2},\ldots)/(s',s_{1},s_{2},\ldots)$};
      \begin{scope}[xshift=3cm]
        \node (d) [draw,circle,obj] at (0,0) {$\mathsf{d}_{\mathsf{X}}$};
        \draw (d) -- node[above=4pt,obj] {$\mathsf{X}$} ++ (1.5,0);
        \draw (d) -- ++(-1.5,0);
        \draw[thick,->] (1.5,-0.15) node[right,obj]{$x$} -- (x);
        \draw[thick,<-] (1.5,0.15) node[right,obj]{$x'$} -- (x');
      \end{scope}
    \end{tikzpicture}
  \end{center}
  for all $s_{1},s_{2},\ldots \in \state{\mathsf{M}}$.
  \begin{proposition}\label{prop:d}
    For any Mealy machine
    $\mathsf{M} \colon \mathsf{I} \multimap \mathsf{X}$,
    \begin{equation*}
      \mathsf{d}_{\mathsf{X}} \circ \oc \mathsf{M}
      \simeq \mathsf{M}.
    \end{equation*}
    Diagrammatically, we have
    \begin{center}
      \begin{tikzpicture}[rounded corners]
        \node (M) [draw,circle,obj] at (0,0) {$\oc \mathsf{M}$};
        \node (d) [draw,circle,obj] at (2,0) {$\mathsf{d}_{\mathsf{X}}$};
        \draw (M) -- node[above,obj]{$\oc\mathsf{X}$} (d);
        \draw (d) -- node[above,obj] {$\mathsf{X}$} ++ (1.5,0);
        \node at (4,0) {$\simeq$};

        \node (N) [draw,circle,obj] at (4.75,0) {$\mathsf{M}$};
        \draw (N) -- node[above,obj]{$\mathsf{X}$} ++(1.5,0);
      \end{tikzpicture}
      .
    \end{center}
  \end{proposition}

  \paragraph{Digging and Contraction}

  For natural numbers $n,m \in \mathbb{N}$, we write
  $\langle n,m \rangle$ for the Cantor pairing $n+(n+m)(n+m+1)/2$, and
  we write $n|_{0}$ and $n|_{1}$ for unique natural numbers such that
  $n = \langle n|_{0},n|_{1} \rangle$. For an $\mathbf{Int}$-object
  $\mathsf{X}$, let
  $\mathsf{dg}_{\mathsf{X}} \colon \oc\mathsf{X} \multimap \oc
  \oc\mathsf{X}$ and
  $\mathsf{con}_{\mathsf{X}} \colon \oc \mathsf{X} \multimap \oc
  \mathsf{X} \otimes \oc \mathsf{X}$ be stateless deterministic Mealy
  machines whose transition functions
  \begin{align*}
    \tran{\mathsf{dg}_{\mathsf{X}}}
    &\colon
      \mathbb{N} \times X^{+} +
      \mathbb{N} \times \mathbb{N} \times X^{-}
      \to
      \mathbb{N} \times \mathbb{N} \times X^{+}
      + \mathbb{N} \times X^{-}, \\
    \tran{\mathsf{con}_{\mathsf{X}}} 
    &\colon
      \mathbb{N} \times X^{+} +
      (\mathbb{N} \times X^{-} + \mathbb{N} \times X^{-})
      \to
      (\mathbb{N} \times X^{+}
      + \mathbb{N} \times X^{+}) + \mathbb{N} \times X^{-}.
  \end{align*}
  are given by
  \begin{align*}
    \tran{\mathsf{dg}_{\mathsf{X}}}(\hh,(\langle n,m \rangle,x))
    &= (\hh,(n,m,x)) \\
    \tran{\mathsf{dg}_{\mathsf{X}}}(\mm,(n,m,x))
    &= (\mm,(\langle n,m \rangle,x)), \\
    \tran{\mathsf{con}_{\mathsf{X}}}(\hh,(n,x))
    &= 
      \begin{cases}
        (\hh,(\hh,(n/2,x)))
        , & \textnormal{if } n \textnormal{ is even}, \\
        (\hh,(1,((n-1)/2,x)))
        , & \textnormal{if } n \textnormal{ is odd}, \\
      \end{cases} \\
    \tran{\mathsf{con}_{\mathsf{X}}}(\mm,(\hh,(u,x)))
    &= (\mm,(2n+1,x)), \\
    \tran{\mathsf{con}_{\mathsf{X}}}(\mm,(\mm,(u,x)))
    &= (\mm,(2n,x)).
  \end{align*}
  These stateless Mealy machines $\mathsf{dg}_{\mathsf{X}}$ and
  $\mathsf{con}_{\mathsf{X}}$ behave as follows: for all
  $n,m \in \mathbb{N}$,
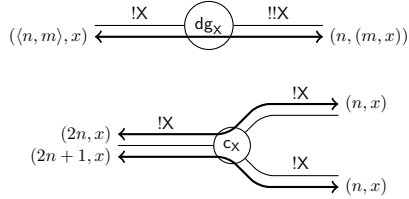
  \begin{center}
    \begin{tikzpicture}
      \node (dg) [draw,circle,obj] at (0,0) {$\mathsf{dg}_{\mathsf{X}}$};
      \draw (dg) -- node[above,obj]{$\oc\oc\mathsf{X}$} ++(1.5,0);
      \draw (dg) -- node[above,obj]{$\oc\mathsf{X}$} ++(-1.5,0);
      \draw[thick,<->] (1.5,-0.15) node[right,obj]{$(n,(m,x))$}
      -- ++(-3,0) node[left,obj]{$(\langle n,m\rangle,x)$};
    \end{tikzpicture}
  \end{center}
  \begin{center}
    \begin{tikzpicture}[rounded corners]
      \node(c) [draw,circle,obj] at (0,0) {$\mathsf{c}_{\mathsf{X}}$};
      \draw (c) -- node[above=4pt,obj]{$\oc\mathsf{X}$} ++(-1.5,0);
      \draw (c) -- ++ (0.4,0.4) -- node[above=4pt,obj]{$\oc\mathsf{X}$}++ (1,0);
      \draw (c) -- ++ (0.4,-0.4) -- node[above,obj]{$\oc\mathsf{X}$}++ (1,0);
      \draw[thick,<->] (-1.5,0.15) node[left,obj]{$(2n,x)$}
      -- ++ (1.45,0) -- ++ (0.4,0.4) -- ++ (1.05,0) node[right,obj]{$(n,x)$};
      \draw[thick,<->] (-1.5,-0.15) node[left,obj]{$(2n+1,x)$}
      -- ++ (1.45,0) -- ++ (0.4,-0.4) -- ++ (1.05,0) node[right,obj]{$(n,x)$};
    \end{tikzpicture}
  \end{center}

  \begin{proposition}\label{prop:dg_and_con}
    For any Mealy machine $\mathsf{M} \colon \mathsf{X} \multimap
    \mathsf{Y}$,
    \begin{equation*}
      \mathsf{dg}_{\mathsf{Y}} \circ \oc \mathsf{M}
      \simeq \oc \oc \mathsf{M}
      \circ \mathsf{dg}_{\mathsf{X}},
      \qquad
      \mathsf{con}_{\mathsf{Y}} \circ \oc \mathsf{M}
      \simeq
      (\oc \mathsf{M} \otimes \oc \mathsf{M}) \circ \mathsf{con}_{\mathsf{X}}.
    \end{equation*}
    Diagrammatically, we have
    
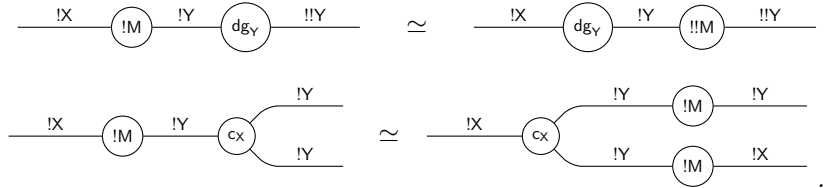
\begin{center}
      \begin{tikzpicture}[rounded corners]
        \node (M)[draw,circle,obj] at (0,0) {$\oc \mathsf{M}$};
        \node (dg)[draw,circle,obj] at (1.5,0) {$\mathsf{dg}_{\mathsf{Y}}$};
        \draw (M) --node[above,obj]{$\oc\mathsf{Y}$} (dg);
        \draw (M)--node[above,obj]{$\oc\mathsf{X}$}++(-1.5,0);
        \draw (dg)--node[above,obj]{$\oc\oc\mathsf{Y}$}++(1.5,0);
        \node at (3.75,0) {$\simeq$};
        \begin{scope}[xshift=6cm]
          \node (M)[draw,circle,obj] at (0,0) {$\mathsf{dg}_{\mathsf{Y}}$};
          \node (dg)[draw,circle,obj] at (1.5,0) {$\oc\oc\mathsf{M}$};
          \draw (M) --node[above,obj]{$\oc\mathsf{Y}$} (dg);
          \draw (M)--node[above,obj]{$\oc\mathsf{X}$}++(-1.5,0);
          \draw (dg)--node[above,obj]{$\oc\oc\mathsf{Y}$}++(1.5,0);
        \end{scope}
      \end{tikzpicture}
    \end{center}
    \begin{center}
      \begin{tikzpicture}[rounded corners]
        \node (M) [draw,circle,obj] at (-0.5,0) {$\oc \mathsf{M}$};
        \node (c) [draw,circle,obj] at (1,0) {$\mathsf{c}_{\mathsf{X}}$};
        \draw (M) -- node[above,obj]{$\oc\mathsf{Y}$}(c);
        \draw (M) -- node[above,obj]{$\oc\mathsf{X}$}++(-1.5,0);
        \draw (c) -- ++(0.4,0.4) -- node[above,obj]{$\oc\mathsf{Y}$}++(1,0);
        \draw (c) -- ++(0.4,-0.4) -- node[above,obj]{$\oc\mathsf{Y}$}++(1,0);
        \node at (3,0) {$\simeq$};
        \begin{scope}[xshift=5cm]
          \node (c) [draw,circle,obj] at (0,0) {$\mathsf{c}_{\mathsf{X}}$};
          \node (M0) [draw,circle,obj] at (2,0.4) {$\oc \mathsf{M}$};
          \node (M1) [draw,circle,obj] at (2,-0.4) {$\oc \mathsf{M}$};
          \draw (M0) -- node[above,obj]{$\oc\mathsf{Y}$} ++(1.5,0);
          \draw (M1) -- node[above,obj]{$\oc\mathsf{X}$} ++(1.5,0);
          \draw (c) -- ++(0.4,0.4) -- node[above,obj]{$\oc\mathsf{Y}$} (M0);
          \draw (c) -- ++(0.4,-0.4) -- node[above,obj]{$\oc\mathsf{Y}$} (M1);
          \draw (c) -- node[above,obj]{$\oc\mathsf{X}$}++(-1.5,0);
        \end{scope}
      \end{tikzpicture}
      .
    \end{center}
  \end{proposition}

  \paragraph{Weakening}

  We define a token machine $\mathsf{w}_{X} \colon \mathsf{X} \to \mathsf{I}$
  by
  \begin{equation*}
    \tran{\mathsf{w}_{X}}
    = \textnormal{the empty partial function}.
  \end{equation*}
  Because the identity is the only Mealy machine from $\mathsf{I}$ to $\mathsf{I}$
  (up to behavioural equivalence),
  we see that for any Mealy machine $\mathsf{M} \colon \mathsf{I} \multimap \mathsf{X}$,
  \begin{equation*}
    \mathsf{w}_{\mathsf{X}} \circ \mathsf{M}
    \simeq \mathsf{id}_{\mathsf{I}}.
  \end{equation*}
  This behavioural equivalence means that we can remove
  \begin{center}
    \begin{tikzpicture}
      \node (M) [draw,circle,obj] at (0,0) {$\mathsf{M}$};
      \node (w) [draw,circle,obj] at (1,0) {$\mathsf{w}_{\mathsf{X}}$};
      \draw (M) --node[above,obj] {$\mathsf{X}$} (w);
    \end{tikzpicture}
  \end{center}
  from any diagram.
}

\subsubsection{Real Numbers}
\label{sec:num}
We define an $\mathbf{Int}$-object
$\mathsf{R}$ to be 
\begin{math}
  (\mathbb{S},\mathbb{S})
\end{math}
where $\mathbb{S}$ is the measurable space of all finite
sequences of real numbers endowed with the following $\sigma$-algebra
\begin{equation*}
  A \in \Sigma_{\mathbb{S}}
  \iff
  A \cap \mathbb{R}^{n} \in \Sigma_{\mathbb{R}^{n}}
  \textnormal{ for all } n \in \mathbb{N}.
\end{equation*}
For $a \in \mathbb{R}$, we define a token machine
$\mathsf{r}_{a} \colon \mathsf{I} \multimap \mathsf{R}$ by
\begin{equation*}
  \raisebox{0.2cm}{$\tran{\mathsf{r}_{a}}(\mm,u) = (\hh,a \mathbin{::} u)$}
  \qquad
  \begin{tikzpicture}
    \node (r) [draw,circle,obj] at (0,0) {$\mathsf{r}_{a}$};
    \draw (r) -- node[above=4,obj] {$\mathsf{R}$} ++ (1.5,0);
    \draw[thick] (1.5,-0.15) node[right,obj] {$u$} -- ++ (-1.2,0);
    \draw[thick,<-] (1.5,0.15) node[right,obj] {$a\mathbin{::}u$} -- ++ (-1.2,0);
    \draw[thick] (0.3,0.15) arc(90:270:0.15);
  \end{tikzpicture}
  .
\end{equation*}
The transition means that given a query $u$ from environment,
$\mathsf{r}_{a}$ answers its value $a$ by appending $a$ to $u$. We
will use $u$ as a stack. \longshortv{See Section~\ref{sec:mf} and
  Section~\ref{sec:score}.}{See Section~\ref{sec:score}.}

\subsubsection{Measurable Functions}
\label{sec:mf}
We associate a measurable function
$f \colon \mathbb{R}^{n} \to \mathbb{R}$ with a token machine
$\mathsf{fn}_{f} \colon \mathsf{R}^{\otimes n} \multimap \mathsf{R}$.
For simplicity, we define $\mathsf{fn}_{f}$
\longshortv{for $n = 1$ and $n = 2$.}{for $n = 1$.}
When $n = 1$, the transition function
\begin{math}
  \tran{\mathsf{fn}_{f}} \colon
  \mathbb{S} + \mathbb{S}
  \to
  \mathbb{S} + \mathbb{S}
\end{math}
is given by
\begin{align*}
  \tran{\mathsf{fn}_{f}}(\mm,u)
  &= (\mm,u), \\
  \tran{\mathsf{fn}_{f}}(\hh,u)
  &= 
    \begin{cases}
      (\hh,f(a)\mathbin{::}u')
      , & \textnormal{if } u = a\mathbin{::}u', \\
      \textnormal{undefined}
      , & \textnormal{otherwise}.
    \end{cases}
\end{align*}
We explain how $\mathsf{fn}_{f}$ simulates $f$ by describing execution
of $\mathsf{fn}_{f} \circ \mathsf{r}_{a}$ for a real number
$a \in \mathbb{R}$. As in the following diagram, given an input
$u \in \mathbb{S}$ from the right $\mathsf{R}$-edge, $\mathsf{fn}_{f}$
sends $u$ to the left $\mathsf{R}$-edge in order to obtain the value
of its argument. The return value to $\mathsf{fn}_{f}$ from
$\mathsf{r}_{a}$ is $a\mathbin{::}u$, by which $\mathsf{fn}_{f}$ sees
that its argument is $a$. Then, $\mathsf{fn}_{f}$ outputs
$f(a)\mathbin{::}u$. As a whole, the following Mealy machine is
behaviourally equivalent to $\mathsf{r}_{f(a)}$.
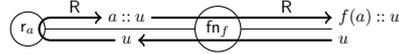
\begin{center}
  \begin{tikzpicture}
    \node (f) [draw,circle,obj] at (0,0) {$\mathsf{fn}_{f}$};
    \node (r) [draw,circle,obj] at (-2.5,0) {$\mathsf{r}_{a}$};
    \node (u) [obj] at (-1.2,-0.15) {$u$} ;
    \node (au) [obj] at (-1.2,0.15) {$a\mathbin{::}u$} ;
    \draw (r) -- node[above=4,obj] {$\mathsf{R}$} ++ (1,0) -- (f);
    \draw (f) -- node[above=4,obj] {$\mathsf{R}$} ++(1.5,0);

    \draw[thick,->] (1.5,-0.15) node[right,obj] {$u$} -- (u);
    \draw[thick] (u) -- ++(-1,0);
    \draw[thick] (-2.2,0.15) arc(90:270:0.15);
    \draw[thick,->] (-2.2,0.15) -- (au);
    \draw[thick,<-] (1.5,0.15) node[right,obj] {$f(a)\mathbin{::}u$} -- (au);
  \end{tikzpicture}
\end{center}
\longv{
  When $n = 2$, the transition function of $\mathsf{fn}_{f} \colon \mathsf{R}
  \otimes \mathsf{R} \multimap \mathsf{R}$ is
  $\tran{\mathsf{fn}_{f}} \colon (\mathbb{S} + \mathbb{S}) + \mathbb{S}
  \to \mathbb{S} + (\mathbb{S} + \mathbb{S})$ given by
  \begin{align*}
    \tran{\mathsf{fn}_{f}}(\hh,(\hh,u))
    &= (\mm,(\hh,u)), \\
    \tran{\mathsf{fn}_{f}}(\hh,(\mm,u))
    &= 
      \begin{cases}
        (\hh,f(a,b) \cons v)
        , & \textnormal{if } u = a \cons b \cons v, \\
        \textnormal{undefined}
        , & \textnormal{otherwise}, \\
      \end{cases} \\
    \tran{\mathsf{fn}_{f}}(\mm,u)
    &= (\mm,(\mm,u)).
  \end{align*}
  As in the following diagram, given an input $u \in \mathbb{S}$ from
  the right $\mathsf{R}$-edge, $\mathsf{fn}_{f}$ first sends $u$ to
  the lower $\mathsf{R}$-edge in the left hand side in order to obtain
  the value of its first argument. The return value to
  $\mathsf{fn}_{f}$ from $\mathsf{r}_{a}$ is $a \cons u$. Next,
  $\mathsf{fn}_{f}$ sends $a \cons u$ to the upper
  $\mathsf{R}$-edge in the left hand side. Then $\mathsf{r}_{b}$
  returns $b \cons a \cons u$. Now, $\mathsf{fn}_{f}$ sees that
  its first argument is $a$ and its second argument is $b$. Finally,
  $\mathsf{fn}_{f}$ outputs $f(a,b) \cons u$.
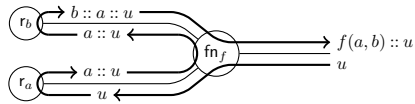
  \begin{center}
    \begin{tikzpicture}[rounded corners]
      \node (t) [draw,circle,obj] at (0,0) {$\mathsf{fn}_{f}$};
      \node (a) [draw,circle,obj] at (-2.5,-0.4) {$\mathsf{r}_{a}$};
      \node (b) [draw,circle,obj] at (-2.5,0.4) {$\mathsf{r}_{b}$};
      \draw (a) -- ++ (2,0) -- (t);
      \draw (b) -- ++ (2,0) -- (t);
      \draw (1.5,0) -- ++ (-1,0) -- (t);

      \node (u) at (-1.5,-0.55) [obj] {$u$};
      \node (au) at (-1.5,-0.25) [obj] {$a \cons u$};
      \node (au') at (-1.5,0.25) [obj] {$a \cons u$};
      \node (bau) at (-1.5,0.55) [obj] {$b \cons a \cons u$};

      \draw[<-,thick] (u) -- ++ (1.05,0) -- ++ (0.5,0.4) -- ++ (1.45,0) node[right,obj] {$u$};
      \draw[->,thick] (bau) -- ++ (1.05,0) -- ++ (0.5,-0.4) -- ++ (1.45,0)
      node[right,obj] {$f(a,b) \cons u$};
      \draw[thick] (u) -- ++(-0.7,0);
      \draw[thick] (-2.2,-0.55) arc(270:90:0.15);
      \draw[thick,->] (-2.2,-0.25) -- (au);
      \draw[thick] (au) -- ++ (1,0);
      \draw[thick,<-] (au') -- ++ (1,0);
      \draw[thick] (-0.5,-0.25) arc(-90:90:0.25);

      \draw[thick] (au') -- ++(-0.7,0);
      \draw[thick] (-2.2,0.25) arc(270:90:0.15);
      \draw[thick,->] (-2.2,0.55) -- (bau);
    \end{tikzpicture}
  \end{center}
  For general cases, $f$ may have more arguments, and $\mathsf{fn}_{f}$
  sequentially sends queries to its arguments storing partial
  information about its arguments on finite sequences of real numbers.
}

\longv{
  \subsubsection{Conditional Branching}
  \label{sec:cond}

  For an $\mathbf{Int}$-object $\mathsf{X}$ such that $X^{-}$ is a
  measurable subspace of $\mathbb{S}$, 
  we define
  \begin{equation*}
    \mathsf{cd} \colon
    \mathsf{R} \otimes (\mathsf{X} \otimes
    \mathsf{X}) \to \mathsf{X} 
  \end{equation*}
  to be a token machine whose transition function
  \begin{equation*}
    \tran{\mathsf{cd}_{\mathsf{X}}} \colon
    (\mathbb{S} + (X^{+} + X^{+})) + X^{-}
    \to
    X^{+} + ((X^{-} + X^{-}) + \mathbb{S})
  \end{equation*}
  is given by
  \begin{align*}
    \tran{\mathsf{cd}_{\mathsf{X}}}
    (\hh,(\hh,u))
    &=
      \begin{cases}
        (\mm,(\hh,(\mm,v)))
        , & \textnormal{if } u = 0 \cons v \textnormal{ and }
        v \in X^{-}, \\
        (\mm,(\hh,(\hh,v)))
        , & \textnormal{if } u = a \cons v \textnormal{ and } a \neq 0 \textnormal{ and }
        v \in X^{-}, \\
        \textnormal{undefined}
        , & \textnormal{otherwise}, \\
      \end{cases} \\
    \tran{\mathsf{cd}_{\mathsf{X}}}
    (\hh,(\mm,(\hh,x)))
    &= (\hh,x), \\
    (\hh,(\mm,(\mm,x)))
    &= (\hh,x), \\
    \tran{\mathsf{cd}_{\mathsf{X}}}
    (\mm,u)
    &= (\mm,(\mm,u)).
  \end{align*}
  For a real number $a \in \mathbb{R}$ and Mealy machines
  $\mathsf{M},\mathsf{N} \colon \mathsf{I} \multimap \mathsf{X}$, we
  describe execution of
  \begin{math}
    \mathsf{cd}_{\mathsf{X}} \circ (\mathsf{r}_{a} \otimes \mathsf{M}
    \otimes \mathsf{N}). 
  \end{math}
  Given an input $u \in X^{-}$, then $\mathsf{cd}_{\mathsf{X}}$ tries
  to check whether $a$ is zero or not by sending $u$ to the
  $\mathsf{R}$-edge. There are two cases: (i) if $a$ is $0$, then
  $\mathsf{r}_{a}$ returns $0 \cons u$, and $\mathsf{cd}_{\mathsf{X}}$
  forwards $u$ to the middle $\mathsf{X}$-edge; (ii) if $a$ is not
  $0$, say $1$, then $\mathsf{r}_{a}$ returns $1 \cons u$, and
  $\mathsf{cd}_{\mathsf{X}}$ forwards $u$ to the upper
  $\mathsf{X}$-edge:
  
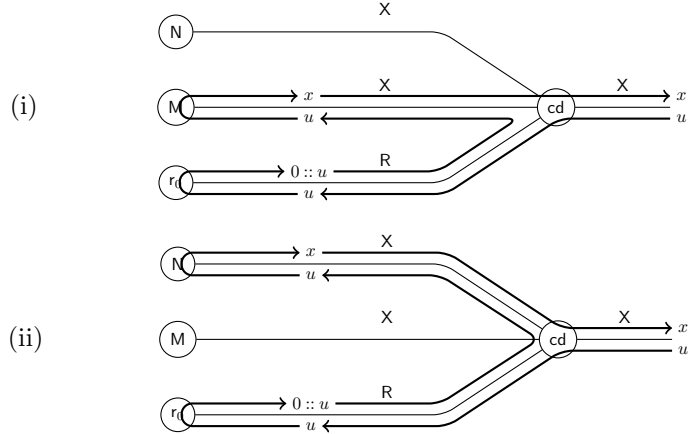
\begin{center}
    \begin{tikzpicture}[rounded corners]
      \node (cd) [draw,circle,obj] at (0,0) {$\mathsf{cd}$};
      \node (r) [draw,circle,obj] at (-5,-1) {$\mathsf{r}_{0}$};
      \node (M) [draw,circle,obj] at (-5,0) {$\mathsf{M}$};
      \node (N) [draw,circle,obj] at (-5,1) {$\mathsf{N}$};

      \draw (r) -- ++(2,0) -- node[above=4pt,obj]{$\mathsf{R}$} ++ (1.5,0) -- (cd);
      \draw (M) -- ++(2,0) -- node[above=4pt,obj]{$\mathsf{X}$} ++ (1.5,0) -- (cd);
      \draw (N) -- ++(2,0) -- node[above=4pt,obj]{$\mathsf{X}$} ++ (1.5,0) -- (cd);
      \draw (cd) -- node[above=4pt,obj]{$\mathsf{X}$} ++ (1.5,0);

      \node (1) [obj] at (-3.25,-1.15) {$u$};
      \node (2) [obj] at (-3.25,-0.85) {$0 \cons u$};
      \node (3) [obj] at (-3.25,-0.15) {$u$};
      \node (4) [obj] at (-3.25,0.15) {$x$};
      \node (5) [obj,right] at (1.5,-0.15) {$u$};
      \node (6) [obj,right] at (1.5,0.15) {$x$};

      \draw[thick,->] (5) -- ++ (-1.6,0) -- ++ (-1.5,-1) -- (1);
      \draw[thick,->] (1) -- ++ (-1.7,0) -- ++ (0,0.3) -- (2);
      \draw[thick,->] (2) -- ++ (1.7,0) -- ++(1.05,0.7) -- ++ (3);
      \draw[thick,->] (3) -- ++ (-1.7,0) -- ++ (0,0.3) -- (4);
      \draw[thick,->] (4) -- (6);

      \node at (-7,0) {(i)};
    \end{tikzpicture}
  \end{center}
  \begin{center}
    \begin{tikzpicture}[rounded corners]
      \node (cd) [draw,circle,obj] at (0,0) {$\mathsf{cd}$};
      \node (r) [draw,circle,obj] at (-5,-1) {$\mathsf{r}_{0}$};
      \node (M) [draw,circle,obj] at (-5,0) {$\mathsf{M}$};
      \node (N) [draw,circle,obj] at (-5,1) {$\mathsf{N}$};

      \draw (r) -- ++(2,0) -- node[above=4pt,obj]{$\mathsf{R}$} ++ (1.5,0) -- (cd);
      \draw (M) -- ++(2,0) -- node[above=4pt,obj]{$\mathsf{X}$} ++ (1.5,0) -- (cd);
      \draw (N) -- ++(2,0) -- node[above=4pt,obj]{$\mathsf{X}$} ++ (1.5,0) -- (cd);
      \draw (cd) -- node[above=4pt,obj]{$\mathsf{X}$} ++ (1.5,0);

      \node (1) [obj] at (-3.25,-1.15) {$u$};
      \node (2) [obj] at (-3.25,-0.85) {$0 \cons u$};
      \node (3) [obj] at (-3.25,0.85) {$u$};
      \node (4) [obj] at (-3.25,1.15) {$x$};
      \node (5) [obj,right] at (1.5,-0.15) {$u$};
      \node (6) [obj,right] at (1.5,0.15) {$x$};

      \draw[thick,->] (5) -- ++ (-1.6,0) -- ++ (-1.5,-1) -- (1);
      \draw[thick,->] (1) -- ++ (-1.7,0) -- ++ (0,0.3) -- (2);
      \draw[thick,->] (2) -- ++ (1.7,0) -- ++(1.3,0.85) --
      ++ (-1.3,0.85) -- ++ (3);
      \draw[thick,->] (3) -- ++ (-1.7,0) -- ++ (0,0.3) -- (4);
      \draw[thick,<-] (6) -- ++ (-1.6,0) -- ++ (-1.5,1) -- (4);

      \node at (-7,0) {(ii)};
    \end{tikzpicture}
  \end{center}
  Because in both cases, all outputs from $\mathsf{M}$ and $\mathsf{N}$
  are sent to the $\mathsf{X}$-edge in the right hand, we see
  that $\mathsf{cd}_{\mathsf{X}} \circ (\mathsf{r}_{a} \otimes
  \mathsf{M} \otimes \mathsf{N})$ simulates $\mathsf{M}$
  when $a = 0$ and simulates $\mathsf{N}$ when $a \neq 0$.
  \begin{proposition}
    For $a \in \mathbb{R}$ and for Mealy machines
    $\mathsf{M},\mathsf{N} \colon \mathsf{I} \to \mathsf{X}$, we have
    \begin{align*}
      \mathsf{cd}_{\mathsf{X}} \circ (\mathsf{r}_{a} \otimes
      \mathsf{M} \otimes \mathsf{N})
      &\simeq
        \begin{cases}
          \mathsf{M}, & \textnormal{if } a = 0, \\
          \mathsf{N}, & \textnormal{if } a \neq 0.
        \end{cases}
    \end{align*}
    Diagrammatically, we have
    
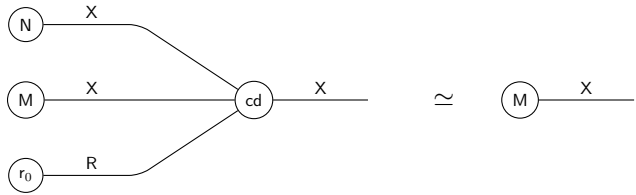
\begin{center}
      \begin{tikzpicture}[rounded corners]
        \node (cd) [draw,circle,obj] at (0,0) {$\mathsf{cd}$};
        \node (r) [draw,circle,obj] at (-3,-1) {$\mathsf{r}_{0}$};
        \node (M) [draw,circle,obj] at (-3,0) {$\mathsf{M}$};
        \node (N) [draw,circle,obj] at (-3,1) {$\mathsf{N}$};

        \draw (r) -- node[above,obj]{$\mathsf{R}$} ++ (1.5,0) -- (cd);
        \draw (M) -- node[above,obj]{$\mathsf{X}$} ++ (1.5,0) -- (cd);
        \draw (N) -- node[above,obj]{$\mathsf{X}$} ++ (1.5,0) -- (cd);
        \draw (cd) -- node[above,obj]{$\mathsf{X}$} ++ (1.5,0);

        \node at (2.5,0) {$\simeq$};

        \node (L) [draw,circle,obj] at (3.5,0) {$\mathsf{M}$};
        \draw (L) -- node[above,obj] {$\mathsf{X}$} ++ (1.5,0);
      \end{tikzpicture}
    \end{center}
    and for any $a \neq 0$,
    
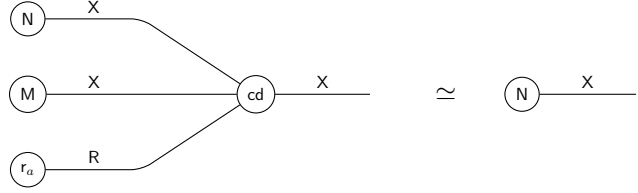
\begin{center}
      \begin{tikzpicture}[rounded corners]
        \node (cd) [draw,circle,obj] at (0,0) {$\mathsf{cd}$};
        \node (r) [draw,circle,obj] at (-3,-1) {$\mathsf{r}_{a}$};
        \node (M) [draw,circle,obj] at (-3,0) {$\mathsf{M}$};
        \node (N) [draw,circle,obj] at (-3,1) {$\mathsf{N}$};

        \draw (r) -- node[above,obj]{$\mathsf{R}$} ++ (1.5,0) -- (cd);
        \draw (M) -- node[above,obj]{$\mathsf{X}$} ++ (1.5,0) -- (cd);
        \draw (N) -- node[above,obj]{$\mathsf{X}$} ++ (1.5,0) -- (cd);
        \draw (cd) -- node[above,obj]{$\mathsf{X}$} ++ (1.5,0);

        \node at (2.5,0) {$\simeq$};

        \node (L) [draw,circle,obj] at (3.5,0) {$\mathsf{N}$};
        \draw (L) -- node[above,obj] {$\mathsf{X}$} ++ (1.5,0);
      \end{tikzpicture}
    \end{center}
  \end{proposition}
  \begin{proof}
    When $a = 0$, the first behavioral equivalence is realized by the
    first projection from
    $1 \times 1 \times \state{\mathsf{M}} \times \state{\mathsf{M}}
    \cong \state{\mathsf{M}} \times \state{\mathsf{M}}$ to
    $\state{\mathsf{M}}$. When $a \neq 0$, the first behavioral
    equivalence is realized by the second projection from
    $1 \times 1 \times \state{\mathsf{M}} \times \state{\mathsf{M}}
    \cong \state{\mathsf{M}} \times \state{\mathsf{M}}$ to
    $\state{\mathsf{N}}$. The second behavioral equivalence is realized
    by the obvious measurable function from
    $1 \times 1 \times \state{\mathsf{M}} \times \state{\mathsf{N}}$ to
    $1$.
  \end{proof}
}

\subsubsection{A State Monad}
\label{sec:state}
Let $\mathbb{T}$ be the subspace of $\mathbb{S}$ consisting of
all finite sequences of real numbers in $\mathbb{R}_{[0,1]}$. Recall that
$\mathbb{R}_{\geq 0} \times \mathbb{T}$ is ``the set of states'' in
sampling-based operational semantics and our idea is to model
$\mathtt{score}$ and $\mathtt{sample}$ by a state monad. In this
section, we give a state monad that we use in our Mealy machine
semantics. We define $\mathbf{Int}$-objects $\mathsf{S}_{0}$ and
$\mathsf{S}$ by
\begin{equation*}
  \mathsf{S}_{0} = (\mathbb{R}_{\geq 0} \times \mathbb{T},\emptyset),
  \qquad
  \mathsf{S} = (\mathbb{R}_{\geq 0} \times \mathbb{T},\mathbb{R}_{\geq 0} \times \mathbb{T}).
\end{equation*}
Then $\mathsf{S} \otimes (-)$ is a state monad (on $\mathbf{Mealy}$)
because for any $\mathbf{Int}$-object $\mathsf{X}$, we have
\begin{math}
  \mathsf{S} \otimes \mathsf{X}
  = ((\mathsf{S}_{0} \otimes \mathsf{X})^{\bot} \otimes \mathsf{S}_{0})^{\bot}.
\end{math}
The unit and the multiplication of this monad are:
\begin{align*}
  \mathsf{e} \otimes \mathsf{X}
  \colon \mathsf{X} \multimap \mathsf{S} \otimes \mathsf{X},
  \qquad
  \mathsf{m} \otimes \mathsf{X}
  \colon \mathsf{S} \otimes \mathsf{S} \otimes \mathsf{X} \multimap
  \mathsf{S} \otimes \mathsf{X}
\end{align*}
where $\mathsf{e} = \mathsf{unit}_{\mathsf{S}_{0}}$ and
$\mathsf{m} = \mathsf{S}_{0} \otimes \mathsf{counit}_{\mathsf{S}_{0}} \otimes
\mathsf{S}_{0}^{\bot}$. Note that $\mathsf{S}$ is equal to
$\mathsf{S}_{0} \otimes \mathsf{S}_{0}^{\bot}$.
We can depict the unit and the multiplication as
follows:
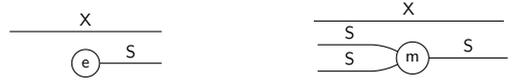
\begin{center}
  \begin{tikzpicture}[rounded corners,yscale=0.7]
    \draw (-1,0.7) -- node[above,obj] {$\mathsf{X}$} ++ (2,0);
    \node (h) [draw,circle,obj] at (0,0.1) {$\mathsf{e}$};
    \draw (h) -- node[above,obj]{$\mathsf{S}$}++ (1,0);
    \begin{scope}[xshift=3cm]
      \draw (0,0.9) -- node[above,obj] {$\mathsf{X}$} ++ (2.5,0);
      \node (mu) [draw,circle,obj] at (1.3,0.2) {$\mathsf{m}$};
      \draw (mu) -- node[above,obj]{$\mathsf{S}$}++(1.25,0);
      \draw (mu) --++(-0.4,-0.25) -- node[above,obj]{$\mathsf{S}$}++(-0.85,0);
      \draw (mu) --++(-0.4,0.25) -- node[above,obj]{$\mathsf{S}$}++(-0.85,0);
    \end{scope}
  \end{tikzpicture}
  .
\end{center}

\subsubsection{Scoring}
\label{sec:score}

We define $\mathsf{sc}$ to be a token machine from
$\mathsf{R}$ to $\mathsf{S}$ whose transition function
\begin{math}
  \tran{\mathsf{sc}} \colon
  \mathbb{S} + \mathbb{R}_{\geq 0} \times \mathbb{T}
  \to \mathbb{R}_{\geq 0} \times \mathbb{T} + \mathbb{S}
\end{math}
is given by
\begin{align*}
  &\tran{\mathsf{sc}}(\mm,(a,u))
  = (\mm,a\mathbin{::}u), \\
  &\tran{\mathsf{sc}}(\hh,u)
  =
  \begin{cases}
    (\hh,(|ab|,u')), & \textnormal{if } u = a\mathbin{::}b\mathbin{::}u' \textnormal{ and }
    u' \in \mathbb{T}, \\
    \textnormal{undefined}, & \textnormal{otherwise}.
  \end{cases}
\end{align*}
The token machine simulates scoring
$(\mathtt{score}(\mathtt{r}_{a}),b,u)
\to (\mathtt{skip},|a|\,b,u)$ as follows:
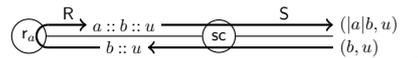
\begin{center}
  \begin{tikzpicture}
    \node (r) [draw,circle,obj] at (0,0) {$\mathsf{r}_{a}$};
    \node (sc) [draw,circle,obj] at (2.5,0) {$\mathsf{sc}$};
    \draw (r) -- node[above=4,obj] {$\mathsf{R}$} ++(0.8,0) -- (sc);
    \draw (sc) -- node[above=4,obj] {$\mathsf{S}$} ++ (1.5,0);
    \node (au) [obj] at (1.25,-0.15) {$b\mathbin{::}u$};
    \node (bau) [obj] at (1.25,0.15) {$a\mathbin{::}b\mathbin{::}u$};
    \draw[thick,<-] (au) -- (4,-0.15) node[right,obj] {$(b,u)$};
    \draw[thick,->] (bau) -- (4,0.15) node[right,obj] {$(|a|b,u)$};
    \draw[thick] (au) -- ++(-1,0);
    \draw[thick,<-] (bau) -- ++(-1,0);
    \draw[thick] (0.25,0.15) arc(90:270:0.15);    
  \end{tikzpicture}
  .
\end{center}

\subsubsection{Sampling}
\label{sec:sampling}
We define a Mealy machine $\mathsf{sa} \colon \mathsf{I}
\multimap \mathsf{S} \otimes \oc \mathsf{R}$ by: the state space
$\state{\mathsf{sa}}$ is defined to be $\{\ast\} \cup \mathbb{R}_{[0,1]}$, and
the initial state $\init{\mathsf{sa}}$ is $\ast$, and the
transition function
\longshortv{
  \begin{equation*}
    \tran{\mathsf{sa}} \colon
    (\emptyset + (\mathbb{N} \times \mathbb{S} + \mathbb{R}_{\geq 0} \times \mathbb{T}))
    \times \state{\mathsf{sa}}
    \to \\
    ((\mathbb{R}_{\geq 0} \times \mathbb{T} + \mathbb{N} \times \mathbb{S}) + \emptyset)
    \times \state{\mathsf{sa}}
  \end{equation*}
}{
  \begin{multline*}
    \tran{\mathsf{sa}} \colon
    (\emptyset + (\mathbb{N} \times \mathbb{S} + \mathbb{R}_{\geq 0} \times \mathbb{T}))
    \times \state{\mathsf{sa}}
    \to \\
    ((\mathbb{R}_{\geq 0} \times \mathbb{T} + \mathbb{N} \times \mathbb{S}) + \emptyset)
    \times \state{\mathsf{sa}}
  \end{multline*}
}
is given by
\longshortv{
  \begin{align*}
    \tran{\mathsf{sa}}((\mm,(\hh,(n,u))),s)
    &=
      \begin{cases}
        \textnormal{undefined}
        , & \textnormal{if } s = \ast, \\
        ((\hh,(\mm,(n,s\mathbin{::}u))),s)
        , & \textnormal{if } s \in \mathbb{R}, \\
      \end{cases} \\
    \tran{\mathsf{sa}}((\mm,(\mm,(a,u))),s)
    &=
      \begin{cases}
        ((\hh,(\hh,(a,v))),b)
        , & \textnormal{if } s=\ast \textnormal{ and }
        u = b\mathbin{::}v \\
        \textnormal{undefined}
        , & \textnormal{otherwise}.
      \end{cases}
  \end{align*}
}{
  \begin{align*}
    &\tran{\mathsf{sa}}((\mm,(\hh,(n,u))),s)
      = \\
    & \qquad
      \begin{cases}
        \textnormal{undefined}
        , & \textnormal{if } s = \ast, \\
        ((\hh,(\mm,(n,s\mathbin{::}u))),s)
        , & \textnormal{if } s \in \mathbb{R}, \\
      \end{cases} \\
    &\tran{\mathsf{sa}}((\mm,(\mm,(a,u))),s)
      = \\
    &\qquad
      \begin{cases}
        ((\hh,(\hh,(a,v))),b)
        , & \textnormal{if } s=\ast \textnormal{ and }
        u = b\mathbin{::}v \\
        \textnormal{undefined}
        , & \textnormal{otherwise}.
      \end{cases}
  \end{align*}
} As we explained in Section~\ref{sec:pn}, the Mealy machine
$\mathsf{sa}$ simulates the evaluation rule
$(\mathtt{sample},a,b\mathbin{::}u) \to (b,a,u)$:
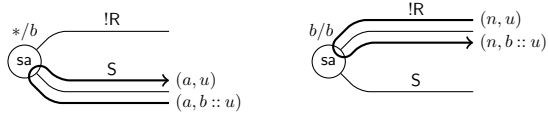
\begin{center}
  \begin{tikzpicture}[rounded corners]
    \node (sa) [draw,circle,obj] at (0,0) {$\mathsf{sa}$};
    \draw (sa) -- ++ (0.4,0.4) -- node[above,obj] {$\oc \mathsf{R}$}++ (1.5,0);
    \draw (sa) -- ++ (0.4,-0.4) -- node[above=4,obj] {$\mathsf{S}$}++ (1.5,0);
    \node[obj] at (0,0.4) {$\ast/b$};
    \draw[thick,->] (1.9,-0.55) node[right,obj] {$(a,b\mathbin{::}u)$}
    -- ++ (-1.55,0) -- ++(-0.35,0.35) -- ++ (0.2,0.2)
    -- ++ (0.25,-0.25) -- ++ (1.45,0) node[right,obj] {$(a,u)$};
    \begin{scope}[xshift=4cm]
      \node (sa) [draw,circle,obj] at (0,0) {$\mathsf{sa}$};
      \draw (sa) -- ++ (0.4,0.4) -- node[above=4,obj] {$\oc \mathsf{R}$}++ (1.5,0);
      \draw (sa) -- ++ (0.4,-0.4) -- node[above,obj] {$\mathsf{S}$}++ (1.5,0);
      \node[obj] at (-0.1,0.4) {$b/b$};
      \draw[thick,->] (1.9,0.55) node[right,obj] {$(n,u)$}
      -- ++ (-1.55,0) -- ++(-0.35,-0.35) -- ++ (0.2,-0.2)
      -- ++ (0.25,0.25) -- ++ (1.45,0) node[right,obj] {$(n,b\mathbin{::}u)$};
    \end{scope}
  \end{tikzpicture}
  .
\end{center}
Namely, $\mathsf{sa}$ pops $b$ from the trace, and then $\mathsf{sa}$
answers queries $(n,u)$ that the result of sampling is $b$.

\subsection{Diagrammatic Reasoning}
\label{sec:dr}

We now give a brief remark on diagrammatic presentation of Mealy
machines. The diagrammatic presentation of a Mealy machine is not only
for intuitive explanation, but also for rigorous reasoning about
behavioural equivalence. This follows from some categorical
observation. Let $\mathbf{Mealy}$ be the category of
$\mathbf{Int}$-objects and behavioural equivalence classes of Mealy
machines, where composition is induced by the composition of Mealy
machines. \longshortv{We will give a proof of the followng proposition
  in the next section.}{ The next proposition follows from
  some categorical argument on Mealy machines. See \cite{dlh2019}.}
\begin{proposition}\label{prop:cpt}
  The category $\mathbf{Mealy}$ is a compact closed category. The dual
  of an $\mathbf{Int}$-object $\mathsf{X}$ is $\mathsf{X}^{\bot}$. The
  unit and the counit arrows are $\mathsf{unit}_{\mathsf{X}}$ and
  $\mathsf{counit}_{\mathsf{X}}$.
\end{proposition}
Therefore, as a consequence of the coherence theorem for compact closed
categories \cite{kl1980,selinger2011}, we see that graph isomorphism
preserves behavioural equivalence.
\begin{proposition}\label{prop:beh_eq}
  If two Mealy machines have the same diagrammatic presentation modulo
  some rearrangement of edges and nodes, then they are behaviourally
  equivalent.
\end{proposition}
\longv{For example, for all Mealy machines
$\mathsf{M} \colon \mathsf{X} \otimes \mathsf{Y} \multimap \mathsf{Z}
\otimes \mathsf{W}$ and
$\mathsf{N} \colon \mathsf{W} \multimap \mathsf{Y}$, we have
\begin{center}
  \begin{tikzpicture}[rounded corners]
    \node (M) [draw,circle,obj] at (0,0) {$\mathsf{M}$};
    \node (N) [draw,circle,obj] at (1,0.4) {$\mathsf{N}$};
    \draw (N) -- ++(-0.5,0) -- node[above,obj]{$\mathsf{W}$} (M);
    \draw (N) -- ++ (0.5,0);
    \draw (2.5,0.4) arc(-90:90:0.2);
    \draw (2.5,0.4) -- ++(-0.3,0) -- ++(-0.4,0.4) -- ++(-0.3,0);
    \draw (2.5,0.8) -- ++(-0.3,0) -- ++(-0.4,-0.4) -- ++(-0.3,0);
    \draw (1.5,0.8) -- ++(-2.5,0);
    \draw (-1,0.8) arc(90:270:0.2);
    \draw (-1,0.4) -- ++(0.5,0) -- node[above,obj]{$\mathsf{Y}$} (M);
    \draw (M) -- ++(0.4,-0.4) -- node[above,obj]{$\mathsf{Z}$} ++(2.3,0);
    \draw (M) -- ++(-0.4,-0.4) -- node[above,obj]{$\mathsf{X}$} ++(-0.8,0);

    \node at (2.9,0) {$\simeq$};
    \begin{scope}[xshift=2.9cm]
      \node (M) [draw,circle,obj] at (2,0) {$\mathsf{M}$};
      \node (N) [draw,circle,obj] at (1,0.4) {$\mathsf{N}$};
      \draw (N) -- ++(0.5,0) --node[above,obj]{$\mathsf{Y}$} (M);
      \draw (N) -- ++ (-0.5,0);
      \draw (0.5,0.4) arc(270:90:0.2);
      \draw (0.5,0.8) -- ++(2.5,0);
      \draw (4,0.4) arc(-90:90:0.2);
      \draw (4,0.4) -- ++(-0.3,0) -- ++(-0.4,0.4) -- ++(-0.3,0);
      \draw (4,0.8) -- ++(-0.3,0) -- ++(-0.4,-0.4) -- ++(-0.3,0);
      \draw (3,0.4) -- node[above,obj]{$\mathsf{W}$} ++(-0.5,0) -- (M);
      \draw (M) -- ++(-0.4,-0.4) -- node[above,obj]{$\mathsf{X}$} ++(-1.3,0);
      \draw (M) -- ++(0.4,-0.4) -- node[above,obj]{$\mathsf{Z}$} ++(1.8,0);
    \end{scope}
  \end{tikzpicture}
  .
\end{center}}

\longv{
\subsection{Proof of Proposition~\ref{prop:cpt}}
\label{sec:proof_cpt}

\subsubsection{The Category of Partial Measurable Functions}
\label{sec:category}

For some basic categorical notions, see standard text books such as
\cite{maclane1998}. We define
$\mathbf{pMeas}$ to be the category of measurable spaces and partial
measurable functions. In $\mathbf{pMeas}$, the empty space $\emptyset$ is the initial
object, and the coproduct space $X + Y$ is the coproduct of $X$ and
$Y$ in the categorical sense. We write
\begin{equation*}
  \mathrm{inl}_{X,Y} \colon X \to X + Y,
  \qquad
  \mathrm{inr}_{X,Y} \colon Y \to X + Y
\end{equation*}
for the left/right injections. For partial measurable functions
$f \colon X \to Y$ and $g \colon Z \to Y$, we define
$[f,g] \colon X + Z \to Y$ to be the cotupling of $f$ and $g$.
For partial measurable functions
$f \colon X \to Y$ and $g \colon Z \to W$, we define
partial measurable functions $f + g \colon X + Z \to Y + W$ and
$f \times g \colon X \times Z \to Y \times W$ by
\begin{align*}
  (f + g)(\hh,x)
  &=
    \begin{cases}
      (\hh,y), & \textnormal{if } f(x) \textnormal{ is defined and is equal to } y, \\
      \textnormal{undefined}, & \textnormal{otherwise},
    \end{cases} \\
  (f + g)(\mm,z)
  &=
    \begin{cases}
      (\mm,w), & \textnormal{if } g(z) \textnormal{ is defined and is equal to } w, \\
      \textnormal{undefined}, & \textnormal{otherwise},
    \end{cases} \\
  (f \times g)(x,z)
  &=
    \begin{cases}
      (y,w), & \textnormal{if } f(x) \textnormal{ is defined and is equal to } y \\
      &\textnormal{and } g(w) \textnormal{ is defined and is equal to } z, \\
      \textnormal{undefined}, & \textnormal{otherwise}.
    \end{cases}
\end{align*}
We note that $(\mathbf{pMeas},1,\times)$ and
$(\mathbf{pMeas},\emptyset,+)$ are symmetric monoidal categories. We
also note that $X \times (-)$ distributes over the coproducts, i.e.,
the canonical arrow
\begin{equation*}
  \mathrm{dst}_{X,Y,Z} = [X \times \mathrm{inl}_{Y,Z},X \times \mathrm{inr}_{Y,Z}]
  \colon X \times Y + X \times Z
  \to X \times (Y + Z)  
\end{equation*}
is an isomorphism.

The notion of trace introduced by Joyal, Street and Verity \cite{jsv}
plays important role in this section.
\begin{definition}
  Let $(\mathcal{C},I,\otimes)$ be a symmetric monoidal
  category. A \emph{trace operator} on $(\mathcal{C},I,\otimes)$
  is a family
  $\left\{\mathbf{tr}_{X,Y}^{Z}\right\}_{X,Y,Z \in \mathcal{C}}$
  satisfying the following axioms:
  \begin{varitemize}
  \item (Dinaturality) For all $\mathcal{C}$-arrows
    $f \colon X \otimes Z \to Y \otimes Z$,
    $g \colon X' \to X$ and
    $h \colon Y \to Y'$,
    \begin{equation*}
      h \circ \mathbf{tr}_{X,Y}^{Z}(f) \circ g
      = \mathbf{tr}_{X',Y'}^{Z}((h \otimes Z) \circ f \circ (g \otimes Z)).
    \end{equation*}
  \item (Sliding) For all $\mathcal{C}$-arrows
    $f \colon X \otimes Z \to Y \otimes W$,
    $g \colon W \to Z$,
    \begin{equation*}
      \mathbf{tr}_{X,Y}^{Z}((Y \otimes g) \circ f)
      = \mathbf{tr}_{X,Y}^{Z}(f \circ (X \otimes g)).
    \end{equation*}
  \item (Vanishing I) For all $\mathcal{C}$-arrows
    $f \colon X \otimes I \to Y \otimes I$,
    \begin{equation*}
      \mathbf{tr}_{X,Y}^{I}(f) = f.
    \end{equation*}
  \item (Vanishing II) For all $\mathcal{C}$-arrows
    $f \colon X \otimes Z \otimes W \to Y \otimes Z \otimes W$,
    \begin{equation*}
      \mathbf{tr}_{X,Y}^{Z \otimes W}(f)
      = \mathbf{tr}_{X,Y}^{Z}\left(
      \mathbf{tr}_{X \otimes Z, Y \otimes Z}^{W}(f)\right).
    \end{equation*}
  \item (Superposing)
    For all $\mathcal{C}$-arrows
    $f \colon X \otimes Z \to Y \otimes Z$,
    \begin{equation*}
      W \otimes \mathbf{tr}_{X,Y}^{Z}(f)
      = \mathbf{tr}_{W \otimes X, W \otimes Y}^{W \otimes Z}(W \otimes f).
    \end{equation*}
  \item (Yanking) For all $X \in \mathcal{C}$,
    \begin{equation*}
      \mathbf{tr}_{X,X}^{X}(\sigma_{X,X})
      = \mathrm{\id}_{X}
    \end{equation*}
    where $\sigma_{X,Y} \colon X \otimes Y \to Y \otimes X$
    is the brading.
  \end{varitemize}
  A symmetric monoidal category $(\mathcal{C},I,\otimes)$
  endowed with a trace operator $\mathbf{tr}$ is called
  a \emph{traced symmetric monoidal category}.
\end{definition}

We give a trace operator on $(\mathbf{pMeas},\emptyset,+)$. The
symmetric monoidal category $(\mathbf{pMeas},\emptyset,+)$ is enriched
over $\omega\mathbf{Cppo}$, which is the cartesian category of pointed
$\omega$-cpos and continuous functions. The partial order on a hom-set
$\mathbf{pMeas}(X,Y)$ is given by
\begin{equation*}
  f \leq g \iff \textnormal{for all $x \in X$, if $f(x)$ is defined,
  then $g(x)$ is defined, and $f(x) = g(x)$.}
\end{equation*}
The least arrow $\bot_{X,Y} \colon X \to Y$ is the empty partial
measurable function. The $\omega\mathbf{Cppo}$-enrichment induces an
iterator
\begin{equation*}
  \mathbf{iter}_{X,Y} \colon \mathbf{pMeas}(X,Y + X)
  \to \mathbf{pMeas}(X, Y)
\end{equation*}
given by
\begin{equation*}
  \mathbf{iter}_{X,Y}(f) =
  \textnormal{the least fixed point of }
  \biggl(
  g \colon X \to Y
  \longmapsto
  [\id_{Y},g] \circ f \colon X \to Y
  \biggr).
\end{equation*}
The operator $\mathbf{iter}$ induces another operator
\begin{equation*}
  \mathbf{tr}_{X,Y}^{Z} \colon
  \mathbf{pMeas}(X + Z, Y + Z) \to \mathbf{pMeas}(X, Y)
\end{equation*}
given by
\begin{equation*}
  \mathbf{tr}_{X,Y}^{Z}(f) =
  [\id_{Y},\mathbf{iter}_{Z,Y}(f \circ \mathrm{inr}_{X,Z})]
  \circ f \circ \mathrm{inl}_{X,Z}.
\end{equation*}
Concretely, for a partial measurable function
$f \colon X + Z \to Y + Z$,
\begin{equation*}
  \mathbf{tr}_{X,Y}^{Z}(f)(x)
  \textnormal{ is defined and is equal to } y
\end{equation*}
if and only if either $f(\hh,x) = (\hh,y)$
or there is a finite sequence $z_{1},\ldots,z_{n} \in Z$ such that
\begin{equation*}
  f(\hh,x) = (\mm,z_{1}),
  \qquad
  f(\mm,z_{1}) = (\mm,z_{2}),
  \qquad
  \cdots
  \qquad
  f(\mm,z_{n-1}) = (\mm,z_{n}),
  \qquad
  f(\mm,z_{n}) = (\hh,y).
\end{equation*}
\begin{proposition}\label{prop:tr}
  The family of operators
  $\left\{\mathbf{tr}_{X,Y}^{Z}\right\}_{X,Y,Z \in \mathbf{pMeas}}$ is
  a trace operator of the symmetric monoidal category
  $(\mathbf{pMeas},\emptyset,+)$. Furthermore, the trace operator is
  \emph{uniform} \cite{hasegawa2003} with respect to partial
  measurable functions : for all partial measurable functions
  $f \colon X + Z \to Y + Z$, $f \colon X + W \to Y + W$ and
  $h \colon Z \to W$, if
  \begin{equation*}
    \xymatrix{
      X + Z
      \ar[r]^-{f}
      \ar[d]_-{X + h}
      &
      Y + Z
      \ar[d]^-{Y + h}
      \\
      X + W
      \ar[r]_-{g}
      &
      Y + W
    }
  \end{equation*}
  commutes, then
  \begin{equation*}
    \mathbf{tr}_{X,Y}^{Z}(f) = \mathbf{tr}_{X,Y}^{W}(g).
  \end{equation*}
\end{proposition}
\begin{proof}
  It is straightforward to adapt the argument in
  \cite[Section~A]{hoshino2014}.
\end{proof}

We will use the next proposition to construct a trace operator
for Mealy machines.

\begin{proposition}\label{prop:distributes}
  For any partial measurable function $f \colon X + Z \to Y + Z$ and a
  measurable space $W$,
   \begin{equation*}
     W \times \mathbf{tr}_{X,Y}^{Z}(f)
     =
     \mathbf{tr}_{W \times X,W \times Y}^{W \times Z}
     \left(
       \mathrm{dst}_{W,Y,Z}^{-1} \circ
       (W \times f) \circ \mathrm{dst}_{W,X,Z}
     \right).
  \end{equation*}
\end{proposition}
\begin{proof}
  For any $w \in W$, we show that
  \begin{equation*}
    \mathbf{tr}_{W \times X,W \times Y}^{W \times Z}
    \left(
      \mathrm{dst}_{W,Y,Z}^{-1} \circ
      (W \times f) \circ \mathrm{dst}_{W,X,Z}
    \right) \circ (w \times \id_{X})
    = w \times \mathbf{tr}_{X,Y}^{Z}(f)
  \end{equation*}
  where we identify $w$ with the arrow from $1=\{\ast\}$ to $W$ that
  sends $\ast$ to $w$. Because
  \begin{equation*}
    \mathrm{dst}_{W,Y,Z}^{-1} \circ
    (W \times f) \circ \mathrm{dst}_{W,X,Z}
    \circ (w \times X + w \times Z)
    = 
    (w \times Y + w \times Z) \circ f,
  \end{equation*}
  it follows from uniformity that
  \begin{equation*}
    \mathbf{tr}_{X,W \times Y}^{W \times Z}
    \left(
      \mathrm{dst}_{W,Y,Z}^{-1} \circ
      (W \times f) \circ \mathrm{dst}_{W,X,Z}
      \circ (w \times X + W \times Z)
    \right)
    = 
    \mathbf{tr}_{W \times X,Y}^{Z}((w \times Y + Z) \circ f).
  \end{equation*}
  By dinaturality, we obtain
  \begin{equation*}
    \mathbf{tr}_{W \times X,W \times Y}^{W \times Z}
    \left(
      \mathrm{dst}_{W,Y,Z}^{-1} \circ
      (W \times f) \circ \mathrm{dst}_{W,X,Z}
    \right) \circ (w \times \id_{X})
    = w \times \mathbf{tr}_{X,Y}^{Z}(f).
  \end{equation*}
  Since this is true for any $w \in W$,
  we see that $W \times \mathbf{tr}_{X,Y}^{Z}(f)$ is equal to
  \begin{equation*}
    \mathbf{tr}_{W \times X,W \times Y}^{W \times Z}
    \left(
      \mathrm{dst}_{W,Y,Z}^{-1} \circ
      (W \times f) \circ \mathrm{dst}_{W,X,Z}
    \right).
  \end{equation*}
\end{proof}

\subsubsection{The Category of Mealy Machines}
\label{sec:category_of_mealy_machines}

\begin{definition}
  We define a category $\mathbf{Mealy}$ by:
  \begin{varitemize}
  \item objects are $\mathbf{Int}$-objects $\mathsf{X}$; and
  \item arrows $f \colon \mathsf{X} \multimap \mathsf{Y}$ are
    behavioural equivalence classes of Mealy machines from
    $\mathsf{X}$ to $\mathsf{Y}$.
  \end{varitemize}
  We denote a wide subcategory of $\mathbf{Mealy}$ consisting of
  $\mathbf{Int}$-object $\mathsf{X}$ such $X^{-} = \emptyset$ by
  $\mathbf{Mealy}_{+}$.
\end{definition}
Intuitively, while arrows in $\mathbf{Mealy}$ are bidirectional Mealy
machines, arrows in $\mathbf{Mealy}_{+}$ are ``one-way'' Mealy
machines. We consider the wide subcategory $\mathbf{Mealy}_{+}$
because categorical structure of $\mathbf{Mealy}_{+}$ is easier to
describe than that of $\mathbf{Mealy}$, and categorical structure of
$\mathbf{Mealy}$ is induced by that of $\mathbf{Mealy}_{+}$.

The identity arrow and the composition of 
$\mathbf{Mealy}_{+}$ is given by the identity Mealy machine
$[\mathsf{id}_{\mathsf{X}}]$ and the composition of Mealy machine:
\begin{equation*}
  [\mathsf{M}] \circ [\mathsf{N}] = [\mathsf{M} \circ \mathsf{N}].
\end{equation*}
Concrete description of the composition of Mealy machines between
$\mathbf{Mealy}_{+}$-objects is easy: for $\mathbf{Mealy}_{+}$-objects
$\mathsf{X}$ and $\mathsf{Y}$, and for Mealy machines
$\mathsf{M} \colon \mathsf{X} \multimap \mathsf{Y}$ and
$\mathsf{N} \colon \mathsf{Y} \multimap \mathsf{Z}$, the composition
$\mathsf{N} \circ \mathsf{M} \colon \mathsf{X} \multimap \mathsf{Y}$
consists of:
\begin{itemize}
\item
  $\state{\mathsf{N} \circ \mathsf{M}} = \state{\mathsf{N}} \times
  \state{\mathsf{M}}$;
\item
  $\init{\mathsf{N} \circ \mathsf{M}} =
  (\init{\mathsf{N}},\init{\mathsf{M}})$;
\item
  $\tran{\mathsf{N} \circ \mathsf{M}}$ given by
  \begin{multline*}
    (X^{+} + \emptyset) \times \state{g} \times \state{f}
    \xrightarrow{\cong}
    (X^{+} + \emptyset) \times \state{f} \times \state{g}
    \xrightarrow{\tran{f} \times \state{g}} {} \\
    (Y^{+} + \emptyset) \times \state{f} \times \state{g}
    \xrightarrow{\cong}
    (Y^{+} + \emptyset) \times \state{g} \times \state{f}
    \xrightarrow{\tran{g} \times \state{f}}
    (Z^{+} + \emptyset) \times \state{g} \times \state{f}.
  \end{multline*}
\end{itemize}
The composition of transition functions makes sense because
$X^{-} = Z^{-} = \emptyset$. From this concrete description, it is
easy to check that the composition of $\mathbf{Mealy}_{+}$-arrows is
well-defined. In fact, if
$h \colon \state{\mathsf{M}} \to \state{\mathsf{M}'}$ realizes
$\mathsf{M} \preceq \mathsf{M}'$ and
$h' \colon \state{\mathsf{N}} \to \state{\mathsf{N}'}$ realizes and
$\mathsf{N} \preceq \mathsf{N}'$, then $h' \times h$ realizes
$\mathsf{N} \circ \mathsf{M} \preceq \mathsf{N}' \circ \mathsf{M}'$.
Therefore, the symmetric transitive closure $\simeq$ is compatible
with the composition.
\begin{proposition}
  The category $\mathbf{Mealy}_{+}$ with $(\mathsf{I},\otimes)$
  is a symmetric monoidal category where the monoidal product
  of $\mathbf{Mealy}_{+}$-arrows
  $[\mathsf{M}] \colon \mathsf{X} \multimap \mathsf{Y}$
  and $[\mathsf{N}] \colon \mathsf{Z} \multimap \mathsf{W}$
  is given by
  \begin{equation*}
    [\mathsf{M}] \otimes [\mathsf{N}]
    = [\mathsf{M} \otimes \mathsf{N}].
  \end{equation*}
\end{proposition}
\begin{proof}
  It is easy to see that objects in $\mathbf{Mealy}_{+}$ are closed
  under the monoidal product of $\mathbf{Int}$-objects. Thanks to
  simplicity of the composition of $\mathbf{Mealy}_{+}$-arrows, we can
  easily check that the monoidal product of Mealy machines between
  $\mathbf{Mealy}_{+}$-objects is compatible with behavioural
  equivalence and that $(\mathbf{Mealy}_{+},I,\otimes)$ is a symmetric
  monoidal category.
\end{proof}
Furthermore, $\mathbf{Mealy}_{+}$ inherits the trace operator of
$\mathbf{pMeas}$. For a $\mathbf{Mealy}_{+}$-arrow
$[\mathsf{M}] \colon \mathsf{X} \otimes \mathsf{Z} \multimap \mathsf{Y}
\otimes \mathsf{Z}$, we define a $\mathbf{Mealy}_{+}$-arrow
$\mathsf{Tr}_{\mathsf{X},\mathsf{Y}}^{\mathsf{Z}}[\mathsf{M}] \colon
\mathsf{X} \multimap \mathsf{Y}$ to be
the equivalence class of a Mealy machine
$\mathsf{N} \colon \mathsf{X} \multimap \mathsf{Y}$ given by
\begin{equation*}
  \state{\mathsf{N}} = \state{\mathsf{M}},
  \qquad
  \init{\mathsf{N}}
  = \init{\mathsf{M}}
\end{equation*}
and
\begin{equation*}
  \tran{\mathsf{N}} =
  \mathbf{tr}_{(X^{+} + \emptyset) \times \state{f},
    (Y^{+} + \emptyset) \times \state{f}}^{(Z^{+} + \emptyset)\times \state{f}}
  \left(
    \vcenter{
      \xymatrix{
        (X^{+} + \emptyset) \times \state{f} + (Z^{+} + \emptyset) \times \state{f}
        \ar[d]^{\cong}
        \\
        ((X^{+} + Z^{+}) + \emptyset) \times \state{f}
        \ar[d]^{\tran{\mathsf{M}}}
        \\
        ((Y^{+} + Z^{+}) + \emptyset) \times \state{f}
        \ar[d]^{\cong}
        \\
        (Y^{+} + \emptyset) \times \state{f} + (Z^{+} + \emptyset) \times \state{f}
      }
    }
  \right)
  .
\end{equation*}

\begin{proposition}
  The family of operators
  $\{\mathsf{Tr}_{\mathsf{X},\mathsf{Y}}^{\mathsf{Z}}\}_{
    \mathsf{X},\mathsf{Y},\mathsf{Z} \in \mathbf{Mealy}_{+}}$ is a
  trace operator on the symmetric monoidal category
  $(\mathbf{Mealy}_{+},\mathsf{I},\otimes)$.
\end{proposition}
\begin{proof}
  Well-definedness of
  $\mathsf{Tr}_{\mathsf{X},\mathsf{Y}}^{\mathsf{Z}}(-)$ follows from
  uniformity of the trace operator on $\mathbf{pMeas}$. Sliding,
  vanishing I, vanishing II, superposing and yanking for $\mathsf{Tr}$
  follow from that of $\mathbf{tr}$. Dinaturality for $\mathsf{Tr}$
  follows from dinaturality of $\mathbf{tr}$ and
  Proposition~\ref{prop:distributes}.
\end{proof}

We recall the notions of \emph{$\mathbf{Int}$-construction} \cite{jsv}
and \emph{compact closed category}.

\begin{definition}[$\mathbf{Int}$-construction]
  Let $(\mathcal{C},I,\otimes,\mathbf{tr})$ be a traced symmetric
  monoidal category. We define a category $\mathbf{Int}(\mathcal{C})$
  by:
  \begin{varitemize}
  \item objects are pairs $(X^{+},X^{-})$ of $\mathcal{C}$-objects;
  \item arrows from $(X^{+},X^{-})$ to $(Y^{+},Y^{-})$ are
    $\mathcal{C}$-arrows from $X^{+} \otimes Y^{-}$ to
    $Y^{+} \otimes X^{-}$.
  \end{varitemize}
  The identity on $(X^{+},X^{-})$ is given by the identity on
  $X^{+} \otimes X^{-}$, and the composition of
  $\mathbf{Int}(\mathcal{C})$-arrows
  $f \colon (X^{+},X^{-}) \to (Y^{+},Y^{-})$ and
  $g \colon (Y^{+},Y^{-}) \to (Z^{+},Z^{-})$ is given by
  \begin{equation*}
    \mathbf{tr}_{X^{+} \otimes Z^{-},Z^{+} \otimes X^{-}}^{Y^{-}}
    \left(
      (X^{+} \otimes \sigma_{Z^{-},Y^{-}}) \circ
      (f \otimes Z^{-}) \circ
      (Y^{+} \otimes \sigma_{X^{-},Z^{-}}) \circ
      (g \otimes X^{-}) \circ
      (Z^{+} \otimes \sigma_{Y^{-},X^{-}})
    \right).
  \end{equation*}
  Here, we omit some coherence isomorphisms.
\end{definition}

\begin{definition}
  A \emph{compact closed category} is a symmetric monoidal category
  $(\mathcal{C},I,\otimes)$ with a function
  $(-)^{\bot} \colon \mathrm{obj}(\mathcal{C}) \to \mathrm{obj}(\mathcal{C})$
  and families of $\mathcal{C}$-arrows
  \begin{equation*}
    \{\eta_{X} \colon I \to X \otimes X^{\bot}\}_{X \in \mathcal{C}},
    \qquad
    \{\epsilon_{X} \colon X^{\bot} \otimes X \to I\}_{X \in \mathcal{C}}
  \end{equation*}
  such that
  \begin{equation*}
    (\epsilon_{X} \otimes X^{\bot})\circ (X^{\bot} \otimes \eta_{X})
    = \id_{X^{\bot}},
    \qquad
    (X \otimes \epsilon_{X}) \circ (\eta_{X} \otimes X)
    = \id_{X}.
  \end{equation*}
  For $X \in \mathcal{C}$, the object $X^{\bot}$ is called the \emph{dual object}
  of $X$.
\end{definition}

\begin{theorem}[\cite{jsv}]
  The category $\mathbf{Int}(\mathcal{C})$
  is a compact closed category. The unit and the monoidal
  product are given by
  \begin{equation*}
    (I,I), \qquad
    (X^{+},X^{-}) \otimes (Y^{+},Y^{-})
    = (X^{+} \otimes Y^{+},Y^{-} \otimes X^{-}).
  \end{equation*}
  The dual object of $(X^{+},X^{-})$ is $(X^{-},X^{+})$.
  The unit arrow $\eta_{(X^{+},X^{-})}$ and
  the counit arrow $\epsilon_{(X^{+},X^{-})}$ are given by
  \begin{equation*}
    \eta_{(X^{+},X^{-})} = \id_{X^{+} \otimes X^{-}},
    \qquad
    \epsilon_{(X^{+},X^{-})} = \id_{X^{-} \otimes X^{+}}.
  \end{equation*}
\end{theorem}

\begin{corollary}
  The category $\mathbf{Mealy}$ is a compact closed category. The
  monoidal structure is given by $(\mathsf{I},\otimes)$, and the unit
  and the counit are given by $\mathsf{unit}_{\mathsf{X}}$ and
  $\mathsf{counit}_{\mathsf{X}}$ respectively.
\end{corollary}
\begin{proof}
  It is straightforward to check that
  $\mathbf{Mealy}$ is isomorphic to $\mathbf{Int}(\mathbf{Mealy}_{+})$,
  and the compact closed structure is given by
  data provided in Section~\ref{sec:mealy}.  
\end{proof}
}

\section{Mealy Machine Semantics for $\PCFSS$}
\label{sec:GoI2}

We interpret a type $\mathtt{A}$
as the $\mathbf{Int}$-object $\sem{\mathtt{A}}$ given by
\begin{equation*}
  \sem{\ttunit} = \mathsf{I},
  \quad
  \sem{\ttreal} = \mathsf{R},
  \quad
  \sem{\mathtt{A}\to\mathtt{B}}
  = \mathsf{S}
  \otimes \oc \sem{\mathtt{B}} \otimes \oc \sem{\mathtt{A}}^{\bot}.
\end{equation*}
We define interpretation of contexts by
\begin{equation*}
  \sem{\mathtt{x}:\mathtt{A},\ldots,
    \mathtt{y}:\mathtt{B}} = \sem{\mathtt{A}} \otimes
  \cdots \otimes \sem{\mathtt{B}}.
\end{equation*}
When $\mathtt{\Delta}$ is the empty sequence, we define
$\sem{\mathtt{\Delta}}$ to be $\mathsf{I}$.

\longv{ For interpreting conditional branching, we use the following
  proposition.
  \begin{proposition}\label{prop:emb}
    For any type $\mathtt{A}$, there is a partial measurable embedding
    $e \colon \mathbb{S} + \mathbb{N} \times \sem{\mathtt{A}}^{-} \to \mathbb{S}$.
  \end{proposition}
  \begin{proof}
    We first define an embedding from $\sem{\mathtt{A}}^{-}$ to
    $\mathbb{S}$ by induction on $\mathtt{A}$. We note that for any type
    $\mathtt{A}$, we have $\sem{\mathtt{A}}^{+} = \sem{\mathtt{A}}^{-}$.
    Base cases are easy. For induction step,
    \begin{align*}
      \sem{\mathtt{A} \to \mathtt{B}}^{-}
      = \mathbb{S} + \mathbb{N} \times \sem{\mathtt{B}}^{-} + \mathbb{N} \times \sem{\mathtt{A}}^{-}
      &\subseteq
        \mathbb{S} + \mathbb{N} \times \mathbb{S} + \mathbb{N} \times \mathbb{S}
        \tag{\textnormal{induction hypothesis}} \\
      &\subseteq
        \mathbb{S} + \mathbb{R} \times \mathbb{S} + \mathbb{R} \times \mathbb{S}
        \tag{$\mathbb{N} \subseteq \mathbb{R}$} \\
      &\subseteq
        \mathbb{S} + \mathbb{S} + \mathbb{S}
        \tag{$\mathbb{R} \times \mathbb{S} \subseteq \mathbb{S}$} \\
      &\cong
        \{0,1,2\} \times \mathbb{S} \\
      &\subseteq
        \mathbb{R} \times \mathbb{S} \\
      &\subseteq
        \mathbb{S}.
    \end{align*}
    The statement follows from
    $\mathbb{S} + \mathbb{N} \times \sem{\mathtt{A}}^{-}
    \subseteq \sem{\ttunit\to\mathtt{A}}^{-}$.
  \end{proof}
}

We interpret terms $\mathtt{\Delta} \vdash \mathtt{M} : \mathtt{A}$
and values $\mathtt{\Delta} \vdash \mathtt{V} : \mathtt{A}$ by
\begin{equation*}
  \sem{\mathtt{\Delta} \vdash \mathtt{M} : \mathtt{A}}
  \colon \oc \sem{\mathtt{\Delta}} \multimap
  \mathsf{S} \otimes \oc \sem{\mathtt{A}},
  \quad
  \psem{\mathtt{\Delta} \vdash \mathtt{V} : \mathtt{A}}
  \colon \oc \sem{\mathtt{\Delta}} \multimap \sem{\mathtt{A}}
\end{equation*}
inductively defined by diagrams in Figure~\ref{fig:int}. In these
definitions, when we can infer $\mathtt{\Delta}$ and $\mathtt{A}$, we
simply write $\sem{\mathtt{M}}$ and $\psem{\mathtt{V}}$ for
$\sem{\mathtt{\Delta} \vdash \mathtt{M} : \mathtt{A}}$ and
$\psem{\mathtt{\Delta} \vdash \mathtt{V} : \mathtt{A}}$ respectively,
and we often apply Convention~\ref{conv:oc} to these Mealy machines.
Extracting precise definitions from these diagrams would be
easy. \shortv{For interpretation of conditional branching,
  see \cite{dlh2019}.}

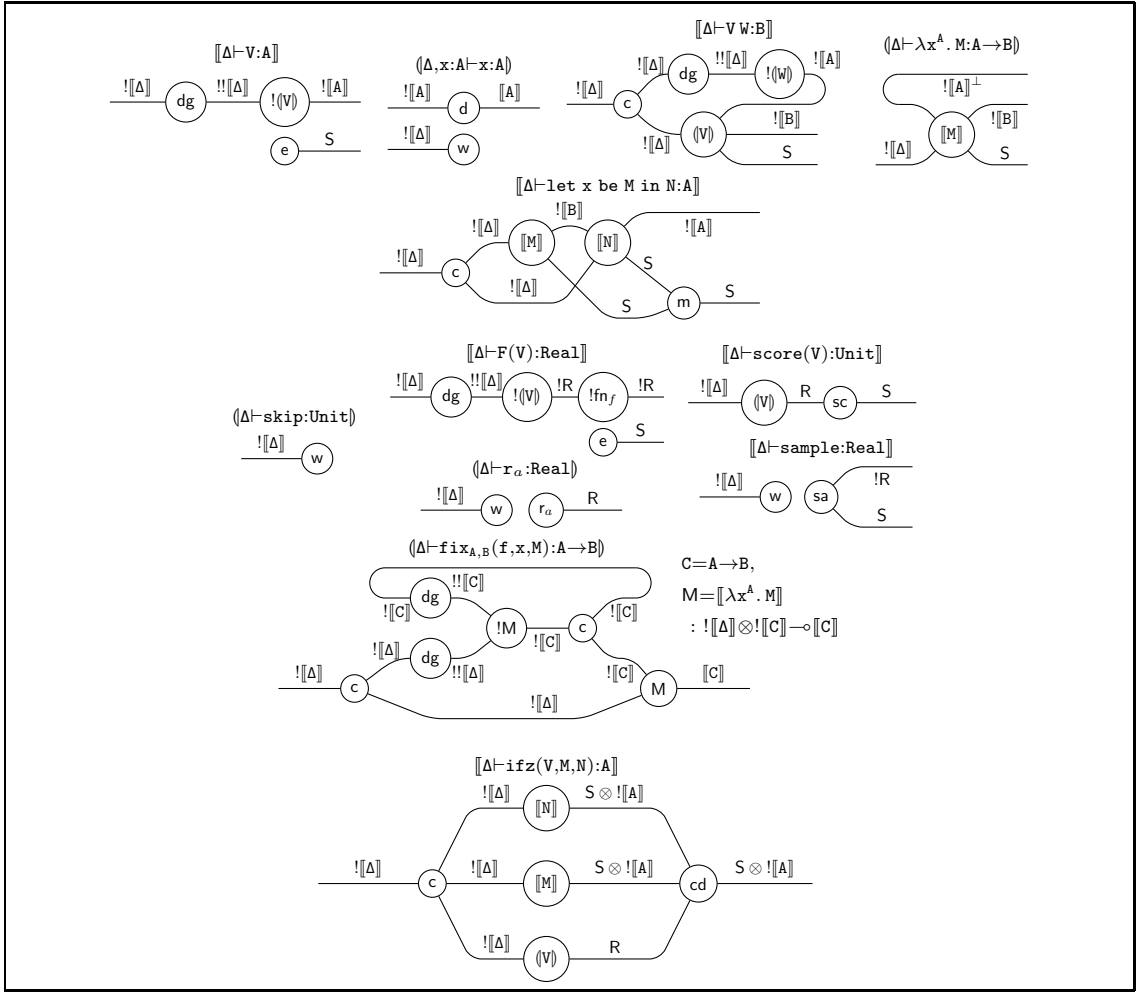
\begin{figure*}[t]
  \begin{center}
    \fbox{\begin{minipage}{.98\textwidth}
        \centering
        \begin{tikzpicture}[rounded corners]
          \node [above] at (0.8,0.4)
          {$\scriptstyle\sem{\mathtt{\Delta} \vdash \mathtt{V} : \mathtt{A}}$};
          \node (dg) [draw,circle,obj] at (0,0) {$\mathsf{dg}$};
          \node (v) [draw,circle,obj] at (1.3,0) {$\oc\psem{\mathtt{V}}$};
          \node (h) [draw,circle,obj] at (1.3,-0.65) {$\mathsf{e}$};
          \draw (dg) -- node[above,obj]{$\oc\sem{\mathtt{\Delta}}$}++(-1,0);
          \draw (dg) -- node[above,obj]{$\oc\oc\sem{\mathtt{\Delta}}$} (v);
          \draw (v) -- node[above,obj]{$\oc\sem{\mathtt{A}}$}++(1,0);
          \draw (h) -- node[above,obj]{$\mathsf{S}$}++(1,0);
        \end{tikzpicture}
        \hspace{3pt}
        \begin{tikzpicture}
          \node (bot) [draw,circle,obj] at (1,-0.15) {$\mathsf{w}$};
          \node (d)[draw,circle,obj] at (1,0.4) {$\mathsf{d}$};
          \draw (bot) -- node[above,obj]{$\oc\sem{\mathtt{\Delta}}$} ++(-1,0);
          \draw (d) -- node[above,obj]{$\oc\sem{\mathtt{A}}$} ++(-1,0);
          \draw (d) -- node[above,obj]{$\sem{\mathtt{A}}$} ++(1,0);
          \node[above=0.1cm] at (d.north) {$\scriptstyle\psem{\mathtt{\Delta}, \mathtt{x}
              : \mathtt{A} \vdash \mathtt{x} : \mathtt{A}}$};
        \end{tikzpicture}
        \hspace{3pt}
        \begin{tikzpicture}[rounded corners]
          \node (c) [draw,circle,obj] at (1,0) {$\mathsf{c}$};
          \node (v) [draw,circle,obj] at (2,-0.4) {$\psem{\mathtt{V}}$};
          \node (dg) [draw,circle,obj] at (1.8,0.4) {$\mathsf{dg}$};
          \node (w) [draw,circle,obj] at (3,0.4) {$\oc\psem{\mathtt{W}}$};
          \draw (dg) -- node[above,obj]{$\oc\oc\sem{\mathtt{\Delta}}$}(w);
          \draw (c) to[out=45,in=180] node[above,obj]{$\oc\sem{\mathtt{\Delta}}$}(dg);
          \draw (c) to[out=-45,in=180] node[below,obj]{$\oc\sem{\mathtt{\Delta}}$}(v);
          \draw (c) -- node[above,obj]{$\oc\sem{\mathtt{\Delta}}$}++(-0.8,0);
          \draw (v) -- ++(0.4,0.4) -- ++(1,0);
          \draw (3.4,0) arc(-90:90:0.2);
          \draw (3.4,0.4) -- node[above right,obj]{$\oc\sem{\mathtt{A}}$}(w);
          \draw (v) -- ++(0.7,0) -- node[above,obj]{$\oc\sem{\mathtt{B}}$}++(0.8,0);
          \draw (v) -- ++(0.4,-0.4) -- ++ (0.3,0) -- node[above,obj]{$\mathsf{S}$}++(0.8,0);
          \node at (2.4,1)
          {$\scriptstyle\sem{\mathtt{\Delta}\vdash \mathtt{V}\,\mathtt{W} :\mathtt{B}}$};
        \end{tikzpicture}
        \hspace{3pt}
        \begin{tikzpicture}[rounded corners]
          \node (m) [draw,circle,obj] at (1,0) {$\sem{\mathtt{M}}$};
          \draw (m) -- ++ (-0.4,-0.4) -- node[above,obj] {$\oc\sem{\mathtt{\Delta}}$} ++ (-0.6,0);
          \draw (m) -- ++ (-0.4,0.4) -- ++ (-0.3,0);
          \draw (0.3,0.4) arc(270:90:0.2);
          \draw (0.3,0.8) -- node[below=-0.05cm,obj]{$\oc\sem{\mathtt{A}}^{\bot}$}++(1.7,0);
          \draw (m) -- ++ (0.4,-0.4) -- node[above,obj] {$\mathsf{S}$} ++ (0.6,0);
          \draw (m) -- ++ (0.4,0.4) -- node[below,obj] {$\oc\sem{\mathtt{B}}$} ++ (0.6,0);
          \node at (1,1.2)
          {$\scriptstyle\psem{\mathtt{\Delta} \vdash \lambda \mathtt{x}^{\mathtt{A}}.\,
              \mathtt{M}: \mathtt{A}\to\mathtt{B}}$};
        \end{tikzpicture}
        \hspace{3pt}
        \begin{tikzpicture}[rounded corners]
          \node (c) [draw,circle,obj] at (1,0) {$\mathsf{c}$};
          \node (m) [draw,circle,obj] at (2,0.4) {$\sem{\mathtt{M}}$};
          \node (n) [draw,circle,obj] at (3,0.4) {$\sem{\mathtt{N}}$};
          \node (mul) [draw,circle,obj] at (4,-0.4) {$\mathsf{m}$};
          \draw (c) -- node[above,obj]{$\oc\sem{\mathtt{\Delta}}$} ++(-1,0);
          \draw (c) -- ++(0.4,0.4) node[above,obj]{$\oc\sem{\mathtt{\Delta}}$} -- (m);
          \draw (m) to[out=30,in=150] node[above,obj]{$\oc\sem{\mathtt{B}}$}(n);
          \draw (c) -- ++(0.4,-0.4) -- node[above,obj]{$\oc\sem{\mathtt{\Delta}}$} ++ (1,0) -- (n);
          \draw (m) -- ++(1,-1) -- node[above,obj]{$\mathsf{S}$} ++(0.5,0) -- (mul);
          \draw (n) -- node[above,obj]{$\mathsf{S}$} (mul);
          \draw (n) -- ++ (0.4,0.4) -- node[below,obj]{$\oc\sem{\mathtt{A}}$}++ (1.6,0);
          \draw (mul) -- node[above,obj]{$\mathsf{S}$} ++(1,0);
          \node [above=0.2cm] at (n.north)
          {$\scriptstyle\sem{\mathtt{\Delta} \vdash \letin{\mathtt{x}}{\mathtt{M}}{\mathtt{N}} :
              \mathtt{A}}$};
        \end{tikzpicture}
        \\[5pt]
        \begin{tikzpicture}
          \node (w) [draw,circle,obj] at (1,0) {$\mathsf{w}$};
          \node at (0,-0.8) {};
          \draw (w) -- node[above,obj]{$\oc\sem{\mathtt{\Delta}}$} ++(-1,0);
          \node [above] at (0.7,0.3)
          {$\scriptstyle\psem{\mathtt{\Delta} \vdash \mathtt{skip}:\ttunit}$};
        \end{tikzpicture}
        \hspace{2pt}
        \begin{tikzpicture}
          \begin{scope}[yshift=-1.5cm,xshift=-0.4cm]
            \node (w) [draw,circle,obj] at (1,0) {$\mathsf{w}$};
            \node (a) [draw,circle,obj] at (1.65,0) {$\mathsf{r}_{a}$};
            \draw (w) -- node[above,obj]{$\oc\sem{\mathtt{\Delta}}$} ++(-1,0);
            \draw (a) -- node[above,obj]{$\mathsf{R}$}++(1,0);
            \node [above] at (1.375,0.3)
            {$\scriptstyle\psem{\mathtt{\Delta} \vdash \mathtt{r}_{a}:\ttreal}$};
          \end{scope}
          \node (v)[draw,circle,obj] at (1,0) {$\oc\psem{\mathtt{V}}$};
          \node (f)[draw,circle,obj] at (2,0) {$\oc\mathsf{fn}_{f}$};
          \node (dg)[draw,circle,obj] at (0,0) {$\mathsf{dg}$};
          \node (e)[draw,circle,obj] at (2,-0.6) {$\mathsf{e}$};
          \draw (e)--node[above,obj]{$\mathsf{S}$}++(0.8,0);
          
          \draw (v) -- node[above,obj]{$\oc\oc\sem{\mathtt{\Delta}}$} (dg);
          \draw (dg) -- node[above,obj]{$\oc\sem{\mathtt{\Delta}}$} ++(-0.8,0);
          \draw (v) -- node[above,obj]{$\oc\mathsf{R}$} (f);
          \draw (f) -- node[above,obj]{$\oc\mathsf{R}$} ++(0.8,0);
          \node [above] at (v.north)
          {$\scriptstyle\sem{\mathtt{\Delta} \vdash \mathtt{F}(\mathtt{V}):\ttreal}$};
        \end{tikzpicture}
        \hspace{2pt}
        \begin{tikzpicture}[rounded corners]
          \begin{scope}[yshift=1.25cm,xshift=-0.75cm]
            \node (v) [draw,circle,obj] at (0,0) {$\psem{\mathtt{V}}$};
            \node (sc) [draw,circle,obj] at (1,0) {$\mathsf{sc}$};
            \draw (v) --node[above,obj]{$\oc\sem{\mathtt{\Delta}}$} ++(-1,0);
            \draw (v) --node[above,obj]{$\mathsf{R}$} (sc);
            \draw (sc) --node[above,obj]{$\mathsf{S}$} ++(1,0);
            \node [above] at (0.5,0.4)
            {$\scriptstyle\sem{\mathtt{\Delta} \vdash \mathtt{score}(\mathtt{V}) :\ttunit}$};
          \end{scope}
          \node (sa) [draw,circle,obj] at (0,0) {$\mathsf{sa}$};
          \node (w) [draw,circle,obj] at (-0.6,0) {$\mathsf{w}$};
          \draw (w) -- node[above,obj]{$\oc\sem{\mathtt{\Delta}}$}++(-1,0);
          \draw (sa) -- ++(0.4,0.4) -- node[below,obj]{$\oc\mathsf{R}$}++(0.8,0);
          \draw (sa) -- ++(0.4,-0.4) -- node[above,obj]{$\mathsf{S}$}++(0.8,0);
          \node [above] at (0,0.4) {$\scriptstyle\sem{\mathtt{\Delta} \vdash \mathtt{sample} : \ttreal}$};
        \end{tikzpicture}
        \hspace{2pt}
        \begin{tikzpicture}[rounded corners]
          \node (cd) [draw,circle,obj] at (0,0) {$\mathsf{c}$};
          \node (dd) [draw,circle,obj] at (1,0.4) {$\mathsf{dg}$};
          \node (dab) [draw,circle,obj] at (1,1.2) {$\mathsf{dg}$};
          \node (m) [draw,circle,obj] at (2,0.8) {$\oc\mathsf{M}$};
          \node (cab) [draw,circle,obj] at (3,0.8) {$\mathsf{c}$};
          \node (n) [draw,circle,obj] at (4,0) {$\mathsf{M}$};
          \draw (cd) -- ++ (1,-0.4) -- ++ (1,0) -- node[above,obj]{$\oc\sem{\mathtt{\Delta}}$}
          ++ (1,0) -- (n);
          \draw (cd) -- node[above,obj]{$\oc\sem{\mathtt{\Delta}}$} ++(-1,0);
          \draw (cd) to[out=35,in=180] node[above,obj]{$\oc\sem{\mathtt{\Delta}}$} (dd);
          \draw (dd) -- ++(0.5,0) node[below, obj]{$\oc\oc\sem{\mathtt{\Delta}}$} -- (m);
          \draw (dab) -- ++(0.5,0) node[above,obj]{$\oc\oc\sem{\mathtt{C}}$} -- (m);
          \draw (m) -- node[below,obj]{$\oc\sem{\mathtt{C}}$} (cab);
          \draw (cab) -- ++(0.35,-0.4) -- node[below,obj]{$\oc\sem{\mathtt{C}}$} ++ (0.3,0) -- (n);
          \draw (cab) -- ++(0.4,0.4) -- node[below,obj]{$\oc\sem{\mathtt{C}}$} ++ (0.3,0);
          \draw (3.7,1.2) arc(-90:90:0.2);
          \draw (3.7,1.6) -- ++ (-3.3,0);
          \draw (0.4,1.6) arc(90:270:0.2);
          \draw (0.4,1.2) -- node[below,obj]{$\oc\sem{\mathtt{C}}$}(dab);
          \draw (n) -- node[above,obj]{$\sem{\mathtt{C}}$}++(1.2,0);
          \node [above] at (2,1.6) {$\scriptstyle\psem{\mathtt{\Delta} \vdash
              \mathtt{fix}_{\mathtt{A},\mathtt{B}}(\mathtt{f},\mathtt{x},\mathtt{M})
              : \mathtt{A} \to \mathtt{B}}$};
          \node[right] at(4,1.25)
          {$\begin{array}{l}
              \scriptstyle
              \mathtt{C}=\mathtt{A}\to\mathtt{B}, \\
              \scriptstyle
              \mathsf{M}=\sem{\lambda \mathtt{x}^{\mathtt{A}}.\,\mathtt{M}} \\
              \scriptstyle
              \; \colon \oc\sem{\mathtt{\Delta}} \otimes
              \oc\sem{\mathtt{C}} \multimap \sem{\mathtt{C}}
            \end{array}$};
        \end{tikzpicture}
        \longv{
          \\[10pt]
          \begin{tikzpicture}[rounded corners]
            \node (M) [draw,circle,obj] at (0,0) {$\sem{\mathtt{M}}$};
            \node (N) [draw,circle,obj] at (0,1) {$\sem{\mathtt{N}}$};
            \node (V) [draw,circle,obj] at (0,-1) {$\psem{\mathtt{V}}$};
            \node (cd) [draw,circle,obj] at (2,0) {$\mathsf{cd}$};
            \node (c) [draw,circle,obj] at (-1.5,0) {$\mathsf{c}$};

            \draw (c) -- node[above,obj]{$\oc\sem{\mathtt{\Delta}}$} ++ (-1.5,0);
            \draw (c) -- node[above,obj]{$\oc\sem{\mathtt{\Delta}}$} (M);
            \draw (c) -- ++ (0.5,1) -- node[above,obj]{$\oc\sem{\mathtt{\Delta}}$} (N);
            \draw (c) -- ++ (0.5,-1) -- node[above,obj]{$\oc\sem{\mathtt{\Delta}}$} (V);

            \draw (cd) -- node[above,obj]{$\mathsf{S} \otimes \oc\sem{\mathtt{A}}$} ++ (1.5,0);
            \draw (cd) -- node[above,obj]{$\mathsf{S} \otimes \oc\sem{\mathtt{A}}$} (M);
            \draw (cd) -- ++ (-0.5,1) -- node[above,obj]{$\mathsf{S} \otimes \oc\sem{\mathtt{A}}$} (N);
            \draw (cd) -- ++ (-0.5,-1) -- node[above,obj]{$\mathsf{R}$} (V);
            
            \node[above] at (N.north) {$\scriptstyle
              \sem{\mathtt{\Delta}\vdash\ifterm{\mathtt{V}}{\mathtt{M}}{\mathtt{N}}:\mathtt{A}}$};
          \end{tikzpicture}
        }
      \end{minipage}}
  \end{center}
  \caption{Interpretation of Terms and Values}
  \label{fig:int}
\end{figure*}

\section{Adequacy Theorems}

Finally, we give our main results. In the proof of our adequacy theorems,
we use logical relations, diagrammatic reasoning of Mealy machines
(Proposition~\ref{prop:beh_eq}), the domain theoretic structure of
Mealy machines (Proposition~\ref{prop:wcpo}), and Fubini-Tonelli
theorem.

\subsection{Sampling-Based Operational Semantics}

For a closed term $\mathtt{M}:\ttreal$, we define a partial
measurable function
$\mathfrak{o}(\mathtt{M}) \colon \mathbb{R}_{\geq 0} \times \mathbb{T}
\to \mathbb{R}_{\geq 0} \times \mathbb{R}$ as follows:
\begin{varitemize}
\item for $(a,u) \in \mathbb{R}_{\geq 0} \times \mathbb{T}$,
  if there are $s,s' \in \state{\sem{\mathtt{M}}}$ such that
  \begin{align*}
    \tran{\sem{\mathtt{M}}}((\mm,(\mm,(a,u))), \init{\sem{\mathtt{M}}})
    &= ((\hh,(\hh,(a',\varepsilon))),s), \\
    \tran{\sem{\mathtt{M}}}((\mm,(\hh,(0,\varepsilon))),s)
    &= (\hh,(\mm,(0,b\mathbin{::}\varepsilon)),s'),
  \end{align*}
  i.e., if we have the following transitions:
  
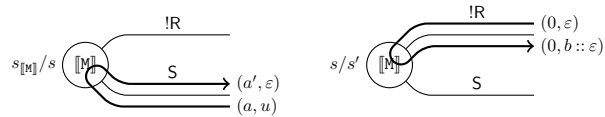
\begin{center}
    \begin{tikzpicture}[rounded corners]
      \node (M) [draw,circle,obj] at (0,0) {$\sem{\mathtt{M}}$};
      \draw (M) -- ++ (0.4,0.4) -- node[above,obj] {$\oc \mathsf{R}$}++ (1.5,0);
      \draw (M) -- ++ (0.4,-0.4) -- node[above=4,obj] {$\mathsf{S}$}++ (1.5,0);
      \node[obj,left] at (M.west) {$\init{\sem{\mathtt{M}}}/s$};
      \draw[thick,->] (1.9,-0.55) node[right,obj] {$(a,u)$}
      -- ++ (-1.55,0) -- ++(-0.4,0.4) -- ++ (0.2,0.2)
      -- ++ (0.3,-0.3) -- ++ (1.45,0) node[right,obj] {$(a',\varepsilon)$};
      \begin{scope}[xshift=4cm]
        \node (M) [draw,circle,obj] at (0,0) {$\sem{\mathtt{M}}$};
        \draw (M) -- ++ (0.4,0.4) -- node[above=4,obj] {$\oc \mathsf{R}$}++ (1.5,0);
        \draw (M) -- ++ (0.4,-0.4) -- node[above,obj] {$\mathsf{S}$}++ (1.5,0);
        \node[obj,left] at (M.west) {$s/s'$};
        \draw[thick,->] (1.9,0.55) node[right,obj] {$(0,\varepsilon)$}
        -- ++ (-1.55,0) -- ++(-0.4,-0.4) -- ++ (0.2,-0.2)
        -- ++ (0.3,0.3) -- ++ (1.45,0) node[right,obj] {$(0,b\mathbin{::}\varepsilon)$};
      \end{scope}
    \end{tikzpicture}
  \end{center}
  then we define $\mathfrak{o}(\mathtt{M})(a,u)$ to be $(a',b)$;
\item otherwise,
  $\mathfrak{o}(\mathtt{M})(a,u)$ is undefined.
\end{varitemize}

\begin{theorem}[Adequacy]\label{thm:adq}
  For any closed term $\vdash \mathtt{M} : \ttreal$ and for any
  $(a,u) \in \mathbb{R}_{\geq 0} \times \mathbb{T}$, we have
  \begin{equation*}
    \mathfrak{o}(\mathtt{M})(a,u) = (a',b)
    \iff
    (\mathtt{M},a,u)\to^{\ast}(b,a',\varepsilon).
  \end{equation*}
\end{theorem}
\begin{corollary}
  For any closed term $\vdash \mathtt{M}:\ttreal$, partial functions
  $\mathtt{weight}(\mathtt{M}) \colon \mathbb{R}_{\geq 0} \times
  \mathbb{T} \to \mathbb{R}_{\geq 0}$ and
  $\mathtt{val}(\mathtt{M}) \colon \mathbb{R}_{\geq 0} \times
  \mathbb{T} \to \mathbb{R}$ given as follows \longshortv{
    \begin{equation*}
      (\mathtt{weight}(\mathtt{M})(a,u),
      \mathtt{val}(\mathtt{M})(a,u)) = \\
      \begin{cases}
        (a',b), & \textnormal{if } (\mathtt{M},a,u) \to^{\ast} (b,a',\varepsilon), \\
        \textnormal{undefined}, & \textnormal{otherwise}
      \end{cases}
    \end{equation*}
  }
  {
    \begin{multline*}
      (\mathtt{weight}(\mathtt{M})(a,u),
      \mathtt{val}(\mathtt{M})(a,u)) = \\
      \begin{cases}
        (a',b), & \textnormal{if } (\mathtt{M},a,u) \to^{\ast} (b,a',\varepsilon), \\
        \textnormal{undefined}, & \textnormal{otherwise}
      \end{cases}
    \end{multline*}} are partial measurable functions.
\end{corollary}

\subsection{Distribution-Based Operational Semantics}

For a closed term $\mathtt{M}:\ttreal$, we define measurable functions
$\mathfrak{o}_{0}(\mathtt{M}) \colon \mathbb{T} \to \mathbb{R}_{\geq
  0}$ and
$\mathfrak{o}_{1}(\mathtt{M}) \colon \mathbb{T} \to \mathbb{R}$ by
\begin{align*}
  \mathfrak{o}_{0}(\mathtt{M})(u) &=
  \begin{cases}
    a, & \textnormal{if } \exists b \in \mathbb{R},\,
    \mathfrak{o}(\mathtt{M})(1,u) = (a,b), \\
    0, & \textnormal{otherwise}, \\
  \end{cases} \\
  \mathfrak{o}_{1}(\mathtt{M})(u) &=
  \begin{cases}
    b, & \textnormal{if }
    \exists a \in \mathbb{R}_{\geq 0},\,
    \mathfrak{o}(\mathtt{M})(1,u) = (a,b), \\
    0, & \textnormal{otherwise}.
  \end{cases}
\end{align*}
Then we define a measure
$\mathfrak{O}(\mathtt{M})$ on
$\mathbb{R}$ by:
\begin{equation*}
  \mathfrak{O}(\mathtt{M})(A) =
  \sum_{n \in \mathbb{N}}
  \int_{\mathbb{R}_{[0,1]}^{n}} \mathfrak{o}_{0}(\mathtt{M})(u)
  \,[\mathfrak{o}_{1}(\mathtt{M})(u) \in A] \;\mathrm{d}u.
\end{equation*}
\begin{theorem}[Adequacy]\label{thm:adq2}
  For any closed term $\vdash \mathtt{M} : \ttreal$,
  we have
  \begin{math}
    \mathtt{M} \Rightarrow_{\infty} \mathfrak{O}(\mathtt{M}).
  \end{math}
\end{theorem}
It follows from our adequacy theorems that sampling-based operational
semantics induces distribution-based operational semantics.
\begin{corollary}\label{coro:equivsem}
  For any closed term $\vdash \mathtt{M}:\ttreal$,
  \begin{equation*}
    \mathtt{M} \Rightarrow_{\infty}
    \sum_{n \in \mathbb{N}}
    \int_{\mathbb{R}_{[0,1]}^{n}} \mathtt{weight}(\mathtt{M})(u)
    \,[\mathtt{val}(\mathtt{M})(u) \in A] \;\mathrm{d}u.
  \end{equation*}
\end{corollary}
A result analogous to Corollary~\ref{coro:equivsem} has already been
proved by way of a purely operational (and quite laburious)
argument in an untyped setting where score is not available in
its full generality~\cite{bdlgs2016}. Here, it is just an easy corollary of
our adequacy theorems.
\longv{
  \section{Proof of Adequacy Theorems}
  \label{sec:proof_of_adq}

  \begin{lemma}
    For any term $\mathtt{\Delta},\mathtt{x}:\mathtt{A} \vdash \mathtt{M}:\mathtt{B}$
    and for any closed value $\vdash \mathtt{V}:\mathtt{A}$,
    \begin{equation*}
      \sem{\mathtt{M}} \circ (\sem{\mathtt{\Delta}} \otimes \oc \psem{\mathtt{V}})
      \simeq \sem{\mathtt{M}\{\mathtt{V}/\mathtt{x}\}}.
    \end{equation*}  
  \end{lemma}
  \begin{proof}
    By induction on $\mathtt{M}$.
  \end{proof}
  \begin{lemma}\label{lem:detred}
    For all closed terms $\mathtt{M},\mathtt{N}:\mathtt{A}$,
    if $\mathtt{M} \detred \mathtt{N}$, then
    $\sem{\mathtt{M}} = \sem{\mathtt{N}}$.
  \end{lemma}
  \begin{proof}
    By case analysis. For the case of recursion, see
    Proposition~\ref{prop:fix} in Section~\ref{sec:fix}.
  \end{proof}
  We first prove soundness.
  \begin{proposition}\label{prop:sound}
    For any closed term $\mathtt{M}:\ttreal$ and for any
    $(a,u) \in \mathbb{R}_{\geq 0} \times \mathbb{T}$, if
    $(\mathtt{M},a,u) \to^{\ast} (b,a',\varepsilon)$, then
    $\mathfrak{o}(\mathtt{M})(a,u) = (a',b)$.
  \end{proposition}
  \begin{proof}
    By induction on the length of $\to^{\ast}$.
    (Base case) Easy. (Induction step) By case analysis on the
    first evaluation step of
    $(\mathtt{M},a,u) \to^{\ast} (b,a',\varepsilon)$.
    \begin{itemize}
    \item If the first evaluation step is of the form
      $(\mathtt{E}[\mathtt{N}],a,u) \to
      (\mathtt{E}[\mathtt{L}],a',u')$ for some
      $\mathtt{N} \detred \mathtt{L}$, then by Lemma~\ref{lem:detred},
      we have $\mathtt{E}[\mathtt{N}] = \mathtt{E}[\mathtt{N}]$.
      Because
      $(\mathtt{E}[\mathtt{L}],a',u') \to^{\ast} (b,a',\varepsilon)$,
      by induction hypothesis, we obtain
      $\mathfrak{o}(\mathtt{E}[\mathtt{L}])(a',u') = (a',b)$. Hence,
      $\mathfrak{o}(\mathtt{E}[\mathtt{N}])(a,u) =
      \mathfrak{o}(\mathtt{E}[\mathtt{L}],a',u') = (a',b)$.
    \item If the first evaluation step is of the form
      $(\mathtt{E}[\mathtt{score}(\mathtt{r}_{c})],a,u) \to
      (\mathtt{E}[\mathtt{skip}],|c|\,a,u)$, then by induction
      hypothesis, we have $\mathfrak{o}(\mathtt{E}[\mathtt{skip}])(|c|\,a,u) =
      (a',b)$. Therefore, by the definition of the Mealy machine
      $\mathsf{sc}$, we see that
      $\mathfrak{o}(\mathtt{E}[\mathtt{score}(\mathtt{r}_{c})])(a,u)$ is
      $(a',b)$.
    \item If the first evaluation step is of the form
      $(\mathtt{E}[\mathtt{sample}],a,c \cons u)
      \to (\mathtt{E}[\mathtt{r}_{c}],a,u)$, then
      by induction hypothesis,
      $\mathfrak{o}(\mathtt{E}[\mathtt{r}_{c}])(a,u)$
      is $(a',b)$. Therefore, by the definition of the Mealy
      machine $\mathtt{sample}$, we see that
      $\mathfrak{o}(\mathtt{E}[\mathtt{sample}])(a,c \cons u)$
      is $(a',b)$.
    \end{itemize}
  \end{proof}

  It remains to prove that $\mathfrak{o}(\mathtt{M})(a,u) = (a',b)$
  implies that $(\mathtt{M},a,u) \to^{\ast} (b,a',\varepsilon)$.
  We use logical relations.
  We define a binary relation $O$ between closed terms of type $\ttreal$
  and Mealy machines from $\mathsf{I}$ to $\mathsf{S} \otimes \oc \mathsf{R}$ by
  \begin{equation*}
    (\mathtt{M},\mathsf{M}) \in O
    \iff
    \textnormal{if }o(\mathsf{M})(a,u) = (a',b),
    \textnormal{ then } (\mathtt{M},a,u) \to^{\ast} (b,a',\varepsilon)
  \end{equation*}
  where
  $o(\mathsf{M}) \colon \mathbb{R}_{\geq 0} \times \mathbb{T} \to
  \mathbb{R}_{\geq 0} \times \mathbb{R}$ is a partial measurable
  function given by: for each
  $(a,u) \in \mathbb{R}_{\geq 0} \times \mathbb{T}$,
  \begin{varitemize}
  \item if there are states $s,s' \in \state{\mathsf{M}}$ such that
    \begin{align*}
      \tran{\mathsf{M}}((\mm,(\mm,(a,u))), \init{\mathsf{M}})
      &= ((\hh,(\hh,(a',\varepsilon))),s), \\
      \tran{\mathsf{M}}((\mm,(\hh,(0,\varepsilon))),s)
      &= (\hh,(\mm,(0,b\mathbin{::}\varepsilon)),s'),
    \end{align*}
    i.e., if we have the following transitions:
    
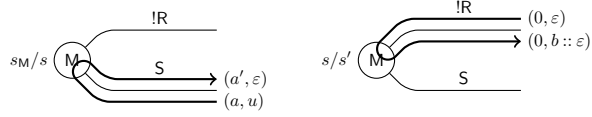
\begin{center}
      \begin{tikzpicture}[rounded corners]
        \node (M) [draw,circle,obj] at (0,0) {$\mathsf{M}$};
        \draw (M) -- ++ (0.4,0.4) -- node[above,obj] {$\oc \mathsf{R}$}++ (1.5,0);
        \draw (M) -- ++ (0.4,-0.4) -- node[above=4,obj] {$\mathsf{S}$}++ (1.5,0);
        \node[obj,left] at (M.west) {$\init{\mathsf{M}}/s$};
        \draw[thick,->] (1.9,-0.55) node[right,obj] {$(a,u)$}
        -- ++ (-1.55,0) -- ++(-0.4,0.4) -- ++ (0.2,0.2)
        -- ++ (0.3,-0.3) -- ++ (1.45,0) node[right,obj] {$(a',\varepsilon)$};
        \begin{scope}[xshift=4cm]
          \node (M) [draw,circle,obj] at (0,0) {$\mathsf{M}$};
          \draw (M) -- ++ (0.4,0.4) -- node[above=4,obj] {$\oc \mathsf{R}$}++ (1.5,0);
          \draw (M) -- ++ (0.4,-0.4) -- node[above,obj] {$\mathsf{S}$}++ (1.5,0);
          \node[obj,left] at (M.west) {$s/s'$};
          \draw[thick,->] (1.9,0.55) node[right,obj] {$(0,\varepsilon)$}
          -- ++ (-1.55,0) -- ++(-0.4,-0.4) -- ++ (0.2,-0.2)
          -- ++ (0.3,0.3) -- ++ (1.45,0) node[right,obj] {$(0,b\mathbin{::}\varepsilon)$};
        \end{scope}
      \end{tikzpicture}
    \end{center}
    then we define $o(\mathsf{M})(a,u)$ to be $(a',b)$;
  \item otherwise,
    $o(\mathsf{M})(a,u)$ is undefined.
  \end{varitemize}

  We then inductively define binary relations
  \begin{align*}
    R_{\mathtt{A}}
    &\subseteq \{\textnormal{closed values of type }\mathtt{A}\}
      \times \{\textnormal{Mealy machines from } \mathsf{I} \textnormal{ to } \sem{\mathtt{A}}\} \\
    R_{\mathtt{A}}^{\top}
    &\subseteq
      \{\textnormal{evaluation contexts } \mathtt{x}:\mathtt{A}
      \vdash \mathtt{E}[\mathtt{x}]:\ttreal \}
      \times \{\textnormal{Mealy machines from } \oc\sem{\mathtt{A}}
      \textnormal{ to } \mathsf{S} \otimes
      \oc \mathsf{R}\} \\
    \overline{R}_{\mathtt{A}}
    &\subseteq \{\textnormal{closed terms of type }\mathtt{A}\}
      \times \{\textnormal{Mealy machines from } \mathsf{I} \textnormal{ to }
      \mathsf{S} \otimes \oc \sem{\mathtt{A}}\}
  \end{align*}
  by
  \begin{align*}
    R_{\ttreal}
    &= \{(\mathtt{r}_{a},\mathsf{r}_{a}) : a \in \mathbb{R}\}, \\
    R_{\ttunit}
    &= \{(\mathtt{skip},\mathsf{id}_{\mathsf{I}})\}, \\
    R_{\mathtt{A}\to\mathtt{B}}
    &= \{(\mathtt{V},\mathsf{M}) :
      \forall (\mathtt{W},\mathtt{N}) \in R_{\mathtt{A}},\,
      (\mathtt{V}\,\mathtt{W},
      (\mathsf{S} \otimes \oc \sem{\mathtt{B}} \otimes \mathsf{counit}_{\oc\sem{\mathtt{A}}})
      \circ (\mathsf{M} \otimes \oc \mathsf{N})) \in \overline{R}_{\mathtt{B}}
      \}, \\
    R_{\mathtt{A}}^{\top}
    &= \{(\mathtt{E}[-],\mathsf{E}) :
      \forall (\mathtt{V},\mathsf{M}) \in R_{\mathtt{A}},\,
      (\mathtt{E}[\mathtt{V}],\mathsf{E} \circ \oc \mathsf{M}) \in O\}, \\
    \overline{R}_{\mathtt{A}}
    &= \{(\mathtt{M},\mathsf{M}):
      \forall (\mathtt{E}[-],\mathtt{E}) \in R_{\mathtt{A}}^{\top},\,
      (\mathtt{E}[\mathtt{M}],(\mathsf{m} \otimes \oc \mathsf{R})
      \circ (\mathsf{S} \otimes \mathsf{E}) \circ \mathsf{M}) \in O\}
  \end{align*}

  We list some properties of the logical relations.
  \begin{lemma}\label{lem:lr}
    Let $\mathtt{A}$ be a type.
    \begin{enumerate}
    \item If $(\mathtt{V},\mathtt{M}) \in R_{\mathtt{A}}$, then
      $(\mathtt{V},\mathsf{e} \otimes \oc \mathsf{M}) \in
      \overline{R}_{\mathtt{A}}$.
    \item If $(\mathtt{M},\mathsf{M}) \in \overline{R}_{\mathtt{A}}$ and
      $\mathtt{N} \detred \mathtt{M}$, then
      $(\mathtt{N},\mathsf{M}) \in \overline{R}_{\mathtt{A}}$.
    \item If $(\mathtt{M},\mathsf{M}) \in \overline{R}_{\mathtt{A}}$ and
      $\mathtt{M} \detred \mathtt{N}$, then
      $(\mathtt{N},\mathsf{M}) \in \overline{R}_{\mathtt{A}}$.
    \item If $(\mathtt{M},\mathsf{M}) \in \overline{R}_{\mathtt{A}}$ and
      $\mathsf{M} \simeq \mathsf{N}$, then
      $(\mathtt{M},\mathsf{N}) \in \overline{R}_{\mathtt{A}}$.
    \item For any closed term $\mathtt{M}:\mathtt{A}$,
      $(\mathtt{M},\mathsf{bot}_{\mathsf{S} \otimes \oc
        \sem{\mathtt{A}}}) \in \overline{R}_{\mathtt{A}}$ where
      $\mathsf{bot}_{\mathsf{X}} \colon \mathsf{I} \multimap \mathsf{X}$
      is a token machine whose transition function is the empty partial
      measurable function.
    \item For any closed value $\mathtt{V}:\mathtt{A}\to\mathtt{B}$,
      $(\mathtt{V},\mathsf{bot}_{\sem{\mathtt{A}\to\mathtt{B}}}) \in
      R_{\mathtt{A}\to\mathtt{B}}$.
    \item If $(\mathtt{M},\mathsf{M}_{i}) \in \overline{R}_{\mathtt{A}}$
      and $[\mathsf{M}_{1}] \leq [\mathsf{M}_{2}] \leq \cdots$, then
      $(\mathtt{M},\mathsf{N}) \in \overline{R}_{\mathtt{A}}$ where
      $[\mathsf{N}]$ is the least upper bound of the $\omega$-chain
      $[\mathsf{M}_{1}] \leq [\mathsf{M}_{2}] \leq \cdots$.
    \end{enumerate}
  \end{lemma}
  \begin{proof}
    We can check these items by unfolding the definition of
    $O$ and the logical relations.
  \end{proof}

  \begin{lemma}[Basic Lemma]\label{lem:basic}
    Let
    $\mathtt{\Delta} = (\mathtt{x}:\mathtt{A}_{1},
    \ldots,\mathtt{x}_{n}:\mathtt{A}_{n})$ be a context.
    \begin{itemize}
    \item For any term $\mathtt{\Delta} \vdash \mathtt{M}:\mathtt{A}$
      and for any
      $(\mathtt{V}_{i},\mathtt{N}_{i}) \in R_{\mathtt{A}_{i}}$ for
      $i = 1,2,\ldots,n$, we have
      \begin{equation*}
        \left(
          \mathtt{M}\{\mathtt{V}_{1}/\mathtt{x}_{1},\ldots,
          \mathtt{V}_{n}/\mathtt{x}_{n}\},
          \sem{\mathtt{M}} \circ
          (\oc\mathsf{N}_{1} \otimes \cdots \otimes \oc \mathsf{N}_{n})
        \right)
        \in \overline{R}_{\mathtt{A}}.
      \end{equation*}
    \item  For any value
      $\mathtt{\Delta} \vdash \mathtt{V}:\mathtt{A}$ and for any
      $(\mathtt{V}_{i},\mathtt{N}_{i}) \in R_{\mathtt{A}_{i}}$ for
      $i = 1,2,\ldots,n$, we have
      \begin{equation*}
        \left(
          \mathtt{V}\{\mathtt{V}_{1}/\mathtt{x}_{1},\ldots,
          \mathtt{V}_{n}/\mathtt{x}_{n}\},
          \psem{\mathtt{M}} \circ
          (\oc\mathsf{N}_{1} \otimes \cdots \otimes \oc \mathsf{N}_{n})
        \right)
        \in R_{\mathtt{A}}.
      \end{equation*}
    \end{itemize}
  \end{lemma}
  \begin{proof}
    By induction on $\mathtt{M}$ and $\mathtt{V}$. Most cases follow
    from Lemma~\ref{lem:lr}. For $\mathtt{M}=\mathtt{sample}$ and
    $\mathtt{M}=\mathtt{score}(\mathtt{V})$, we check the statement by
    unfolding the definition of $\mathsf{sa}$ and $\mathsf{sc}$. Here,
    we only check for
    $\mathtt{M}=\mathtt{sample}$ and
    $\mathtt{M}=\mathtt{fix}_{\mathtt{A},\mathtt{B}}(\mathtt{f},
    \mathtt{x},\mathtt{N})$.
    \begin{itemize}
    \item When $\mathtt{M}=\mathtt{sample}$,
      for any $(\mathtt{E},\mathsf{E})$ in $R_{\ttreal}^{\top}$,
      if 
      \begin{equation*}
        \mathfrak{o}((\mathsf{m} \otimes \oc \mathsf{R})
        \circ (\mathsf{S} \otimes \mathsf{E}) \circ \mathsf{sa})(a,u)
        = (a',b),
      \end{equation*}
      then by the definition of $\mathsf{sa}$, we see that
      $u = c \cons u'$ for some $c \in \mathbb{R}_{[0,1]}$
      and $u' \in \mathbb{T}$ such that
      \begin{equation*}
        \mathfrak{o}(\mathsf{E} \circ \oc \mathsf{r}_{c})(a,u')
        = (a',b).
      \end{equation*}
      Because $(\mathtt{E},\mathsf{E}) \in R_{\ttreal}^{\top}$,
      we obtain $(\mathtt{E}[\mathtt{r}_{c}],a,u')
      \to^{\ast} (b,a',\varepsilon)$.
      Hence,
      \begin{equation*}
        (\mathtt{E}[\mathtt{sample}],a,u) \to^{\ast}
        (b,a',\varepsilon). 
      \end{equation*}
    \item When
      $\mathtt{M}=\mathtt{fix}_{\mathtt{A},\mathtt{B}}(\mathtt{f},\mathtt{x},\mathtt{N})$,
      for simplicity, we suppose that $\mathtt{M}$ is a closed term. By
      induction hypothesis, we can check that
      \begin{equation*}
        (\mathtt{M},
        \sem{\lambda \mathtt{x}^{\mathtt{A}}.\,\mathtt{N}}
        \circ \oc\sem{\lambda \mathtt{x}^{\mathtt{A}}.\,\mathtt{N}}
        \circ \cdots
        \circ \oc^{k}\sem{\lambda \mathtt{x}^{\mathtt{A}}.\,\mathtt{N}}
        \circ \mathsf{bot}_{\mathsf{I},\oc^{k}\sem{\mathtt{A}\to\mathtt{B}}})
        \in R_{\mathtt{A}\to\mathtt{B}}
      \end{equation*}
      by induction on $n$. Because $[\sem{\mathtt{M}}]$ is the least
      upper bound of
      $\sem{\lambda \mathtt{x}^{\mathtt{A}}.\,\mathtt{M}} \circ
      \oc\sem{\lambda \mathtt{x}^{\mathtt{A}}.\,\mathtt{N}} \circ
      \cdots \circ \oc^{k}\sem{\lambda
        \mathtt{x}^{\mathtt{A}}.\,\mathtt{N}} \circ
      \mathsf{bot}_{\mathsf{I},\oc^{k}\sem{\mathtt{A}\to\mathtt{B}}}$
      (Proposition~\ref{prop:iter}),
      we obtain
      $(\mathtt{M},\sem{\mathtt{M}}) \in R_{\mathtt{A}\to\mathtt{B}}$
      by Lemma~\ref{lem:lr}.
    \end{itemize}
  \end{proof}

  \begin{theorem}
    For any closed term $\vdash \mathtt{M} : \ttreal$ and for any
    $(a,u) \in \mathbb{R}_{\geq 0} \times \mathbb{T}$, we have
    \begin{equation*}
      \mathfrak{o}(\mathtt{M})(a,u) = (a',b)
      \iff
      (\mathtt{M},a,u)\to^{\ast}(b,a',\varepsilon).
    \end{equation*}
  \end{theorem}
  \begin{proof}
    If $(\mathtt{M},a,u)\to^{\ast}(b,a',\varepsilon)$,
    then we have $\mathfrak{o}(\mathtt{M})(a,u) = (a',b)$
    by Proposition~\ref{prop:sound}. If
    $\mathfrak{o}(\mathtt{M})(a,u) = (a',b)$,
    then because $([-],\mathsf{e} \otimes \mathsf{id}_{\oc \mathsf{R}})$
    is an element of $R_{\ttreal}^{\top}$,
    we obtain $(\mathtt{M},a,u)\to^{\ast}(b,a',\varepsilon)$
    by Lemma~\ref{lem:basic}.
  \end{proof}
}

\longv{
  \section{Approximation Lemma}
\label{sec:fix}

Let $\mathsf{M} \colon \oc \mathsf{X} \multimap \mathsf{X}$ be a
Mealy machine. In this section, we give a way to
calculate a Mealy machine
$\mathsf{M}^{\dagger} \colon \mathsf{I} \to \oc\mathsf{X}$ given by
\begin{equation*}
  \mathsf{M}^{\dagger} =
  (\mathsf{M} \otimes \mathsf{counit}_{\oc\mathsf{X}})
  \circ ((\mathsf{c}_{\mathsf{X}} \circ \oc\mathsf{M}
  \circ \mathsf{dg}_{\mathsf{X}})
  \otimes \mathsf{id}_{\oc\mathsf{X}})
  \circ \mathsf{unit}_{\oc\mathsf{X}}
\end{equation*}
Diagrammatically, $\mathsf{M}^{\dagger}$ consists of digging, contraction
and a feed back loop:
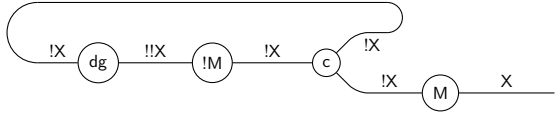
\begin{center}
  \begin{tikzpicture}[rounded corners]
    \node (!m)[draw,circle,obj] at (0,0) {$\oc\mathsf{M}$};
    \node (dg)[draw,circle,obj] at (-1.5,0) {$\mathsf{dg}$};
    \node (c)[draw,circle,obj] at (1.5,0) {$\mathsf{c}$};
    \node (m)[draw,circle,obj] at (3,-0.4) {$\mathsf{M}$};

    \draw (dg) -- node[above,obj]{$\oc\oc\mathsf{X}$} (!m);
    \draw (!m) -- node[above,obj]{$\oc\mathsf{X}$} (c);
    \draw (c) -- ++ (0.4,-0.4) -- node[above,obj]{$\oc\mathsf{X}$}(m);
    \draw (m) --node[above,obj]{$\mathsf{X}$} ++(1.5,0);
    \draw (c) -- ++(0.4,0.4) -- node[below,obj]{$\oc\mathsf{X}$} ++(0.4,0);
    \draw (2.3,0.4) arc(-90:90:0.2);
    \draw (2.3,0.8) -- (-2.3,0.8);
    \draw (dg) -- node[above,obj]{$\oc\mathsf{X}$} ++(-0.8,0);
    \draw (-2.3,0) arc(270:90:0.4);
  \end{tikzpicture}
\end{center}
This construction already appeared in the interpretation of the fixed point
operator. In fact, for a term
$\mathtt{f}:\mathtt{A} \to \mathtt{B}, \mathtt{x}:\mathtt{A} \vdash
\mathtt{M}:\mathtt{B}$, we have
$\sem{\mathtt{fix}_{\mathtt{A},\mathtt{B}}(\mathtt{f},\mathtt{x},\mathtt{M})}
= \sem{\lambda \mathtt{x}^{\mathtt{A}}.\,\mathtt{M}}^{\dagger}$.

The goal of this section is to show that $\mathsf{M}^{\dagger}$ is a
fixed point of $\mathsf{M}$ and can be approximated by a family of
Mealy machines
\begin{equation*}
  \mathsf{M} \circ \mathsf{bot}_{\mathsf{I},\oc \mathsf{X}},\;
  \mathsf{M} \circ \oc(\mathsf{M} \circ
  \mathsf{bot}_{\mathsf{I},\oc \mathsf{X}}),\;
  \mathsf{M} \circ \oc(\mathsf{M} \circ \oc(
  \mathsf{M} \circ \mathsf{bot}_{\mathsf{I},\oc \mathsf{X}})),\ldots
  \colon \mathsf{I} \multimap \mathsf{X}.
\end{equation*}

\subsubsection{Parametrized Modal Operator and Parametrized Loop Operator}

We introduce parametrization of the modal operator $\oc$. For a subset
$\alpha \subseteq \mathbb{N}$ and for a Mealy machine
$\mathsf{M} \colon \mathsf{X} \multimap \mathsf{Y}$, we define a Mealy
machine
$\oc_{\alpha}\mathsf{M} \colon \oc \mathsf{X} \multimap \oc\mathsf{Y}$
by: the state space and the initial state of $\oc_{\alpha} \mathsf{M}$
are given by
\begin{equation*}
  \state{\oc_{\alpha}\mathsf{M}}
  = \state{\oc\mathsf{M}} = \state{\mathsf{M}}^{\mathbb{N}},
  \qquad
  \init{\oc_{\alpha}\mathsf{M}} = \init{\oc\mathsf{M}}
\end{equation*}
and $\tran{\oc_{\alpha}\mathsf{M}}$ is given by
\begin{multline*}
  \tran{\oc \mathsf{M}}((i,(n,z)),(s_{n})_{n \in \mathbb{N}})
  = \\
  \begin{cases}
    ((j,(n,w)),(s_{1},\ldots,s_{n-1},t,s_{n+1},\ldots)),
    & \textnormal{if } n \in \alpha
    \textnormal{ and } \tran{\mathsf{M}}((i,z),s_{n}) = ((j,w),t), \\
    \textnormal{undefined}, & \textnormal{otherwise}
  \end{cases}
\end{multline*}
where $i,j \in \{0,1\}$ and $z,w$ vary over the corresponding sets.
For example, if we have
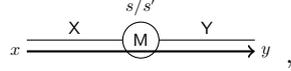
\begin{center}
  \begin{tikzpicture}
    \node (m)[draw,circle,obj] at (0,0) {$\mathsf{M}$};
    \draw (m)--node[above,obj]{$\mathsf{Y}$}++(1.5,0);
    \draw (m)--node[above,obj]{$\mathsf{X}$}++(-1.5,0);
    \node [above,obj] at (m.north) {$s/s'$};
    \draw[thick,->] (-1.5,-0.15) node[left,obj]{$x$} --
    (1.5,-0.15) node[right,obj]{$y$};
  \end{tikzpicture}
  ,
\end{center}
then for any $t_{1},t_{2},\ldots \in \state{\mathsf{M}}$ and
for any $n \in \mathbb{N}$, we have
\begin{center}
  \begin{tikzpicture}
    \node (m)[draw,circle,obj] at (0,0) {$\oc\mathsf{M}$};
    \draw (m)--node[above,obj]{$\mathsf{Y}$}++(1.5,0);
    \draw (m)--node[above,obj]{$\mathsf{X}$}++(-1.5,0);
    \node [above,obj] at (m.north) {$(t_{1},\ldots, t_{n-1},s,t_{n+1},\ldots)
      /(t_{1},\ldots, t_{n-1},s',t_{n+1},\ldots)$};
    \draw[thick,->] (-1.5,-0.15) node[left,obj]{$(n,x)$} --
    (1.5,-0.15) node[right,obj]{$(n,y)$};
  \end{tikzpicture}
\end{center}
whenever $n \in \alpha$. When $n \notin \alpha$, there is no output
from $\oc_{\alpha}\mathsf{M}$. We can think $\oc_{\alpha}\mathsf{M}$
as a ``restriction'' of $\oc \mathsf{M}$ to $\alpha$. In fact,
$\oc \mathsf{M}$ is equal to 
$\oc_{\mathbb{N}} \mathsf{M}$.

We are interested in restrictions of $\oc$ to
subsets $\alpha_{n},\beta_{n} \subseteq \mathbb{N}$ inductively given
by
\begin{equation*}
  \alpha_{0} = \emptyset,
  \qquad
  \beta_{n} =
  \{\langle i,j \rangle :
  i \in \alpha_{n} \textnormal{ and } j \in \mathbb{N}\},
  \qquad
  \alpha_{n+1}
  = \{2i : i \in \mathbb{N}\} \cup
  \{2i+1 : i \in \beta_{n}\}.
\end{equation*}
The definition of $\alpha_{n}$ and $\beta_{n}$ are motivated by the
following lemma.
\begin{lemma}\label{lem:alphabeta}
  For any $n \in \mathbb{N}$ and for any
  $\mathsf{M} \colon \mathsf{X} \multimap \mathsf{Y}$, we have
  \begin{equation*}
    \mathsf{c}_{\mathsf{Y}} \circ \oc_{\alpha_{n+1}} \mathsf{M}
    \simeq
    (\oc \mathsf{M} \otimes \oc_{\beta_{n}} \mathsf{M})
    \circ \mathsf{c}_{\mathsf{X}},
    \qquad
    \mathsf{dg}_{\mathsf{Y}} \circ \oc_{\beta_{n+1}} \mathsf{M}
    \simeq
    \oc_{\alpha_{n}} \oc \mathsf{M} \circ \mathsf{dg}_{\mathsf{X}}.
  \end{equation*}
\end{lemma}

By means of $\oc_{\alpha}$, we can also parametrize the operator
$(-)^{\dagger}$. For $\alpha \subseteq \mathbb{N}$, and for
$\mathsf{M} \colon \oc\mathsf{X} \multimap \mathsf{X}$, we define
$\mathsf{M}^{\dagger,\alpha} \colon \mathsf{I} \to \oc\mathsf{X}$ by
\begin{equation*}
  \mathsf{M}^{\dagger,\alpha} =
  (\mathsf{M} \otimes \mathsf{counit}_{\oc\mathsf{X}})
  \circ ((\mathsf{c}_{\mathsf{X}} \circ
  \oc_{\alpha}\mathsf{M}) \circ \mathsf{dg}_{\mathsf{X}})
  \otimes \mathsf{id}_{\oc\mathsf{X}}
  \circ \mathsf{unit}_{\oc\mathsf{X}}.
\end{equation*}
It is easy to see that $\mathsf{M}^{\dagger}$ is equal to
$\mathsf{M}^{\dagger,\mathbb{N}}$.

\begin{lemma}\label{lem:expand}
  For any Mealy machine
  $\mathsf{M} \colon \oc \mathsf{X} \multimap \mathsf{X}$, we have
  \begin{equation*}
    \mathsf{M}^{\dagger,\alpha_{n + 1}}
    \simeq \mathsf{M} \circ (\oc \mathsf{M})^{\dagger,\alpha_{n}}.
  \end{equation*}
\end{lemma}
\begin{proof}
  See Figure~\ref{fig:fix}.
\end{proof}

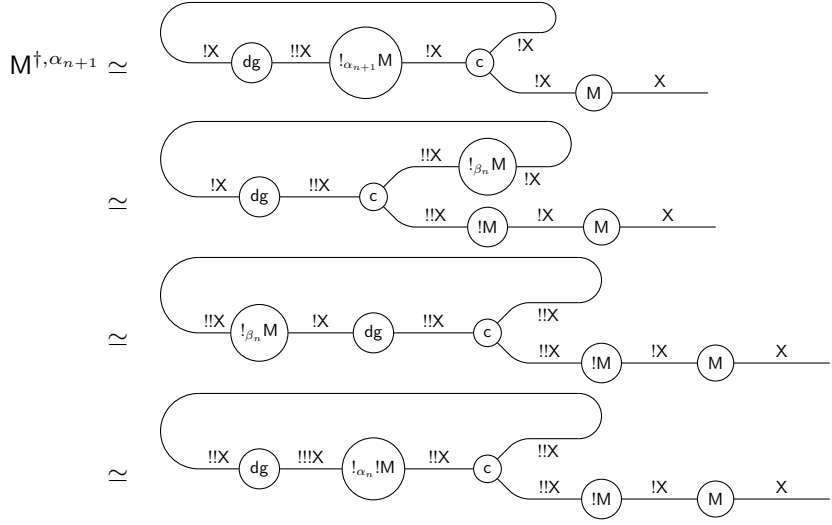
\begin{figure}
  \centering
  \begin{center}
    \begin{tabular}{rl}
      \raisebox{0.5cm}{$\mathsf{M}^{\dagger,\alpha_{n+1}} \simeq$}
      &
        \begin{tikzpicture}[rounded corners]
          \node (!m)[draw,circle,obj] at (0,0) {$\oc_{\alpha_{n+1}}\mathsf{M}$};
          \node (dg)[draw,circle,obj] at (-1.5,0) {$\mathsf{dg}$};
          \node (c)[draw,circle,obj] at (1.5,0) {$\mathsf{c}$};
          \node (m)[draw,circle,obj] at (3,-0.4) {$\mathsf{M}$};

          \draw (dg) -- node[above,obj]{$\oc\oc\mathsf{X}$} (!m);
          \draw (!m) -- node[above,obj]{$\oc\mathsf{X}$} (c);
          \draw (c) -- ++ (0.4,-0.4) -- node[above,obj]{$\oc\mathsf{X}$}(m);
          \draw (m) --node[above,obj]{$\mathsf{X}$} ++(1.5,0);
          \draw (c) -- ++(0.4,0.4) -- node[below,obj]{$\oc\mathsf{X}$} ++(0.4,0);
          \draw (2.3,0.4) arc(-90:90:0.2);
          \draw (2.3,0.8) -- (-2.3,0.8);
          \draw (dg) -- node[above,obj]{$\oc\mathsf{X}$} ++(-0.8,0);
          \draw (-2.3,0) arc(270:90:0.4);
        \end{tikzpicture} \\
      \raisebox{0.5cm}{$\simeq$}
      &
        \begin{tikzpicture}[rounded corners]
          \node (!bm)[draw,circle,obj] at (1.5,0.4) {$\oc_{\beta_{n}}\mathsf{M}$};
          \node (!m)[draw,circle,obj] at (1.5,-0.4) {$\oc\mathsf{M}$};
          \node (dg)[draw,circle,obj] at (-1.5,0) {$\mathsf{dg}$};
          \node (c)[draw,circle,obj] at (0,0) {$\mathsf{c}$};
          \node (m)[draw,circle,obj] at (3,-0.4) {$\mathsf{M}$};

          \draw (dg) -- node[above,obj]{$\oc\oc\mathsf{X}$} (c);
          \draw (c) -- ++ (0.4,-0.4) -- node[above,obj]{$\oc\oc\mathsf{X}$}(!m);
          \draw (c) -- ++ (0.4,0.4) -- node[above,obj]{$\oc\oc\mathsf{X}$}(!bm);
          \draw (!m) -- node[above,obj]{$\oc\mathsf{X}$} (m);
          \draw (m) --node[above,obj]{$\mathsf{X}$} ++(1.5,0);
          \draw (!bm) -- node[below,obj]{$\oc\mathsf{X}$} ++(0.8,0);
          \draw (2.3,0.4) arc(-90:90:0.3);
          \draw (2.3,1) -- (-2.3,1);
          \draw (dg) -- node[above,obj]{$\oc\mathsf{X}$} ++(-0.8,0);
          \draw (-2.3,0) arc(270:90:0.5);
        \end{tikzpicture} \\
      \raisebox{0.5cm}{$\simeq$}
      &
        \begin{tikzpicture}[rounded corners]
          \node (!bm)[draw,circle,obj] at (-3,0) {$\oc_{\beta_{n}}\mathsf{M}$};
          \node (!m)[draw,circle,obj] at (1.5,-0.4) {$\oc\mathsf{M}$};
          \node (dg)[draw,circle,obj] at (-1.5,0) {$\mathsf{dg}$};
          \node (c)[draw,circle,obj] at (0,0) {$\mathsf{c}$};
          \node (m)[draw,circle,obj] at (3,-0.4) {$\mathsf{M}$};

          \draw (dg) -- node[above,obj]{$\oc\oc\mathsf{X}$} (c);
          \draw (c) -- ++ (0.4,-0.4) -- node[above,obj]{$\oc\oc\mathsf{X}$}(!m);
          \draw (c) -- ++ (0.4,0.4) -- node[below,obj]{$\oc\oc\mathsf{X}$} ++(0.8,0);
          \draw (!m) -- node[above,obj]{$\oc\mathsf{X}$} (m);
          \draw (m) --node[above,obj]{$\mathsf{X}$} ++(1.5,0);
          \draw (1.2,0.4) arc(-90:90:0.3);
          \draw (1.2,1) -- (-3.8,1);
          \draw (dg) -- node[above,obj]{$\oc\mathsf{X}$} (!bm);
          \draw (!bm) -- node[above,obj]{$\oc\oc\mathsf{X}$} ++(-0.8,0);
          \draw (-3.8,0) arc(270:90:0.5);
        \end{tikzpicture} \\
      \raisebox{0.5cm}{$\simeq$}
      &
        \begin{tikzpicture}[rounded corners]
          \node (!bm)[draw,circle,obj] at (-1.5,0) {$\oc_{\alpha_{n}}\oc\mathsf{M}$};
          \node (!m)[draw,circle,obj] at (1.5,-0.4) {$\oc\mathsf{M}$};
          \node (dg)[draw,circle,obj] at (-3,0) {$\mathsf{dg}$};
          \node (c)[draw,circle,obj] at (0,0) {$\mathsf{c}$};
          \node (m)[draw,circle,obj] at (3,-0.4) {$\mathsf{M}$};

          \draw (!bm) -- node[above,obj]{$\oc\oc\mathsf{X}$} (c);
          \draw (c) -- ++ (0.4,-0.4) -- node[above,obj]{$\oc\oc\mathsf{X}$}(!m);
          \draw (c) -- ++ (0.4,0.4) -- node[below,obj]{$\oc\oc\mathsf{X}$} ++(0.8,0);
          \draw (!m) -- node[above,obj]{$\oc\mathsf{X}$} (m);
          \draw (m) --node[above,obj]{$\mathsf{X}$} ++(1.5,0);
          \draw (1.2,0.4) arc(-90:90:0.3);
          \draw (1.2,1) -- (-3.8,1);
          \draw (dg) -- node[above,obj]{$\oc\oc\oc\mathsf{X}$}(!bm);
          \draw (dg) -- node[above,obj]{$\oc\oc\mathsf{X}$} ++(-0.8,0);
          \draw (-3.8,0) arc(270:90:0.5);
        \end{tikzpicture}
    \end{tabular}
  \end{center}
  \caption{A Diagrammatic Proof of Lemma~\ref{lem:expand}}
  \label{fig:fix}
\end{figure}

\begin{lemma}\label{lem:phi'}
  For any Mealy machine
  $\mathsf{M} \colon \oc \mathsf{X} \multimap \mathsf{X}$ and for any
  $n \in \mathbb{N}$,
  \begin{equation*}
    (\oc \mathsf{M})^{\dagger,\alpha_{n}}
    \simeq
    \oc (\mathsf{M}^{\dagger,\alpha_{n}}).
  \end{equation*}
\end{lemma}
\begin{proof}
  We prove the statement by induction on $n$. The base case follows
  from that the transition functions of
  $(\oc \mathsf{M})^{\dagger,\alpha_{0}}$ and
  $\oc (\mathsf{M}^{\dagger,\alpha_{0}})$ are equal to the empty
  partial function.
  We next check the induction step. We have
  \begin{align*}
    (\oc \mathsf{M})^{\dagger,\alpha_{n+1}}
    &\simeq \oc \mathsf{M} \circ (\oc \oc \mathsf{M})^{\dagger,\alpha_{n}}
    \tag{Lemma~\ref{lem:expand}}\\
    &\simeq \oc \mathsf{M} \circ \oc (\oc \mathsf{M})^{\dagger,\alpha_{n}}
      \tag{Induction hypothesis} \\
    &\simeq \oc (\mathsf{M} \circ (\oc \mathsf{M})^{\dagger,\alpha_{n}})
      \tag{Functoriality} \\
    &\simeq \oc (\mathsf{M}^{\dagger,\alpha_{n+1}}).
      \tag{Lemma~\ref{lem:expand}}
  \end{align*}
\end{proof}

\begin{proposition}\label{prop:iter}
  For a Mealy machine
  $\mathsf{M} \colon \oc \mathsf{X} \multimap \mathsf{X}$, we
  inductively define
  $\mathsf{iter}_{n}(\mathsf{M}) \colon \mathsf{I} \multimap
  \mathsf{X}$ by
  \begin{equation*}
    \mathsf{iter}_{0}(\mathsf{M}) = \mathsf{bot}_{\mathsf{I},\mathsf{X}},
    \qquad
    \mathsf{iter}_{n+1}(\mathsf{M})
    = \mathsf{M} \circ \oc (\mathsf{iter}_{n}(\mathsf{M})).
  \end{equation*}
  For all $n \in \mathbb{N}$, we have
  \begin{equation*}
    \state{\mathsf{M}^{\dagger}} =
    \state{\mathsf{M}^{\dagger,\alpha_{n}}}, \qquad
    \init{\mathsf{M}^{\dagger}} =
    \init{\mathsf{M}^{\dagger,\alpha_{n}}}, \qquad
    \mathsf{M}^{\dagger,\alpha_{n}}
    \simeq
    \mathsf{iter}_{n}(\mathsf{M}),
  \end{equation*}
  and
  \begin{equation*}
    \tran{\mathsf{M}^{\dagger,\alpha_{0}}} \leq 
    \tran{\mathsf{M}^{\dagger,\alpha_{1}}} \leq
    \tran{\mathsf{M}^{\dagger,\alpha_{2}}} \leq
    \cdots, \qquad
    \tran{\mathsf{M}^{\dagger}} =
    \bigvee_{n \geq 0} \tran{\mathsf{M}^{\dagger,\alpha_{n}}}.
  \end{equation*}
  Hence, we have an ascending chain
  \begin{equation*}
    [\mathsf{M}^{\dagger,\alpha_{0}}]
    \leq
    [\mathsf{M}^{\dagger,\alpha_{1}}]
    \leq
    [\mathsf{M}^{\dagger,\alpha_{2}}]
    \leq \cdots,
  \end{equation*}
  and $[\mathsf{M}^{\dagger}]$ is the least
  upper bound of the ascending chain $[\mathsf{iter}_{n}(\mathsf{M})]$.
\end{proposition}
\begin{proof}
  It follows from the definition of $\oc_{\alpha_{n}}$, we have
  \begin{equation*}
    \state{\oc\mathsf{M}} = \state{\oc_{\alpha_{n}} \mathsf{M}}, \qquad
    \init{\oc\mathsf{M}} = \init{\oc_{\alpha_{n}} \mathsf{M}}
  \end{equation*}
  for all $n \in \mathbb{N}$, and
  \begin{equation*}
    \tran{\oc_{\alpha_{0}}\mathsf{M}}
    \leq \tran{\oc_{\alpha_{1}} \mathsf{M}}
    \leq \tran{\oc_{\alpha_{2}} \mathsf{M}} \leq \cdots,
    \qquad
    \tran{\oc\mathsf{M}} =
    \bigvee_{n \geq 0} \tran{\oc_{\alpha_{n}}\mathsf{M}}
  \end{equation*}
  Hence, by the definition of the composition and the monoidal
  product, we have
  \begin{equation*}
    \state{\mathsf{M}^{\dagger}} =
    \state{\mathsf{M}^{\dagger,\alpha_{n}}}, \qquad
    \init{\mathsf{M}^{\dagger}} =
    \init{\mathsf{M}^{\dagger,\alpha_{n}}}
  \end{equation*}
  for all $n \in \mathbb{N}$, and
  \begin{equation*}
    \tran{\mathsf{M}^{\dagger,\alpha_{0}}} \leq 
    \tran{\mathsf{M}^{\dagger,\alpha_{1}}} \leq
    \tran{\mathsf{M}^{\dagger,\alpha_{2}}} \leq
    \cdots, \qquad
    \tran{\mathsf{M}^{\dagger}} =
    \bigvee_{n \geq 0} \tran{\mathsf{M}^{\dagger,\alpha_{n}}}.
  \end{equation*}
  It remains to check
  $\mathsf{iter}_{n}(\mathsf{M}) \simeq
  \mathsf{M}^{\dagger,\alpha_{n}}$. We show this by induction on $n$.
  For the base case, we have
  $\mathsf{M}^{\dagger,\emptyset} \simeq
  \mathsf{iter}_{0}(\mathsf{M})$ because these Mealy machines
  $\mathsf{M}^{\dagger,\emptyset}$ and $\mathsf{iter}_{0}(\mathsf{M})$
  are behaviorally equivalent to
  $\mathsf{bot}_{\mathsf{I},\mathsf{X}}$. For the induction step,
  \begin{equation*}
    \mathsf{M} \circ \oc (\mathsf{iter}_{n}(\mathsf{M}))
    \overset{\textnormal{induction hypothesis}}{\simeq}
    \mathsf{M} \circ \oc (\mathsf{M}^{\dagger,\alpha_{n}})
    \overset{\textnormal{Lemma~\ref{lem:phi'}}}{\simeq}
    \mathsf{M} \circ (\oc \mathsf{M})^{\dagger,\alpha_{n}}
    \overset{\textnormal{Lemma~\ref{lem:expand}}}{\simeq}
    \mathsf{M}^{\dagger,\alpha_{n+1}}.
  \end{equation*}
  Because $\mathsf{iter}_{n+1}(\mathsf{M})$ is equal to
  $\mathsf{M} \circ \oc (\mathsf{iter}_{n}(\mathsf{M}))$,
  we obtain
  $\mathsf{iter}_{n+1}(\mathsf{M}) \simeq
  \mathsf{M}^{\dagger,\alpha_{n+1}}$.
\end{proof}

\begin{proposition}\label{prop:fix}
  For any Mealy machine
  $\mathsf{M} \colon \oc \mathsf{X} \multimap \mathsf{X}$,
  \begin{equation*}
    \mathsf{M} \circ \oc (\mathsf{M}^{\dagger})
    \simeq \mathsf{M}^{\dagger}.
  \end{equation*}
\end{proposition}
\begin{proof}
  Because $\oc(-)$ and the composition
  of Mealy machines are continuous, we have
  \begin{equation*}
    \mathsf{M} \circ \oc (\mathsf{M}^{\dagger})
    \simeq
    \mathsf{M} \circ \bigvee_{n \in \mathbb{N}} \oc (\mathsf{M}^{\dagger,\alpha_{n}}) 
    \simeq
    \bigvee_{n \in \mathbb{N}} (\mathsf{M} \circ \oc (\mathsf{M}^{\dagger,\alpha_{n}})) 
    \simeq
    \bigvee_{n \in \mathbb{N}} \mathsf{iter}_{n+1}(\mathsf{M}) 
    \simeq
    \bigvee_{n \in \mathbb{N}} \mathsf{M}^{\dagger,\alpha_{n+1}} 
    \simeq \mathsf{M}^{\dagger}.
  \end{equation*}
\end{proof}
}

\section{How About S-Finite Kernels?}\label{sec:sfinite}
The reader experienced with the semantics of probabilistic programming
languages have probably already wondered whether a GoI model for
$\PCFSS$ could be given out of s-finite kernels instead of measurable
functions, following Staton's work on the semantics of a first-order
probabilistic programming language~\cite{staton2017}.

The answer is indeed positive: the kind of construction we have
presented in Section~\ref{sec:mealy} can in fact be adapted to the
category of measurable spaces \emph{and s-finite kernels}. The latter,
being traced monoidal, has all the necessary structure one needs
\cite{ahs2002}. What one obtains proceeding this way is indeed a
GoI model, but adequate only for the distribution-based operational
semantics.

The interpretation of any program in this alternative GoI can be seen
as structurally identical to the one from Section~\ref{sec:mealy} once
the sample and score operators are interpreted as usual, namely as
those s-finite kernels which actually perform sampling and scoring
\emph{internally}. \longshortv{Below, we first recall the definition
  of s-finite kernel, and then we introduce Mealy machines whose
  transition is described in terms of an s-finite kernel, and we give
  some basic Mealy machines. Finally, we give an adequate GoI model
  for the distribution-based operational semantics. }{More details on
  this alternate model is available in a longer version of this
  paper~\cite{dlh2019}. }

Being adequate for the distribution-based semantics directly (and not
by way of integration as in Theorem~\ref{thm:adq2}) has the pleasant consequence of validating a
number of useful program transformations, and in particular commutation
of sampling and scoring effects, see~\cite{bdlgs2016} for a thorough discussion
about this topic, and about how s-finite kernels are a particularly nice
way of achieving commutativity in presence of scoring.

\longv{
  \subsection{S-finite Kernels}

  Let $k \colon X \leadsto Y$ be a kernel. We say that $k$ is
  \emph{finite} when there is a real number $c > 0$ such that for all
  $x \in X$ and $A \in \Sigma_{Y}$, we have $k(x,A) < c$. An
  \emph{s-finite kernel} is a kernel $k \colon X \leadsto Y$ such that
  there is a countable family $\{k_{n} \colon X \leadsto Y\}$ of finite
  kernels such that $k(x,A) = \sum_{n \in \mathbb{N}} k_{n}(x,A)$ for
  all $x \in X$ and $A \in \Sigma_{Y}$. It is easy to see that s-finite
  kernels are closed under the pointwise addition. We write
  $\sum_{i \in I} k_{i} \colon X \leadsto Y$ for the pointwise addition
  of s-finite kernels $k_{i} \colon X \leadsto Y$. A
  \emph{(sub)probability kernel} is a kernel $k \colon X \leadsto Y$
  such that $k(x,-)$ is a (sub)probability measure on $X$ for all
  $x \in X$. Every (sub)probability kernel is a finite kernel.

  Every measurable function $f \colon X \to Y$ gives rise to a
  probability kernel $\hat{f} \colon X \leadsto Y$ given by
  \begin{equation*}
    \hat{f}(x,A) = [f(x) \in A].
  \end{equation*}
  We denote the probability kernel induced by the identity measurable
  function by $\id_{X} \colon X \leadsto X$. Concretely, this is given
  by $\id_{X}(x,A) = [x \in A]$.

  We recall two constructions of s-finite kernels.
  \begin{itemize}
  \item (Composition) For s-finite kernels $k \colon X \leadsto Y$ and
    $h \colon Y \leadsto Z$, we define an s-finite kernel
    $h \circ k \colon X \leadsto Z$ by
    \begin{equation*}
      (h \circ k)(x,C) =
      \int h(y, C) k(x,\mathrm{d}y).
    \end{equation*}
    The composition of s-finite kernels is associative and satisfies the
    unit laws, namely, we have $k \circ \id_{X} = k$ and $\id_{Y} \circ k = k$.
  \item (Tensor product) For s-finite kernels $k \colon X \leadsto Y$ and
    $h \colon Z \leadsto W$, we define an s-finite kernel
    $k \otimes h \colon X \times Z \leadsto Y \times W$ to be the unique
    s-finite kernel such that for all $(x,z) \in X \times Z$ and for all
    $A \in \Sigma_{Y}$ and $B \in \Sigma_{W}$,
    \begin{equation*}
      (k \otimes h)((x,z),A \times B) = k(x,A)h(z,B).
    \end{equation*}
  \end{itemize}
  The tensor product and the coproduct of s-finite kernels is
  functorial. This means that these constructors are compatible with the
  composition and preserve identities. The following proposition
  summarizes catagorical status of these structures.
  \begin{proposition}
    The category of measurable spaces and s-finite kernels with
    $\otimes$ forms a symmetric monoidal category where the unit object
    is $1$. The object $\emptyset$ is the zero object, and $X + Y$ with
    \begin{equation*}
      \widehat{\mathrm{inl}_{X,Y}} \colon X \leadsto X + Y,
      \qquad
      \widehat{\mathrm{inr}_{X,Y}} \colon Y \leadsto X + Y,
    \end{equation*}
    forms the coproduct of $X$ and $Y$ where
    $\mathrm{inl}_{X,Y} \colon X \to X + Y$ and
    $\mathrm{inr}_{X,Y} \colon X \to X + Y$ are the first and the second
    injections. Furthermore, the monoidal product distributes over the
    coproducts. Namely, the canonical s-finite kernel
    $\widehat{\mathrm{dst}_{X,Y,Z}} \colon X \times Z + Y \times Z
    \leadsto (X + Y) \times Z$ given by
    \begin{equation*}
      \mathrm{dst}_{X,Y,Z}(\hh,(x,z)) = ((\hh,x),z),
      \qquad
      \mathrm{dst}_{X,Y,Z}(\mm,(y,z)) = ((\mm,y),z)
    \end{equation*}
    is a natural isomorphism.
  \end{proposition}
  \begin{proof}
    For associativity of the composition, see
    \cite[Lemma~3]{staton2017}, and for functoriality of $\otimes$, see
    \cite[Proposition~5]{staton2017}. It is not difficult to check that
    the category of measurable spaces and s-finite kernels associated
    with $\otimes$ and $1$ forms a symmetric monoidal category. For
    s-finite kernels $f \colon X \leadsto Z$ and
    $g \colon Y \leadsto Z$, the cotupling
    $[f,g] \colon X + Y \leadsto Z$ is given by
    \begin{equation*}
      [f,g]((\hh,x),A) = f(x,A),
      \qquad
      [f,g]((\mm,y),A) = g(y,A).
    \end{equation*}
    It follows from universality of coproducts that
    $\widehat{\mathrm{dst}_{X,Y,Z}}$ is a natural isomorphism.
  \end{proof}
  For s-finite kernels $k \colon X \leadsto Y$ and
  $h \colon Z \leadsto W$, we define an s-finite kernel
  $k \oplus h \colon X + Z \leadsto Y + W$ by
  \begin{align*}
    (k \oplus h)((\hh,x),A)
    &= k(x,A_{Y}) \textnormal{ where } A_{Y} = \{y : (\hh,y) \in A\}, \\
    (k \oplus h)((\mm,z),A)
    &= k(x,A_{W}) \textnormal{ where } A_{W} = \{w : (\mm,w) \in A\}.
  \end{align*}
  This is the unique s-finite kernel satisfying
  \begin{equation*}
    (k \oplus h) \circ \mathrm{inl}_{X,Z} = \mathrm{inl}_{Y,W} \circ k,
    \qquad
    (k \oplus h) \circ \mathrm{inr}_{X,Z} = \mathrm{inr}_{Y,W} \circ h.
  \end{equation*}

  \subsection{Probabilistic Mealy Machine}
  \label{sec:ppmm}

  \begin{definition}\label{def:ppmm}
    For $\mathbf{Int}$-objects $\mathsf{X}$ and $\mathsf{Y}$, a
    \emph{probabilistic Mealy machine} $\mathsf{M}$ from $\mathsf{X}$ to
    $\mathsf{Y}$ consists of
    \begin{itemize}
    \item a measurable space $\state{\mathsf{M}}$ called the \emph{state
        space} of $\mathsf{M}$;
    \item an element $\init{\mathsf{M}} \in \state{\mathsf{M}}$
      called the \emph{initial state} of $\mathsf{M}$;
    \item an s-finite kernel
      $\tran{\mathsf{M}} \colon (X^{+} + Y^{-}) \times \state{\mathsf{M}}
      \leadsto (Y^{+} + X^{-}) \times \state{\mathsf{M}}$ called the
      \emph{transition relation}.
    \end{itemize}
    When $\mathsf{M}$ is a probabilistic Mealy machine from
    $\mathsf{X}$ to $\mathsf{Y}$, we write
    $\mathsf{M} \colon \mathsf{X} \rightarrowtriangle \mathsf{Y}$.
  \end{definition}

  We can regard a Mealy machine
  $\mathsf{M} \colon \mathsf{X} \multimap \mathsf{Y}$ as a
  probabilistic Mealy machine from $\mathsf{X}$ to $\mathsf{Y}$ by
  identifying the transition function
  $\tran{\mathsf{M}} \colon (X^{+} + Y^{-}) \times \state{\mathsf{M}}
  \to (Y^{+} + X^{-}) \times \state{\mathsf{M}}$ with the
  correspondnig s-finite kernel
  $\widehat{\tran{\mathsf{M}}} \colon (X^{+} + Y^{-}) \times
  \state{\mathsf{M}} \leadsto (Y^{+} + X^{-}) \times
  \state{\mathsf{M}}$. In the sequel, we confuse Mealy machines (and
  token machines) with corresponding probabilistic Mealy machines.

  Let $\mathsf{X}_{1},\ldots,\mathsf{X}_{n}$ and
  $\mathsf{Y}_{1},\ldots,\mathsf{Y}_{m}$ be $\mathbf{Int}$-object.
  Just like Mealy machines, we depict a probabilistic Mealy machine
  $\mathsf{M}$ from
  $\mathsf{X}_{1} \otimes \cdots \otimes \mathsf{X}_{n}$ to
  $\mathsf{Y}_{1} \otimes \cdots \otimes \mathsf{Y}_{m}$ as a box with
  edges labeled by $\mathsf{X}_{1},\ldots,\mathsf{X}_{n}$ on the left
  hand side and edges labeled by
  $\mathsf{Y}_{1},\ldots,\mathsf{Y}_{m}$ on the right hand side:
  
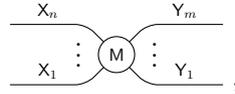
\begin{center}
    \begin{tikzpicture}[rounded corners]
      \node (M) [draw,circle,obj] at (0,0) {$\mathsf{M}$};
      \draw (M) -- ++ (0.4,0.4) -- node[above,obj] {$\mathsf{Y}_{m}$} ++ (1,0);
      \node at (0.5,0.1) {$\vdots$};
      \draw (M) -- ++ (0.4,-0.4) -- node[above,obj] {$\mathsf{Y}_{1}$} ++ (1,0);
      \draw (M) -- ++ (-0.4,0.4) -- node[above,obj] {$\mathsf{X}_{n}$} ++ (-1,0);
      \node at (-0.5,0.1) {$\vdots$};
      \draw (M) -- ++ (-0.4,-0.4) -- node[above,obj] {$\mathsf{X}_{1}$} ++ (-1,0);
    \end{tikzpicture}
    ,
  \end{center}
  and we depict transitions as arrows. For example, when $n=m=1$, we depict
  \begin{equation*}
    \tran{\mathsf{M}}(((\mm,y),s),A) = 0.4 [((\mm,x),s_{1}) \in A],
  \end{equation*}
  for $y \in Y_{1}^{-}$, $x \in X_{1}^{-}$ and
  $s,s_{1},s_{2} \in \state{\mathsf{M}}$ as the following arrow
  
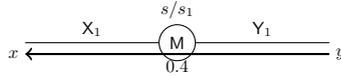
\begin{center}
    \begin{tikzpicture}[rounded corners]
      \node (M) [draw,circle,obj] at (0,0) {$\mathsf{M}$};
      \node [above,obj] at (M.north) {$s/s_{1}$};
      \draw (M) -- node [above,obj] {$\mathsf{Y}_{1}$} ++ (2,0);
      \draw (M) -- node [above,obj] {$\mathsf{X}_{1}$} ++ (-2,0);
      \draw[->,thick] (2,-0.15) node[right,obj]{$y$} -- node[below,obj]{$0.4$} ++ (-4,0)
      node [left,obj]{$x$};
    \end{tikzpicture}
  \end{center}
  where the positive real on the arrow indicate probabilities of the
  transition. Below, we may omit states and probabilities of transitions when
  they are not important or are easy to infer.

  \subsection{Behavioral Equivalence}
  \label{sec:pbhe}

  We give an equivalence relation between probabilistic Mealy machines
  so as to identify probabilistic Mealy machines that behaves in the
  same way. Let $\mathsf{X}$ and $\mathsf{Y}$ be
  $\mathbf{Int}$-objects, and let $\mathsf{M}$ and $\mathsf{N}$ be
  probabilistic Mealy machines from $\mathsf{X}$ to $\mathsf{Y}$. We
  write $\mathsf{M} \sim_{\mathsf{X},\mathsf{Y}} \mathsf{N}$ when
  there is a measurable function
  $f \colon \state{\mathsf{M}} \to \state{\mathsf{N}}$ such that
  $f(\init{\mathsf{M}}) = \init{\mathsf{N}}$ and the following diagram
  commutes:
  \begin{equation*}
    \xymatrix@C=20mm{
      (X^{+} + Y^{-}) \times \state{\mathsf{M}}
      \ar@{~>}[r]^-{\id \otimes \widehat{f}}
      \ar@{~>}[d]_{\tran{\mathsf{M}}}
      & (X^{+} + Y^{-}) \times \state{\mathsf{N}}
      \ar@{~>}[d]^{\tran{\mathsf{N}}}
      \\
      (Y^{+} + X^{-}) \times \state{\mathsf{M}}
      \ar@{~>}[r]^-{\id \otimes \widehat{f}}
      &
      (Y^{+} + X^{-}) \times \state{\mathsf{N}}
      \nulldot
    }
  \end{equation*}
  We define an equivalence relation $\simeq_{\mathsf{X},\mathsf{Y}}$
  to be the symmetric transitive closure of
  $\sim_{\mathsf{X},\mathsf{Y}}$. A probabilistic Mealy machine
  $\mathsf{M} \colon \mathsf{X} \rightarrowtriangle \mathsf{Y}$ is
  \emph{behaviorally equivalent} to
  $\mathsf{N} \colon \mathsf{X} \rightarrowtriangle \mathsf{Y}$ when
  we have $\mathsf{M} \simeq_{\mathsf{X},\mathsf{Y}} \mathsf{N}$. When
  we can infer subscripts of $\simeq_{\mathsf{X},\mathsf{Y}}$, we omit
  them. We say that a measurable function
  $f \colon \state{\mathsf{M}} \to \state{\mathsf{N}}$ realizes a
  behavioral equivalence $\mathsf{M} \simeq \mathsf{N}$ (realizes
  $\mathsf{M} \sim \mathsf{N}$) when 
  $\mathsf{M} \sim \mathsf{N}$ is witnessed by $f$.
}
\longv{
  \subsection{Construction of probabilistic Mealy Machines}
  \label{sec:pconst}
  We introduce probabilistic Mealy machines and their constructions that
  are building blocks of our denotational semantics. Most of them are
  adoptation of Mealy machines in Section~\ref{sec:const}, and we just
  give their formal definitions.

  \subsubsection{Composition/Cut}

  For probabilistic Mealy machines
  $\mathsf{M} \colon \mathsf{X} \rightarrowtriangle \mathsf{Y}$ and
  $\mathsf{N} \colon \mathsf{Y} \rightarrowtriangle \mathsf{Z}$, we
  define the state space and the initial states of
  $\mathsf{N} \circ \mathsf{M}$ by
  $\state{\mathsf{N} \circ \mathsf{M}} = \state{\mathsf{M}} \times
  \state{\mathsf{N}}$,
  $\init{\mathsf{N} \circ \mathsf{M}} =
  (\init{\mathsf{M}},\init{\mathsf{N}})$ and we define the transition
  relation $\tran{\mathsf{N} \circ \mathsf{M}}$ by
  \begin{equation*}
    \tran{\mathsf{N} \circ \mathsf{M}} = k_{X^{+} , Z^{-},Z^{+} , X^{-}}
    \vee
    \bigvee_{n \in \mathbb{N}} k_{Y^{+} , Y^{-},Z^{+} , X^{-}} \circ
    k_{Y^{+} , Y^{-},Y^{+} , Y^{-}}^{n} \circ k_{X^{+} , Z^{-},Y^{+} ,
      Y^{-}}
  \end{equation*}
  where $k_{A,B,C,D} \colon (A + B) \times \state{\mathsf{N} \circ \mathsf{M}}
  \leadsto (C + D) \times \state{\mathsf{N} \circ \mathsf{M}}$ are restrictions of
  the following s-finite kernel
  \begin{equation*}
    \xymatrix@R=3mm{
      (X^{+} + Z^{-} + Y^{+} + Y^{-}) \times
      \state{\mathsf{N} \circ \mathsf{M}}
      \ar@{~>}[d]^-{\cong} \\
      (X^{+} + Y^{-}) \times \state{\mathsf{M}} \times \state{\mathsf{N}}
      +
      (Y^{+} + Z^{-}) \times \state{\mathsf{N}} \times \state{\mathsf{M}}
      \ar@{~>}[d]^-{(\tran{\mathsf{M}} \otimes \state{\mathsf{N}})
        \oplus (\tran{\mathsf{N}} \otimes \state{\mathsf{M}})} \\
      (Y^{+} + X^{-}) \times \state{\mathsf{M}} \times \state{\mathsf{N}}
      +
      (Z^{+} + Y^{-}) \times \state{\mathsf{N}} \times \state{\mathsf{M}}
      \ar@{~>}[d]^-{\cong} \\
      (Z^{+} + X^{-} + Y^{+} + Y^{-}) \times
      \state{\mathsf{N} \circ \mathsf{M}}
      \nullcomma
    }
  \end{equation*}
  namely, the s-finite kernels $k_{A,B,C,D}$ satisfies
  \begin{equation*}
    \xymatrix@C=17mm{
      A \times \state{\mathsf{N} \circ \mathsf{M}}
      \ar@{~>}[d]_{k_{A,B}}
      \ar@{~>}[r]^-{\mathrm{inj},\; \mathrm{dst}}
      &
      ((X^{+} + Y^{-}) \times \state{\mathsf{M}} \times \state{\mathsf{N}})
      + ((Y^{+} + Z^{-}) \times \state{\mathsf{N}} \times \state{\mathsf{M}})
      \ar@{~>}[d]^{(\tran{\mathsf{M}} \otimes \id_{\state{\mathsf{N}}})
        \oplus (\tran{\mathsf{N}} \otimes \id_{\state{\mathsf{M}}})} \\
      B \times \state{\mathsf{N} \circ \mathsf{M}}
      \ar@{~>}[r]_-{\mathrm{inj},\; \mathrm{dst}}
      &
      ((Y^{+} + X^{-}) \times \state{\mathsf{M}} \times \state{\mathsf{N}})
      + ((Z^{+} + Y^{-}) \times \state{\mathsf{N}} \times \state{\mathsf{M}})
      \nulldot \\
    }
  \end{equation*}
  Here, the horizontal arrows consists of the injection from $A$ into
  $X^{+} + Y^{-} + Y^{+} + Z^{-}$ followed by distributivity and
  symmetry. For example, when $A = X^{+} + Z^{-}$, the upper horizontal
  arrow is given by
  \begin{equation*}
    \xymatrix{
      (X^{+} + Z^{-}) \times \state{\mathsf{N} \circ \mathsf{M}}
      \ar@{~>}[d]^-{(\mathrm{inl}_{X^{+},Y^{-}} \oplus \mathrm{inr}_{Y^{+},Z^{-}})
        \otimes \mathrm{id}_{\state{\mathsf{N} \circ \mathsf{M}}}} \\
      ((X^{+} + Y^{-}) + (Y^{+} + Z^{-})) \times \state{\mathsf{N} \circ \mathsf{M}}
      \ar@{~>}[d]^-{\mathrm{dst}_{X^{+} + Y^{-},Y^{+} + Z^{-},\state{\mathsf{N} \circ \mathsf{M}}}} \\
      ((X^{+} + Y^{-}) \times \state{\mathsf{N} \circ \mathsf{M}})
      + ((Y^{+} + Z^{-}) \times \state{\mathsf{N} \circ \mathsf{M}})
      \ar@{~>}[d]^-{\mathrm{id}_{(X^{+} + Y^{-}) \times \state{\mathsf{N} \circ \mathsf{M}}}
        \oplus (\mathrm{id}_{Y^{+} + Z^{-}} \otimes \mathrm{sym}_{\state{\mathsf{N}},
          \state{\mathsf{M}}})} \\
      ((X^{+} + Y^{-}) \times \state{\mathsf{M}} \times \state{\mathsf{N}})
      + ((Y^{+} + Z^{-}) \times \state{\mathsf{N}} \times \state{\mathsf{M}})
      \nulldot
    }
  \end{equation*}
  Joins in the definition of the composition of probabilistic Mealy
  machines are the pointwise ordoer. We can check that the composition
  of probabilistic Mealy machines is compatible with behavioural
  equivalence and that $\mathbf{Int}$-objects and the composition of
  probabilistic Mealy machines is a category where the identity on an
  $\mathbf{Int}$-object $\mathsf{X}$ is
  $\mathsf{id}_{\mathsf{X}} \colon \mathsf{X} \rightarrowtriangle \mathsf{X}$
  (regarded as a probabilistic Mealy machine).

  \subsubsection{Monoidal Products}
  \label{sec:pmon}

  We give monoidal products of probabilistic Mealy machines.
  For probabilistic Mealy machines $\mathsf{M} \colon \mathsf{X} \rightarrowtriangle \mathsf{Z}$
  and $\mathsf{N} \colon \mathsf{Y} \rightarrowtriangle \mathsf{W}$, we define a
  probabilistic Mealy machine $\mathsf{M} \otimes \mathsf{N}\colon
  \mathsf{X} \otimes \mathsf{Y} \rightarrowtriangle \mathsf{Z} \otimes \mathsf{W}$ by:
  $\state{\mathsf{M} \otimes \mathsf{N}} = \state{\mathsf{M}} \times
  \state{\mathsf{N}}$,
  $\init{\mathsf{M} \otimes \mathsf{N}} =
  (\init{\mathsf{M}},\init{\mathsf{N}})$ and
  $\tran{\mathsf{M} \otimes \mathsf{N}}$ is given by
  \begin{equation*}
    \xymatrix@R=3mm{
      ((X^{+} + Y^{+}) + (W^{-} + Z^{-})) \times
      \state{\mathsf{M} \otimes \mathsf{N}}
      \ar@{~>}[d]^-{\cong} \\
      (X^{+} + Z^{-}) \times \state{\mathsf{M}} \times \state{\mathsf{N}}
      +
      (Y^{+} + W^{-}) \times \state{\mathsf{N}} \times \state{\mathsf{M}}
      \ar@{~>}[d]^-{(\tran{\mathsf{M}} \otimes \state{\mathsf{N}})
        \oplus (\tran{\mathsf{N}} \otimes \state{\mathsf{M}})} \\
      (Z^{+} + X^{-}) \times \state{\mathsf{M}} \times \state{\mathsf{N}}
      +
      (W^{+} + Y^{-}) \times \state{\mathsf{N}} \times \state{\mathsf{M}}
      \ar@{~>}[d]^-{\cong} \\
      ((Z^{+} + W^{+}) + (Y^{-} + X^{-})) \times
      \state{\mathsf{M} \otimes \mathsf{N}}
      \nulldot
    }
  \end{equation*}
  It is not difficult to check that the monoidal product is compatible
  with behavioural equivalence. We depict
  $\mathsf{M} \otimes \mathsf{N} \colon (\mathsf{X} \otimes
  \mathsf{Y}) \rightarrowtriangle (\mathsf{Z} \otimes \mathsf{W})$ as follows:
  \begin{center}
    \begin{tikzpicture}
      \node (M) [draw,circle,obj] at (0,0) {$\mathsf{M}$};
      \node (N) [draw,circle,obj] at (0,0.6) {$\mathsf{N}$};
      \draw (M) -- node[above,obj]{$\mathsf{Z}$} ++(1.5,0);
      \draw (M) -- node[above,obj]{$\mathsf{X}$} ++(-1.5,0);
      \draw (N) -- node[above,obj]{$\mathsf{W}$} ++(1.5,0);
      \draw (N) -- node[above,obj]{$\mathsf{Y}$} ++(-1.5,0);
    \end{tikzpicture}
  \end{center}
  For $\mathsf{unit}_{\mathsf{X}}, \mathsf{counit}_{\mathsf{X}}$ and
  $\mathsf{sym}_{\mathsf{X},\mathsf{Y}}$, we adopt the same
  diagrammatic presentation.

  \subsubsection{A Modal Operator}
  \label{sec:ocp}

  Let $\mathsf{M} \colon \mathsf{X} \rightarrowtriangle \mathsf{Y}$ be
  a probabilistic Mealy machine. We define a probabilistic Mealy
  machine
  $\oc \mathsf{M} \colon \oc \mathsf{X} \rightarrowtriangle \oc
  \mathsf{Y}$ by: the state space of $\oc \mathsf{M}$ is defined to be
  $|\mathsf{M}|^{\mathbb{N}}$ associated with the least
  $\sigma$-algebra such that for all
  $A_{1},A_{2},\ldots \in \Sigma_{\mathsf{M}}$,
  \begin{equation*}
    A_{1} \times A_{2} \times \cdots \in \Sigma_{\oc\mathsf{M}};
  \end{equation*}
  the initial state $\init{\oc\mathsf{M}}$ is
  $(\init{\mathsf{M}},\init{\mathsf{M}},\ldots)$; the transition
  function $\tran{\oc \mathsf{M}}$ is the unique
  partial measurable function satisfying
  \begin{equation*}
    \xymatrix@C=70pt{
      (X^{+} + Y^{-})
      \times \state{\mathsf{M}} \times \state{\mathsf{M}}^{\mathbb{N}}
      \ar@{~>}[r]^-{(\mathrm{inj}_{n} \oplus \mathrm{inj}_{n}) \otimes \mathrm{ins}_{n}}
      \ar@{~>}[d]_-{\tran{\mathsf{M}} \otimes \state{\mathsf{M}}^{\mathbb{N}}}
      &
      (\mathbb{N} \times X^{+} + \mathbb{N} \times Y^{-})
      \times \state{\mathsf{M}}^{\mathbb{N}}
      \ar@{~>}[d]^-{\tran{\oc\mathsf{M}}} \\
      (Y^{+} + X^{-})
      \times \state{\mathsf{M}} \times \state{\mathsf{M}}^{\mathbb{N}}
      \ar@{~>}[r]^-{(\mathrm{inj}_{n} \oplus \mathrm{inj}_{n}) \otimes \mathrm{ins}_{n}}
      &
      (\mathbb{N} \times Y^{+} + \mathbb{N} \times X^{-})
      \times \state{\mathsf{M}}^{\mathbb{N}}
    }
  \end{equation*}
  for all $n \in \mathbb{N}$. Here,
  $\mathrm{inj}_{n} \colon (-) \to \mathbb{N} \times (-)$ are the
  $n$th injections, and
  $\mathrm{ins}_{n} \colon \state{\mathsf{M}} \times
  \state{\mathsf{M}}^{\mathbb{N}} \to
  \state{\mathsf{M}}^{\mathbb{N}}$ sends
  $(s,\{s_{n}\}_{n \in \mathbb{N}})$ to
  $(s_{0},\ldots,s_{n-1},s,s_{n},s_{n+1},\ldots)$.

}
\longv{
  \subsubsection{Diagrammatic Reasoning on Probabilistic Mealy Machines}
  \label{sec:drp}
  
  Diagrammatic reasoning is valid also for probabilistic Mealy
  machines.

  \begin{proposition}\label{prop:cptp}
    The category $\mathbf{pMealy}$ of $\mathbf{Int}$-object and
    probabilistic Mealy machines (modulo behavioural equivalence) is a
    compact closed category. The dual of an $\mathbf{Int}$-object
    $\mathsf{X}$ is $\mathsf{X}^{\bot}$. The unit and the counit
    arrows are $\mathsf{unit}_{\mathsf{X}}$ and
    $\mathsf{counit}_{\mathsf{X}}$.
  \end{proposition}
  \begin{proposition}\label{prop:beh_eqp}
    If two probabilistic Mealy machines have the same diagrammatic
    presentation modulo some rearrangement of edges and nodes, then
    they are behaviourally equivalent.
  \end{proposition}
  We can check Proposition~\ref{prop:cptp} by replacing the category
  of partial measurable functions by the category of s-finite kernels
  in Section~\ref{sec:proof_cpt}.
}
\longv{

  \subsubsection{A State Monad}

  We define an $\mathbf{Int}$-object $\mathsf{J}$ by $(1,1)$ and
  define an $\mathbf{Int}$-object $\mathsf{J}_{0}$ by $(1,\emptyset)$.
  Then $\mathsf{J} \otimes (-)$ is a state monad (on
  $\mathbf{pMealy}$), whose unit and multiplication are given by:
  \begin{equation*}
    \mathsf{j} \otimes \mathsf{X} \colon
    \mathsf{X} \rightarrowtriangle \mathsf{J} \otimes \mathsf{X},
    \qquad
    \mathsf{n} \otimes \mathsf{X} \colon
    \mathsf{J} \otimes \mathsf{J} \otimes \mathsf{X} \rightarrowtriangle \mathsf{J} \otimes \mathsf{X}
  \end{equation*}
  where $\mathsf{j} = \mathsf{unit}_{\mathsf{J}_{0}}$ and
  $\mathsf{n} = \mathsf{J}_{0} \otimes \mathsf{counit_{\mathsf{J}_{0}}}
  \otimes \mathsf{J}_{0}^{\bot}$.

  \subsubsection{Scoring}

  We construct a probabilistic Mealy machine
  $\mathsf{Sc} \colon \mathsf{R} \rightarrowtriangle \mathsf{J}$ by:
  \begin{equation*}
    \state{\mathsf{Sc}} = 1,
    \quad
    \init{\mathsf{Sc}} = \ast
  \end{equation*}
  and
  \begin{equation*}
    \tran{\mathsf{Sc}}(((\hh,u),\ast),A)
    = 
    \begin{cases}
      |a|\,[((\hh,\ast),\ast) \in A], & \textnormal{if } u = a \cons v, \\
      0, & \textnormal{otherwise},
    \end{cases}
    \qquad
    \tran{\mathsf{Sc}}(((\mm,\ast),\ast),A)
    = [((\mm, \varepsilon),\ast) \in A].
  \end{equation*}
  The probabilistic Mealy machine simulates scoring
  $\mathtt{score}(\mathtt{r}_{a})$ as follows:
  
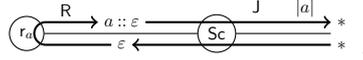
\begin{center}
    \begin{tikzpicture}
      \node (r) [draw,circle,obj] at (0,0) {$\mathsf{r}_{a}$};
      \node (sc) [draw,circle,obj] at (2.5,0) {$\mathsf{Sc}$};
      \draw (r) -- node[above=4,obj] {$\mathsf{R}$} ++(0.8,0) -- (sc);
      \draw (sc) -- node[above=4,obj] {$\mathsf{J}$\qquad$|a|$} ++ (1.5,0);
      \node (au) [obj] at (1.25,-0.15) {$\varepsilon$};
      \node (bau) [obj] at (1.25,0.15) {$a \mathbin{::} \varepsilon$};
      \draw[thick,<-] (au) -- (4,-0.15) node[right,obj] {$\ast$};
      \draw[thick,->] (bau) -- (4,0.15) node[right,obj] {$\ast$};
      \draw[thick] (au) -- ++(-1,0);
      \draw[thick,<-] (bau) -- ++(-1,0);
      \draw[thick] (0.25,0.15) arc(90:270:0.15);    
    \end{tikzpicture}
    .
  \end{center}

  \subsubsection{Sampling}
  \label{sec:samplingp}

  We define a Mealy machine $\mathsf{Sa} \colon \mathsf{I}
  \rightarrowtriangle \mathsf{J} \otimes \oc \mathsf{R}$ by: the state space
  $\state{\mathsf{Sa}}$ is defined to be $\{\ast\} \cup \mathbb{R}_{[0,1]}$, and
  the initial state $\init{\mathsf{Sa}}$ is $\ast$, and the
  transition function
  \begin{equation*}
    \tran{\mathsf{Sa}} \colon
    (\emptyset + (\mathbb{N} \times \mathbb{S} + 1))
    \times \state{\mathsf{Sa}}
    \leadsto \\
    ((1 + \mathbb{N} \times \mathbb{S}) + \emptyset)
    \times \state{\mathsf{Sa}}
  \end{equation*}
  is given by
  \begin{align*}
    \tran{\mathsf{Sa}}(((\mm,(\hh,(n,u))),s),A)
    &= [((\hh,(\mm,(n,s\mathbin{::}u))),s) \in A] \, [s \in \mathbb{R}], \\
    \tran{\mathsf{Sa}}(((\mm,(\mm,\ast)),s),A)
    &=
      [s = \ast]\, \mu_{\mathrm{Borel}}(\{a \in \mathbb{R}_{[0,1]} : ((\hh,(\hh,\ast)),a) \in A\}).
  \end{align*}
  The probabilistic Mealy machine behaves as follows:
  \begin{itemize}
  \item In the initial state $\ast$, given $\ast$ from the
    $\mathsf{J}$-edge, $\mathsf{Sa}$ draws a real number from the
    uniform distribution and stores the real number:
    
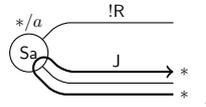
\begin{center}
      \begin{tikzpicture}[rounded corners]
        \node (sa) [draw,circle,obj] at (0,0) {$\mathsf{Sa}$};
        \draw (sa) -- ++ (0.4,0.4) -- node[above,obj] {$\oc \mathsf{R}$}++ (1.5,0);
        \draw (sa) -- ++ (0.4,-0.4) -- node[above=4,obj] {$\mathsf{J}$}++ (1.5,0);
        \node[obj] at (0,0.4) {$\ast/a$};
        \draw[thick,->] (1.9,-0.55) node[right,obj] {$\ast$}
        -- ++ (-1.55,0) -- ++(-0.35,0.35) -- ++ (0.2,0.2)
        -- ++ (0.25,-0.25) -- ++ (1.45,0) node[right,obj] {$\ast$};
      \end{tikzpicture}
      .
    \end{center}
    For example, the probability of the state being a real number in $[0,0.3]$
    after this transition is $0.3$.
  \item After this transition, $\mathsf{Sa}$ returns $(n,a \cons u)$
    to each ``query'' $(n,u)$:
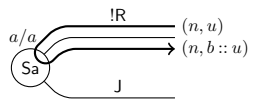
    \begin{center}
      \begin{tikzpicture}[rounded corners]
        \node (sa) [draw,circle,obj] at (0,0) {$\mathsf{Sa}$};
          \draw (sa) -- ++ (0.4,0.4) -- node[above=4,obj] {$\oc \mathsf{R}$}++ (1.5,0);
          \draw (sa) -- ++ (0.4,-0.4) -- node[above,obj] {$\mathsf{J}$}++ (1.5,0);
          \node[obj] at (-0.1,0.4) {$a/a$};
          \draw[thick,->] (1.9,0.55) node[right,obj] {$(n,u)$}
          -- ++ (-1.55,0) -- ++(-0.35,-0.35) -- ++ (0.2,-0.2)
          -- ++ (0.25,0.25) -- ++ (1.45,0) node[right,obj] {$(n,b\mathbin{::}u)$};
      \end{tikzpicture}
      .
    \end{center}
  \end{itemize}
}

\longv{
  \section{Probabilistic Mealy Machine Semantics for $\PCFSS$}
  \label{sec:pGoI2}

  We interpret a type $\mathtt{A}$
  as the $\mathbf{Int}$-object $\bsem{\mathtt{A}}$ given by
  \begin{equation*}
    \bsem{\ttunit} = \mathsf{I},
    \quad
    \bsem{\ttreal} = \mathsf{R},
    \quad
    \bsem{\mathtt{A}\to\mathtt{B}}
    = \mathsf{J}
    \otimes \oc \bsem{\mathtt{B}} \otimes \oc \bsem{\mathtt{A}}^{\bot}.
  \end{equation*}
  We define interpretation of contexts by
  \begin{equation*}
    \bsem{\mathtt{x}:\mathtt{A},\ldots,
      \mathtt{y}:\mathtt{B}} = \bsem{\mathtt{A}} \otimes
    \cdots \otimes \bsem{\mathtt{B}}.
  \end{equation*}
  When $\mathtt{\Delta}$ is the empty sequence, we define
  $\bsem{\mathtt{\Delta}}$ to be $\mathsf{I}$.

  We interpret terms $\mathtt{\Delta} \vdash \mathtt{M} : \mathtt{A}$
  and values $\mathtt{\Delta} \vdash \mathtt{V} : \mathtt{A}$ by
  \begin{equation*}
    \bsem{\mathtt{\Delta} \vdash \mathtt{M} : \mathtt{A}}
    \colon \oc \bsem{\mathtt{\Delta}} \rightarrowtriangle
    \mathsf{S} \otimes \oc \bsem{\mathtt{A}},
    \quad
    \Bsem{\mathtt{\Delta} \vdash \mathtt{V} : \mathtt{A}}
    \colon \oc \bsem{\mathtt{\Delta}} \rightarrowtriangle \bsem{\mathtt{A}}
  \end{equation*}
  inductively defined by diagrams in Figure~\ref{fig:intp} where
  Mealy machines are regarded as probabilistic Mealy machines
  in the obvious manner.
  \begin{figure*}[t]
    \begin{center}
      \fbox{\begin{minipage}{.98\textwidth}
          \centering
          \begin{tikzpicture}[rounded corners]
            \node [above] at (0.8,0.4)
            {$\scriptstyle\bsem{\mathtt{\Delta} \vdash \mathtt{V} : \mathtt{A}}$};
            \node (dg) [draw,circle,obj] at (0,0) {$\mathsf{dg}$};
            \node (v) [draw,circle,obj] at (1.3,0) {$\oc\Bsem{\mathtt{V}}$};
            \node (h) [draw,circle,obj] at (1.3,-0.65) {$\mathsf{j}$};
            \draw (dg) -- node[above,obj]{$\oc\bsem{\mathtt{\Delta}}$}++(-1,0);
            \draw (dg) -- node[above,obj]{$\oc\oc\bsem{\mathtt{\Delta}}$} (v);
            \draw (v) -- node[above,obj]{$\oc\bsem{\mathtt{A}}$}++(1,0);
            \draw (h) -- node[above,obj]{$\mathsf{J}$}++(1,0);
          \end{tikzpicture}
          \hspace{3pt}
          \begin{tikzpicture}
            \node (bot) [draw,circle,obj] at (1,-0.15) {$\mathsf{w}$};
            \node (d)[draw,circle,obj] at (1,0.4) {$\mathsf{d}$};
            \draw (bot) -- node[above,obj]{$\oc\bsem{\mathtt{\Delta}}$} ++(-1,0);
            \draw (d) -- node[above,obj]{$\oc\bsem{\mathtt{A}}$} ++(-1,0);
            \draw (d) -- node[above,obj]{$\bsem{\mathtt{A}}$} ++(1,0);
            \node[above=0.1cm] at (d.north) {$\scriptstyle\Bsem{\mathtt{\Delta}, \mathtt{x}
                : \mathtt{A} \vdash \mathtt{x} : \mathtt{A}}$};
          \end{tikzpicture}
          \hspace{3pt}
          \begin{tikzpicture}[rounded corners]
            \node (c) [draw,circle,obj] at (1,0) {$\mathsf{c}$};
            \node (v) [draw,circle,obj] at (2,-0.4) {$\Bsem{\mathtt{V}}$};
            \node (dg) [draw,circle,obj] at (1.8,0.4) {$\mathsf{dg}$};
            \node (w) [draw,circle,obj] at (3,0.4) {$\oc\Bsem{\mathtt{W}}$};
            \draw (dg) -- node[above,obj]{$\oc\oc\bsem{\mathtt{\Delta}}$}(w);
            \draw (c) to[out=45,in=180] node[above,obj]{$\oc\bsem{\mathtt{\Delta}}$}(dg);
            \draw (c) to[out=-45,in=180] node[below,obj]{$\oc\bsem{\mathtt{\Delta}}$}(v);
            \draw (c) -- node[above,obj]{$\oc\bsem{\mathtt{\Delta}}$}++(-0.8,0);
            \draw (v) -- ++(0.4,0.4) -- ++(1,0);
            \draw (3.4,0) arc(-90:90:0.2);
            \draw (3.4,0.4) -- node[above right,obj]{$\oc\bsem{\mathtt{A}}$}(w);
            \draw (v) -- ++(0.7,0) -- node[above,obj]{$\oc\bsem{\mathtt{B}}$}++(0.8,0);
            \draw (v) -- ++(0.4,-0.4) -- ++ (0.3,0) -- node[above,obj]{$\mathsf{J}$}++(0.8,0);
            \node at (2.4,1)
            {$\scriptstyle\bsem{\mathtt{\Delta}\vdash \mathtt{V}\,\mathtt{W} :\mathtt{B}}$};
          \end{tikzpicture}
          \hspace{3pt}
          \begin{tikzpicture}[rounded corners]
            \node (m) [draw,circle,obj] at (1,0) {$\bsem{\mathtt{M}}$};
            \draw (m) -- ++ (-0.4,-0.4) -- node[above,obj] {$\oc\bsem{\mathtt{\Delta}}$} ++ (-0.6,0);
            \draw (m) -- ++ (-0.4,0.4) -- ++ (-0.3,0);
            \draw (0.3,0.4) arc(270:90:0.2);
            \draw (0.3,0.8) -- node[below=-0.05cm,obj]{$\oc\bsem{\mathtt{A}}^{\bot}$}++(1.7,0);
            \draw (m) -- ++ (0.4,-0.4) -- node[above,obj] {$\mathsf{J}$} ++ (0.6,0);
            \draw (m) -- ++ (0.4,0.4) -- node[below,obj] {$\oc\bsem{\mathtt{B}}$} ++ (0.6,0);
            \node at (1,1.2)
            {$\scriptstyle\Bsem{\mathtt{\Delta} \vdash \lambda \mathtt{x}^{\mathtt{A}}.\,
                \mathtt{M}: \mathtt{A}\to\mathtt{B}}$};
          \end{tikzpicture}
          \hspace{3pt}
          \begin{tikzpicture}[rounded corners]
            \node (c) [draw,circle,obj] at (1,0) {$\mathsf{c}$};
            \node (m) [draw,circle,obj] at (2,0.4) {$\bsem{\mathtt{M}}$};
            \node (n) [draw,circle,obj] at (3,0.4) {$\bsem{\mathtt{N}}$};
            \node (mul) [draw,circle,obj] at (4,-0.4) {$\mathsf{n}$};
            \draw (c) -- node[above,obj]{$\oc\bsem{\mathtt{\Delta}}$} ++(-1,0);
            \draw (c) -- ++(0.4,0.4) node[above,obj]{$\oc\bsem{\mathtt{\Delta}}$} -- (m);
            \draw (m) to[out=30,in=150] node[above,obj]{$\oc\bsem{\mathtt{B}}$}(n);
            \draw (c) -- ++(0.4,-0.4) -- node[above,obj]{$\oc\bsem{\mathtt{\Delta}}$} ++ (1,0) -- (n);
            \draw (m) -- ++(1,-1) -- node[above,obj]{$\mathsf{J}$} ++(0.5,0) -- (mul);
            \draw (n) -- node[above,obj]{$\mathsf{J}$} (mul);
            \draw (n) -- ++ (0.4,0.4) -- node[below,obj]{$\oc\bsem{\mathtt{A}}$}++ (1.6,0);
            \draw (mul) -- node[above,obj]{$\mathsf{J}$} ++(1,0);
            \node [above=0.2cm] at (n.north)
            {$\scriptstyle\bsem{\mathtt{\Delta} \vdash \letin{\mathtt{x}}{\mathtt{M}}{\mathtt{N}} :
                \mathtt{A}}$};
          \end{tikzpicture}
          \\[5pt]
          \begin{tikzpicture}
            \node (w) [draw,circle,obj] at (1,0) {$\mathsf{w}$};
            \node at (0,-0.8) {};
            \draw (w) -- node[above,obj]{$\oc\bsem{\mathtt{\Delta}}$} ++(-1,0);
            \node [above] at (0.7,0.3)
            {$\scriptstyle\Bsem{\mathtt{\Delta} \vdash \mathtt{skip}:\ttunit}$};
          \end{tikzpicture}
          \hspace{2pt}
          \begin{tikzpicture}
            \begin{scope}[yshift=-1.5cm,xshift=-0.4cm]
              \node (w) [draw,circle,obj] at (1,0) {$\mathsf{w}$};
              \node (a) [draw,circle,obj] at (1.65,0) {$\mathsf{r}_{a}$};
              \draw (w) -- node[above,obj]{$\oc\bsem{\mathtt{\Delta}}$} ++(-1,0);
              \draw (a) -- node[above,obj]{$\mathsf{R}$}++(1,0);
              \node [above] at (1.375,0.3)
              {$\scriptstyle\Bsem{\mathtt{\Delta} \vdash \mathtt{r}_{a}:\ttreal}$};
            \end{scope}
            \node (v)[draw,circle,obj] at (1,0) {$\oc\Bsem{\mathtt{V}}$};
            \node (f)[draw,circle,obj] at (2,0) {$\oc\mathsf{fn}_{f}$};
            \node (dg)[draw,circle,obj] at (-0.5,0) {$\mathsf{dg}$};
            \node (e)[draw,circle,obj] at (2,-0.6) {$\mathsf{j}$};
            \draw (e)--node[above,obj]{$\mathsf{J}$}++(0.8,0);
            
            \draw (v) -- node[above,obj]{$\oc\oc\bsem{\mathtt{\Delta}}$} (dg);
            \draw (dg) -- node[above,obj]{$\oc\bsem{\mathtt{\Delta}}$} ++(-0.8,0);
            \draw (v) -- node[above,obj]{$\oc\mathsf{R}$} (f);
            \draw (f) -- node[above,obj]{$\oc\mathsf{R}$} ++(0.8,0);
            \node [above] at (v.north)
            {$\scriptstyle\bsem{\mathtt{\Delta} \vdash \mathtt{F}(\mathtt{V}):\ttreal}$};
          \end{tikzpicture}
          \hspace{2pt}
          \begin{tikzpicture}[rounded corners]
            \begin{scope}[yshift=1.25cm,xshift=-0.75cm]
              \node (v) [draw,circle,obj] at (0,0) {$\Bsem{\mathtt{V}}$};
              \node (sc) [draw,circle,obj] at (1,0) {$\mathsf{Sc}$};
              \draw (v) --node[above,obj]{$\oc\bsem{\mathtt{\Delta}}$} ++(-1,0);
              \draw (v) --node[above,obj]{$\mathsf{R}$} (sc);
              \draw (sc) --node[above,obj]{$\mathsf{J}$} ++(1,0);
              \node [above] at (0.5,0.4)
              {$\scriptstyle\bsem{\mathtt{\Delta} \vdash \mathtt{score}(\mathtt{V}) :\ttunit}$};
            \end{scope}
            \node (sa) [draw,circle,obj] at (0,0) {$\mathsf{Sa}$};
            \node (w) [draw,circle,obj] at (-0.6,0) {$\mathsf{w}$};
            \draw (w) -- node[above,obj]{$\oc\bsem{\mathtt{\Delta}}$}++(-1,0);
            \draw (sa) -- ++(0.4,0.4) -- node[below,obj]{$\oc\mathsf{R}$}++(0.8,0);
            \draw (sa) -- ++(0.4,-0.4) -- node[above,obj]{$\mathsf{J}$}++(0.8,0);
            \node [above] at (0,0.4) {$\scriptstyle\bsem{\mathtt{\Delta} \vdash \mathtt{sample} : \ttreal}$};
          \end{tikzpicture}
          \hspace{2pt}
          \begin{tikzpicture}[rounded corners]
            \node (cd) [draw,circle,obj] at (0,0) {$\mathsf{c}$};
            \node (dd) [draw,circle,obj] at (1,0.4) {$\mathsf{dg}$};
            \node (dab) [draw,circle,obj] at (1,1.2) {$\mathsf{dg}$};
            \node (m) [draw,circle,obj] at (2,0.8) {$\oc\mathsf{M}$};
            \node (cab) [draw,circle,obj] at (3,0.8) {$\mathsf{c}$};
            \node (n) [draw,circle,obj] at (4,0) {$\mathsf{M}$};
            \draw (cd) -- ++ (1,-0.4) -- ++ (1,0) -- node[above,obj]{$\oc\bsem{\mathtt{\Delta}}$}
            ++ (1,0) -- (n);
            \draw (cd) -- node[above,obj]{$\oc\bsem{\mathtt{\Delta}}$} ++(-1,0);
            \draw (cd) to[out=35,in=180] node[above,obj]{$\oc\bsem{\mathtt{\Delta}}$} (dd);
            \draw (dd) -- ++(0.5,0) node[below, obj]{$\oc\oc\bsem{\mathtt{\Delta}}$} -- (m);
            \draw (dab) -- ++(0.5,0) node[above,obj]{$\oc\oc\bsem{\mathtt{C}}$} -- (m);
            \draw (m) -- node[below,obj]{$\oc\bsem{\mathtt{C}}$} (cab);
            \draw (cab) -- ++(0.35,-0.4) -- node[below,obj]{$\oc\bsem{\mathtt{C}}$} ++ (0.3,0) -- (n);
            \draw (cab) -- ++(0.4,0.4) -- node[below,obj]{$\oc\bsem{\mathtt{C}}$} ++ (0.3,0);
            \draw (3.7,1.2) arc(-90:90:0.2);
            \draw (3.7,1.6) -- ++ (-3.3,0);
            \draw (0.4,1.6) arc(90:270:0.2);
            \draw (0.4,1.2) -- node[below,obj]{$\oc\bsem{\mathtt{C}}$}(dab);
            \draw (n) -- node[above,obj]{$\bsem{\mathtt{C}}$}++(1.2,0);
            \node [above] at (2,1.6) {$\scriptstyle\Bsem{\mathtt{\Delta} \vdash
                \mathtt{fix}_{\mathtt{A},\mathtt{B}}(\mathtt{f},\mathtt{x},\mathtt{M})
                : \mathtt{A} \to \mathtt{B}}$};
            \node[right] at(4,1.25)
            {$\begin{array}{l}
                \scriptstyle
                \mathtt{C}=\mathtt{A}\to\mathtt{B}, \\
                \scriptstyle
                \mathsf{M}=\bsem{\lambda \mathtt{x}^{\mathtt{A}}.\,\mathtt{M}} \\
                \scriptstyle
                \; \colon \oc\bsem{\mathtt{\Delta}} \otimes
                \oc\bsem{\mathtt{C}} \rightarrowtriangle \bsem{\mathtt{C}}
              \end{array}$};
          \end{tikzpicture}
          \\[10pt]
          \begin{tikzpicture}[rounded corners]
            \node (M) [draw,circle,obj] at (0,0) {$\bsem{\mathtt{M}}$};
            \node (N) [draw,circle,obj] at (0,1) {$\bsem{\mathtt{N}}$};
            \node (V) [draw,circle,obj] at (0,-1) {$\Bsem{\mathtt{V}}$};
            \node (cd) [draw,circle,obj] at (2,0) {$\mathsf{cd}$};
            \node (c) [draw,circle,obj] at (-1.5,0) {$\mathsf{c}$};

            \draw (c) -- node[above,obj]{$\oc\bsem{\mathtt{\Delta}}$} ++ (-1.5,0);
            \draw (c) -- node[above,obj]{$\oc\bsem{\mathtt{\Delta}}$} (M);
            \draw (c) -- ++ (0.5,1) -- node[above,obj]{$\oc\bsem{\mathtt{\Delta}}$} (N);
            \draw (c) -- ++ (0.5,-1) -- node[above,obj]{$\oc\bsem{\mathtt{\Delta}}$} (V);

            \draw (cd) -- node[above,obj]{$\mathsf{J} \otimes \oc\bsem{\mathtt{A}}$} ++ (1.5,0);
            \draw (cd) -- node[above,obj]{$\mathsf{J} \otimes \oc\bsem{\mathtt{A}}$} (M);
            \draw (cd) -- ++ (-0.5,1) -- node[above,obj]{$\mathsf{J} \otimes \oc\bsem{\mathtt{A}}$} (N);
            \draw (cd) -- ++ (-0.5,-1) -- node[above,obj]{$\mathsf{R}$} (V);
            
            \node[above] at (N.north) {$\scriptstyle
              \bsem{\mathtt{\Delta}\vdash\ifterm{\mathtt{V}}{\mathtt{M}}{\mathtt{N}}:\mathtt{A}}$};
          \end{tikzpicture}
        \end{minipage}}
    \end{center}
    \caption{Interpretation of Terms and Values}
    \label{fig:intp}
  \end{figure*}
}

\longv{
  \subsection{Soundness and Adequacy}

  \subsubsection{Observation}

  Let $\mathsf{M} \colon \mathsf{I} \rightarrowtriangle \mathsf{J} \otimes \oc \mathsf{R}$ be
  a probabilistic Mealy machine. We define s-finite kernels $t_{0} \colon 1
  \leadsto \state{\mathsf{M}}$ and $t_{1} \colon \state{\mathsf{M}} \leadsto
  \mathbb{R}$ by
  \begin{align*}
    t_{0}^{\mathsf{M}}(\ast,A) &= \tran{\mathsf{M}}(((\mm,(\mm,\ast)),\init{\mathsf{M}}),
                    \{((\hh,(\hh,\ast)),s) : s \in A\}), \\
    t_{1}^{\mathsf{M}}(s,A) &= \tran{\mathsf{M}}
                 (((\mm,(\hh,(0,\varepsilon))),s),
                 \{(\hh,(\mm,(0,a \cons \varepsilon))) : a \in A\}).
  \end{align*}
  Then we define a measure $\mathsf{obs}(\mathsf{M})$ on $\mathbb{R}$ to be
  $t_{1}^{\mathsf{M}} \circ t_{0}^{\mathsf{M}}(\ast,-)$.
  Intuitively, $\mathsf{obs}(\mathsf{M})$ is a measure that describes
  distribution of real numbers obtained by the following process:
  \begin{itemize}
  \item We first input $\ast$ to the $\mathsf{J}$-wire of $\mathsf{M}$.
  \item If $\mathsf{M}$ outputs $\ast$ to the $\mathsf{J}$-wire, then we input
    $(0,\varepsilon)$ to the $\oc \mathsf{R}$-wire of $\mathsf{M}$.
  \item We only observe outputs of the form $(0,a \cons \varepsilon)$ for some
    $a \in \mathbb{R}$.
  \end{itemize}
  For example, $\mathsf{obs}(\mathsf{sa})$ is the uniform distribution over $\mathbb{R}_{[0,1]}$.

  \begin{theorem}[Soundness and Adequacy]\label{thm:padq}
    For any closed term $\vdash \mathtt{M} : \ttreal$, if
    $\mathtt{M} \Rightarrow_{\infty} \mu$, then
    $\mathsf{obs}(\sem{\mathtt{M}}) = \mu$.
  \end{theorem}

  Below, we give a proof of Theorem~\ref{thm:padq}.

  \subsubsection{Proof of Adequacy Theorem}
  \label{sec:lr}

  \begin{lemma}
    For any term $\mathtt{\Delta},\mathtt{x}:\mathtt{A} \vdash \mathtt{M}:\mathtt{B}$
    and for any closed value $\vdash \mathtt{V}:\mathtt{A}$,
    \begin{equation*}
      \bsem{\mathtt{M}} \circ (\bsem{\mathtt{\Delta}} \otimes \oc \Bsem{\mathtt{V}})
      \simeq \bsem{\mathtt{M}\{\mathtt{V}/\mathtt{x}\}}.
    \end{equation*}  
  \end{lemma}
  \begin{proof}
    By induction on $\mathtt{M}$.
  \end{proof}
  \begin{lemma}\label{lem:detredp}
    For all closed terms $\mathtt{M},\mathtt{N}:\mathtt{A}$,
    if $\mathtt{M} \detred \mathtt{N}$, then
    $\sem{\mathtt{M}} = \sem{\mathtt{N}}$.
  \end{lemma}
  \begin{proof}
    By case analysis. For the case of recursion, see
    Corollary~\ref{cor:fixp}.
  \end{proof}
  We first prove soundness.
  \begin{proposition}\label{prop:soundp}
    For any closed term $\mathtt{M}:\ttreal$,
    if $\mathtt{M} \Rightarrow_{n} \mu$, then
    $\mu \leq \mathsf{obs}(\bsem{\mathtt{M}})$.
  \end{proposition}
  \begin{proof}
    By induction on $n$.
    (Base case) Easy. (Induction step) By case analysis.
    \begin{itemize}
    \item If $\mathtt{M} = \mathtt{r}_{a}$, then $\mathsf{obs}{\bsem{\mathtt{M}}}
      = \delta_{a}$.
    \item If $\mathtt{M} =\mathtt{E}[\mathtt{N}]$ and $\mathtt{N} \detred \mathtt{L}$,
      then $\mathsf{obs}(\bsem{\mathtt{M}}) = \mathsf{obs}(\bsem{\mathtt{E}[\mathtt{L}]})
      \geq \mu$.
    \item If $\mathtt{M} = \mathtt{E}[\mathtt{score}(\mathtt{r}_{a})]$
      and $\mathtt{E}[\mathtt{skip}] \Rightarrow_{n-1} \mu$, then by
      the definition of $\mathsf{obs}(-)$, we see that
      $\mathsf{obs}(\bsem{\mathtt{M}}) = |a| \,
      \mathsf{obs}(\bsem{\mathtt{E}[\mathtt{skip}]}) \geq |a|\, \mu$.
    \item If $\mathtt{M} = \mathtt{E}[\mathtt{sample}]$ and
      $\mathtt{E}[\mathtt{r}_{a}] \Rightarrow_{n-1} k(a,-)$ for some
      finite kernel $k$, then by the definition of $t_{0}$ and $t_{1}$,
      we see that
      \begin{align*}
        t_{0}^{\bsem{\mathtt{M}}}
        &=
          \xymatrix{
          1 \ar@{~>}[rr]^-{\mu_{\mathrm{Borel}}}
        &&
          \mathbb{R} \ar@{~>}[r]^-{h}
        &
          \state{\bsem{\mathtt{M}}} \times \mathbb{R} \ar@{~>}[rr]^-{\mathrm{injection}}
        &&
          \state{\bsem{\mathtt{M}}} \times (1 + \mathbb{R})
          },
        \\
        t_{1}^{\bsem{\mathtt{M}}}((s,a),-)
        &=
          t_{1}^{\bsem{\mathtt{E}[\mathtt{r}_{a}]}}(s,-), \\
        t_{1}^{\bsem{\mathtt{M}}}((s,0),-)
        &= 0
      \end{align*}
      for some $h$ such that
      $h(a,-) = t_{0}^{\bsem{\mathtt{E}[\mathtt{r}_{a}]}}
      \times \delta_{a}$. Hence,
      \begin{equation*}
        (t_{1}^{\bsem{\mathtt{M}}} \circ 
        t_{0}^{\bsem{\mathtt{M}}})(\ast,A)
        = \int_{\mathbb{R}_{[0,1]}}
        (t_{1}^{\bsem{\mathtt{E}[\mathtt{r}_{a}]}} \circ 
        t_{0}^{\bsem{\mathtt{E}[\mathtt{r}_{a}]}})(\ast,A) \; \mathrm{d}a
        \geq \int_{\mathbb{R}_{[0,1]}}
        k(a,A) \; \mathrm{d}a.
      \end{equation*}
    \end{itemize}
  \end{proof}

  It remains to prove that $\mathtt{M} \Rightarrow_{\infty} \mu$
  implies $\mathsf{obs}(\bsem{\mathtt{M}}) \leq \mu$. We use logical
  relations. We define a binary relation $O_{\mathrm{d}}$ between
  closed terms of type $\ttreal$ and probabilistic Mealy machines from $\mathsf{I}$
  to $\mathsf{J} \otimes \oc \mathsf{R}$ by
  \begin{equation*}
    (\mathtt{M},\mathsf{M}) \in O_{\mathrm{d}}
    \iff
    \textnormal{if } \mathtt{M}\Rightarrow_{\infty} \mu
    \textnormal{ then } \mathsf{obs}(\mathsf{M}) \leq \mu.
  \end{equation*}
  We then inductively define binary relations
  \begin{align*}
    S_{\mathtt{A}}
    &\subseteq \{\textnormal{closed values of type }\mathtt{A}\}
      \times \{\textnormal{Mealy machines from } \mathsf{I} \textnormal{ to } \sem{\mathtt{A}}\} \\
    S_{\mathtt{A}}^{\top}
    &\subseteq
      \{\textnormal{evaluation contexts } \mathtt{x}:\mathtt{A}
      \vdash \mathtt{E}[\mathtt{x}]:\ttreal \}
      \times \{\textnormal{Mealy machines from } \oc\sem{\mathtt{A}}
      \textnormal{ to } \mathsf{J} \otimes
      \oc \mathsf{R}\} \\
    \overline{S}_{\mathtt{A}}
    &\subseteq \{\textnormal{closed terms of type }\mathtt{A}\}
      \times \{\textnormal{Mealy machines from } \mathsf{I} \textnormal{ to }
      \mathsf{J} \otimes \oc \sem{\mathtt{A}}\}
  \end{align*}
  by
  \begin{align*}
    S_{\ttreal}
    &= \{(\mathtt{r}_{a},\mathsf{r}_{a}) : a \in \mathbb{R}\}, \\
    S_{\ttunit}
    &= \{(\mathtt{skip},\mathsf{id}_{\mathsf{I}})\}, \\
    S_{\mathtt{A}\to\mathtt{B}}
    &= \{(\mathtt{V},\mathsf{M}) :
      \forall (\mathtt{W},\mathtt{N}) \in S_{\mathtt{A}},\,
      (\mathtt{V}\,\mathtt{W},
      (\mathsf{J} \otimes \oc \bsem{\mathtt{B}} \otimes \mathsf{counit}_{\oc\bsem{\mathtt{A}}})
      \circ (\mathsf{M} \otimes \oc \mathsf{N})) \in \overline{S}_{\mathtt{B}}
      \}, \\
    S_{\mathtt{A}}^{\top}
    &= \{(\mathtt{E}[-],\mathsf{E}) :
      \forall (\mathtt{V},\mathsf{M}) \in S_{\mathtt{A}},\,
      (\mathtt{E}[\mathtt{V}],\mathsf{E} \circ \oc \mathsf{M}) \in O_{\mathrm{d}}\}, \\
    \overline{S}_{\mathtt{A}}
    &= \{(\mathtt{M},\mathsf{M}):
      \forall (\mathtt{E}[-],\mathtt{E}) \in S_{\mathtt{A}}^{\top},\,
      (\mathtt{E}[\mathtt{M}],(\mathsf{n} \otimes \oc \mathsf{R})
      \circ (\mathsf{J} \otimes \mathsf{E}) \circ \mathsf{M}) \in O_{\mathrm{d}}\}
  \end{align*}

  We list some properties of the logical relations.
  \begin{lemma}\label{lem:lrp}
    Let $\mathtt{A}$ be a type.
    \begin{enumerate}
    \item If $(\mathtt{V},\mathtt{M}) \in S_{\mathtt{A}}$, then
      $(\mathtt{V},\mathsf{j} \otimes \oc \mathsf{M}) \in
      \overline{S}_{\mathtt{A}}$.
    \item If $(\mathtt{M},\mathsf{M}) \in \overline{S}_{\mathtt{A}}$ and
      $\mathtt{N} \detred \mathtt{M}$, then
      $(\mathtt{N},\mathsf{M}) \in \overline{S}_{\mathtt{A}}$.
    \item If $(\mathtt{M},\mathsf{M}) \in \overline{S}_{\mathtt{A}}$ and
      $\mathtt{M} \detred \mathtt{N}$, then
      $(\mathtt{N},\mathsf{M}) \in \overline{S}_{\mathtt{A}}$.
    \item If $(\mathtt{M},\mathsf{M}) \in \overline{S}_{\mathtt{A}}$ and
      $\mathsf{M} \simeq \mathsf{N}$, then
      $(\mathtt{M},\mathsf{N}) \in \overline{S}_{\mathtt{A}}$.
    \item For any closed term $\mathtt{M}:\mathtt{A}$,
      $(\mathtt{M},\mathsf{bot}_{\mathsf{J} \otimes \oc
        \sem{\mathtt{A}}}) \in \overline{S}_{\mathtt{A}}$ where
      $\mathsf{bot}_{\mathsf{X}} \colon \mathsf{I} \rightarrowtriangle \mathsf{X}$
      is a token machine whose transition function is the zero kernel.
    \item For any closed value $\mathtt{V}:\mathtt{A}\to\mathtt{B}$,
      $(\mathtt{V},\mathsf{bot}_{\sem{\mathtt{A}\to\mathtt{B}}}) \in
      S_{\mathtt{A}\to\mathtt{B}}$.
    \item If $(\mathtt{M},\mathsf{M}_{i}) \in \overline{S}_{\mathtt{A}}$
      and $\state{\mathsf{M}_{1}} = \state{\mathsf{M}_{2}} = \cdots$ and
      $\init{\mathsf{M}_{1}} = \init{\mathsf{M}_{2}} = \cdots$ and
      $\tran{\mathsf{M}_{1}} \leq \tran{\mathsf{M}_{2}} \leq \cdots$, then
      $(\mathtt{M},\mathsf{N}) \in \overline{S}_{\mathtt{A}}$ where
      $\mathsf{N}$ is given by
      $\state{\mathsf{N}} = \state{\mathsf{M}_{1}}$,
      $\init{\mathsf{N}} = \init{\mathsf{M}_{1}}$,
      $\tran{\mathsf{N}} = \bigvee_{n} \tran{\mathsf{M}_{n}}$.
    \end{enumerate}
  \end{lemma}
  \begin{proof}
    We can check these items by unfolding the definition of
    $O_{\mathrm{d}}$ and the logical relations.
  \end{proof}

  \begin{lemma}[Basic Lemma]\label{lem:basicp}
    Let
    $\mathtt{\Delta} = (\mathtt{x}:\mathtt{A}_{1},
    \ldots,\mathtt{x}_{n}:\mathtt{A}_{n})$ be a context.
    \begin{itemize}
    \item For any term $\mathtt{\Delta} \vdash \mathtt{M}:\mathtt{A}$
      and for any
      $(\mathtt{V}_{i},\mathtt{N}_{i}) \in S_{\mathtt{A}_{i}}$ for
      $i = 1,2,\ldots,n$, we have
      \begin{equation*}
        \left(
          \mathtt{M}\{\mathtt{V}_{1}/\mathtt{x}_{1},\ldots,
          \mathtt{V}_{n}/\mathtt{x}_{n}\},
          \bsem{\mathtt{M}} \circ
          (\oc\mathsf{N}_{1} \otimes \cdots \otimes \oc \mathsf{N}_{n})
        \right)
        \in \overline{S}_{\mathtt{A}}.
      \end{equation*}
    \item  For any value
      $\mathtt{\Delta} \vdash \mathtt{V}:\mathtt{A}$ and for any
      $(\mathtt{V}_{i},\mathtt{N}_{i}) \in S_{\mathtt{A}_{i}}$ for
      $i = 1,2,\ldots,n$, we have
      \begin{equation*}
        \left(
          \mathtt{V}\{\mathtt{V}_{1}/\mathtt{x}_{1},\ldots,
          \mathtt{V}_{n}/\mathtt{x}_{n}\},
          \Bsem{\mathtt{M}} \circ
          (\oc\mathsf{N}_{1} \otimes \cdots \otimes \oc \mathsf{N}_{n})
        \right)
        \in S_{\mathtt{A}}.
      \end{equation*}
    \end{itemize}
  \end{lemma}
  \begin{proof}
    By induction on $\mathtt{M}$ and $\mathtt{V}$. Most cases follow
    from Lemma~\ref{lem:lrp}. For $\mathtt{M}=\mathtt{sample}$ and
    $\mathtt{M}=\mathtt{score}(\mathtt{V})$, we check the statement by
    unfolding the definition of $\mathsf{Sa}$ and $\mathsf{sc}$. Here,
    we only check for
    $\mathtt{M}=\mathtt{sample}$ and
    $\mathtt{M}=\mathtt{fix}_{\mathtt{A},\mathtt{B}}(\mathtt{f},
    \mathtt{x},\mathtt{N})$.
    \begin{itemize}
    \item When $\mathtt{M}=\mathtt{sample}$, let
      $(\mathtt{E},\mathsf{E})$ be a pair in $R_{\ttreal}^{\top}$. We
      write $\mathsf{N}$ for
      $(n \otimes \oc \mathsf{R}) \circ (\mathsf{J} \otimes \mathsf{E})
      \circ \mathsf{Sa}$. By the definition of $t_{0}$ and $t_{1}$, we
      have
      \begin{align*}
        &t_{0}^{\mathsf{N}}
        =
          \xymatrix{
          1 \ar@{~>}[rr]^-{\mu_{\mathrm{Borel}}}
        &&
          \mathbb{R} \ar@{~>}[r]^-{h}
        &
          \state{\mathsf{N}} \times \mathbb{R} \ar@{~>}[rr]^-{\mathrm{injection}}
        &&
          \state{\mathsf{N}} \times (1 + \mathbb{R})
          },
        \\
        &t_{1}^{\mathsf{N}}((s,a),-)
        =
          t_{1}^{\mathsf{E} \circ \mathsf{r}_{a}}(s,-), \\
        &t_{1}^{\mathsf{N}}((s,0),-)
        = 0
      \end{align*}
      for some $h$ such that
      $h(a,-) = t_{0}^{\mathsf{E} \circ \mathsf{r}_{a}}
      \times \delta_{a}$. Hence,
      \begin{equation*}
        (t_{1}^{\mathsf{N}} \circ 
        t_{0}^{\mathsf{N}})(\ast,A)
        = \int_{\mathbb{R}_{[0,1]}}
        (t_{1}^{\mathsf{E} \circ \mathsf{r}_{a}} \circ 
        t_{0}^{\mathsf{E} \circ \mathsf{r}_{a}})(\ast,A) \; \mathrm{d}a
        = \int_{\mathbb{R}_{[0,1]}}
        k(a,A) \; \mathrm{d}a.
      \end{equation*}
    \item When
      $\mathtt{M}=\mathtt{fix}_{\mathtt{A},\mathtt{B}}(\mathtt{f},\mathtt{x},\mathtt{N})$,
      for simplicity, we suppose that $\mathtt{M}$ is a closed term. By
      induction hypothesis, we can check that
      \begin{equation*}
        (\mathtt{M},
        \Bsem{\lambda \mathtt{x}^{\mathtt{A}}.\,\mathtt{N}}
        \circ \oc\Bsem{\lambda \mathtt{x}^{\mathtt{A}}.\,\mathtt{N}}
        \circ \cdots
        \circ \oc^{k}\Bsem{\lambda \mathtt{x}^{\mathtt{A}}.\,\mathtt{N}}
        \circ \mathsf{bot}_{\mathsf{I},\oc^{k}\bsem{\mathtt{A}\to\mathtt{B}}})
        \in R_{\mathtt{A}\to\mathtt{B}}
      \end{equation*}
      by induction on $n$. By 
      Lemma~\ref{lem:lrp} and Proposition~\ref{prop:iterp}, 
      we obtain
      $(\mathtt{M},\sem{\mathtt{M}}) \in R_{\mathtt{A}\to\mathtt{B}}$.
    \end{itemize}
  \end{proof}

  \begin{theorem}
    For any closed term $\vdash \mathtt{M} : \ttreal$, if
    $\mathtt{M} \Rightarrow_{\infty} \mu$, then
    \begin{equation*}
      \mathsf{obs}(\bsem{\mathtt{M}}) \leq \mu.
    \end{equation*}
  \end{theorem}
  \begin{proof}
    By soundness, we have
    $\mu \leq \mathsf{obs}(\bsem{\mathtt{M}})$.
    On the other hand, because $([-],\mathsf{j} \otimes \mathsf{id}_{\oc \mathsf{R}})$
    is an element of $S_{\ttreal}^{\top}$,
    we obtain the other inequality
    by Lemma~\ref{lem:basicp}.
  \end{proof}

  \subsubsection{Induction step on recursion}
  \label{sec:expandp}

  \paragraph{Our Goal: Approximation Lemma}

  Let $\mathsf{M} \colon \oc \mathsf{X} \rightarrowtriangle \mathsf{X}$ be a
  Mealy machine. In this section, we show that
  a Mealy machine
  $\mathsf{M}^{\dagger} \colon \mathsf{I} \to \oc\mathsf{X}$ given by
  \begin{equation*}
    \mathsf{M}^{\dagger} =
    (\mathsf{counit}_{\oc\mathsf{X}} \otimes \mathsf{M})
    \circ ((\mathsf{con}_{\mathsf{X}} \circ \oc\mathsf{M} \circ \mathsf{dg}_{\mathsf{X}})
    \otimes \mathsf{id}_{\oc\mathsf{X}})
    \circ \mathsf{unit}_{\oc\mathsf{X}}.
  \end{equation*}
  is a ``least'' fixed point of $\mathsf{M}$. Diagrammatically,
  $\mathsf{M}^{\dagger}$ consists of digging, contraction and a feed
  back loop:
  
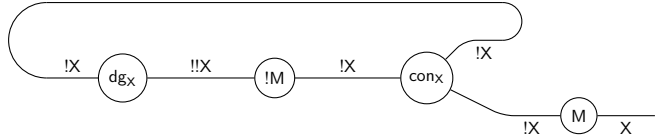
\begin{center}
    \begin{tikzpicture}[rounded corners]
      \node [draw,circle,obj] (dg) at (0,0) {$\mathsf{dg}_{\mathsf{X}}$};
      \node [draw,circle,obj] (ocM) at (2,0) {$\oc\mathsf{M}$};
      \node [draw,circle,obj] (c) at (4,0) {$\mathsf{con}_{\mathsf{X}}$};
      \node [draw,circle,obj] (M) at (6,-0.5) {$\mathsf{M}$};

      \draw (dg) -- node[above,obj] {$\oc\oc \mathsf{X}$} (ocM);
      \draw (ocM) -- node[above,obj] {$\oc \mathsf{X}$} (c);
      \draw (c) -- ++ (1,-0.5) -- node[below,obj] {$\oc \mathsf{X}$} (M);
      \draw (M) -- node[below,obj] {$\mathsf{X}$} ++ (1,0);
      \draw (dg) -- node[above,obj] {$\oc \mathsf{X}$} ++ (-1,0);
      \draw (-1,0) arc(270:90:0.5);
      \draw (-1,1) -- ++ (6,0);
      \draw (c) -- ++ (0.5,0.5) -- node[below,obj] {$\oc \mathsf{X}$} ++ (0.5,0);
      \draw (5,0.5) arc(-90:90:0.25);
    \end{tikzpicture}
  \end{center}
  This construction already appeared in the interpretation of the fixed point
  operator. In fact, for a term
  $\mathtt{f}:\mathtt{A} \to \mathtt{B}, \mathtt{x}:\mathtt{A} \vdash
  \mathtt{M}:\mathtt{B}$, we have
  $\sem{\mathtt{fix}_{\mathtt{A},\mathtt{B}}(\mathtt{f},\mathtt{x},\mathtt{M})}
  = \sem{\lambda \mathtt{x}^{\mathtt{A}}.\,\mathtt{M}}^{\dagger}$.

  \paragraph{Parametrized Modal Operator and Parametrized Loop Operator}

  We introduce parametrization of the modal operator $\oc$ and the
  loop operator $(-)^{\dagger}$. For $\alpha \subseteq \mathbb{N}$,
  we define $\oc_{\alpha} \mathsf{M}$ by:
  the state space and the initial state of $\oc \mathsf{M}$ are given
  by
  \begin{equation*}
    \state{\oc_{\alpha}\mathsf{M}}
    = \state{\oc\mathsf{M}} = \state{\mathsf{M}}^{\mathbb{N}},
    \qquad
    \init{\oc_{\alpha}\mathsf{M}} = \init{\oc\mathsf{M}}
  \end{equation*}
  and $\tran{\oc_{\alpha}\mathsf{M}}$ is a unique s-finite kernel such
  that the following diagrams commute:
  \begin{itemize}
  \item for any $n \in \alpha$,
    \begin{equation*}
      \xymatrix@C=30mm{
        (X^{+} + Y^{-})
        \times \state{\mathsf{M}} \times \state{\oc_{\alpha}\mathsf{M}}
        \ar@{~>}[d]_{\tran{\mathsf{M}} \otimes \state{\oc_{\alpha}\mathsf{M}}}
        \ar@{~>}[r]^-{(\mathrm{inj}_{n} \oplus \mathrm{inj}_{n}) \otimes \mathrm{ins}_{n}}
        &
        ((\mathbb{N} \times X^{+}) + (\mathbb{N} \times Y^{-}))
        \times \state{\oc_{\alpha}\mathsf{M}}
        \ar@{~>}[d]^{\tran{\oc_{\alpha}\mathsf{M}}} \\
        (Y^{+} + X^{-}) \times \state{\mathsf{M}} \times \state{\oc_{\alpha}\mathsf{M}}
        \ar@{~>}[r]^-{(\mathrm{inj}_{n} \oplus \mathrm{inj}_{n}) \otimes \mathrm{ins}_{n}}
        &
        ((\mathbb{N} \times Y^{+}) + (\mathbb{N} \times X^{-}))
        \times \state{\oc_{\alpha}\mathsf{M}}
        \nullcomma
      }
    \end{equation*}
  \item for any $n \not\in \alpha$,
    \begin{equation*}
      \xymatrix@C=30mm{
        (X^{+} + Y^{-})
        \times \state{\mathsf{M}} \times \state{\oc_{\alpha}\mathsf{M}}
        \ar@{~>}[d]_{\tran{\mathsf{M}} \otimes \state{\oc_{\alpha}\mathsf{M}}}
        \ar@{~>}[r]^-{(\mathrm{inj}_{n} \oplus \mathrm{inj}_{n}) \otimes \mathrm{ins}_{n}}
        &
        ((\mathbb{N} \times X^{+}) + (\mathbb{N} \times Y^{-}))
        \times \state{\oc_{\alpha}\mathsf{M}}
        \ar@{~>}[d]^{\emptyset} \\
        (Y^{+} + X^{-}) \times \state{\mathsf{M}} \times \state{\oc_{\alpha}\mathsf{M}}
        \ar@{~>}[r]^-{(\mathrm{inj}_{n} \oplus \mathrm{inj}_{n}) \otimes \mathrm{ins}_{n}}
        &
        ((\mathbb{N} \times Y^{+}) + (\mathbb{N} \times X^{-}))
        \times \state{\oc_{\alpha}\mathsf{M}}
        \nullcomma
      }
    \end{equation*}
  \end{itemize}
  Let $\alpha_{n},\beta_{n} \subseteq \mathbb{N}$ be
  \begin{equation*}
    \alpha_{0} = \emptyset,
    \qquad
    \beta_{n} = \{\langle i,j \rangle : i \in \alpha_{n} \textnormal{ and } j \in \mathbb{N}\},
    \qquad
    \alpha_{n+1}
    = \{2i : i \in \mathbb{N}\} \cup
    \{2i+1 : i \in \beta_{n}\}.
  \end{equation*}
  The definition of $\alpha_{n}$ and $\beta_{n}$ are motivated by the following lemma.
  \begin{lemma}\label{lem:alphabetap}
    For any $n \in \mathbb{N}$ and for any
    $\mathsf{M} \colon \mathsf{X} \rightarrowtriangle \mathsf{Y}$, we have
    \begin{equation*}
      \mathsf{con}_{\mathsf{Y}} \circ \oc_{\alpha_{n+1}} \mathsf{M}
      \simeq
      (\oc \mathsf{M} \otimes \oc_{\beta_{n}} \mathsf{M}) \circ \mathsf{con}_{\mathsf{X}},
      \qquad
      \mathsf{dg}_{\mathsf{Y}} \circ \oc_{\beta_{n+1}} \mathsf{M}
      \simeq
      \oc_{\alpha_{n}} \oc \mathsf{M} \circ \mathsf{dg}_{\mathsf{X}}.
    \end{equation*}
  \end{lemma}

  By means of $\oc_{\alpha}$, we also parametrize the operator
  $(-)^{\dagger}$. For $\alpha \subseteq \mathbb{N}$, and for
  $\mathsf{M} \colon \oc\mathsf{X} \rightarrowtriangle \mathsf{X}$, we define
  $\mathsf{M}^{\dagger,\alpha} \colon \mathsf{I} \to \oc\mathsf{X}$ by
  \begin{equation*}
    \mathsf{M}^{\dagger,\alpha} =
    (\mathsf{counit}_{\oc\mathsf{X}} \otimes \mathsf{M})
    \circ ((\mathsf{con}_{\mathsf{X}} \circ \oc_{\alpha}\mathsf{M}) \circ \mathsf{dg}_{\mathsf{X}})
    \otimes \mathsf{id}_{\oc\mathsf{X}}
    \circ \mathsf{unit}_{\oc\mathsf{X}}.
  \end{equation*}
  Because $\oc_{\mathbb{N}} \mathsf{M} = \oc \mathsf{M}$,
  we have $\mathsf{M}^{\dagger} = \mathsf{M}^{\dagger,\mathbb{N}}$.
  \begin{lemma}
    For any $\alpha \subseteq \mathbb{N}$, we have
    \begin{equation*}
      \state{\mathsf{M}^{\dagger,\alpha}} =
      \state{\mathsf{M}^{\dagger}},
      \qquad
      \init{\mathsf{M}^{\dagger,\alpha}} =
      \init{\mathsf{M}^{\dagger}}.
    \end{equation*}
    Below, we write
    $h(\mathsf{M}) \colon \state{\mathsf{M}^{\dagger}} \to
    \state{\mathsf{M}} \times \state{\mathsf{M}}^{\mathbb{N}}$ for the
    isomorphism obtained by applying
    $1 \times (-) \cong (-)$ to $\state{\mathsf{M}}$.
  \end{lemma}
  
  \begin{lemma}\label{lem:phip}
    There is a family of measurable functions
    $\phi_{X} \colon (X^{\mathbb{N}})^{\mathbb{N}} \to
    (X^{\mathbb{N}})^{\mathbb{N}}$ such that the following diagram
    commutes:
    \begin{equation*}
      \xymatrix@C=15mm{
        (X^{\mathbb{N}})^{\mathbb{N}}
        \ar[d]_-{u_{X^{\mathbb{N}}}}
        \ar[rrdd]^{\phi_{X}} & & \\
        (X^{\mathbb{N}})^{\mathbb{N}}
        \times ((X^{\mathbb{N}})^{\mathbb{N}})^{\mathbb{N}}
        \ar[d]_-{(X^{\mathbb{N}})^{\mathbb{N}} \times \phi_{X^{\mathbb{N}}}} 
        & & \\
        (X^{\mathbb{N}})^{\mathbb{N}} \times
        ((X^{\mathbb{N}})^{\mathbb{N}})^{\mathbb{N}}
        \ar[r]_-{\cong}
        &
        (X^{\mathbb{N}} \times (X^{\mathbb{N}})^{\mathbb{N}})^{\mathbb{N}}
        \ar[r]_-{(u_{X}^{-1})^{\mathbb{N}}}
        &
        (X^{\mathbb{N}})^{\mathbb{N}} \nulldot
      }
    \end{equation*}
    where
    $u_{X} \colon X^{\mathbb{N}} \to X^{\mathbb{N}} \times
    (X^{\mathbb{N}})^{\mathbb{N}}$ is a measurable isomorphism given by
    \begin{equation*}
      u_{X}(x_{n})_{n \in \mathbb{N}}
      = ((x_{2n})_{n \in \mathbb{N}},
      ((x_{2\langle m_{0},m_{1} \rangle+1}
      )_{m_{0} \in \mathbb{N}})_{m_{1} \in \mathbb{N}}).
    \end{equation*}
  \end{lemma}
  \begin{proof}
    In this proof, for sets $N_{1},N_{2},\ldots,N_{k}$, we identify
    elements in $(((X^{N_{1}})^{N_{2}}) \cdots)^{N_{k}}$ with functions
    from $N_{1} \times N_{2} \times \cdots \times N_{k}$ to $X$. For
    $x \in (X^{\mathbb{N}})^{\mathbb{N}}$ and $(a,b) \in \mathbb{N}$, We
    define $(\phi_{X}(x))(a,b)$ by induction on $a$:
    \begin{equation*}
      (\phi_{X}(x))(a,b) =
      \begin{cases}
        x(a',2b)
        , & \textnormal{if } a = 2a'; \\
        (\phi_{X^{\mathbb{N}}}(x'))(a_{0},a_{1},b)
        , & \textnormal{if } a = 2 \langle a_{0},a_{1} \rangle + 1.
      \end{cases}
    \end{equation*}
    where $x' \in ((X^{\mathbb{N}})^{\mathbb{N}})^{\mathbb{N}}$ is given by
    \begin{equation*}
      x'(n,m,l) = x(n,2\langle m,l \rangle + 1).
    \end{equation*}
    We note that this definition makes sense
    because $a = 2 \langle a_{0},a_{1} \rangle + 1$ implies
    $a_{1} < a$. It is straightforward to check the family $\phi_{X}$
    makes the above diagram commute.
  \end{proof}

  \begin{lemma}\label{lem:expandp}
    For any Mealy machine $\mathsf{M} \colon \oc \mathsf{X}
    \rightarrowtriangle \mathsf{X}$, we have
    \begin{equation*}
      \mathsf{M}^{\dagger,\alpha_{n + 1}}
      \simeq \mathsf{M} \circ (\oc \mathsf{M})^{\dagger,\alpha_{n}},
    \end{equation*}
    which is realized by $u'_{\mathsf{M}}$ given by
    \begin{equation*}
      \state{\mathsf{M}^{\dagger,\alpha_{n+1}}}
      \cong \state{\mathsf{M}} \times \state{\mathsf{M}}^{\mathbb{N}}
      \xrightarrow{\state{\mathsf{M}} \times u_{\state{\mathsf{M}}}}
      \state{\mathsf{M}} \times \state{\mathsf{M}}^{\mathbb{N}}
      \times (\state{\mathsf{M}}^{\mathbb{N}})^{\mathbb{N}}
      \cong
      \state{\mathsf{M} \circ (\oc \mathsf{M})^{\dagger,\alpha_{n}}}
    \end{equation*}
    where the first and the last isomorphisms are
    obtained by applying
    canonical isomorphisms $1 \times (-) \cong (-)$ and
    $(-) \times 1 \cong (-)$.
  \end{lemma}
  \begin{proof}
    See Figure~\ref{fig:fix}.
  \end{proof}
  
  \begin{lemma}\label{lem:phi'p}
    For any Mealy machine
    $\mathsf{M} \colon \oc \mathsf{X} \rightarrowtriangle \mathsf{X}$ and for any
    $n \in \mathbb{N}$,
    \begin{equation*}
      (\oc \mathsf{M})^{\dagger,\alpha_{n}}
      \simeq
      \oc (\mathsf{M}^{\dagger,\alpha_{n}})
    \end{equation*}
    which is realized by $\phi'_{\mathsf{M}}$ given by
    \begin{equation*}
      \xymatrix@C=15mm{
        \makebox[0pt][r]{$\state{(\oc \mathsf{M})^{\dagger,\alpha_{n}}}
          \cong \;$}
        \state{\mathsf{M}}^{\mathbb{N}}
        \times (\state{\mathsf{M}}^{\mathbb{N}})^{\mathbb{N}}
        \ar[r]^-{\state{\mathsf{M}}^{\mathbb{N}} \times
          \phi_{\state{\mathsf{M}}}}
        &
        \state{\mathsf{M}}^{\mathbb{N}}
        \times (\state{\mathsf{M}}^{\mathbb{N}})^{\mathbb{N}}
        \ar[r]^-{\cong}
        &
        (\state{\mathsf{M}}
        \times
        \state{\mathsf{M}}^{\mathbb{N}})^{\mathbb{N}}
        \makebox[0pt][l]{$ \; \cong
          \state{\oc (\mathsf{M}^{\dagger,\alpha_{n}})}$}
      }
    \end{equation*}
    where the first and the last isomorphisms are
    obtained by applying
    canonical isomorphisms $1 \times (-) \cong (-)$ and
    $(-) \times 1 \cong (-)$.
  \end{lemma}
  \begin{proof}
    We prove the statement by induction on $n$. The base case follows
    from that the transition relations of
    $(\oc \mathsf{M})^{\dagger,\alpha_{0}}$ and
    $\oc (\mathsf{M}^{\dagger,\alpha_{0}})$ are the zero kernels.
    We next check the induction step. We have
    \begin{align*}
      (\oc \mathsf{M})^{\dagger,\alpha_{n+1}}
      &\simeq \oc \mathsf{M} \circ (\oc \oc \mathsf{M})^{\dagger,\alpha_{n}}
        \tag{Lemma~\ref{lem:expandp}}\\
      &\simeq \oc \mathsf{M} \circ \oc (\oc \mathsf{M})^{\dagger,\alpha_{n}}
        \tag{Induction hypothesis} \\
      &\simeq \oc (\mathsf{M} \circ (\oc \mathsf{M})^{\dagger,\alpha_{n}})
        \tag{Proposition~\ref{prop:oc_f}} \\
      &\simeq \oc (\mathsf{M}^{\dagger,\alpha_{n+1}}).
        \tag{Lemma~\ref{lem:expandp}}
    \end{align*}
    By Lemma~\ref{lem:phip}, this behavioral equivalence
    is realized by $\phi'_{\mathsf{M}}$.
  \end{proof}

  \begin{proposition}\label{prop:iterp}
    For a Mealy machine
    $\mathsf{M} \colon \oc \mathsf{X} \rightarrowtriangle \mathsf{X}$, we
    inductively define
    $\mathsf{iter}_{n}(\mathsf{M}) \colon \mathsf{I} \rightarrowtriangle
    \mathsf{X}$ by
    \begin{equation*}
      \mathsf{iter}_{0}(\mathsf{M}) = \mathsf{bot}_{\mathsf{I},\mathsf{X}},
      \qquad
      \mathsf{iter}_{n+1}(\mathsf{M})
      = \mathsf{M} \circ \oc (\mathsf{iter}_{n}(\mathsf{M})).
    \end{equation*}
    For all $n \in \mathbb{N}$, we have
    \begin{equation*}
      \state{\mathsf{M}^{\dagger}} =
      \state{\mathsf{M}^{\dagger,\alpha_{n}}}, \qquad
      \init{\mathsf{M}^{\dagger}} =
      \init{\mathsf{M}^{\dagger,\alpha_{n}}}, \qquad
      \mathsf{M}^{\dagger,\alpha_{n}}
      \simeq
      \mathsf{iter}_{n}(\mathsf{M}),
    \end{equation*}
    and
    \begin{equation*}
      \tran{\mathsf{M}^{\dagger,\alpha_{0}}} \leq 
      \tran{\mathsf{M}^{\dagger,\alpha_{1}}} \leq
      \tran{\mathsf{M}^{\dagger,\alpha_{2}}} \leq
      \cdots, \qquad
      \tran{\mathsf{M}^{\dagger}} =
      \bigvee_{n \geq 0} \tran{\mathsf{M}^{\dagger,\alpha_{n}}}.
    \end{equation*}
  \end{proposition}
  \begin{proof}
    It follows from the definition of $\oc_{\alpha_{n}}$, we have
    \begin{equation*}
      \state{\oc\mathsf{M}} = \state{\oc_{\alpha_{n}} \mathsf{M}}, \qquad
      \init{\oc\mathsf{M}} = \init{\oc_{\alpha_{n}} \mathsf{M}}
    \end{equation*}
    for all $n \in \mathbb{N}$, and
    \begin{equation*}
      \tran{\oc_{\alpha_{0}}\mathsf{M}}
      \leq \tran{\oc_{\alpha_{1}} \mathsf{M}}
      \leq \tran{\oc_{\alpha_{2}} \mathsf{M}} \leq \cdots,
      \qquad
      \tran{\oc\mathsf{M}} =
      \bigvee_{n \geq 0} \tran{\oc_{\alpha_{n}}\mathsf{M}}
    \end{equation*}
    Hence, by continuity of the composition, the coproduct and the
    monodal product of s-finite kernels, we have
    \begin{equation*}
      \state{\mathsf{M}^{\dagger}} =
      \state{\mathsf{M}^{\dagger,\alpha_{n}}}, \qquad
      \init{\mathsf{M}^{\dagger}} =
      \init{\mathsf{M}^{\dagger,\alpha_{n}}}
    \end{equation*}
    for all $n \in \mathbb{N}$, and
    \begin{equation*}
      \tran{\mathsf{M}^{\dagger,\alpha_{0}}} \leq 
      \tran{\mathsf{M}^{\dagger,\alpha_{1}}} \leq
      \tran{\mathsf{M}^{\dagger,\alpha_{2}}} \leq
      \cdots, \qquad
      \tran{\mathsf{M}^{\dagger}} =
      \bigvee_{n \geq 0} \tran{\mathsf{M}^{\dagger,\alpha_{n}}}.
    \end{equation*}
    By induction on $n \in \mathbb{N}$, we show that we have
    \begin{equation*}
      \mathsf{iter}_{n}(\mathsf{M})
      \simeq \mathsf{M}^{\dagger,\alpha_{n}},
    \end{equation*}
    which is realized by $f_{n}$. For the base case, we have
    $\mathsf{M}^{\dagger,\emptyset} \simeq
    \mathsf{iter}_{0}(\mathsf{M})$ because these Mealy machines
    $\mathsf{M}^{\dagger,\emptyset}$ and $\mathsf{iter}_{0}(\mathsf{M})$
    are behaviorally equiavalent to
    $\mathsf{emp}_{\mathsf{I},\mathsf{X}}$. The induction step
    follows from Lemma~\ref{lem:expandp}.
  \end{proof}

  \begin{corollary}\label{cor:fixp}
    For any Mealy machine
    $\mathsf{M} \colon \oc \mathsf{X} \rightarrowtriangle \mathsf{X}$,
    \begin{equation*}
      \mathsf{M} \circ \oc \mathsf{M}^{\dagger}
      \simeq \mathsf{M}^{\dagger}.
    \end{equation*}
  \end{corollary}
  \begin{proof}
    By Proposition~\ref{prop:iter}, we have
    \begin{equation*}
      \state{\mathsf{M} \circ \oc (\mathsf{M}^{\dagger})}
      = \state{\mathsf{M} \circ \oc (\mathsf{M}^{\dagger,\alpha_{n}})},
      \qquad
      \init{\mathsf{M} \circ \oc (\mathsf{M}^{\dagger})}
      = \init{\mathsf{M} \circ \oc (\mathsf{M}^{\dagger,\alpha_{n}})},
    \end{equation*}
    and
    \begin{equation*}
      \tran{\mathsf{M} \circ \oc (\mathsf{M}^{\dagger})}
      = \bigvee_{n \in \mathbb{N}}
      \tran{\mathsf{M} \circ \oc (\mathsf{M}^{\dagger,\alpha_{n}})}.
    \end{equation*}
    Therefore, by Lemma~\ref{lem:expandp} and Lemma~\ref{lem:phi'p}, the
    following diagram commutes:
    \begin{equation*}
      \xymatrix@C=30mm{
        X^{-} \times \state{\mathsf{M}}^{\mathbb{N}}
        \ar@{~>}[r]^{X^{+} \otimes
          (u_{\mathsf{M}}' \circ (\state{\mathsf{M}}
          \times \phi'_{\mathsf{M}}))^{\wedge}}
        \ar@{~>}[d]_-{\bigvee_{n \in \mathbb{N}}
          \tran{\mathsf{M} \circ \oc (\mathsf{M}^{\dagger,\alpha_{n + 1}})}}
        &
        X^{-} \times \state{\mathsf{M}}^{\mathbb{N}}
        \ar@{~>}[d]^-{\bigvee_{n \in \mathbb{N}}
          \tran{\mathsf{M}^{\dagger,\alpha_{n + 1}}}} \\
        X^{-} \times \state{\mathsf{M}}^{\mathbb{N}}
        \ar@{~>}[r]_-{X^{+} \otimes
          (u_{\mathsf{M}}' \circ (\state{\mathsf{M}}
          \times \phi'_{\mathsf{M}}))^{\wedge}}
        &
        X^{-} \times \state{\mathsf{M}}^{\mathbb{N}}
        \nulldot
      }
    \end{equation*}
    It is easy to see that
    $u_{\mathsf{M}}' \circ (\state{\mathsf{M}} \times
    \phi'_{\mathsf{M}})$ preserves the initial states.
    Hence, $\mathsf{M} \circ \oc (\mathsf{M}^{\dagger})
    \simeq \mathsf{M}^{\dagger}$.
  \end{proof}

  \subsubsection{Commutativity Modulo Observational Equivalence}

  \begin{definition}
    For terms $\vdash \mathtt{M},\mathtt{N} : \mathtt{A}$,
    we say that $\mathtt{M}$ is \emph{observationally equivalent} to
    $\mathtt{N}$ when for all context $\mathtt{C}[-]$,
    if $\mathtt{C}[\mathtt{M}] \Rightarrow_{\infty} \mu$, then
    $\mathtt{C}[\mathtt{N}] \Rightarrow_{\infty} \mu$.
  \end{definition}
  In this sectin, as an application of our GoI semantics,
  we show that for all
  $\vdash \mathtt{M} : \mathtt{A}$, $\vdash \mathtt{N} : \mathtt{B}$
  and
  $\mathtt{x}:\mathtt{A},\mathtt{y}:\mathtt{B} \vdash \mathtt{L} :
  \mathtt{C}$,
  \begin{equation*}
    \letin{\mathtt{x}}{\mathtt{M}}{
      \letin{\mathtt{y}}{\mathtt{N}}{\mathtt{L}}}
  \end{equation*}
  is observationally equivalent to
  \begin{equation*}
    \letin{\mathtt{y}}{\mathtt{N}}{
      \letin{\mathtt{x}}{\mathtt{M}}{\mathtt{L}}}.
  \end{equation*}
  To prove this equivalence, let $O_{\mathrm{d}}'$ be
  a binary relation between closed terms of type $\ttreal$
  and probabilistic Mealy machines from $\mathsf{I}$ to $\mathsf{J}
  \otimes \oc \mathsf{R}$ by
  \begin{equation*}
    (\mathtt{M},\mathsf{M}) \in O_{\mathrm{d}}'
    \iff (\mathtt{M},\mathsf{M}) \in O_{\mathrm{d}}
    \text{ and } (\mathrm{Condition 1}) \text{ and }
    (\mathrm{Condition 2})
  \end{equation*}
  where (Condition 1) is: for any
  $A \in \Sigma_{((1 + \mathbb{N} \times \mathbb{S}) +
    \emptyset) \times \state{\mathsf{M}}}$ such that
  \begin{equation*}
    A = \{((\hh,(\mm,(n,u))),s) \mid ((\hh,(\mm,(n,u))),s) \in A\},
  \end{equation*}
  and for any $s \in \state{\mathsf{M}}$, 
  \begin{math}
    \tran{\mathsf{M}}(((\mm,(\mm,\ast)),s),A) = 0;
  \end{math}
  (Condition 2) is: for any
  $A \in \Sigma_{((1 + \mathbb{N} \times \mathbb{S}) +
    \emptyset) \times \state{\mathsf{M}}}$ such that
  \begin{equation*}
    A = \{((\hh,(\hh,\ast)),s) \mid ((\hh,(\hh,\ast)),s) \in A\},
  \end{equation*}
  and for any $(n,u) \in \mathbb{N} \times \mathbb{S}$,
  for any $s \in \state{\mathsf{M}}$,
  \begin{math}
    \tran{\mathsf{M}}(((\mm,(\hh,(n,u))),s),A) = 0.
  \end{math}
  We then inductively define binary relations
  \begin{align*}
    T_{\mathtt{A}}
    &\subseteq \{\textnormal{closed values of type }\mathtt{A}\}
      \times \{\textnormal{Mealy machines from } \mathsf{I} \textnormal{ to } \sem{\mathtt{A}}\} \\
    T_{\mathtt{A}}^{\top}
    &\subseteq
      \{\textnormal{evaluation contexts } \mathtt{x}:\mathtt{A}
      \vdash \mathtt{E}[\mathtt{x}]:\ttreal \}
      \times \{\textnormal{Mealy machines from } \oc\sem{\mathtt{A}}
      \textnormal{ to } \mathsf{J} \otimes
      \oc \mathsf{R}\} \\
    \overline{T}_{\mathtt{A}}
    &\subseteq \{\textnormal{closed terms of type }\mathtt{A}\}
      \times \{\textnormal{Mealy machines from } \mathsf{I} \textnormal{ to }
      \mathsf{J} \otimes \oc \sem{\mathtt{A}}\}
  \end{align*}
  by replacing $O_{\mathrm{d}}$ in the definition of
  $S_{\mathtt{A}}$, $S_{\mathtt{A}}^{\top}$ and $\overline{S}_{\mathtt{A}}$
  with $O_{\mathtt{d}}'$. Then we can prove basic lemma for this logical relation.
  \begin{lemma}[Basic Lemma]\label{lem:basicp2}
    Let
    $\mathtt{\Delta} = (\mathtt{x}:\mathtt{A}_{1},
    \ldots,\mathtt{x}_{n}:\mathtt{A}_{n})$ be a context.
    \begin{itemize}
    \item For any term $\mathtt{\Delta} \vdash \mathtt{M}:\mathtt{A}$
      and for any
      $(\mathtt{V}_{i},\mathtt{N}_{i}) \in T_{\mathtt{A}_{i}}$ for
      $i = 1,2,\ldots,n$, we have
      \begin{equation*}
        \left(
          \mathtt{M}\{\mathtt{V}_{1}/\mathtt{x}_{1},\ldots,
          \mathtt{V}_{n}/\mathtt{x}_{n}\},
          \bsem{\mathtt{M}} \circ
          (\oc\mathsf{N}_{1} \otimes \cdots \otimes \oc \mathsf{N}_{n})
        \right)
        \in \overline{T}_{\mathtt{A}}.
      \end{equation*}
    \item  For any value
      $\mathtt{\Delta} \vdash \mathtt{V}:\mathtt{A}$ and for any
      $(\mathtt{V}_{i},\mathtt{N}_{i}) \in T_{\mathtt{A}_{i}}$ for
      $i = 1,2,\ldots,n$, we have
      \begin{equation*}
        \left(
          \mathtt{V}\{\mathtt{V}_{1}/\mathtt{x}_{1},\ldots,
          \mathtt{V}_{n}/\mathtt{x}_{n}\},
          \Bsem{\mathtt{M}} \circ
          (\oc\mathsf{N}_{1} \otimes \cdots \otimes \oc \mathsf{N}_{n})
        \right)
        \in T_{\mathtt{A}}.
      \end{equation*}
    \end{itemize}
  \end{lemma}
  \begin{proof}
    Almost equivalent to the proof of Lemma~\ref{lem:basicp}.
  \end{proof}
  \begin{corollary}\label{cor:sem}
    For any $\vdash \mathtt{M} : \mathtt{A}$,
    \begin{itemize}
    \item for any
      $A \in \Sigma_{((1 + \mathbb{N} \times \bsem{\mathtt{A}}^{+}) +
        \emptyset) \times \state{\mathsf{M}}}$ such that
      \begin{equation*}
        A = \{((\hh,(\mm,(n,a))),s) \mid ((\hh,(\mm,(n,a))),s) \in A\},
      \end{equation*}
      and for any $s \in \state{\mathsf{M}}$,
      \begin{math}
        \tran{\mathsf{M}}(((\mm,(\mm,\ast)),s),A) = 0;
      \end{math}
    \item for any
      $A \in \Sigma_{((1 + \mathbb{N} \times \bsem{\mathtt{A}}^{+}) +
        \emptyset) \times \state{\mathsf{M}}}$ such that
      \begin{equation*}
        A = \{((\hh,(\hh,\ast)),s) \mid ((\hh,(\hh,\ast)),s) \in A\},
      \end{equation*}
      and for any $(n,a) \in \mathbb{N} \times \bsem{\mathtt{A}}^{-}$,
      for any $s \in \state{\mathsf{M}}$,
      \begin{math}
        \tran{\mathsf{M}}(((\mm,(\hh,(n,a))),s),A) = 0.
      \end{math}
    \end{itemize}
  \end{corollary}
  By Corollary~\ref{cor:sem} and
  by the definition of composition of probabilistic Mealy machines,
  we see that if
  \begin{equation}\label{eq:symmon}
    \xymatrix{
      S_{1} \otimes S_{2}
      \ar[r]^-{k \otimes S_{2}}
      &
      S_{1} \otimes S_{2}
      \ar[r]^-{\cong}
      &
      S_{2} \otimes S_{1}
      \ar[r]^-{h \otimes S_{1}}
      &
      S_{2} \otimes S_{1}
      \ar[r]^-{\cong}
      &
      S_{1} \otimes S_{2}
    }
    = \xymatrix{
      S_{1} \otimes S_{2}
      \ar[r]^-{k \otimes h}
      &
      S_{1} \otimes S_{2}
    }
  \end{equation}
  then
  \begin{equation}\label{eq:MNNM}
    \bsem{\letin{\mathtt{x}}{\mathtt{M}}{
        \letin{\mathtt{y}}{\mathtt{N}}{\mathtt{L}}}}
    =
    \bsem{\letin{\mathtt{y}}{\mathtt{N}}{
      \letin{\mathtt{x}}{\mathtt{M}}{\mathtt{L}}}}
  \end{equation}
  where $k$ and $h$ are s-finite kernels given by
  restricting the domain and the codomain of
  $\tran{\bsem{\mathtt{M}}}$ and
  $\tran{\bsem{\mathtt{N}}}$ respectively.
  Because of commuativity for s-finite kernels \cite{staton2017},
  the equality \eqref{eq:symmon} is true.
  Hence, \eqref{eq:MNNM} holds. Then by adequacy,
  we see that
  \begin{equation*}
    \letin{\mathtt{x}}{\mathtt{M}}{
        \letin{\mathtt{y}}{\mathtt{N}}{\mathtt{L}}}
  \end{equation*}
  is observationally equivalent to
  \begin{equation*}
    \letin{\mathtt{y}}{\mathtt{N}}{
      \letin{\mathtt{x}}{\mathtt{M}}{\mathtt{L}}}.
  \end{equation*}

}

\section{Conclusion}\label{sec:conclusion}
We introduced a denotational semantics for $\PCFSS$, a
higher-order functional language with sampling from a uniform
continuous distribution and scoring.
Following~\cite{bdlgs2016}, we considered two operational semantics,
namely a \emph{distribution-based} operational semantics, which
associates terms with distributions over real numbers, and a
\emph{sampling-based} operational semantics, which associates each
term with a weight along every probabilistic branch. Our main results are
adequacy theorems for both kinds of operational semantics, and it
follows from these theorems that sampling-based operational semantics
is essentially equivalent to distribution-based operational semantics.
Another consequence of adequacy theorems is the possibility of
diagrammatic reasoning for observational equivalence of programs. It
follows from the observation in Section~\ref{sec:dr} and the adequacy theorems,
that diagrammatic equivalence for denotation of terms implies
observational equivalence. It would be interesting to explore
possible connections between our work and other works on diagrammatic
reasoning for probabilistic computation, such as \cite{cj2017,jzk2018}.

At this point, our language does not support normalisation mechanism
as a first class operator, and we are negative about extending our
semantics to capture normalisation mechanism. However, capturing
sampling algorithms such as the Metropolis-Hastings algorithm
\cite{mrrtt1953,hastings1970}, which consists of a number of
interactions between programs and their environment seems plausible.
Exploring the relationships between ``idealised'' normalisation
mechanisms and such ``approximating'' normalisation mechanisms from the
point of view of GoI is an interesting topic for future work.

\section*{Acknowledgment}

The authors are partially supported by the INRIA/JSPS project
``CRECOGI'', and would like to thank Michele Pagani for many fruitful
discussions about an earlier version of this work. Naohiko Hoshino is
supported by JST ERATO HASUO Metamathematics for Systems Design
Project (No. JPMJER1603).

\newpage



\bibliographystyle{IEEEtran}
\bibliography{./ref}
%



\end{document}

